\documentclass[english,11pt]{article}
 
\makeatletter

\usepackage[T1]{fontenc}
\usepackage[utf8]{inputenc}
\usepackage{amsthm}
\usepackage{amsmath}
\usepackage{amssymb}
\usepackage{graphicx}
\usepackage{dsfont}
\usepackage[noadjust]{cite}
\usepackage{amsmath}
\usepackage{array}
\usepackage{multirow}
\usepackage{caption}
\usepackage{color}

\usepackage{multicol}

\theoremstyle{plain}

\theoremstyle{plain}

\theoremstyle{plain}
\newtheorem{lem}{\protect\lemmaname}
\theoremstyle{plain}
\newtheorem{thm}{\protect\theoremname}
\theoremstyle{plain}
\newtheorem{cor}{\protect\corollaryname}  
\theoremstyle{definition}

\theoremstyle{definition}

\theoremstyle{definition}

\makeatother
  
\usepackage{babel} 

\providecommand{\claimname}{Claim}
\providecommand{\lemmaname}{Lemma}
\providecommand{\propositionname}{Proposition}
\providecommand{\theoremname}{Theorem}
\providecommand{\corollaryname}{Corollary} 
\providecommand{\definitionname}{Definition}
\providecommand{\assumptionname}{Assumption}
\providecommand{\remarkname}{Remark}

\DeclareMathOperator*{\argmin}{arg\,min}

\newcommand{\Psf}{\mathsf{P}}

\newcommand{\Shat}{\widehat{S}}

\newcommand{\sdif}{s_{\mathrm{dif}}}

\newcommand{\seq}{s_{\mathrm{eq}}}

\newcommand{\Xvdif}{\mathbf{X}_{s_{\mathrm{dif}}}}
\newcommand{\Xveq}{\mathbf{X}_{s_{\mathrm{eq}}}}

\newcommand{\xveq}{\mathbf{x}_{s_{\mathrm{eq}}}}

\newcommand{\Bernoulli}{\mathrm{Bernoulli}}

\newcommand{\pe}{P_{\mathrm{e}}}

\newcommand{\xv}{\mathbf{x}}

\newcommand{\Xv}{\mathbf{X}}
\newcommand{\yv}{\mathbf{y}}
\newcommand{\Yv}{\mathbf{Y}}

\newcommand{\Zv}{\mathbf{Z}}

\newcommand{\EE}{\mathbb{E}}
\newcommand{\PP}{\mathbb{P}}

\newcommand{\Tv}{\mathbf{T}}

\usepackage{xcolor}
\usepackage{hyperref} 
\usepackage{enumitem}
\usepackage{amsthm}
\usepackage{mathrsfs}
\usepackage{float}
\usepackage{setspace}
\usepackage{bbm}
\usepackage{bm}
\providecommand{\tabularnewline}{\\}
\floatstyle{ruled}
\newfloat{algorithm}{tbp}{loa}
\providecommand{\algorithmname}{Algorithm}
\floatname{algorithm}{\protect\algorithmname}

\setenumerate[1]{label=\arabic*.}

\usepackage{geometry} 
\newcommand{\stirlingii}{\genfrac{\{}{\}}{0pt}{}}

\usepackage{listings}

\numberwithin{equation}{section}
\numberwithin{thm}{section}
\numberwithin{lem}{section}
\numberwithin{cor}{section}

\makeatletter
\newcommand{\manuallabel}[2]{\def\@currentlabel{#2}\label{#1}}
\makeatother

\begin{document} 

\title{Exact Thresholds for Noisy Non-Adaptive Group Testing}
\author{Junren Chen \and Jonathan Scarlett}
% \author{Anonymous Authors}
\date{}

\maketitle
\singlespacing

\long\def\symbolfootnote[#1]#2{\begingroup\def\thefootnote{\fnsymbol{footnote}}\footnote[#1]{#2}\endgroup}

 \symbolfootnote[0]{J.~Chen is with the Department of Mathematics, The University of Hong Kong. The work was done when he was
 a visiting Ph.D.~student in the Department of Computer Science, National University of Singapore (NUS).  J.~Scarlett is with the Department of Computer Science, Department of Mathematics, and Institute of Data Science, National University of Singapore (NUS).  e-mail: \url{chenjr58@connect.hku.hk}; \url{scarlett@comp.nus.edu.sg}. This work was supported by the Singapore National Research Foundation (NRF) under grant number A-0008064-00-00.
 }

\begin{abstract}
    In recent years, the mathematical limits and algorithmic bounds for probabilistic group testing have become increasingly well-understood, with exact asymptotic thresholds now being known in general scaling regimes for the noiseless setting.  In the noisy setting where each test outcome is flipped with constant probability, there have been similar developments, but the overall understanding has lagged significantly behind the noiseless setting.  In this paper, we substantially narrow this gap by deriving exact asymptotic thresholds for the noisy setting under two widely-studied random test designs:~i.i.d.~Bernoulli and near-constant tests-per-item.  These thresholds are established by combining components of an existing information-theoretic threshold decoder with a novel analysis of maximum-likelihood decoding (upper bounds), and deriving a novel set of impossibility results by analyzing certain failure events for optimal maximum-likelihood decoding (lower bounds).   % Our results show that previous algorithmic upper bounds for the noisy setting are strictly suboptimal. %, and leave open the interesting question of whether our thresholds can be attained using computationally efficient algorithms.
\end{abstract}

% \newpage
{ \small 
\thispagestyle{empty}
\tableofcontents
\thispagestyle{empty}}

\setcounter{page}{0}
\newpage

\section{Introduction}

Group testing is a fundamental combinatorial and statistical problem with origins in medical testing and a variety of more recent applications including fault detection, DNA testing, data storage, communication protocols, and COVID-19 testing.  The problem involves a collection of $n$ items, with a smaller subset of $k$  items being defective; here we adopt widely-used terminology from fault detection applications, but other terminology may be substituted in other applications (e.g., in medical testing, replace ``items'' by ``individuals'' and ``defective'' by ``infected'').  There is a mechanism by which items can be pooled together and tested in groups, and in the noiseless scenario, a test returns positive if it contains at least one defective item, and returns negative otherwise.  The goal is to reliably identify the set of defective items using as few tests as possible.

Alongside its uses in applications, group testing has had a long history of mathematical developments characterizing the number of tests required for reliable recovery, with increased interest in recent years \cite{Du93,Ald19}.  
In particular, following a number of recent advances (which we describe in more detail in Section \ref{sec:related}), a complete characterization is now available for the asymptotic number of tests required to achieve an error probability approaching zero in the noiseless setting \cite{Ald14a,Sca15b,Ald15,Joh16,coja2020information,Coj19a,Bay22}.  Moreover, efficient algorithms are available whose performance matches these fundamental limits, with the Definite Defective (DD) algorithm \cite{Ald14,Joh16} or individual testing \cite{Ald18} sufficing in ``less sparse'' settings, and a spatial coupling approach \cite{Coj19a} closing the remaining gaps in ``sparser'' settings.

Similar advances have also been made in \emph{noisy} group testing, in which some positive tests get flipped to negative and/or vice versa.  However, the understanding of this setting has lagged significantly behind the noiseless setting.  In particular, under the widely-considered symmetric noise model, optimal thresholds on the number of tests are only known in a narrow range of (sufficiently sparse) scaling regimes \cite{Sca15b}, and the degree of (sub-)optimality of practical algorithms is unknown \cite{Sca17b,Sca20,Geb21}.

In this paper, we address this gap by obtaining \emph{exact information-theoretic thresholds} for the symmetric noise setting under two of the most commonly-considered test designs: i.i.d.~Bernoulli and near-constant tests-per-item.  Our results provide a significant improvement over the existing state-of-the-art, and demonstrate that all of the existing bounds for practical algorithms (to our knowledge) fall short of the information-theoretic thresholds.  We believe that this is a significant step forward in understanding noisy group testing, though we highlight two important challenges that we do not attempt to address: Proving a tight converse for \emph{arbitrary test designs}, and devising \emph{efficient algorithms} whose required number of tests matches the information-theoretic thresholds. % that we derive.

\subsection{Problem Setup}

Let $p$ denote the number of items, $k$ the number of defective items, and $n$ the number of tests. We consider the standard probabilistic group testing setup in which the defective set $S$ is uniform over the $p \choose k$ subsets of $[p]$ of cardinality $k$, where $[p] := \{1,\dotsc,p\}$.  Each test is represented by a length-$p$ binary vector $X = (X_1,\dotsc,X_p)$, with $X_j = 1$ if item $j$ is included in the test, and $X_j = 0$ otherwise.  We adopt the standard symmetric noise model (e.g., see \cite{Cha11,Sca15b}), in which the resulting test outcome is given by
\begin{equation}
    Y = \Big(\bigvee_{j \in S} X_j\Big) \oplus Z, \label{eq:gt_symm_model}
\end{equation}
with noise term $Z \sim \Bernoulli(\rho)$ for some noise level $\rho \in \big(0,\frac{1}{2}\big)$, where $\oplus$ denotes modulo-2 addition.\footnote{It would also be of interest to extend our analyses to asymmetric binary noise models (where $1 \to 0$ and $0 \to 1$ flips have different probabilities), but we stick to the symmetric noise due to its widespread consideration in previous works, and to avoid adding another layer of complexity in an already highly technical analysis.} 
With $n$ tests, there are $n$ such test vectors $X^{(1)},\dotsc,X^{(n)}$, and we represent these via an $n \times p$ matrix $\Xv$.  We only consider non-adaptive testing, in which the entire matrix $\Xv$ must be specified prior to observing any test results; such designs are often preferred due to allowing the implementation of the tests in parallel.  The test outcomes $Y^{(1)},\dotsc,Y^{(n)}$ are represented via a length-$n$ vector $\Yv$, and we will sometimes use $\Zv$ to denote the corresponding vector of noise terms.  We assume independent noise across tests, i.e., $Y^{(1)},\dotsc,Y^{(n)}$ are conditionally independent given $\Xv$.

Given the test matrix $\Xv$ and the corresponding outcomes $\Yv$, a \emph{decoder} forms an estimate $\Shat$, and the error probability is given by
\begin{equation}
    \pe := \PP[\Shat \ne S]. \label{eq:pe}
\end{equation}
Here the probability is taken with respect to the uniformly random defective set, the possibly-randomized tests, and the noise.  As was done in prior works such as \cite{Sca15b,Coj19a}, we consider the goal of attaining $\pe \to 0$ as $p \to \infty$ (with as few tests as possible). 
We will state our theorems in terms of bounds on the number of tests $n$, but in our discussions and figures, we will instead use the notion of \emph{rate} \cite{Ald19}:
\begin{equation}
    {\rm Rate} = \lim_{p \to \infty} \frac{\log_2{p \choose k}}{n} \quad \text{(bits/test)}, \label{eq:rate}
\end{equation}
where $k$ and $n$ implicitly depend on $p$. 
Thus, a higher rate amounts to fewer tests, and the rate is in $[0,1]$ since the tests are binary-valued.  

We focus on the sublinear sparsity regime, in which $k = \Theta(p^{\theta})$ for some sparsity parameter $\theta \in (0,1)$.  While the linear regime $k = \Theta(p)$ is also of interest, there are strong hardness results in this regime showing that even in the noiseless setting, (almost) every item must be tested individually to attain $\pe \to 0$ \cite{Ald18}, or similarly even to attain $\pe \not\to1$ \cite{Bay22}.  Thus, exact recovery is overly stringent for this regime, and one needs to move to other recovery criteria (e.g., approximate recovery \cite{tan2022performance}).

Finally, we describe the two randomized test designs that we will consider in this paper:
\begin{itemize}
    \item Under the \emph{Bernoulli design} (e.g, see \cite{Mal78,Cha11,Ald14a}), each item is independently placed in each test with probability $\frac{\nu}{k}$ for some design parameter $\nu > 0$ that dictates the density of 1s in $\Xv$.
    \item Under the \emph{near-constant weight design} (e.g., see \cite{Joh16,coja2020information}), each item is independently placed in $\Delta = \frac{\nu n}{k}$ tests chosen uniformly at random with replacement, for some $\nu > 0$.  (We leave integer rounding implicit, as it does not affect our results.)  Since the same test may be chosen multiple times, the weight of a given column in $\Xv$ may be (slightly) below $\Delta$.
\end{itemize}
Below we will outline the existing results regarding these designs, and then state and discuss our new results.

{\bf Note on notation.} In most cases, capital letters (e.g., $S$ and $\Xv$) are used for random variables (or random vectors/matrices), and lower-case letters (e.g., $s$ and $\xv$) are used for specific realizations.  We use bold/non-bold symbols to distinguish multi-test/single-test quantities; for instance, even though $X^{(i)}$ is a vector, is non-bold because it corresponds to a single test.  Asymptotic notation such as $O(\cdot)$ and $o(\cdot)$ is with respect to $p \to \infty$, with quantities such as $k$ and $n$ implicitly depending on $p$.  We will often use the binary entropy function $H_2(\rho) = \rho\log\frac{1}{\rho} + (1-\rho)\log\frac{1}{1-\rho}$, the binary KL divergence $D(a\|b) = a\log\frac{a}{b} + (1-a)\log\frac{1-a}{1-b}$, and the shorthand $a\star b = ab+(1-a)(1-b)$.  The function $\log(\cdot)$ has base $e$, and information quantities are measured in nats (except in our plots, where we convert to bits).

\subsection{Related Work} \label{sec:related}

Two important defining features of the group testing problem are probabilistic vs.~combinatorial (i.e., whether the recovery guarantee is high-probability or definite) and non-adaptive vs.~adaptive (i.e., whether or not the tests may be designed sequentially based on previous outcomes).  We will mainly summarize existing work for the probabilistic non-adaptive setting with sublinear sparsity (i.e., $k = \Theta(n^{\theta})$ for some $\theta \in (0,1)$), since this is the focus of our work.  % As mentioned in the introduction, characterizations of the linear regime $k = \Theta(p)$ and the ``mildly sublinear'' regime $k = p^{1-o(1)}$ can be found in \cite{Ald18,Bay22}.

{\bf Noiseless group testing.} Some of the main developments in the theory of noiseless probabilistic group testing in the sublinear regime are outlined as follows:
\begin{enumerate}[itemsep=0ex]
    \item Information-theoretic characterization of the very sparse regime $k=O(1)$ \cite{Mal80};
    \item Initial algorithmic bounds for Bernoulli designs via a simple algorithm that marks items in negative tests as non-defective an the rest as defective, sometimes referred to as the Combinatorial Orthogonal Matching Pursuit (COMP) algorithm \cite{Cha11};
    \item Improved algorithmic bounds for Bernoulli designs via an algorithm that declares a given item as defective if and only if it appears in any test in which all other items are known to be non-defective, referred to as the Definite Defectives (DD) algorithm \cite{Ald14};
    \item Precise information-theoretic bounds for Bernoulli designs \cite{Sca15b};
    \item Ensemble tightness\footnote{By ensemble tightness, we mean establishing that the bounds derived are the best possible for the specific random design considered, implying an exact threshold.  This leaves open the possibility that other designs may be better.} of these information-theoretic bounds for Bernoulli designs \cite{Ald15};
    \item Improved algorithmic bounds for COMP and DD via the near-constant weight design \cite{Joh16};
    \item Precise information-theoretic bounds for the near-constant weight design, including a proof of their ensemble tightness \cite{coja2020information};
    \item A proof that these bounds for the near-constant weight design are optimal among \emph{all} designs, and a computationally efficient method for attaining this optimal threshold using a spatially coupled randomized design \cite{Coj19a}.  
\end{enumerate}
These results for the noiseless setting are summarized in Figure \ref{fig:existing_noiseless} (Left).  We also briefly note that computational considerations for the standard (non spatially coupled) random designs were studied in \cite{Ili21,Coj22}, with evidence both for and against the prospect of a statistical-computational gap.

\begin{figure}
    \begin{centering}
        \includegraphics[width=0.425\columnwidth]{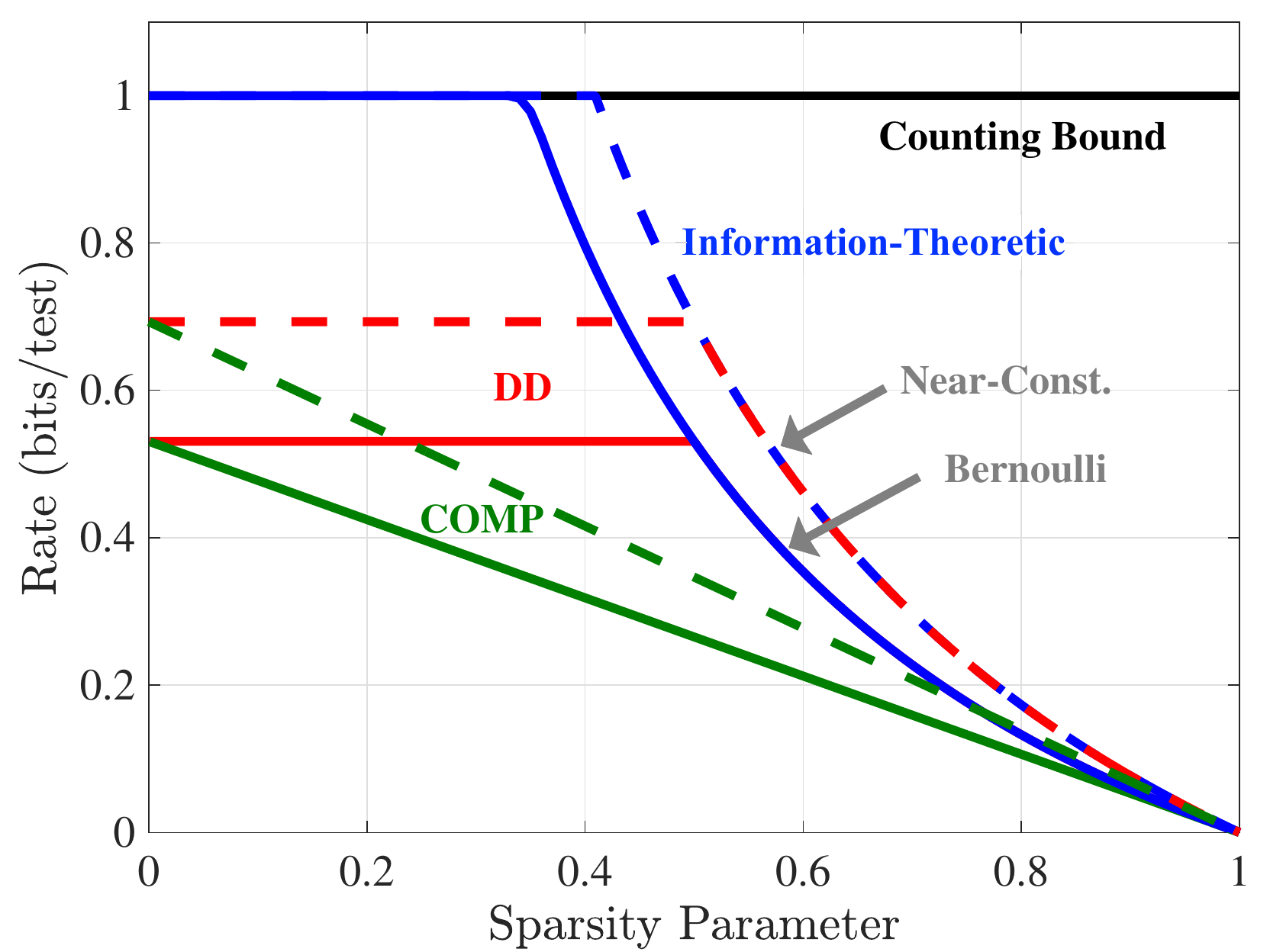} \quad \includegraphics[width=0.425\columnwidth]{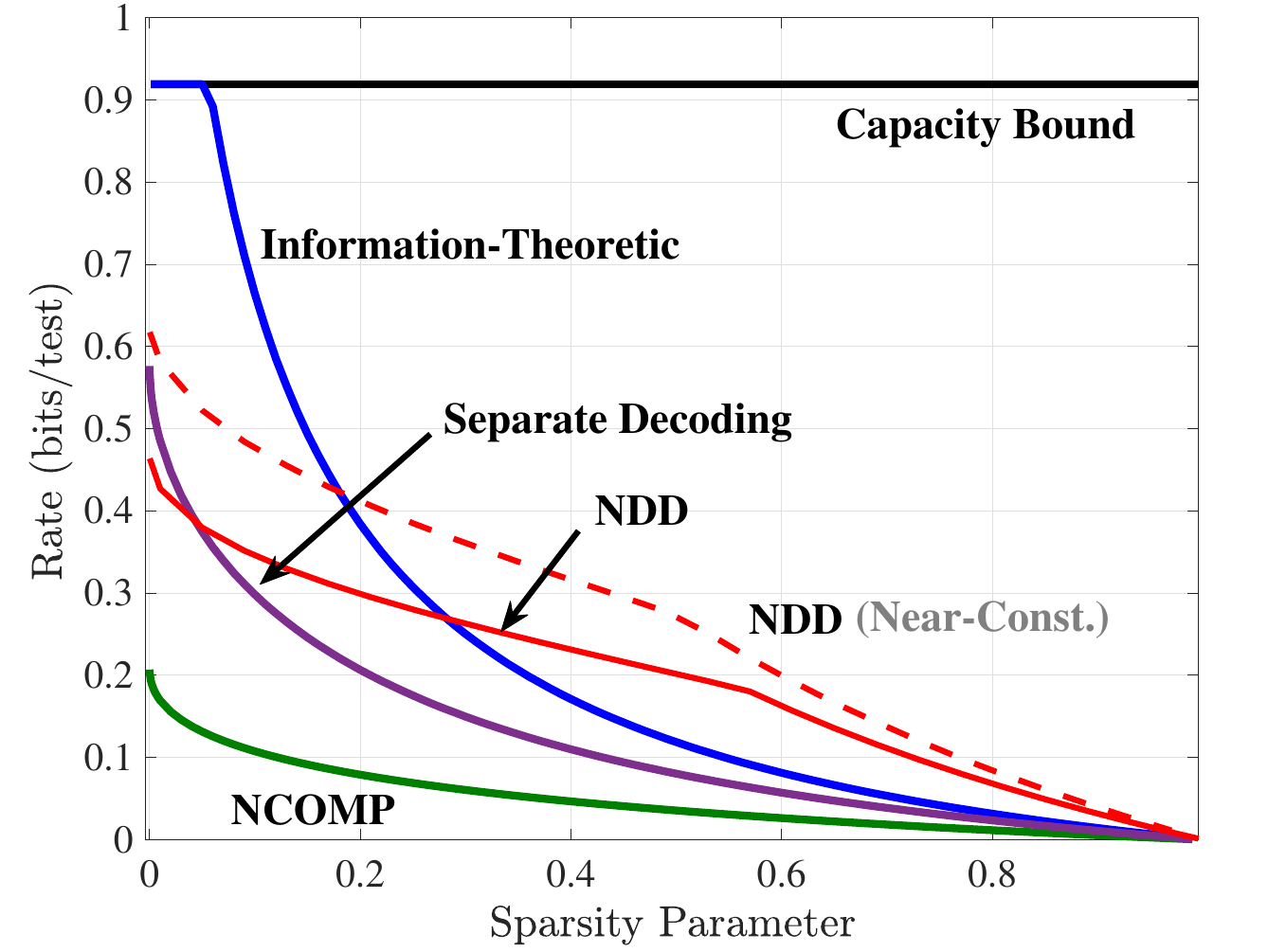}
        \par
    \end{centering}
    
    \caption{Existing bounds for the noiseless setting (Left) and the noisy setting with $\rho = 0.01$ (Right).  The sparsity parameter $\theta \in (0,1)$ is the value such that $k = \Theta(p^{\theta})$, and the rate is the asymptotic limit of $\frac{1}{n}\log_2{p \choose k}$.  In the noiseless setting, the dashed curve labeled ``Information-Theoretic'' is not only optimal for the near-constant weight design, but also for arbitrary test designs. The curves ``Counting Bound'' and ``Capacity Bound'' represent simple algorithm-independent lower bounds on $n$ for arbitrary test designs (e.g., see \cite{Mal78,Bal13}). \label{fig:existing_noiseless}}
\end{figure}

{\bf Noisy group testing.} In the noisy setting (with the symmetric noise model \eqref{eq:gt_symm_model}), several counterparts to the above developments are available, but the list is significantly less complete:
\begin{itemize}
    \item Regarding \#1, the very sparse regime $k = O(1)$ has long been well-understood \cite{Mal80}.
    \item Regarding \#2, algorithmic bounds for COMP-like algorithms were established in \cite{Cha11,Sca17b}.
    \item As an attempt towards \#4, information-theoretic upper bounds were established in \cite{Sca20}, but with tightness only shown for very small values of the sparsity parameter $\theta$.
    \item Regarding \#3, and \#6, bounds were derived for a noisy DD-like algorithm in \cite{Sca20} for the Bernoulli design, and in \cite{Geb21} for the near-constant weight design.  Exact thresholds were established in certain (dense) regimes for one-sided noise models, but not for the more commonly-considered symmetric noise model.
\end{itemize}
These existing results are summarized in Figure \ref{fig:existing_noiseless} (Right).  

We will outline our contributions below, but they can be succinctly summarized as follows:
\begin{gather*}
    \text{\emph{In this paper, we establish exact thresholds for the Bernoulli and near-constant weight designs}} \\ \text{\emph{in the noisy setting, thus establishing counterparts to items \#4, \#5, and \#7 in the above list.}}
\end{gather*}
% In particular, we complete item \#4 (precise-information-theoretic bounds; Bernoulli) which was only partially commenced in \cite{Sca15b}, and we complete items \#5 (ensemble tightness; Bernoulli) and \#7 (counterparts for the near-constant weight design), which have not been attempted before.
{\bf Other settings.} We only provide a brief summary of other aspects of group testing that are less relevant to the present paper:
\begin{itemize}
    \item For adaptive designs, near-optimal algorithms have long been known for the noiseless setting \cite{Hwa72}, and analogs were recently derived for the noisy setting \cite{Sca18,Teo22} and shown to significantly improve over the best non-adaptive results in denser (high $\theta$) regimes.  We note that even following the improved non-adaptive bounds that we derive in the present paper, significant improvements via adaptivity still remain.
    % We note that these improvements are still clearly maintained even under the new results of the present paper.
    \item Combinatorial non-adaptive group testing is a widely-studied problem \cite{Kau64,Dya82,Dya83,Du93}, and has also been studied under adversarial noise \cite{Che09}.  Under the condition of exact recovery, the stricter recovery criterion (with zero probability of error) increases the dependence on $k$ from linear to quadratic, meaning that it comes at a significant price in the number of tests.  To our knowledge, precise constant factors have not been sought in this line of works.
    \item If the exact recovery criterion \eqref{eq:pe} is relaxed to approximate recovery, where $\alpha k$ false positives and false negatives are allowed in the reconstruction for some $\alpha > 0$, then the study of noisy (and noiseless) probabilistic group testing becomes much simpler, with \cite{Sca15b} giving upper and lower bounds that match in the limit of small $\alpha$.  (With reference to the bound \eqref{eq:ber_threshold} below, the ``difficult'' second term disappears, and only the ``simple'' first term remains.)  Analogous findings for the near-constant weight design can be inferred from \cite[Sec.~5]{Coj19a}, and related studies of ``all-or-nothing'' phase transitions in group testing can be found in \cite{Tru20,Nil21}.
\end{itemize}

\subsection{Overview of Our Results}

In this paper, we obtain exact information-theoretic thresholds for both the Bernoulli design and the near-constant weight design.  To do so, we prove four results separately: Bernoulli achievability, Bernoulli converse, near-constant weight achievability, and near-constant-weight converse, where we use the term ``converse'' in a design-specific sense (i.e., ensemble tightness).  We emphasize that all of our achievability results are based on a computationally intractable decoder that explicitly searches over all $p \choose k$ possible defective sets; thus, we only prove achievability in an information-theoretic sense, without requiring computational efficiency. 

% Thus, our work leaves open the very interesting open problem of matching our information-theoretic thresholds using an efficient algorithm, even if only in restricting regimes of $\theta \in (0,1)$.  

Our results are summarized in the four theorems stated below, whose proofs are outlined in Section \ref{sec:overview_ach} (achievability) and Section \ref{sec:overview_conv} (converse), with the formal details deferred to the appendices.  We briefly recall the notation $a \star b = ab + (1-a)(1-b)$,\footnote{This is interpreted as the probability of getting a ``1'' when a ${\rm Bern}(1-a)$ variable is generated and then flipped with probability $b$.} the binary KL divergence $D(a\|b) = a\log\frac{a}{b} + (1-a)\log\frac{1-a}{1-b}$, and  the binary entropy function $H_2(a) = a\log\frac{1}{a} + (1-a)\log\frac{1}{1-a}$.  The function $\log(\cdot)$ and the information measures are in units of nats, except where stated otherwise.

We note that both of the thresholds that we derive are somewhat complicated to state, but we will provide some intuition in Sections \ref{sec:overview_conv} and \ref{sec:overview_ach} regarding how the various terms arise.  
In all of our results, we make use of the design parameter $\nu > 0$, noise level $\rho \in \big(0,\frac{1}{2}\big)$, and sparsity parameter $\theta \in (0,1)$, which is the constant such that $k = \Theta(p^{\theta})$.  (Since $\theta \in (0,1)$, we have $k \to \infty$ and $k = o(p)$.)  We also give some brief hints on the meaning of other notation that will appear:
\begin{itemize}
    \item The thresholds both contain minimizations over two parameters denoted by $C$ and $\zeta$; we will see that $C$ is related to the number of tests a defective item is placed in without any other defectives, and $\zeta$ is related to how many of those tests are flipped.
    \item The thresholds both contain quantities denoted by $f_1$ and $f_2$, which are related to the occurrence of ``bad events'' associated with defective items and non-defective items, respectively.  Moreover, $f_2$ contains an additional optimization parameter $d$, which is related to how many ``relevant'' positive tests the ``bad'' non-defective item is placed in.
\end{itemize}
We again refer the reader to Sections \ref{sec:overview_conv} and \ref{sec:overview_ach} for more detailed intuition.

\textbf{Threshold for Bernoulli designs:} The threshold for Bernoulli design with i.i.d. ${\rm Bernoulli}(\frac{\nu}{k})$ entries is given as follows: 
\begin{align}
    n^*_{\rm Bern} = \max\left\{\frac{k\log\frac{p}{k}}{H_2(e^{-\nu}\star\rho)-H_2(\rho)},\frac{k\log\frac{p}{k}}{(1-\theta)\nu e^{-\nu}\min_{ C>0,\zeta\in(0,1) }\max\{\frac{1}{\theta}f_1^{\rm Bern}(C,\zeta,\rho),f_2^{\rm Bern}(C,\zeta,\rho)\}}\right\},\label{eq:ber_threshold}
\end{align}
where we define
\begin{gather}
    f_1^{\rm Bern}(C,\zeta,\rho)=C\log C-C+C\cdot D(\zeta\|\rho)+1, \\
    f_2^{\rm Bern}(C,\zeta,\rho)= \min_{d \ge \max\{0,C(1-2\zeta)/\rho\}} g^{\rm Bern}(C,\zeta,d,\rho), \label{eq:f2_bern}
\end{gather}
and where $g^{\rm Bern}(C,\zeta,\rho,d)$ is defined as
\begin{gather}\label{eq:g_ber}
    g^{\rm Bern}(C,\zeta,d,\rho) = \rho d\log d+\big(\rho d-C(1-2\zeta)\big)\log\Big(\frac{\rho d-C(1-2\zeta)}{1-\rho}\Big) + 1-2\rho d+C(1-2\zeta).
\end{gather}
We will also show in the proof that the minimizing choice of $d$ in \eqref{eq:f2_bern} has the following closed-form expression:
\begin{align}
    d^*_{\rm Bern}(C,\zeta,\rho) = \frac{C(1-2\zeta)+\sqrt{C^2(1-2\zeta)^2 + 4\rho(1-\rho)}}{2\rho}.
\end{align}
The achievability and converse statements are formally given as follows.

\begin{thm} \label{thm:ach_bernoulli}
    {\em (Achievability via Bernoulli Designs)} Under the Bernoulli design with i.i.d.~entries following ${\rm Bernoulli}(\frac{\nu}{k})$ for fixed $\nu > 0$, there exists a decoding strategy such that $P_e\to 0$ as $p\to\infty$ with a number of tests satisfying
    \begin{equation}\label{eq:ach_bern}
        n \leq (1+\eta)n^*_{\rm Bern}
    \end{equation}
    for an arbitrary constant $\eta>0$.
\end{thm}

\begin{thm}\label{thm2}
    {\em (Converse for Bernoulli Designs)} Under the Bernoulli design with i.i.d.~entries following ${\rm Bernoulli}(\frac{\nu}{k})$ for fixed $\nu > 0$, for any decoding strategy to have $P_e \not\to 1$ as $p\to\infty$ it must be the case that 
    \begin{equation}\label{eq:conv_bern}
        n \geq (1-\eta)n^*_{\rm Bern}
    \end{equation}
    for an arbitrary constant $\eta>0$.
\end{thm}

\textbf{Threshold for Near-Constant Weight Designs:} The threshold for the near-constant weight design with $\Delta=\frac{\nu n}{k}$ placements per item is given as follows:
\begin{align}\label{eq:NC_threshold}
    n^*_{\rm NC}= \max \left\{\frac{k\log\frac{p}{k}}{H_2(e^{-\nu}\star\rho)-H_2(\rho)},\frac{k\log\frac{p}{k}}{(1-\theta)\nu e^{-\nu}\min_{C\in(0,e^\nu),\zeta\in(0,1)}\max\{\frac{1}{\theta}f^{\rm NC}_1(C,\zeta,\rho,\nu),f^{\rm NC}_2(C,\zeta,\rho,\nu)\}}\right\},
\end{align}
where 
\begin{gather}
    f_1^{\rm NC}(C,\zeta,\rho,\nu)=e^\nu D(Ce^{-\nu}\|e^{-\nu})+C\cdot D(\zeta\|\rho) \\
    f_2^{\rm NC}(C,\zeta,\rho,\nu) = \min_{d \,:\, |C(1-2\zeta)|\le d\le e^{\nu}} g^{\rm NC}(C,\zeta,d,\rho,\nu), \label{eq:f2_NC}
\end{gather}
and where $g^{\rm NC}(C,\zeta,d,\rho,\nu)$ is given by
\begin{gather}
    \label{eq:gnc}
    g^{\rm NC}(C,\zeta,d,\rho,\nu) = e^{\nu}\cdot D\big(de^{-\nu}\|e^{-\nu}\big)+d\cdot D\Big(\frac{1}{2}+\frac{C(1-2\zeta)}{2d}\big\|\rho\Big). % \\
      %d^*_{\rm NC} = \argmin_d~g(C,\zeta,d,\rho,\nu)~~\text{subject to }.
\end{gather}
We will also derive a closed-form expression for the value $d^*_{\rm NC}(C,\zeta,\nu)$ achieving the minimum in \eqref{eq:f2_NC}, but it is omitted here due to being rather complicated (see (\ref{eq:d_star_pinpoint}) in Appendix \ref{app:nc_converse}).

\begin{thm}\label{thm3}
    {\em (Achievability via Near-Constant Weight Designs)} Under the near-constant weight design where each item is placed   $\Delta = \frac{\nu n}{k}$ tests uniformly at random with replacement with $\nu = \Theta(1)$, there exists a decoding strategy such that $P_e\to 0$ as $p\to\infty$ with a number of tests satisfying
    \begin{equation}\label{eq:ach_ncc}
        n \leq (1+\eta) n^*_{\rm NC}
    \end{equation}
    for an arbitrary constant $\eta>0$.
\end{thm}

\begin{thm}\label{thm:nc_convese}
    {\em (Converse for Near-Constant Weight Designs)} Under the near-constant weight design where each item is placed $\Delta = \frac{\nu n}{k}$ tests uniformly at random with replacement with $\nu = \Theta(1)$, for any decoding strategy to have $P_e \not\to 1$ as $p\to\infty$ it must be the case that 
    \begin{equation}\label{eq:conv_ncc}
        n\geq (1-\eta) n^*_{\rm NC}
    \end{equation}
    for an arbitrary constant $\eta>0$.
\end{thm}

We note that the subscripts/superscripts ``Bern'' and ``NC'' (e.g., on $f_1$ and $f_2$) will be omitted throughout the proofs in the appendices, since the choice of design will be clear from the context.

\begin{figure}
    \begin{centering}
        \includegraphics[width=0.425\columnwidth]{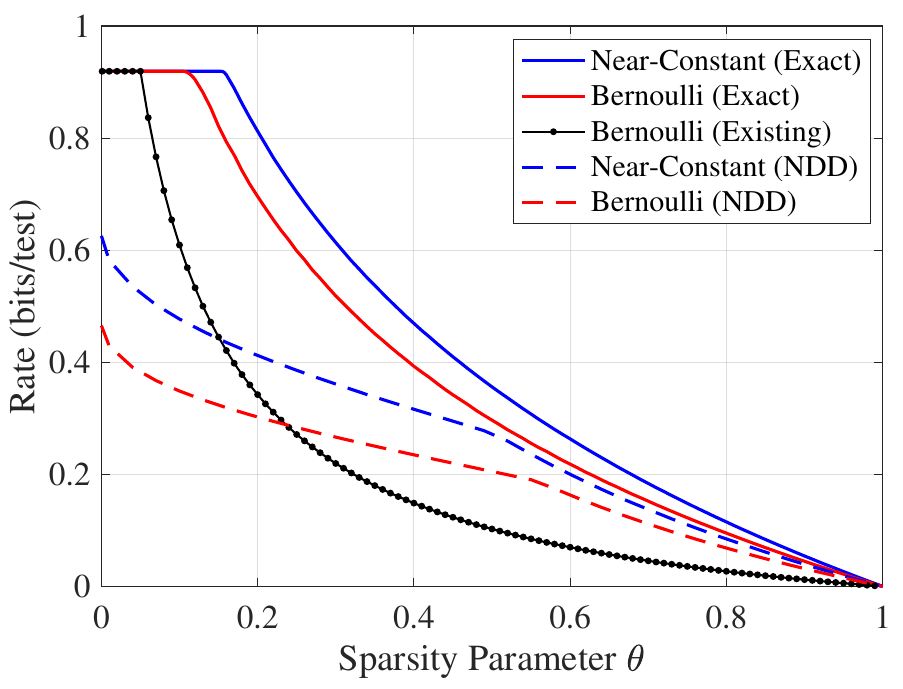} \quad \includegraphics[width=0.425\columnwidth]{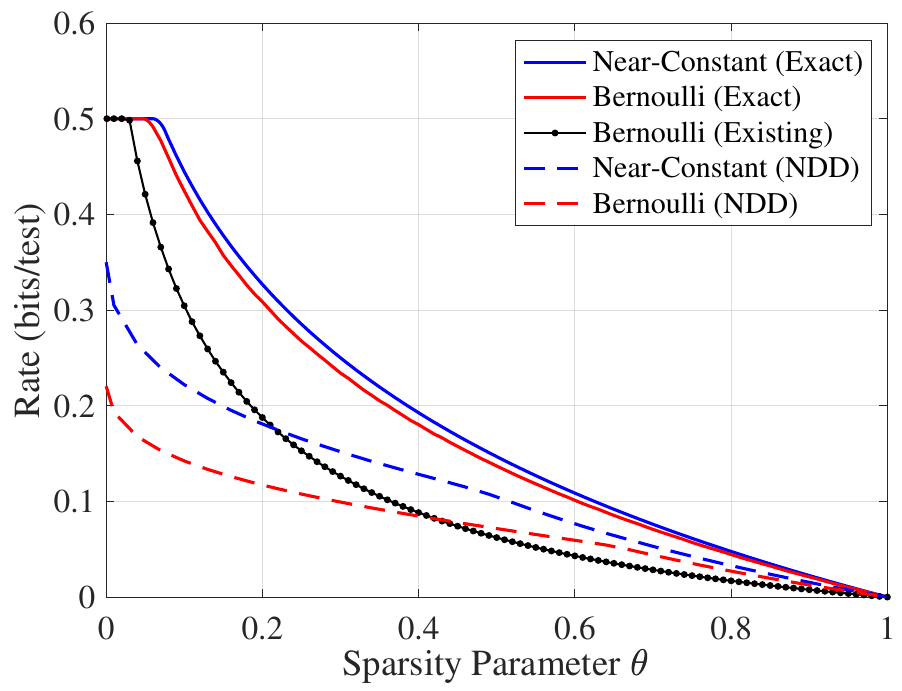}
        \par
    \end{centering}
    
    \caption{Our exact thresholds and the state-of-the-art existing results for noise levels $\rho = 0.01$ (Left) and $\rho = 0.11$ (Right).  The existing Bernoulli achievability result is from \cite{Sca15b}, and the existing noisy definite defectives (NDD) curves are from \cite{Sca20,Geb21}.  All curves are optimized over $\nu$, except that of \cite{Sca15b} which was only proved for $\nu=\ln 2$. \label{fig:comparison}}
\end{figure}

{\bf Discussion.} We illustrate our results numerically using the notion of rate (measured in bits/test), as defined in \eqref{eq:rate}. 
% , we convert the number of tests $n$ to a ``rate'' $\lim_{p \to \infty} \frac{1}{n}\log{{p \choose k}}$, which can roughly be interpreted as the amount of information learned per test.  We also convert from nats/test to bits/test (amounting to multiplication by $\frac{1}{\ln 2}$) for easier comparison with existing results.
The bounds for both designs are illustrated in Figure \ref{fig:comparison}, where we observe a substantial improvement over the best known existing bounds.  
%Our exact thresholds also strictly improve over those of NDD for all $\theta \in (0,1)$, which is in stark contrast to the noiseless setting and one-sided noise models \cite{Sca20}, where DD-type algorithms are asymptotically optimal when $\theta$ is large enough.  We note, however, that these curves for NDD are only achievability results, and to our knowledge it is unknown whether the analysis of NDD could be further tightened.
The horizontal part at small $\theta$ in Figure \ref{fig:comparison} corresponds to the capacity of the binary symmetric channel (namely, $\log 2 - H_2(\rho)$), and the corresponding impossibility result holds for arbitrary test designs and even adaptive algorithms (e.g., see \cite{Mal78,Bal13}).  Thus, in this regime, we have particularly strong guarantees of asymptotic optimality.  
By comparison, for higher $\theta$, we have not established optimality with respect to arbitrary designs, but we have at least established ensemble tightness (and hence exact thresholds) for the two random designs under consideration.

Among the practical algorithms, the most promising one for being asymptotically optimal in some regimes (namely, $\theta$ close to one) is noisy DD, since DD is known to exhibit such optimality in the noiseless setting \cite{Ald14a} and under the one-sided ``Z'' and ``reverse Z'' noise models \cite{Sca20,Geb21}.  However, our results indicate that under symmetric noise, the existing noisy DD bounds fall short of the information-theoretic threshold for all $\theta \in (0,1)$ (at least for the $\rho$ values we considered).  % Hence, more sophisticated techniques appear to be needed to attain our thresholds in a computationally efficient manner, with the spatial coupling approach of \cite{Coj19a} perhaps being a promising candidate.

Another interesting implication of our results concerns the optimal choice of $\nu$ for the two designs.  In the noiseless setting, the following is well known \cite{Ald15,Joh16}:
\begin{itemize}
    \item For the Bernoulli design, $\nu = \log 2$ is optimal for $\theta \le \frac{1}{3}$, and $\nu = 1$ is optimal for almost all higher $\theta$ except for a narrow ``transition region''.
    \item For the near-constant weight design, $\nu = \log 2$ is optimal for all $\theta$.
\end{itemize}
In the noisy case, under both designs, the choice $\nu = \log 2$ is still optimal for sufficiently small $\theta$ such that the capacity bound $n = \frac{k\log\frac{p}{k}}{\log 2 - H_2(\rho)} (1+o(1))$ is achieved.  This has the natural interpretation that in this regime, $\nu$ should be chosen to make the tests maximally  in informative sense of attaining $H(Y) = \log 2$ (i.e., each test is equally likely to be positive or negative).

On the other hand, when $\theta$ is large enough, the division by $1-\theta$ makes the second term in each bound become dominant, and optimizing $\nu$ eventually amounts to optimizing that term.  (There is also a ``transition region'' where the optimal $\nu$ is chosen to equate the two terms.)  In the Bernoulli design, the only $\nu$ dependence is via $\nu e^{-\nu}$, which is the same term that enters in the noiseless setting, so we again have $\nu = 1$ being optimal for large enough $\theta$ (see Figure \ref{fig:nu_plot} (Right) for an illustration).  However, for the near-constant weight design, the choice of $\nu$ becomes more complicated even at high $\theta$ values.  In particular, we found that in contrast to the noiseless setting, the choice $\nu = \log 2$ can be strictly suboptimal (again see Figure \ref{fig:nu_plot}).  Having said this, the degree of suboptimality is small, and we found that $\nu = \log 2$ is still a generally good choice.

For the near-constant weight design, the unusual shape of the optimal $\nu$ curve in Figure \ref{fig:nu_plot} (Right) is explained as follows:
\begin{itemize}
    \item For small enough $\theta$, the first branch of the maximum in \eqref{eq:NC_threshold} is strictly dominant, and it is maximized by $\nu = \ln 2$, which corresponds to having half positive and half negative test outcomes.
    \item For large enough $\theta$, the second branch of the maximum in \eqref{eq:NC_threshold} is strictly dominant, and the corresponding maximizer turns out to be a slowly decreasing function;
    \item In between these two regimes, there is a narrow ``transition region'' in which $\nu$ is chosen to equate the two branches of the maximum.
\end{itemize}
A similar kind of transition region is observed for the Bernoulli design (analogous to the noiseless setting \cite{Ald15}), but in that case the maximizer remains at $\nu = 1$ after the transition region.

\begin{figure}
    \begin{centering}
        \includegraphics[width=0.445\columnwidth]{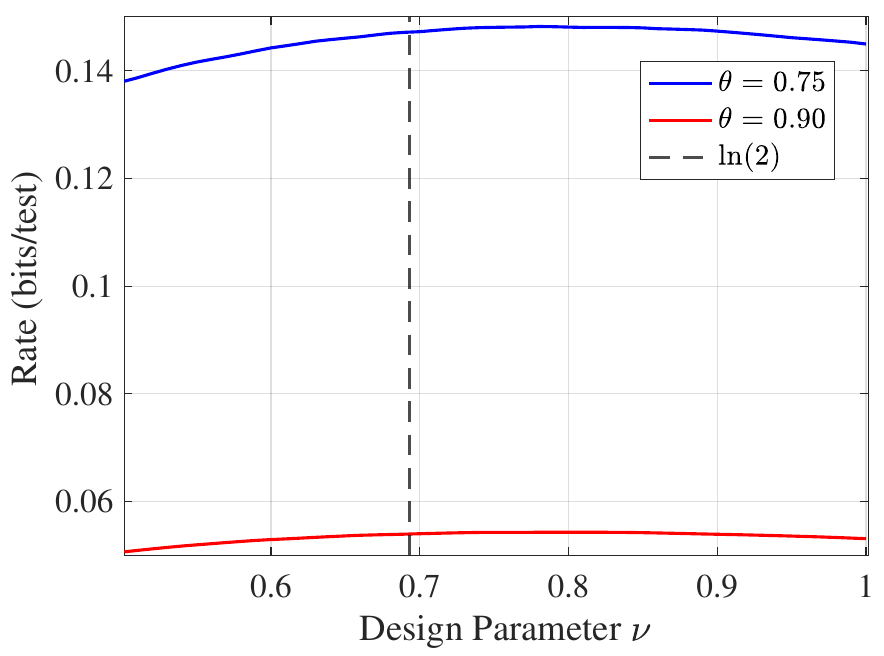} \quad \includegraphics[width=0.45\columnwidth]{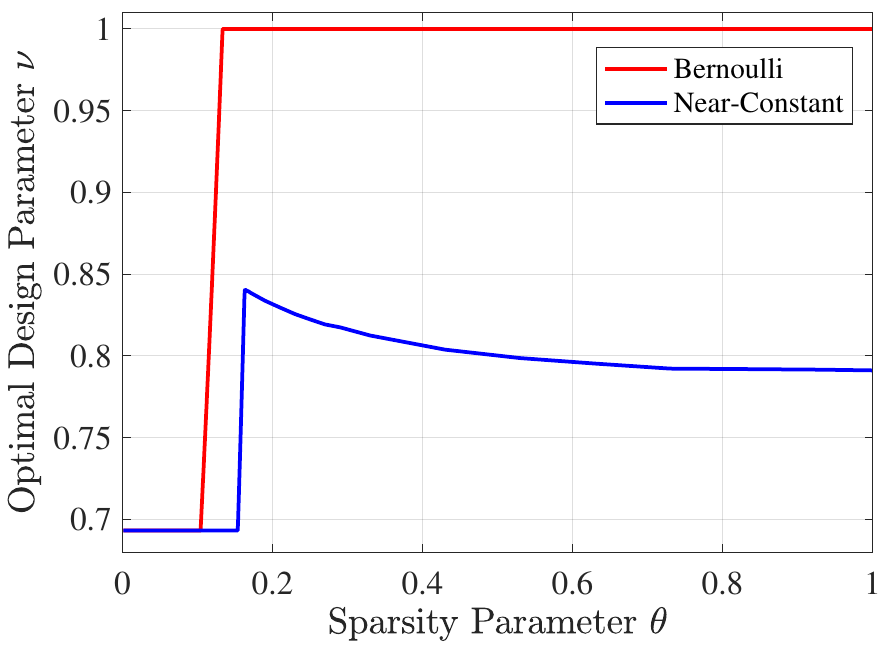}
        \par
    \end{centering}
    
    \caption{ Illustrations regarding the optimal design parameter $\nu$ (varying with the sparsity parameter $\theta$) with noise level $\rho = 0.01$. (Left) Example plots of the rate as a function of $\nu$ under the near-constant weight design.  (Right) Optimal $\nu$ as a function of the sparsity parameter $\theta$ under both designs. We note that the sudden increases after the horizontal parts are not instantaneous; both curves are continuous with respect to $\theta$. \label{fig:nu_plot}}
\end{figure}

While the proofs of our main results all have highly technical components and are quite lengthy, the general approaches taken are conceptually simple.  We provide conceptual outlines in Sections \ref{sec:overview_conv} and \ref{sec:overview_ach}, and provide the technical details in the appendices, starting with a roadmap in Appendix \ref{app:roadmap}.  The outlines will mostly be the same for the two designs, despite the differing technical details.

{\bf Note on concurrent work.} The concurrent work of Coja-Oghlan \emph{et al.}~\cite{Coj24} established the same achievability threshold as that of Theorem \ref{thm3} (stated in a different form) using a computationally efficient spatially coupled design, and also established Theorem \ref{thm:nc_convese}.  Thus, their work makes the major contribution of showing that $n^*_{\rm NC}$ is also attainable using a polynomial-time algorithm when spatial coupling is incorporated.  Despite attaining the same threshold, the two sets of achievability analyses are very different, and we believe that there is significant interest in understanding the fundamental limits of each distinct test design; in this sense, their achievability result does not imply Theorem \ref{thm3} nor vice versa, and the results for Bernoulli designs are unique to the present paper.  We also note that asymmetric binary noise models are additionally handled in \cite{Coj24}, whereas we focus on symmetric noise.

\section{Outline of Converse Proofs} \label{sec:overview_conv}

We first outline the proofs of our converse results (i.e., lower bounds on the number of tests).

{\bf Discussion of the first term.} The thresholds in \eqref{eq:ber_threshold} and \eqref{eq:NC_threshold} consist of two terms.  The converse for the first term \emph{maximized over $\nu$} is known even for \emph{arbitrary test designs} via a simple channel capacity argument \cite{Mal78,Sca16b}, and for arbitrary $\nu$, the converse for Bernoulli designs is also known via \cite{Sca15b}.  We required non-minor additional technical effort to handle the near-constant weight design with arbitrary $\nu$,\footnote{A simpler approach based on Fano's inequality (e.g., see \cite{Mal78}) could also be used, but would only give the weaker statement $\PP({\rm err}) \not\to 0$ rather than $\PP({\rm err}) \to 1$. } but the intuition behind the bound is simple: With parameter $\nu$, each test has probability roughly $e^{-\nu}$ of containing no defectives, and thus a probability roughly $e^{-\nu} \star \rho$ of being positive.  A standard information-theoretic argument then shows that each test can only reveal $H_2(e^{-\nu} \star \rho) - H_2(\rho)$ bits of information.  On the other hand, with ${p \choose k}$ possible defective sets, we require roughly $\log{p \choose k} = \big(k \log_2\frac{p}{k}\big)(1+o(1))$ bits to identify the correct one.  Comparing these quantities leads to the first term in \eqref{eq:NC_threshold}.  This part of the analysis will be detailed in Appendix \ref{app:high_conv}.

% but it is relatively simple to derive via similar but \emph{much simpler} steps to the achievability analysis (with $\seq = \emptyset$).  The much more significant challenge lies in deriving the second terms {\bf \color{blue} [TODO: May become `second and third terms']} in \eqref{eq:conv_bern} and \eqref{eq:conv_ncc}, which are tighter in denser regimes (i.e., higher $\theta$).  

{\bf Discussion of the second term.} The second terms in \eqref{eq:conv_bern} and \eqref{eq:conv_ncc} are dominant at high values of $\theta$.  Terms of this kind (albeit much simpler ones) are widely known in the noiseless setting \cite{Ald14a,Ald15,Joh16,Coj19a}, and are based on the idea of that if a defective item is \emph{masked} (i.e., every test it is in also contains at least one other defective), then even an optimal decoder will be unable to identify it.  Intuitively, this is because the tests results are unchanged when that item is changed to be non-defective.

In the noisy setting, \emph{complete} masking is no longer the dominant error event, but we can use a similar idea.  Informally, an optimal decoder will fail when the following events hold simultaneously (see Corollary \ref{cor:restate} in Appendix \ref{app:mle} for the formal version):
\begin{itemize}
    \item[(i)] There exists a defective item $j$ that has \emph{relatively few tests} in which it is the only defective, and not too few of those tests are flipped by the noise.
    \item[(ii)] There exists a non-defective item $j'$ appearing in \emph{sufficiently many tests that don't contain any defective items}, and not too few of those tests are flipped by the noise.
\end{itemize}
The idea is that when both of these occur with suitable notions of ``few'' and ``many'', the set $(S \setminus \{j\})\cup \{j'\}$ will be favored to $S$ itself.  We proceed with further details towards making this intuition rigorous, and we will see that the formal version of conditions (i) and (ii) above lead to the terms $f_1$ and $f_2$ in the final results.  
This part of the analysis will be detailed in Appendix \ref{app:ber_con} (Bernoulli design) and Appendix \ref{app:nc_converse} (near-constant weight design), and in the following, we provide a more detailed outline of some of the main ideas and analysis techniques that we use.

{\bf Maximum-likelihood decoding.} Under a uniform prior on the defective set $S$, it is well known that the optimal decoding rule is \emph{maximum-likelihood decoding}.  Under our symmetric noise model, this is equivalent to choosing an estimate $\hat{S}$ such that as many tests as possible match what the noiseless outcome would be.  Thus, referring back to the above outline, we seek to find $j \in S$ and $j' \notin S$ such that $(S \setminus \{j\}) \cup \{j'\}$ has more such tests than $S$ itself. 

{\bf Identifying a suitable defective item.} Suppose without loss of generality that $S = \{1,\dotsc,k\}$.  We fix parameters $C >0$ and $\zeta \in (0,1)$ to be optimized at the end of our analysis (thus appearing as optimization parameters in our theorems), and consider the following ``bad event'':\footnote{The parametrization $\frac{C n \nu e^{-\nu}}{k}$ turns out to be more notationally convenient than $\frac{C' n}{k}$, so we adopt the former.} \emph{A given defective item has $\frac{C n \nu e^{-\nu}}{k}$ tests in which it is the only defective, and a fraction $\zeta$ of those tests are flipped.}  We show that if $n$ is below the stated threshold depending on $f_1$, then some $j \in \{1,\dotsc,k\}$ satisfies this with high probability.

For the Bernoulli design, the number of tests where an item is the only defective follows a multinomial distribution, which is tricky to study directly.  However, by restricting attention to $j \in \{1,\dotsc,k^{\xi}\}$ for $\xi$ close to one, we are able to use a (rigorous) Poisson approximation \cite{arenbaev1977asymptotic} that is simpler to analyze due to having independence across the $k^{\xi}$ items.  For the near-constant design, some additional thought is needed, and we adopt an idea from \cite{Coj19a}, outlined as follows:
\begin{itemize}[itemsep=0ex]
    \item Interpret the placements of items into tests as edges in a bipartite graph (or more precisely a multi-graph, since a given item can be placed in the same test more than once).
    \item Establish that among the $k\Delta$ edges connecting defective items to tests, roughly $e^{-\nu} k \Delta$ of them are connected to tests containing exactly one defective item.
    \item Use a symmetry argument to show that for a given defective $j$, the number of tests in which it is the unique defective roughly follows a hypergeometric distribution, ${\rm Hg}(k\Delta,e^{-\nu} k \Delta,\Delta)$.
\end{itemize}
We need to take this idea further by considering the \emph{conditional} distribution for defective item $j$ given the placements of defective items $1,\dotsc,j-1$,
%{\color{blue}[TODO: need editing since I update the proof]} (JS: I think it should be OK given the way it's worded), 
but the conditioning turns out to have a minimal effect unless $j = \Theta(k)$.  To avoid such cases, we use the same trick as the Bernoulli design, and only run up to $j=k^{\xi}$ for $\xi < 1$ very close to one.

We omit any further details here, but re-iterate that this part of the analysis leads to the $f_1$ term in the final bound.  

{\bf Identifying a suitable non-defective item.} For both designs, the test placements from one item to the next are independent, and also independent of the noise.  As a result, we can condition on the defective placements and the noise, subject to the above ``bad event'' with parameters $(C,\zeta)$ occurring for some $j \in S$.  Even after this conditioning, each non-defective item is placed into tests independently of the others, which is convenient for the analysis.

Consider any fixed non-defective $j' \in \{k+1,\dotsc,p\}$.  When comparing the likelihoods of $(S  \setminus \{j\})\cup \{j'\}$ and $S$, there are only two types of tests that ultimately matter:
\begin{enumerate}
    \item Tests including $j$ but no item from $S \setminus \{j\}$;
    \item Tests containing no item from $S$ at all.
\end{enumerate}
This is because all other tests contribute the same amount to the likelihood even after removing $j$ and adding $j'$ (due to the ``OR'' operation in \eqref{eq:gt_symm_model}).  We are interested in counting how many tests of the above kind $j'$ is placed in, and moreover, how many of those tests are negative vs.~positive (both are possible due to the noise).  We note that the parameter $d$ (appearing in $f_2$) is related to the number of such tests that are positive.  
This part of the analysis comes down to the behavior of various binomial distributions, though it is complicated by delicate dependencies.  We show that tests of the second kind above are dominant, whereas in this part of the analysis, the tests of the first kind above only contribute to lower-order asymptotic terms. 

We omit any further details here, but re-iterate that this part of the analysis leads to the $f_2$ term in the final bound.

\section{Outline of Achievability Proofs} \label{sec:overview_ach}

In this section, we outline the proofs of our achievability results (i.e., upper bounds on the required number of tests), which come with additional challenges compared to the converse results.

The analysis will be split into error events in which the final estimate $\hat{S}$ has ``low overlap'' with $S$ or ``high overlap''.  The main novelty is in the high overlap part, whereas for the low overlap part we will build on the information-theoretic framework of \cite{Sca15b}, which we now proceed to introduce.
 
\subsection{Information-Theoretic Tools} \label{sec:it_tools}

Let $\Xv \in \{0,1\}^{n \times p}$ denote the test matrix, $\Yv \in \{0,1\}^n$ the test results, and $S \subseteq [p]$ the defective set.  Since both test designs are symmetric with respect to re-ordering items, we may focus on a specific choice of $S = s$ in our achievability analysis, say $s = [k] := \{1,\dotsc,k\}$.  Let $\Xv_s$ denote the resulting $n \times k$ sub-matrix of $\Xv$ obtained by taking the columns indexed by $s$, and similarly when $s$ is replaced by any other subset of $[p] := \{1,\dotsc,p\}$. 

An error occurs when some $\bar{s} \ne s$ is favored by the decoder (to be defined below), and the contribution to the error probability can differ significantly depending on the amount of overlap between $\bar{s}$ and $s$.  For instance, there are only relatively few sets $\bar{s}$ with $|s \setminus \bar{s}| = 1$, but the individual probability of favoring such $\bar{s}$ is relatively high due to the low overlap.  To capture this, we consider partitioning $s$ into $(\sdif,\seq)$ with $\sdif \ne \emptyset$, and we denote $\ell = |\sdif| \in \{1,\dotsc,k\}$ (so that $|\seq| = k-\ell$).  Intuitively, $\seq$ represents where an incorrect estimate overlaps with $s$, and $\sdif$ represents the differing part, and $\ell = |\sdif|$ represents the amount of deviation.

For each such $\sdif$, we introduce the quantity
\begin{gather}
    \label{12.2}\imath^n(\Xvdif;\Yv|\Xveq): = \log \frac{\PP (\Yv|\Xvdif,\Xveq)}{\PP(\Yv|\Xveq)}, 
\end{gather}
which is the log-likelihood ratio of $\Yv$ given the full test sub-matrix $\Xv_s$ vs.~the smaller sub-matrix $\Xveq$ alone.  Following the extensive literature on similar techniques for channel coding \cite{Han03}, we refer to this quantity as the \emph{information density}. Notice that the average of (\ref{12.2}) with respect to $(\Xv,\Yv)$ is the following conditional mutual information:
\begin{equation} \label{eq:mi}
    I_\ell^n:= I(\Xvdif;\Yv|\Xveq), %=\mathbbm{E}\big[\imath^n(\Xvdif;\Yv|\Xveq)\big],
\end{equation}
which depends on $(\sdif,\seq)$ only through $\ell:=|s_{\text{dif}}|$ by the symmetry of the test designs. 

The information-theoretic threshold decoder introduced in \cite{Sca15b} can be described as follows:  Fix the constants $\{\gamma_{\ell}\}_{\ell = 1}^{k}$, and search for a set $s$ of cardinality $k$ such that 
\begin{equation}
    \imath^n(\Xvdif;\Yv|\Xveq) \ge \gamma_{|\sdif|}, \quad \forall (\sdif,\seq)\text{ such that }|\sdif| \ne 0. \label{eq:dec}
\end{equation}
If multiple such $s$ exist, or if none exist, then an error is declared.  This decoder is inspired by analogous thresholding techniques from the channel coding literature \cite{Fei54,Han03}, with the rough idea being that the numerator in \eqref{12.2} (i.e., the overall likelihood) tends to be much higher than the denominator for the correct set $s$, whereas for an incorrect $s$ that has overlap $\seq$ with the correct one, such behavior is highly unlikely.

{\bf Limitation of existing approach.} While the analysis of the decoder \eqref{eq:dec} in \cite{Sca15b} leads to optimal thresholds in the noiseless setting, we found that even a sharpened analysis under their framework leads to a suboptimal result in the noisy setting.  In more detail, their analysis leads to three terms: (i) a mutual information based term for $\ell = k$; (ii) a mutual information based term for $\ell =1$; and (iii) a term capturing the concentration behavior of the information density.  We found that a tight concentration analysis can lead to improvements in the third of these.\footnote{Namely, in the denominator of the final result we get $\nu e^{-\nu}(1-e^{-D(1/2 \| \rho)})$ for the Bernoulli design, and $\nu^2 - \nu \log(e^{\nu} - 1 + e^{-D(1/2\|\rho)})$ for the near-constant weight design.  We state these without proof because they would significantly lengthen the paper but still give a suboptimal final result.}  However, there still exists a parameter (called $\delta_2$ in \cite{Sca15b}) that trades off the second and third terms and leaves the result significantly suboptimal, and we were unable to identify any promising route to avoiding this limitation when using the framework of \cite{Sca15b}.

\subsection{A Hybrid Decoding Rule} \label{sec:hybrid}

In view of the above limitations, we introduce a hybrid decoding rule that allows us to use a novel maximum-likelihood analysis in the high-overlap regime, while still relying on the techniques of \cite{Sca15b} in the low-overlap regime.  The decoder is as follows: \emph{Search for a set $s$ of cardinality $k$ such that both of the following are true:
\begin{itemize}
    \item[(i)] It holds that 
    \begin{equation}
        \PP(\Yv|\Xv_{s}) > \PP(\Yv|\Xv_{s'}), \quad \forall s' \text{ such that }1\le |s \setminus s'| \le \frac{k}{\log k}, \label{eq:dec_restricted0}
    \end{equation}
    where we implicitly also constrain $s'$ to have cardinality $k$.
    \item[(ii)] It holds (for suitably chosen $\{\gamma_{\ell}\}_{\frac{k}{\log k} < \ell \le k}$) that
    \begin{equation}
        \imath^n(\Xvdif;\Yv|\Xveq) \ge \gamma_{|\sdif|}, \quad \forall (\sdif,\seq)\text{ such that }|\sdif| > \frac{k}{\log k}, \label{eq:dec_restricted}
    \end{equation}
    where we implicitly also constrain $(\sdif,\seq)$ to be a disjoint partition of $s$.
\end{itemize}}
\noindent If no unique $s$ exists satisfying both of these conditions, then an error is declared.  Since we are using $\ell$ to represent the size of the set difference between the true defective set and an incorrect estimate (i.e., $\ell = |S \setminus \hat{S}| = |\hat{S} \setminus S|$), we will refer to the above cases as the low-$\ell$ (high overlap) and high-$\ell$ (low overlap) regimes respectively.

\subsection{Analysis of the Low Overlap (High $\ell$) Regime}

For the Bernoulli design, the analysis of the threshold decoder \eqref{eq:dec_restricted} for $\ell > \frac{k}{\log k}$ will be taken directly from \cite{Sca15b}, and no change is needed.  However, most of the analysis in \cite{Sca15b} relied heavily on the test matrix having independent rows, so it is not applicable to the near-constant weight design.  Thus, we need to analyze the information density and mutual information to fill in these gaps, with the main steps being as follows:
\begin{itemize}
    \item We show that the mutual information $I_{\ell}^n$ has the same asymptotic behavior as that of the Bernoulli design;
    \item We establish a suitable concentration bound for the information density $\imath^n(\Xvdif;\Yv|\Xveq)$.
\end{itemize}
Due to the lack of independence across tests, both of these require more effort than their counterpart for the Bernoulli design.  On the other hand, for the concentration bound we are in the fortunate position of not needing precise constant factors, as those are only needed in the low-$\ell$ regime which we will handle using different methods.

This part of the analysis will be detailed in Appendix \ref{app:high_both} (common analysis for both designs) and Appendix \ref{app:high_ncc} (proofs of additional technical lemmas for the near-constant weight design).

\subsection{Analysis of the High Overlap (Low $\ell$) Regime}

For the low-$\ell$ regime, we follow a similar structure to the converse bound, but we now need to consider all $\ell=1,\dotsc,\frac{k}{\log k}$ instead of only $\ell=1$.  Specifically, we again consider (optimal) maximum-likelihood decoding, and accordingly, we seek to show that $S$ is preferred to any other set of the form $(S \setminus \mathcal{J}) \cup \mathcal{J}'$ with $\mathcal{J} \subset S$ and $\mathcal{J}' \subset \{1,\dotsc,p\} \setminus S$ both having cardinality $\ell$ for some $\ell\in \big\{1,\dotsc,\frac{k}{\log k}\big\}$.  Recall that ``being preferred'' is equivalent to having more tests that match what the noiseless outcome would have been.

Generalizing the idea from the converse proof, we consider the event that a given defective subset $\mathcal{J} \subset S$ of size $\ell$ has $\frac{C n \nu e^{-\nu} \ell}{k}$ tests where at least one of its items is included but none from $S \setminus \mathcal{J}$ are included, and a fraction $\zeta$ of those tests are flipped.  As well as generalizing beyond $\ell = 1$, we now also need to simultaneously consider \emph{all} possible choices of $(C,\zeta)$ that can occur. 

We let $k_{\ell,C,\zeta}$ denote the (random) number of size-$\ell$ subsets of $S$ such that the above event holds with parameters $(C,\zeta)$, and the first step is to rule out many $(\ell,C,\zeta)$ triplets by showing that $k_{\ell,C,\zeta}=0$ with high probability.  This is done by directly studying $\EE[k_{\ell,C,\zeta}]$ under the respective design, and then applying Markov's inequality and a union bound.  This part of the analysis leads to the term $f_1$ in the final bound.

For the $(\ell,C,\zeta)$ triplets for which we cannot guarantee $k_{\ell,C,\zeta}=0$, we need to consider the placements of non-defectives.  We again follow a similar argument to the converse, but with more general $\ell \ge 1$ and all relevant $(C,\zeta)$ pairs being considered simultaneously, with the help of high-probability events and union bounds.  
Moving from $\ell=1$ to $\ell > 1$ adds particular technical challenges, but conceptually we still follow similar steps as in the converse analysis, and this part of the analysis leads to the term $f_2$ in the final bound.

This part of the analysis will be detailed in Appendix \ref{app:ber_achi} (Bernoulli design) and Appendix \ref{app:nc_achi} (near-constant weight design).

\section{Conclusion}

We have derived exact asymptotic thresholds on the number of tests required for probabilistic group testing under binary symmetry noise, under both the Bernoulli design and the near-constant weight design.  Our results bring the degree of understanding of this noisy group testing model significantly closer to the noiseless setting.  With a spatial coupling design concurrently being shown to attain the same threshold as the near-constant weight design in a computationally efficient manner \cite{Coj24}, we believe that the main remaining open problem is as follows: {\em Is this threshold also asymptotically optimal among the entire class of non-adaptive designs?}  The analogous question is indeed true in the noiseless case \cite{Coj19a}, but answering it in the noisy case may be significantly more challenging.  While an affirmative answer may be expected based on the noiseless setting, a negative answer also seems highly plausible given that different error events are dominant here; in fact, we have already observed notable differences to the noiseless setting, such as the design parameter $\nu = \log 2$ not always being optimal.

% but two major open problems remain:
% \begin{itemize}
%     \item It is currently unclear whether the exact threshold for the near-constant weight design (with optimized $\nu$) is also asymptotically optimal among \emph{all} non-adaptive designs.
%     \item The upper bounds were only derived using computationally intractable information-theoretic decoding rules, and establishing the same thresholds using polynomial-time algorithms (possibly with the test design slightly modified, e.g., via spatial coupling techniques \cite{Coj19a}) would be of significant interest.
% \end{itemize}
% Noiseless counterparts for both of these points were established in \cite{Coj19a}, so their techniques may serve as the natural starting point for studying the same questions in the noisy setting.

\newpage
% \bigskip% \newpage
{\centering \huge \bf Appendix \par}

\begin{appendix}

    \section{Roadmap of the Proofs} \label{app:roadmap}

    In all four of our main results, the threshold consists of a maximum of two terms, the second of which itself consists of an optimization of the parameters $(C,\zeta)$.  In each case, these two terms will arise either from two separate analyses, or from two different error events within the same analysis.  Forward-references to the relevant subsequent appendices will be given in Section \ref{app:overview} below.

    {\bf Converse results.} For the converse analysis, the ``max'' operation comes from simply proving two different converse bounds and taking the stronger of the two.  While we do not explicitly split the analysis into $\ell=1,\dotsc,k$ for the converse, we still refer to the first term as the ``high-$\ell$'' analysis, and the second term as the ``low-$\ell$'' analysis.  This is because:
    \begin{itemize}
        \item The starting point for the first term is a lower bound on the error probability from \cite{Sca15b} that depends on $\ell$ but we specialize to $\ell = k$, which amounts to considering error events where the incorrect estimate is completely disjoint from the true defective set.
        \item For the second term, we consider error events under which an estimate of the form $\hat{S} = (S  \setminus \{j\})\cup \{j'\}$ is favored over $S$ by the decoder, thus corresponding to $\ell=1$ since we interpret $\ell = |S \setminus \hat{S}| = |\hat{S} \setminus S|$.
    \end{itemize}

    {\bf Achievability results.} For the achievability results, we recall the hybrid decoder introduced in  \eqref{eq:dec_restricted0}--\eqref{eq:dec_restricted} in Section \ref{sec:hybrid}, which for a given candidate set $s$, does two different checks that correspond to considering high-$\ell$ and low-$\ell$ errors separately:
    \begin{itemize}
        \item The first check (corresponding to low $\ell$) is simply whether $s$ has a higher likelihood than all $s'$ within distance $\frac{k}{\log k}$, so our analysis comes down to that of maximum-likelihood decoding.
        \item The second check (corresponding to high $\ell$) is based on thresholding the information density (as defined in \eqref{12.2}), and we study its conditions for success by following the framework of \cite{Sca15b} (but with significant gaps to fill in the case of the near-constant weight design).
    \end{itemize}
    We claim that in order for our decoder to succeed, the following conditions are sufficient:
    \begin{enumerate}
        \item The true defective set (say $s = \{1,\dotsc,k\}$ without loss of generality) satisfies \eqref{eq:dec_restricted0}--\eqref{eq:dec_restricted};
        \item For any other set $\tilde{s}$ of cardinality $k$ with $|s \setminus \tilde{s}| > \frac{k}{\log k}$, it holds that $\imath^n(\Xv_{\tilde{s} \setminus s}; \Yv | \Xv_{\tilde{s} \cap s}) < \gamma_{\ell}$, where $\ell = |\tilde{s} \setminus s| = |s \setminus \tilde{s}|$.
    \end{enumerate}
    To see that these conditions are sufficient, first consider some incorrect estimate $\tilde{s}$ with $|s \setminus \tilde{s}| \le \frac{k}{\log k}$.  Such $\tilde{s}$ is ruled out by the low-$\ell$ part of our decoder (i.e., \eqref{eq:dec_restricted0}), because the correct $s$ is within distance $\frac{k}{\log k}$ and has a higher likelihood due to condition 1 above.  On the other hand, for any incorrect estimate $\tilde{s}$ with $|s \setminus \tilde{s}| > \frac{k}{\log k}$, condition 2 above immediately implies that the high-$\ell$ condition of the decoder (i.e., \eqref{eq:dec_restricted}) is not satisfied for $\tilde{s}$.  Thus, we conclude that the true defective set $s$ satisfies \eqref{eq:dec_restricted0}--\eqref{eq:dec_restricted} but none of the incorrect sets $\tilde{s}$ do so, as desired.

    In the subsequent appendices, we will write $\PP({\rm err})$ to denote the failure probabilities associated with the above two sufficient conditions, where it will be clear from context whether we mean condition 1 or condition 2.  We will show that attaining $\PP({\rm err}) \to 0$ for condition 1 leads to the first term in each achievability result's requirement on $n$, and attaining $\PP({\rm err}) \to 0$ for condition 2 leads to the second term.  The overall requirement on $n$ is then the stricter of these two requirements.

    \subsection{Overview of the Appendices} \label{app:overview}

    We outline the structure of the remaining appendices as follows: 
    \begin{itemize}
        \item In Appendix \ref{app:mle}, we present necessary and sufficient conditions for maximum-likelihood decoding to fail, which will be used in all four of our low-$\ell$ proofs.
        \item In Appendix \ref{app:ber_con}, we establish the low-$\ell$ converse result for the Bernoulli design.
        \item In Appendix \ref{app:ber_achi}, we establish the low-$\ell$ achievability result for the Bernoulli design.
        \item In Appendix \ref{app:nc_converse}, we establish the low-$\ell$ converse result for the near-constant weight design.
        \item In Appendix \ref{app:nc_achi}, we establish the low-$\ell$ achievability result for the near-constant weight design.
        \item In Appendix \ref{app:nc_lemms}, we state and prove various technical lemmas that are used throughout Appendices \ref{app:nc_converse} and \ref{app:nc_achi}.
        \item In Appendix \ref{app:high_both}, we establish the high-$\ell$ achievability result for both designs (with the Bernoulli design coming from \cite{Sca15b} with only minor changes).
        \item In Appendix \ref{app:high_conv}, we establish the high-$\ell$ converse result for both designs (with the Bernoulli design coming directly from \cite{Sca15b}).
        \item In Appendix \ref{app:high_ncc}, we prove two technical lemmas used for the high-$\ell$ results under the near-constant weight design.
    \end{itemize}
    Before proceeding, we also state a useful observation regarding the scaling of the number of tests $n$, and some useful technical lemmas.

    \subsection{Note on the Scaling of $n$} \label{sec:n_scaling}

    Both our achievability and converse results are stated in terms of a threshold $n^*$ scaling as $\Theta(k \log p)$.  Recall that throughout the entire paper, we consider the regime $k = \Theta(p^{\theta})$ with $\theta \in (0,1)$, which implies that $\log k = \Theta(\log p)$.  This implies that $n^* = \Theta(k \log p) = \Theta(k \log k)$ in all of our results.

    In view of this observation, we will assume throughout the analysis that the number of tests $n$ itself also scales as $\Theta(k \log p) = \Theta(k \log k)$.  For the achievability part, this is justified by the fact that the theorems state that we can achieve $\pe \to 0$ while satisfying $n \le (1+\eta)n^*$, and we are free to let this inequality hold with equality or ``near-equality'' (say $n \in [n^*,(1+\eta)n^*]$). For the converse part, the theorems state that attaining $\pe \not\to 1$ requires $n\geq (1-\eta) n^*$, and we will prove the contrapositive statement that $\pe \to 1$ whenever $n < (1-\eta) n^*$.  The assumption $n = \Theta(k \log p)$ is then justified by noting that if $n = o(k \log p)$, then existing results already show that $\pe \to 1$ even in the case of noiseless tests \cite{Coj19a}.
    
    % and a converse for the former immediately implies a converse for the latter anyway (since the decoder could always choose to ignore some of the test results).

    \subsection{Useful Technical Lemmas}

    The following (anti)-concentration bounds for binomial random variables will be used frequently in our analysis.  Recall the notation $D(a\|b) = a\log\frac{a}{b} + (1-a)\log\frac{1-a}{1-b}$ for binary KL divergence.

\begin{lem}\label{binoconcen}
    {\rm((Anti-)Concentration of Binomial Random Variables, e.g., \cite[Sec.~4.7]{ash1990information}, \cite{arratia1989tutorial})} For $X\sim \mathrm{Bin}(N,q)$, we have the following:
    \begin{itemize}
        \item {\rm (Chernoff bound)} If 
        $k\leq Nq$, then we have 
    \begin{equation}
        \PP(X\leq k)\leq \exp\Big(-N \cdot D\Big(\frac{k}{N}\big\|q\Big)\Big),\label{eq:chernoff1}
    \end{equation}      
    which remains true when replacing $\PP(X\le k)$ with $\PP(X\ge k)$ for $k\ge Nq$. This implies 
          \begin{equation}
              \PP(X\leq k)\leq \exp \Big(-Nq\Big(\frac{k}{Nq} \log\frac{k}{Nq}+1-\frac{k}{Nq}\Big)\Big)\label{eq:chernoff2}
          \end{equation}
for $k\le Nq$, which remains true when replacing $\mathbbm{P}(X\le k)$ with $\PP(X\ge k)$ for $k\ge Nq$.  
        \item {\rm (Anti-concentration)} For any $k \in \{1,\dotsc,N-1\}$, we have
        \begin{equation}
            \label{eq:anti1}
            \PP(X = k) \ge \frac{1}{2\sqrt{2k(1-\frac{k}{N})}}\exp\Big(-N\cdot  D\Big(\frac{k}{N}\big\|q\Big)\Big),
        \end{equation}
        Moreover, for $k \in \{0,1,\dotsc,N\}$, we have  
        \begin{equation}
            \label{eq:anti2}
            \PP(X=k) \geq \frac{1}{\sqrt{2N}}\exp\Big(-N \cdot D\Big(\frac{k}{N}\big\|q\Big)\Big).
        \end{equation}
        % since $2\sqrt{\frac{k}{N}(1-\frac{k}{N})}\leq 1$, and note that the latter also holds for $k=0,N$. 
    \end{itemize}
\end{lem}

We also state the following related lemma regarding binomial coefficients.

\begin{lem} \label{lem:coeff_bounds}
    {\em (Bounds on Binomial Coefficients, e.g., \cite[Sec.~4.7]{ash1990information})} Given positive integers $N$ and $k$ with $k \in \{1,\dotsc,N-1\}$, we have
    \begin{equation}
        \frac{\sqrt{\pi}}{2} \cdot \frac{\exp(N H_2(k/N))}{\sqrt{2\pi k(1-k/N)}} \leq \binom{N}{k}\leq \frac{\exp(N H_2(k/N))}{\sqrt{2\pi k(1-k/N)}}.
    \end{equation}
    Moreover, for $k \in \{0,1,\dotsc,N\}$, we have
    \begin{equation}
         \frac{\exp(N H_2(k/N))}{\sqrt{2N}} \leq \binom{N}{k}\leq \exp(N H_2(k/N)).
    \end{equation}
\end{lem}

\subsection{Table of Notation}

The main recurring notation used throughout the appendices is summarized in Table \ref{tbl:notation}.

\begin{table}

\caption{List of recurring notation. \label{tbl:notation}}

{\small
\begin{tabular}{|c|c|}
\hline 
\multicolumn{2}{|c|}{\textbf{Introduced in the main text}}\tabularnewline
\hline 
$p,k,n$ & Number of items, defectives, and tests\tabularnewline
\hline 
$\rho,\nu,\theta$ & Noise level, design parameter, sparsity parameter $(k=\Theta(p^{\theta})$)\tabularnewline
\hline 
$\Delta$ & Number of placements in near-constant weight design ($\Delta=\frac{\nu n}{k}$)\tabularnewline
\hline 
$S,s$ & Defective set (random variable, specific realization)\tabularnewline
\hline 
$\mathbf{X},\mathbf{Y},\mathbf{Z}$ & Test matrix, test results, noise variables\tabularnewline
\hline 
$(s_{{\rm dif}},s_{{\rm eq}})$ & Partition of $s$ into two disjoint sets\tabularnewline
\hline 
$\ell$ & Size of $s_{{\rm dif}}$ and/or size of $s\setminus s^{\prime}$ for
true $s$ and estimated $s^{\prime}$\tabularnewline
\hline 
$C$ & Parameter associated with number of non-masked tests for defectives\tabularnewline
\hline 
$\zeta$ & Parameter associated with fraction of non-masked tests that are flipped\tabularnewline
\hline 
$f_{1}$ & Function associated with some defective satisfying a ``bad'' event\tabularnewline
\hline 
$f_{2},g$ & Functions associated with some non-defective satisfying a ``bad''
event\tabularnewline
\hline 
$d,d^{*}$ & Parameter associated with number of ``relevant'' positive tests a non-defective is in
\tabularnewline
\hline 
$\imath^{n},I_{\ell}^{n}$ & Information density, conditional mutual information\tabularnewline
\hline 
$j,\mathcal{J}$ & Generic defective item / size-$\ell$ set of defective items \tabularnewline
\hline 
$j^{\prime},\mathcal{J}^{\prime}$ & Generic non-defective item / size-$\ell$ set of non-defective items\tabularnewline
\hline 
\multicolumn{2}{|c|}{\textbf{Introduced in the appendices}}\tabularnewline
\hline 
$\mathcal{L}$ & Likelihood function\tabularnewline
\hline 
$N_{\mathbf{X},\mathbf{Y}}(S)$ & Number of tests that are correct under $S$\tabularnewline
\hline 
$\mathcal{M}_{\mathcal{J}}$ & Tests that contain defectives from $\mathcal{J}$ but not other
defectives\tabularnewline
\hline 
$\mathcal{M}_{\mathcal{J}0},\mathcal{M}_{\mathcal{J}1}$ & Negative (resp.~positive) tests in $\mathcal{M}_{\mathcal{J}}$\tabularnewline
\hline 
$\mathcal{N}_{0}$ & Tests that contain no defectives\tabularnewline
\hline 
$\mathcal{N}_{00},\mathcal{N}_{01}$ & Negative (resp.~positive) tests in $\mathcal{N}_{0}$\tabularnewline
\hline 
$M_{(\cdot)},N_{(\cdot)}$ & Cardinalities of $\mathcal{M}_{(\cdot)}$ and $\mathcal{N}_{(\cdot)}$\tabularnewline
\hline 
$\mathcal{K}_{\ell,C,\zeta}$,$k_{\ell,C,\zeta}$ & Set of size-$\ell$ subsets of $s$ satisfying $(C,\zeta)$ conditions,
its cardinality\tabularnewline
\hline 
$G_{\mathcal{J},\mathcal{J}^{\prime},1},G_{\mathcal{J},\mathcal{J}^{\prime},2}$ & Tests in $\mathcal{N}_{01}\cup\mathcal{M}_{\mathcal{J}1}$ (resp.~$\mathcal{N}_{00}\cup\mathcal{M}_{\mathcal{J}0}$)
containing an item from $\mathcal{J}$\tabularnewline
\hline 
$\mathcal{M}_{j}^{\prime}$, $M_{j}^{\prime}$, $\mathcal{K}_{C,\zeta}^{\prime}$,$k_{C,\zeta}^{\prime}$ & Variants that exclude multiple identical placements (near-const.~design)\tabularnewline
\hline 
$\mathbf{T}_{s}$ & Unordered multi-set indicating test placements (near-const.~design)\tabularnewline
\hline 
$\mathcal{M}_{\mathcal{J},\mathcal{J}^{\prime}}$,$\mathcal{N}_{0,\mathcal{J}^{\prime}}$ & Tests in $\mathcal{M}_{\mathcal{J}}$ (resp.~$\mathcal{N}_{0}$)
containing an item from $\mathcal{J}^{\prime}$\tabularnewline
\hline 
$\text{\ensuremath{\mathscr{A}_{1},\mathscr{A}_{2},}}\,{\rm etc.}$ & Used for high-probability events
\tabularnewline
\hline 
$A$ & Equal to $C(1-2\zeta)$, captures the dependence of $g$ on $(C,\zeta)$
\tabularnewline
\hline 
$\hat{f}_{2}$ & Variant of $f_{2}$ equaling 0 outside a certain range\tabularnewline
\hline 
$\xi$ & Generic parameter taken close to $1$\tabularnewline
\hline 
$\psi_{\ell}$ & Function representing concentration bound for $\imath^{n}$\tabularnewline
\hline 
$\alpha$ & Limiting value of $\frac{\ell}{k}$\tabularnewline
\hline 
$n_{1},n_{2}$ & Number of tests with (resp.~without) an item from $s_{{\rm eq}}$\tabularnewline
\hline 
$\mathbf{Y}_{1},\mathbf{Y_{2}}$ & Results of tests with (resp.~without) an item from $s_{{\rm eq}}$\tabularnewline
\hline 
$M$ & Number of tests with an item from $s_{{\rm dif}}$ but not $s_{{\rm eq}}$\tabularnewline
\hline 
$V$ & Number of positive tests in $\mathbf{Y}_{2}$\tabularnewline
\hline 
$\mathsf{P}$ & Probabilities implicitly conditioned on $\mathbf{X}_{s_{{\rm eq}}}$\tabularnewline
\hline 
\end{tabular}
}
\end{table}

% \section{Proof of Theorem \ref{thm:nc_convese} (Converse Bound for Near-Constant Weight Designs)} \label{sec:pf_conv_nc}

% where in (\ref{eq:con_M}) we condition on the values of $M$ given in (\ref{eq:M_val1}),   (\ref{eq:divide2})  follows by writing $(\rho\star e^{-\nu}+o(1))n=(1-\rho)(1-e^{-\nu})n + (\rho e^{-\nu}+o(1))n$ and then letting the two terms being the realizations of the two independent binomial variables. All that remains is to show the two binomial probability terms in
% (\ref{eq:term1}) and (\ref{eq:divide2}) scale as $\exp(o(n))$. Note that they share a common expression of $\PP\big(\mathrm{Bin}(\tilde{n},\tilde{\rho})=\tilde{n}\tilde{\rho}(1+o(1))\big)$ where $\tilde{n}=\Theta(n)$, $\tilde{\rho}=\rho\text{ or }1-\rho$. To this end, we apply $\binom{N}{\beta N}\gtrsim \frac{\exp(NH(\beta))}{\sqrt{N\beta (1-\beta)}}$ and perform some algebra:\begin{align*}
%    & \PP\Big(\mathrm{Bin}(\tilde{n},\tilde{\rho})=\tilde{n}\tilde{\rho}(1+o(1))\Big)\\=&\binom{\tilde{n}}{\tilde{n}\tilde{\rho}(1+o(1))} \tilde{\rho}^{\tilde{n}\tilde{\rho}(1+o(1))}(1-\tilde{\rho})^{\tilde{n}(1-\tilde{\rho})(1+o(1))}\\
%     \gtrsim & \frac{\exp[\tilde{n}H(\tilde{\rho}+o(1))]}{\sqrt{\tilde{n}\tilde{\rho}(1-\tilde{\rho})(1+o(1))}}  \exp\Big(\tilde{n}\tilde{\rho}(1+o(1))\log\tilde{\rho}+\tilde{n}(1-\tilde{\rho})(1+o(1))\log(1-\tilde{\rho})\Big)\\
%     =& \exp(o(\tilde{n}))=\exp(o(n)).
% \end{align*}
% Thus, $\PP(v)=\exp(o(n))$ with $1-o(1)$ probability. 

\section{Necessary and Sufficient Conditions for Maximum-Likelihood Decoding} \label{app:mle}

As noted in Section \ref{sec:overview_conv}, the optimal decoding rule for minimizing the overall error probability (under a uniform prior on the defective set $S$) is maximum-likelihood decoding:
\begin{equation}\label{eq:mle}
    \widehat{S} = \text{arg~max}_{s'\in \mathcal{S}_k} \mathcal{L}(s'), \text{   where }\mathcal{L}(S):=\PP(\Yv|S,\Xv),
\end{equation}
with $\mathcal{S}_k$ being the set of all $p \choose k$ subsets of $[p]$ of size $k$.  Our converse analysis will characterize conditions under which MLE has $\pe \to 1$, which implies the same for any decoder.

In our achievability analysis, in view of the two sufficient conditions identified in Appendix \ref{app:roadmap}, we will consider a \emph{restricted} MLE decoder such that the argmax in \eqref{eq:mle} is restricted to sets within distance $\frac{k}{\log k}$ of the true one.  Note that this does not imply that our decoder (introduced in Section \ref{sec:hybrid}) actually has knowledge of the true defective set; rather, restricted MLE is a hypothetical decoder whose success coincides with one of the sufficient conditions identified in Appendix \ref{app:roadmap}.  Thus, when the true defective set is $s$, we are interested in the following:
\begin{equation}\label{eq:mle_restricted}
    \widehat{S}' = \text{arg~max}_{s' \in \mathcal{S}_k \,:\, |s \setminus s'| \le \frac{k}{\log k}} \mathcal{L}(s').
\end{equation}
In fact, it suffices to also consider this restricted problem in our converse analysis, observing that if $\widehat{S}' \ne s$ in \eqref{eq:mle_restricted} then we clearly also have $\widehat{S} \ne s$ in \eqref{eq:mle}.

% which succeeds when $\widehat{S}$ recovers the actual defective set, and fails otherwise.
%succeed for the given $\ell\in [1,\frac{k}{\log k}]$ if the actual defective $s$ is favored to all length-$k$ $s'\subset [p]$ with $|s\setminus s'|=\ell$.\footnote{\color{blue}[Might need some cautiousness to define the success of MLE]} Thus, MLE succeeds   for some $\ell \in [1,\frac{k}{\log k}]$ if $\mathcal{L}(s)> \mathcal{L}(s')$ holds for all length-$k$ $s'\subset [p]$ with $|s\setminus s'|=\ell$. Otherwise, MLE is likely to fail for this $\ell$ if there exists some length-$k$ $s' \subset [p]$ with $|s\setminus s'|=\ell$ such that $\mathcal{L}(s')\ge \mathcal{L}(s)$, meaning that $s'$ is at least as favored to the MLE decoder as the actual defective set $s$. Furthermore, MLE is bound to fail for this $\ell$ if there exists some length-$k$ $s'\subset [p]$ with $|s\setminus s'|=\ell$ such that $\mathcal{L}(s')>\mathcal{L}(s)$.

For a given set $s'\subset [p]$ with $|s'|=k$, we say that a test is correct (with respect to $s'$) if its result matches what the noiseless result would be under defective set $s'$.  Given the test matrix and outcomes, $(\Xv,\Yv)$, we define $N_{\Xv,\Yv}(s')$ as the number of correct tests. Observe that \begin{align}\label{eq:likeli_via_corr}
    \mathcal{L}(s')=\rho^{n-N_{\Xv,\Yv}(s')}(1-\rho)^{N_{\Xv,\Yv}(s')},
\end{align} 
so due to the fact that $\rho \in \big(0,\frac{1}{2}\big)$, the MLE decoder finds $\widehat{S}$ with the most correct tests, i.e., $N_{\Xv,\Yv}(\widehat{S})\geq N_{\Xv,\Yv}(s')$ holds for any $s' \in \mathcal{S}_k$ (or similarly for $\widehat{S}'$ in the restricted version).

We introduce some notation before proceeding.
Let $\nu>0$ be the parameter associated with the Bernoulli designs or near-constant weight designs.   We say that the pair $(C,\zeta)\in [0,\infty)\times [0,1]$ is \emph{feasible} with respect to some given $\ell\in[1,\frac{k}{\log k}]$ if $\frac{Cn\nu e^{-\nu}\ell}{k}$ and $\frac{\zeta\cdot Cn\nu e^{-\nu}\ell}{k}$ are integers; when $C=0$, we only view $(C,\zeta)=(0,0)$ as feasible (rather than all $\zeta\in[0,1]$).
 Given $\mathcal{J}\subset s$ with $|\mathcal{J}|=\ell$ for some specific $ \ell\in[1,\frac{k}{\log k}]$, we define the following useful quantities:
 \begin{itemize}
    \item $\mathcal{M}_{\mathcal{J}}\subset [n]$ indexes the tests that include some defective from $\mathcal{J}$ but no defective from $s\setminus \mathcal{J}$;
    \item $\mathcal{M}_{\mathcal{J}0}$ indexes the negative tests in $\mathcal{M}_{\mathcal{J}}$ (i.e., the tests in $\mathcal{M}_{\mathcal{J}}$ that are flipped by noise);
    \item $\mathcal{M}_{\mathcal{J}1}$ indexes the positive tests in $\mathcal{M}_{\mathcal{J}}$;
    \item We let their cardinalities be $M_{\mathcal{J}}:=|\mathcal{M}_{\mathcal{J}}|$, $M_{\mathcal{J}0}:=|\mathcal{M}_{\mathcal{J}0}|$ and $M_{\mathcal{J}1}:=|\mathcal{M}_{\mathcal{J}1}|$.
 \end{itemize}
 Similarly,
 \begin{itemize}
     \item $\mathcal{N}_0\subset [n]$ indexes the tests that include no defectives;
     \item $\mathcal{N}_{00}$ indexes the negative tests in $\mathcal{N}_0$;
     \item $\mathcal{N}_{01}$ indexes the positive tests in $\mathcal{N}_0$ (i.e., the tests in $\mathcal{N}_0$ that are flipped by noise);
     \item We denote their cardinalities by $|\mathcal{N}_0|=N_0$, $|\mathcal{N}_{00}|=N_{00}$ and $|\mathcal{N}_{01}|=N_{01}$. 
 \end{itemize}
In the following lemma, we provide deterministic conditions for the failure of the restricted MLE decoder.  

\begin{lem}\label{lem:restate}
    {\rm (Conditions for the Failure of MLE)} Let $s\subset [p]$ with cardinality $k$ be the actual defective set, and consider the restricted MLE decoder in \eqref{eq:mle_restricted}. Given $\ell\in [1,\frac{k}{\log k}]$, for a   feasible pair $(C,\zeta)\in [0,\infty) \times [0,1]$ such that $\frac{Cn\nu e^{-\nu}\ell}{k}$ and $\frac{\zeta\cdot Cn\nu e^{-\nu}\ell}{k}$ are integers, we define 
    \begin{align}
        \label{eq:defi_KellCzeta}\mathcal{K}_{\ell,C,\zeta} = \Big\{\mathcal{J}\subset s \,:\, |\mathcal{J}|=\ell, M_{\mathcal{J}}= \frac{Cn\nu e^{-\nu}\ell}{k}, ~M_{\mathcal{J}0}=\frac{\zeta\cdot Cn\nu e^{-\nu}\ell}{k}\Big\}
    \end{align}
    and let $k_{\ell,C,\zeta}:=|\mathcal{K}_{\ell,C,\zeta}|$. Given $\mathcal{J}\in \mathcal{K}_{\ell,C,\zeta}$ and $\mathcal{J}'\subset [p]\setminus s$ with $|\mathcal{J}'|=\ell$, let $G_{\mathcal{J},\mathcal{J}',1}$ be the number of tests in $\mathcal{N}_{01}\cup \mathcal{M}_{\mathcal{J}1}$ that contain some item from $\mathcal{J}'$, and let $G_{\mathcal{J},\mathcal{J}',2}$ be the number of tests in $\mathcal{N}_{00}\cup \mathcal{M}_{\mathcal{J}0}$ that contain some item from $\mathcal{J}'$. Then we have the following two statements: 

    \noindent{(a)} {\rm(Sufficient condition for failure)} If for some feasible pair $(C,\zeta)\in[0,\infty)\times [0,1]$ with respect to some $\ell\in [1,\frac{k}{\log k}]$, the set $\mathcal{K}_{\ell,C,\zeta}$ is non-empty, and there exist some $\mathcal{J}\in \mathcal{K}_{\ell,C,\zeta}$ and some $\mathcal{J}'\subset [p]\setminus s$ with $|\mathcal{J}'|=\ell$ such that
    \begin{align}\label{eq:greater}
        G_{\mathcal{J},\mathcal{J}',1} - G_{\mathcal{J},\mathcal{J}',2} > (1-2\zeta) \frac{Cn\nu e^{-\nu}\ell}{k},
    \end{align}
    then $\mathcal{L}(s')>\mathcal{L}(s)$ with $s':=(s\setminus \mathcal{J})\cup \mathcal{J}'$, which implies the failure of the restricted MLE decoder (\ref{eq:mle_restricted}).

    \noindent{(b)} {\rm(Necessary condition for failure)} If the restricted MLE decoder (\ref{eq:mle_restricted}) fails (i.e., it returns $\widehat{S} = s'$ for some $s'\neq s$ with $|s\setminus s'|\in [1,\frac{k}{\log k}]$), then   for some feasible pair $(C,\zeta)\in[0,\infty)\times [0,1]$ with respect to $\ell:=|s\setminus \widehat{S}|$, the set $\mathcal{K}_{\ell,C,\zeta}$ is non-empty, and there exist some $\mathcal{J}\in \mathcal{K}_{\ell,C,\zeta}$ and some $\mathcal{J}'\subset [p]\setminus s$ with $|\mathcal{J}'|=\ell$ such that \begin{align}\label{eq:greater_equal}
        G_{\mathcal{J},\mathcal{J}',1} - G_{\mathcal{J},\mathcal{J}',2} \ge (1-2\zeta) \frac{Cn\nu e^{-\nu}\ell}{k}.
    \end{align}
\end{lem}
\begin{proof}Recall that we let $N_{\Xv,\Yv}(s)$ be the number of correct tests under the defective set $s$, which enters the likelihood via (\ref{eq:likeli_via_corr}). We separately prove the statements (a) and (b) by analyzing the number of correct tests.

    \textbf{Proof of (a):} For $\mathcal{J}\in \mathcal{K}_{\ell,C,\zeta}$ and $\mathcal{J}'\subset [p]\setminus s$ with $|\mathcal{J}'|=\ell$, we note that the size-$k$ set $s':=(s\setminus \mathcal{J})\cup \mathcal{J}'$ satisfies $|s\setminus s'|=\ell$, and $\mathcal{L}(s')>\mathcal{L}(s)$ evidently implies the failure of MLE. Thus we only need to prove  $\mathcal{L}(s')>\mathcal{L}(s)$, and it suffices to show   $N_{\Xv,\Yv}(s')> N_{\Xv,\Yv}(s)$. To this end, we compare the number of correct tests under $s$ and $s'$ as follows: \begin{itemize}
    \item We start with the true defective set $s$.  By removing $\mathcal{J}$ from the defective set, the correct tests in $\mathcal{M}_{\mathcal{J}1}$ become incorrect, while the
    incorrect tests in $\mathcal{M}_{\mathcal{J}0}$ become correct. Thus, there is a loss of $|\mathcal{M}_{\mathcal{J}1}|-|\mathcal{M}_{\mathcal{J}0}|=(1-\zeta)M_{\mathcal{J}}-\zeta M_{\mathcal{J}}=(1-2\zeta)M_{\mathcal{J}}= (1-2\zeta)\frac{Cn\nu e^{-\nu}\ell}{k}$ correct tests. 
    \item Next, we add $\mathcal{J}'$ to $s\setminus \mathcal{J}$ to arrive at $s'=(s\setminus\mathcal{J})\cup \mathcal{J}'$. Then, the incorrect tests in $\mathcal{N}_{01}\cup \mathcal{M}_{\mathcal{J}1}$ become correct if they contain some item from $\mathcal{J}'$, while the correct tests in $\mathcal{N}_{00}\cup \mathcal{M}_{j0}$ become incorrect if they contain some item from $\mathcal{J}'$. This leads to a gain of $G_{\mathcal{J},\mathcal{J}',1}-G_{\mathcal{J},\mathcal{J}',2}$ correct tests. 
\end{itemize}  
Therefore, by (\ref{eq:greater}) we have 
$$N_{\Xv,\Yv}(s')-N_{\Xv,\Yv}(s) = G_{\mathcal{J},\mathcal{J}',1}-G_{\mathcal{J},\mathcal{J}',2}-(1-2\zeta)\frac{Cn\nu e^{-\nu}\ell}{k}>0,$$ 
which leads to $\mathcal{L}(s')>\mathcal{L}(s)$ as claimed.

\textbf{Proof of (b):} If restricted MLE fails and returns a set $s'$ with  $|s\setminus s'|=\ell\in [1,\frac{k}{\log k}]$, then we note from the definition   (\ref{eq:mle_restricted}) that $\mathcal{L}(s')\geq \mathcal{L}(s)$, and thus $N_{\Xv,\Yv}(s')\ge N_{\Xv,\Yv}(s)$. Based on $s'$, we construct the feasible $(C,\zeta)$ for which the conclusion of (b) holds. Specifically, we let $\mathcal{J}=s\setminus s'$ with $|\mathcal{J}|=\ell$ and recall the definitions of $\mathcal{M}_{\mathcal{J}},\mathcal{M}_{\mathcal{J}0},\mathcal{M}_{\mathcal{J}1}$ and their cardinalities $M_{\mathcal{J}},M_{\mathcal{J}0},M_{\mathcal{J}1}$; then, we set $(C,\zeta)= (\frac{M_{\mathcal{J}}}{k^{-1}n\nu e^{-\nu}\ell},\frac{M_{\mathcal{J} 0}}{M_{\mathcal{J}}})$ if $M_{\mathcal{J}}>0$, or $(C,\zeta)=(0,0)$ if $M_{\mathcal{J}}=0$. By construction   $(C,\zeta)\in [0,\infty)\times [0,1]$ is feasible, and by (\ref{eq:defi_KellCzeta}) it is easy to see that $\mathcal{J}\in \mathcal{K}_{\ell,C,\zeta}$. Moreover, we let $\mathcal{J}'=s'\setminus s\subset [p]\setminus  s$   satisfying $|\mathcal{J}'|=\ell$, and note that by the definition of $(\mathcal{J},\mathcal{J}')$ we have $s'=(s\setminus \mathcal{J})\cup\mathcal{J}'$. It remains to show that (\ref{eq:greater_equal}) holds for the $(\mathcal{J},\mathcal{J}')$ we are considering. To this end, we compare the number of correct tests under $s$ and $s'$ exactly the same as the above two dot points in the proof of (a): Starting with the true defective set $s$, removing $\mathcal{J}$ from the defective set leads to a loss of $(1-2\zeta)\frac{Cn\nu e^{-\nu}\ell}{k}$ correct tests, while adding $\mathcal{J}'$ to the current $s\setminus \mathcal{J}$  leads to a gain of $G_{\mathcal{J},\mathcal{J}',1}-G_{\mathcal{J},\mathcal{J}',2}$ correct tests. Overall, we obtain $N_{\Xv,\Yv}(s')-N_{\Xv,\Yv}(s)=G_{\mathcal{J},\mathcal{J}',1}-G_{\mathcal{J},\mathcal{J}',2}-(1-2\zeta)\frac{Cn\nu e^{-\nu}\ell}{k}$, and thus $N_{\Xv,\Yv}(s')\ge N_{\Xv,\Yv}(s)$ yields (\ref{eq:greater_equal}) as claimed. 
\end{proof}

In the achievability analysis, we seek to establish the threshold for $n$ above which the MLE fails with $o(1)$ probability. Since \ref{lem:restate}(b) provides a necessary condition for the failure of MLE, we can instead   bound the probability of the events stated in Lemma \ref{lem:restate}(b).

% To establish the converse bound, we note that MLE is the optimal decoder under a uniform prior on $s$ (i.e., $s\sim \text{Unif}(\mathcal{S}_k)$) and thus it suffices to seek the condition under which     MLE fails with $1-o(1)$ probability. 

Regarding the converse part, while we stated Lemma \ref{lem:restate}(a) for all $\ell \in \big[1,\frac{k}{\log k}\big]$ for consistency with part (b), it will turn out to be sufficient to restrict attention to $\ell=1$ with a single feasible $(C,\zeta)\in (0,\infty)\times (0,1)$.  We thus specialize Lemma \ref{lem:restate}(a) to $\ell=1 $ with a simpler set of notation. 
Specifically, if $\mathcal{J}=\{j\}$ for some $j\in s$, then we will simply write $\mathcal{M}_{\mathcal{J}},\mathcal{M}_{\mathcal{J}1},\mathcal{M}_{\mathcal{J}0}$ as $\mathcal{M}_j,\mathcal{M}_{j1},\mathcal{M}_{j0}$ respectively, and their cardinalities are denoted by $M_j,M_{j1},M_{j0}$ respectively.  Similarly, when it further holds that $\mathcal{J}'=\{j'\}$ for some $j'\in [p]\setminus s$, we write $G_{\mathcal{J},\mathcal{J}',1},G_{\mathcal{J},\mathcal{J}',2},\mathcal{K}_{\ell,C,\zeta},k_{\ell,C,\zeta}$ in Lemma \ref{lem:restate} as $G_{j,j',1},G_{j,j',2},\mathcal{K}_{C,\zeta},k_{C,\zeta}$, respectively. Then, Lemma 8(a) specializes to the following. 

\begin{cor}\label{cor:restate}
    {\rm (Simplified Sufficient Condition for the   Failure of MLE)} Let $s\subset [p]$ with cardinality $k$ be the actual defective set, and consider the MLE decoder and $\ell=1$. For a feasible pair $(C,\zeta)\in(0,\infty)\times (0,1)$ such that $\frac{Cn\nu e^{-\nu}}{k}$ and $\frac{\zeta\cdot Cn\nu e^{-\nu}}{k}$ are integers, we define 
    \begin{align}\label{eq:defi_KCzeta}
        \mathcal{K}_{C,\zeta} = \Big\{j\in s:M_j = \frac{Cn\nu e^{-\nu}}{k},~M_{j0}=\frac{\zeta\cdot Cn\nu e^{-\nu}}{k}\Big\}
    \end{align}
    and let $k_{C,\zeta}:=|\mathcal{K}_{C,\zeta}|$. Given $j\in \mathcal{K}_{C,\zeta}$ and $j'\in [p]\setminus s$, we let $G_{j,j',1}$ be the number of tests in $\mathcal{N}_{01}\cup \mathcal{M}_{j1}$ that contain item $j'$, and $G_{j,j',2}$ be the number of tests in $\mathcal{N}_{00}\cup \mathcal{M}_{j0}$ that contain item $j'$. Then if for some feasible $(C,\zeta)\in (0,\infty)\times (0,1)$ we have that  $\mathcal{K}_{C,\zeta}$ is non-empty, and that there exist some $j\in\mathcal{K}_{C,\zeta}$ and some $j'\in[p]\setminus s$ such that 
    \begin{align}\label{eq:con_failure_ell1}
        G_{j,j',1}-G_{j,j',2} >(1-2\zeta)\frac{Cn\nu e^{-\nu}}{k},
    \end{align}
    then $\mathcal{L}(s')>\mathcal{L}(s)$ with $s'=(s\setminus \{j\})\cup\{j'\}$, which implies the failure of the MLE decoder (\ref{eq:mle}) . 
\end{cor}
\section{Low-$\ell$ Converse Analysis for Bernoulli Designs} \label{app:ber_con}
%\subsubsection*{\color{blue}[To decide:] Shall we turn the reduction/reformulation as a Lemma?}

Let $\mathcal{S}_k=\{S'\subset [p]:|S'|=k\}$ and recall that we consider the prior $S\sim \mathrm{Unif}(\mathcal{S}_k)$.  As noted in Appendix \ref{app:mle}, it suffices to prove the converse bound for the optimal decoder, that is, maximum likelihood estimation (MLE) given in (\ref{eq:mle}).
By symmetry with respect to re-ordering items, we can simply assume that the actual defective set is $s=[k]$,  i.e., items $1,\dotsc,k$ are defective, and items $k+1,\dotsc,p$ are non-defective.

%\subsubsection*{Some Notation}
 %We let $\mathcal{M}_j\subset [n]$ index the tests where $j$ is the only defective, $\mathcal{M}_{j0}$ index the negative tests in $\mathcal{M}_j$ (i.e., the tests in $\mathcal{M}_j$ that are flipped by noise), and $\mathcal{M}_{j1}$ index the positive tests in $\mathcal{M}_j$; we let their cardinalities be $|\mathcal{M}_j|:=M_j$, $|\mathcal{M}_{j0}|:=M_{j0}$, and  $|\mathcal{M}_{j1}|:=M_{j1}$. 
%is the number of tests where the $j$-th item appears as the only defective.
%Given $C>0$ and $\zeta\in(0,1)$ such that $\frac{Cn\nu e^{-\nu}}{k}$ and $\frac{\zeta\cdot Cn\nu e^{-\nu}}{k}$ are integers, henceforth referred to as feasible, we define  
%\begin{equation}\label{eq:defi_K_c_zeta}
 %   \mathcal{K}_{C,\zeta} = \Big\{j\in s: M_j =\frac{Cn\nu e^{-\nu}}{k},~M_{j0}=\zeta M_j\Big\},
%\end{equation}
%and let $k_{C,\zeta} = |\mathcal{K}_{C,\zeta}|$.  Similarly, let $\mathcal{N}_{0}\subset [n]$   index the tests that contain no defectives,  $\mathcal{N}_{00}$ index the negative tests in $\mathcal{N}_0$, and $\mathcal{N}_{01}$ index the positive tests in $\mathcal{N}_0$ (i.e., the tests in $\mathcal{N}_0$ that are flipped by noise). We denote their cardinalities by $|\mathcal{N}_0|=N_0$, $|\mathcal{N}_{00}|=N_{00}$, and $|\mathcal{N}_{01}|=N_{01}$. 
%Note that given  $\Xv_s$ and     $\Yv$, the index sets $\mathcal{K}_{C,\zeta},\mathcal{N}_{0},\mathcal{N}_{00},\mathcal{N}_{01},\mathcal{M}_{j},\mathcal{M}_{j0},\mathcal{M}_{j1}$ (and hence their cardinalities) are known. 

\subsubsection*{Reduction to two conditions}
For $j\in s$ and $j'\in [p]\setminus s$ we will use the notation $\mathcal{M}_j,\mathcal{M}_{j0},\mathcal{M}_{j1},\mathcal{N}_0,\mathcal{N}_{00},\mathcal{N}_{01}$ and the corresponding cardinalities  from Corollary \ref{cor:restate}. Note that these sets and their cardinalities are deterministic given $\Xv_s$ and $\Yv$.  We consider $\ell=1$, and say that $(C,\zeta)\in(0,\infty)\times(0,1)$ is \emph{feasible} if $\frac{Cn\nu e^{-\nu}}{k}$ and $\frac{\zeta\cdot Cn\nu e^{-\nu}}{k}$ are integers. 
 Given some feasible pair $(C,\zeta)$, we will consider $\mathcal{K}_{C,\zeta}=\{j\in s:M_j = \frac{Cn\nu e^{-\nu}}{k},~M_{j0}=\frac{\zeta\cdot Cn\nu e^{-\nu}}{k}\}$ with cardinality $k_{C,\zeta}=|\mathcal{K}_{C,\zeta}|$ as in (\ref{eq:defi_KCzeta}).  For a given pair $(j,j')\in \mathcal{K}_{C,\zeta}\times([p]\setminus s)$, recall that 
     $G_{j,j',1}$ denotes the number of tests in $\mathcal{N}_{01}\cup \mathcal{M}_{j1}$ that contain the   non-defective item $j'$, and $G_{j,j',2}$ denotes number of tests in $\mathcal{N}_{00}\cup \mathcal{M}_{j0}$ that contain the non-defective item $j'$; see Corollary \ref{cor:restate}.

We first invoke Corollary \ref{cor:restate} to identify sufficient conditions for the failure of MLE. In particular, MLE fails  if for some feasible   $(C,\zeta)\in(0,\infty)\times (0,1)$, 
the following two conditions \textbf{(C1)} and \textbf{(C2)} simultaneously hold:\begin{itemize}
    \item \textbf{(C1)} $k_{C,\zeta}\geq 1$ (i.e., $\mathcal{K}_{C,\zeta}\neq \varnothing$); 
    \item \textbf{(C2)} There exist $j\in \mathcal{K}_{C,\zeta}$ and $j'\in \{k+1,\dotsc,p\}$ such that 
    \begin{equation}\label{eq:failure}
        G_{j,j',1}-G_{j,j',2}>(1-2\zeta)\frac{Cn\nu e^{-\nu}}{k}.
    \end{equation} 
\end{itemize}
In the rest of the proof, we identify explicit conditions in the form of sufficient conditions on $n$ for ensuring \textbf{(C1)} and \textbf{(C2)}. We proceed to study the two conditions under a fixed feasible pair $(C,\zeta)$.

\subsubsection*{Condition for (C1) under the given $(C,\zeta)$}
We first derive a condition for ensuring $k_{C,\zeta}\geq 1$ with $1-o(1)$ probability. Fix a constant $\xi<1$ (to be chosen  close to $1$ later) and suppose $k^{\xi}$ is an integer for notational convenience. For a given $j\in s$ and any test (say, with index $i$), we define the probability of test $i$ containing only one defective $j$ as  
\begin{equation}\label{eq:P1}
    P_1 = \mathbb{P}(i\in \mathcal{M}_j) = \frac{\nu}{k}\Big(1-\frac{\nu}{k}\Big)^{k-1} = \frac{\nu e^{-\nu}(1+o(1))}{k}.
\end{equation}
Due to the i.i.d.~nature of the Bernoulli design, $\bm{M}_{\xi}:=[M_1,M_2,...,M_{k^{\xi}}]$ is determined by a multinomial distribution with $n$ trials in which each $M_i$ has probability $P_1$ (and another variable, say $M_0$, captures the remaining probability of $1-k^{\xi}P_1$).  This appears to be difficult to analyze directly; to overcome this, we let $\hat{\bm{M}}_{\xi}:=[\hat{M}_1,\hat{M}_2,...,\hat{M}_{k^{\xi}}]$ have independent components following $\mathrm{Poi}(nP_1)$, and use the following lemma (stated using generic notation).

\begin{lem}{\em (Poisson Approximation of Multinomial, \cite[Thm.~1]{arenbaev1977asymptotic})}
    \label{lem:pos_apro}
    Consider a multinomial random vector $\bm{S}=[S_0,S_1,...,S_k]\sim \mathrm{Mn}(N;[p_0,p_1,...,p_k])$ (for some $p_j\geq 0$ satisfying $\sum_{j=0}^k p_j=1$), its sub-vector $\bm{S}'=[S_1,...,S_k]$, and a Poisson vector   $\bm{T} = [T_1,T_2,...,T_k]$ with independent Poisson components $T_i\sim \mathrm{Poi}(Np_j),j\in[k]$. If $\sum_{j=1}^kp_j>0$, then as $N\to\infty$ we have %{\color{red}(\textbf{TODO:} Statement in \cite{deheuvels1988poisson} requires $p_1,...,p_k$ to be fixed, but in our application this is not true. To check.)}
    \begin{equation}
        \|\bm{S}'-\bm{T}\|_{\mathrm{TV}} = \left(\sum_{j=1}^k p_j\right)\left(\frac{1}{\sqrt{2\pi e}}+O\left(\frac{1}{\sqrt{N\sum_{j=1}^k p_j}}\right)\right). \label{eq:TV_bound}
    \end{equation}
\end{lem}

By Lemma \ref{lem:pos_apro} and the fact that $n,k \to \infty$ as $p \to \infty$, we have \begin{equation}\begin{aligned}\label{tverror}
    d_{\mathrm{TV}}(\bm{M}_\zeta,\hat{\bm{M}}_{\zeta})&=k^{\xi} P_1 \Big(\frac{1}{\sqrt{2\pi e}}+O\Big(\frac{1}{\sqrt{nk^\zeta P_1}}\Big)\Big)\\
    &= O\Big(\frac{1}{k^{1-\xi}}+\sqrt{\frac{1}{nk^{1-\xi}}}\Big)=o(1).
    \end{aligned}
\end{equation} 
In addition, we have 
\begin{align}
    k_{C,\zeta}&=|\mathcal{K}_{C,\zeta}|\geq |\mathcal{K}_{C,\zeta}\cap [k^{\xi}]|=\sum_{j=1}^{k^{\xi}}\mathbbm{1}\Big(M_j = \frac{Cn\nu e^{-\nu}}{k},~M_{j0}=\frac{\zeta Cn\nu e^{-\nu}}{k}\Big).%\\
%&=\sum_{j=1}^{k^{\xi}}\mathbbm{1}\Big(M_j= \frac{Cn\nu e^{-\nu}}{k},~\text{Bin}\Big(\frac{Cn\nu e^{-\nu}}{k},\rho\Big) = \frac{\zeta Cn\nu e^{-\nu}}{k}\Big).
\end{align}
By the definition of $d_{\text{TV}}(\cdot,\cdot)$, we can further replace $M_j$ by the more convenient independent Poisson variables $\{\hat{M}_j:j\in [k^{\xi}]\}$ with $o(1)$ difference in probability: 
\begin{align}
  \label{eq:313}  
  \mathbb{P}(k_{C,\zeta}=0) &\leq \mathbb{P}\left(\sum_{j=1}^{k^{\xi}}\mathbbm{1}\Big(M_j= \frac{Cn\nu e^{-\nu}}{k},~M_{j0} = \frac{\zeta Cn\nu e^{-\nu}}{k}\Big) = 0\right)\\\label{eq:poi_approx}
    &= \mathbb{P}\left(\sum_{j=1}^{k^{\xi}}\mathbbm{1}\Big(\hat{M}_j= \frac{Cn\nu e^{-\nu}}{k},~M_{j0} = \frac{\zeta Cn\nu e^{-\nu}}{k}\Big) = 0\right) +o(1)\\ \label{eq:inde1}
    &= \left[1-\mathbb{P}\left(\hat{M}_j= \frac{Cn\nu e^{-\nu}}{k},~M_{j0} = \frac{\zeta Cn\nu e^{-\nu}}{k}\right)\right]^{k^{\xi}}+o(1)\\\label{eq:inde2}
    &=\left[1-\mathbb{P}\left(\hat{M}_j= \frac{Cn\nu e^{-\nu}}{k}\right)\mathbb{P}\left( \text{Bin}\Big(\frac{Cn\nu e^{-\nu}}{k},\rho\Big)=\frac{\zeta Cn\nu e^{-\nu}}{k}\right)\right]^{k^{\xi}}+o(1), 
\end{align}
where  (\ref{eq:inde1}) holds because the $\hat{M}_j$ are independent, and the $\frac{Cn\nu e^{-\nu}}{k}$ tests in $\mathcal{M}_{j}$ are  disjoint for different $j\in [k^{\xi}]$ (by the definition of $\mathcal{M}_j$), so the $k^{\xi}$ summands in (\ref{eq:poi_approx}) are independent; then, (\ref{eq:inde2}) holds because $\hat{M}_j$ is only related to the randomness of $\Xv_s$, while given $M_j = \frac{Cn\nu e^{-\nu}}{k}$, $M_{j0}\sim \text{Bin}(\frac{Cn\nu e^{-\nu}}{k},\rho)$ is only related to the randomness of $\Zv$. Since $\hat{M}_j\sim\text{Poi}(nP_1)$, we have 
\begin{align}
  \label{eq:substi1}  \mathbb{P}\left(\hat{M}_j = \frac{Cn\nu e^{-\nu}}{k}\right)  &= e^{-nP_1}\frac{(nP_1)^{\frac{Cn\nu e^{-\nu}}{k}}}{\big(\frac{Cn\nu e^{-\nu}}{k}\big)!}\\&= \frac{\exp\big(-\frac{n\nu e^{-\nu}(1+o(1))}{k}\big)}{(C(1+o(1)))^{\frac{Cn\nu e^{-\nu}}{k}}}\cdot \frac{\big(\frac{Cn\nu e^{-\nu}}{k}\big)^{\frac{Cn\nu e^{-\nu}}{k}}}{\big(\frac{Cn\nu e^{-\nu}}{k}\big)!}\\
  &\geq \exp \left(-\frac{n\nu e^{-\nu}}{k}\big(C\log C +1+o(1)\big)\right)\cdot \exp\left(\frac{Cn\nu e^{-\nu}}{k}(1+o(1))\right) \label{eq:stirling1}\\
  & = \exp\left(-\frac{n\nu e^{-\nu}}{k}\big(C\log C -C+1+o(1)\big)\right),\label{eq:bound_need1}
\end{align}
where (\ref{eq:substi1}) follows from (\ref{eq:P1}), and (\ref{eq:stirling1}) follows from Stirling's approximation (specifically, $\frac{w^w}{w!}\geq \frac{e^{w-1}}{w}=e^{w(1+o(1))}$ for $w\to \infty$, with $w=\frac{Cn\nu e^{-\nu}}{k}$). Moreover, by anti-concentration ((\ref{eq:anti2}) in Lemma \ref{binoconcen}), we have 
\begin{align}
    \mathbb{P}\left(\text{Bin}\Big(\frac{Cn\nu e^{-\nu}}{k},\rho\Big) = \frac{\zeta Cn\nu e^{-\nu}}{k}\right)&\geq \frac{1}{\sqrt{\frac{2Cn\nu e^{-\nu}}{k}}}\exp\left(-\frac{Cn\nu e^{-\nu}}{k}D(\zeta\|\rho)\right)\\&=\exp\left(-\frac{n\nu e^{-\nu}}{k}\big(C\cdot D(\zeta\|\rho)+o(1)\big)\right),\label{eq:absorb1}
\end{align}
where (\ref{eq:absorb1}) holds because $\frac{n}{k}=\Theta(\log k)\to \infty$ allows us to incorporate the leading factor into the $o(1)$ term in the exponent; analogous simplifications will be used in the subsequent analysis. Combining the bounds in (\ref{eq:bound_need1}) and (\ref{eq:absorb1}), we obtain \begin{align}
    &\mathbb{P}\left(\hat{M}_j= \frac{Cn\nu e^{-\nu}}{k}\right)\mathbb{P}\left(\text{Bin}\Big(\frac{Cn\nu e^{-\nu}}{k},\rho\Big) = \frac{\zeta Cn\nu e^{-\nu}}{k}\right)\\
    &~~~~~~~~~~~~~~~~~~~\geq \exp\left(-\frac{n\nu e^{-\nu}}{k}\big[C\log C-C+C\cdot D(\zeta\|\rho)+1+o(1)\big]\right):=P_2. \label{eq:P2_bernoulli}
\end{align}
Substituting this into  (\ref{eq:inde2}), we obtain \begin{align}
    \mathbb{P}(k_{C,\zeta}=0) &\leq (1-P_2)^{k^{\xi}}+o(1)%\\
   % &\leq \exp\big(k^{\xi}\cdot \log(1-P_2)\big)+o(1)
   \\&\leq \exp\big(-k^{\xi}P_2\big)+o(1).
\end{align}
Thus, to guarantee $\mathbb{P}(k_{C,\zeta}=0)=o(1)$, or equivalently $\mathbb{P}(k_{C,\zeta}> 0)=1-o(1)$, it suffices to have $k^{\xi}P_2\to \infty$. Defining the shorthand \begin{equation}\label{eq:f1_ber_con}
    f_1(C,\zeta,\rho):= C\log C-C+C\cdot D(\zeta\|\rho)+1,
\end{equation}
and noting that $\eta'_1 \log\frac{p}{k}\to\infty $ for any given $\eta'_1>0$, the condition 
$k^{\xi}P_2\to \infty$ can be ensured by  
\begin{equation}
    \label{eq:Bconver_con1}
    \xi \log k- \frac{n\nu e^{-\nu}}{k}\big(f_1(C,\zeta,\rho)+o(1)\big)\geq \eta'_1 \log\frac{p}{k}.
\end{equation}
Moreover, by $k=\Theta(p^\theta)$ we have $\log k \sim \frac{\theta}{1-\theta}\log\frac{p}{k}$. Substituting this into (\ref{eq:Bconver_con1}), letting $\xi$ be arbitrarily close to $1$ and rearranging, we find that \textbf{(C1)} holds with $1-o(1)$ probability if the following holds for some $\eta_1>0$:
\begin{equation}\label{eq:ber_fail_1}
    n\leq  \frac{(1-\eta_1)\frac{\theta}{1-\theta}k\log\frac{p}{k}}{\nu e^{-\nu}\big(f_1(C,\zeta,\rho)+o(1)\big)}.
\end{equation}
% holds for some $\eta_1>0$% possibly different from the one in (\ref{eq:Bconver_con1}).

\subsubsection*{Condition for (C2) under the given $(C,\zeta)$}

Given   $k_{C,\zeta}\geq 1$ and fixing an arbitrary index $j\in\mathcal{K}_{C,\zeta}$, in this part we   identify a sufficient condition for \textbf{(C2)}. %which is mainly based on the randomness of $\Xv_{S^c}$ by conditioning on $\Xv_S,\Zv$. 
Most of our analysis will be conditioned on $(\Xv_s,\Yv)$ and only based on the randomness of $\Xv_{[p]\setminus s}$, but we first deduce the high-probability behavior of the quantities  $N_{01}=|\mathcal{N}_{01}|$ and $N_{00}=|\mathcal{N}_{00}|$ depending on randomness of $(\Xv_s,\Yv)$. Specifically, $\mathcal{N}_{01}$ (resp., $\mathcal{N}_{00}$) consists of the positive (resp., negative) tests containing no defective, and since the probability of a specific test containing no defective is given by $(1-\frac{\nu}{k})^k= e^{-\nu}(1+o(1))$, we have $N_{01}\sim \text{Bin}(n,\rho e^{-\nu}(1+o(1)))$  and $N_{00}\sim \text{Bin}(n,(1-\rho)e^{-\nu}(1+o(1)))$. Then, by Hoeffding's inequality, we have that
\begin{equation}\label{eq:N00N01}
   \mathscr{A}_1:=\big\{ N_{01}= \rho ne^{-\nu}(1+o(1)),~N_{00}=(1-\rho)ne^{-\nu}(1+o(1))\big\} 
\end{equation}
holds with $1-o(1)$ probability. We will subsequently suppose that we are on the event $\mathscr{A}_1$, which is justified by a simple union bound.

For the specific $j\in\mathcal{K}_{C,\zeta}$ we have $M_{j0}=\frac{\zeta \cdot Cn\nu e^{-\nu}}{k}=o(n)$ and $M_{j1}=\frac{(1-\zeta)\cdot Cn\nu e^{-\nu}}{k}=o(n)$, and thus they are negligible compared to $N_{01},N_{00}$ specified in $\mathscr{A}_1$ in  (\ref{eq:N00N01}): 
\begin{gather}
    |\mathcal{N}_{01}\cup \mathcal{M}_{j1}|=N_{01}+M_{j1} = \rho ne^{-\nu}(1+o(1)),\\
     |\mathcal{N}_{00}\cup \mathcal{M}_{j0}|=N_{00}+M_{j0}=(1-\rho)ne^{-\nu}(1+o(1)).
\end{gather}
Note that these cardinalities correspond to counts of ``relevant'' tests with positive and negative outcomes respectively, and recall from Corollary \ref{cor:restate} that $G_{j,j',1}$ and $G_{j,j',2}$ count how many such tests include the non-defective $j'$ (and thus impact the likelihood function of $(S \setminus \{j\}) \cup \{j'\}$ relative to $S$).

For a specific non-defective $j' \in \{k+1,\dotsc,n\}$, conditioning on $(\Xv_s,\Yv)$ satisfying $\mathscr{A}_1$, the randomness of the $j'$-th column of $\Xv$ (i.e., the placements of the $j'$-th item) gives
\begin{align}
    &G_{j,j',1}\sim \text{Bin}\Big(\rho n e^{-\nu}(1+o(1)),\frac{\nu}{k}\Big),\\
    &G_{j,j',2}\sim \text{Bin}\Big((1-\rho)ne^{-\nu}(1+o(1)),\frac{\nu}{k}\Big).
\end{align}
In order to establish a lower bound on the probability of the ``bad event'' (\ref{eq:failure}), we introduce a parameter $d$ (to be chosen later) satisfying 
\begin{align}
    \label{eq:ber_con_d_range}
    d > \max \Big\{\frac{C(1-2\zeta)}{\rho},0\Big\}
\end{align}
and consider the event $G_{j,j',1}=\frac{d\rho \nu e^{-\nu}n}{k}$ (assumed to be positive integer).  Thus, $d$ parametrizes the number of ``relevant'' positive tests that $j'$ is placed in.  The idea is that while the condition \eqref{eq:failure} could hold for a wide range of $d$ values, it suffices for our purposes to identify a single choice that dominates the error probability.

Accordingly, we use the anti-concentration bound (\ref{eq:anti1}) from Lemma \ref{binoconcen} to obtain the following:
\begin{align}
    &\mathbb{P}\left(G_{j,j',1}-G_{j,j',2}>(1-2\zeta)\frac{Cn\nu e^{-\nu}}{k}\right)\\
    &\geq\mathbb{P}\left(G_{j,j',1}=\frac{d\rho \nu e^{-\nu}n}{k}\right)\mathbb{P}\left(G_{j,j',2}<\frac{n\nu e^{-\nu}}{k}\big(d\rho - (1-2\zeta)C\big)\right)\\
    &\geq  \Theta \Big(\frac{1}{\sqrt{n/k}}\Big)\exp \left(-\rho ne^{-\nu}(1+o(1))\cdot D\Big(\frac{d(1+o(1))\nu}{k}\big\|\frac{\nu}{k}\Big)\right)
    \\ &\quad\times\Theta \Big(\frac{1}{\sqrt{n/k}}\Big)\exp\left(-(1-\rho)ne^{-\nu}(1+o(1))\cdot D\Big(\frac{(d\rho - (1-2\zeta)C)(1+o(1))}{1-\rho}\frac{\nu}{k}\big\|\frac{\nu}{k}\Big)\right)\label{eq:problem1}\\ 
    &=\exp \left(-\frac{n\nu e^{-\nu}}{k}\big[g(C,\zeta,d,\rho)+o(1)\big]\right), \label{eq:problem2}
\end{align}
where in (\ref{eq:problem2}) we define $g(C,\zeta,d,\rho)$ that depends on $(C,\zeta)$ only through $A:=C(1-2\zeta)$ as follows: 
\begin{equation}\label{eq:g_ber_con}
    g(C,\zeta,d,\rho) = \rho d\log d+\big(\rho d-A\big)\log\Big(\frac{\rho d-A}{1-\rho}\Big) + 1-2\rho d+A,
\end{equation}
and (\ref{eq:problem2}) itself follows from $D(\frac{a\nu}{k}\|\frac{\nu}{k})=\frac{\nu}{k}(a\log a-a+1+o(1))$ and writing the pre-factors $\Theta((\frac{n}{k})^{-1/2})$ as $\exp(\frac{o(1)n}{k})$.

Next, we seek to find the choice of $d$ that gives the strongest bound.  Differentiating $g(C,\zeta,d,\rho)$ with respect to $d$, we obtain \begin{equation}\label{eq:338}
    \frac{\partial g}{\partial d} = \rho \log \Big(\frac{d(\rho d-A)}{1-\rho}\Big),
\end{equation}
which is zero if and only if $\rho d^2-Ad -(1-\rho) = 0$.  Thus, setting $\frac{\partial g}{\partial d}=0$, we readily obtain that
\begin{align}
\label{eq:f2_ber_con}&\min_{d>\max\{0, \frac{A}{\rho}\}}~g(C,\zeta,d,\rho)=g(C,\zeta,d^*,\rho):= f_2(C,\zeta,\rho),\\
    &\text{where }d^* = \frac{A+\sqrt{A^2+4\rho(1-\rho)}}{2\rho}.\label{eq:340}
\end{align}
Therefore, we let $d= d^*+o(1)$ (here the $o(1)$ term only serves to ensure that $\frac{d\rho \nu e^{-\nu}n}{k}$ is a positive integer) to get from (\ref{eq:problem2}) the strongest bound: 
\begin{equation}
    \mathbb{P}\left(G_{j,j',1}-G_{j,j',2}>(1-2\zeta)\frac{Cn\nu e^{-\nu}}{k}\right) \geq \exp\left(-\frac{n\nu e^{-\nu}}{k}\big(f_2(C,\zeta,\rho)+o(1)\big)\right).
\end{equation}
Recall  that to ensure the failure of MLE, with a chosen $j\in\mathcal{K}_{C,\zeta}$ at hand, the condition (\ref{eq:failure}) should hold for some $k+1\leq j'\leq p$. The placements of the different non-defective items are independent, and by implicitly conditioning on $(\Xv_s,\Yv)$ we have
\begin{align}
    \label{eq:ber_conclu_c2_start}
    &\mathbb{P}\left(\exists j \in \{k+1,\dotsc, p\}\text{ s.t. }G_{j,j',1}-G_{j,j',2}>(1-2\zeta)\frac{Cn\nu e^{-\nu}}{k}\right)\\
    &\quad =1- \mathbb{P}\left(\forall j \in \{k+1,\dotsc, p\},~G_{j,j',1}-G_{j,j',2}\leq (1-2\zeta)\frac{Cn\nu e^{-\nu}}{k}\right)\\
    &\quad =1-\left[\mathbb{P}\left(G_{j,j',1}-G_{j,j',2}\leq (1-2\zeta)\frac{Cn\nu e^{-\nu}}{k}\right)\right]^{p-k}\\
    &\quad \geq 1-\left[1-\exp\left(-\frac{n\nu e^{-\nu}}{k}\big(f_2(C,\zeta,\rho)+o(1)\big)\right)\right]^{p-k}.\label{eq:345}
\end{align}
Moreover, using $1-a \le e^{-a}$, we have 
\begin{align}
    \left[1-\exp\left(-\frac{n\nu e^{-\nu}}{k}\big(f_2(C,\zeta,\rho)+o(1)\big)\right)\right]^{p-k}%\\
    %=&\exp\left((p-k)\log\Big[1-\exp\Big(-\frac{n\nu e^{-\nu}}{k}\big(f_2(C,\zeta,\rho)+o(1)\big)\Big)\Big]\right)
    \leq\exp\left(-(p-k)\exp\Big(-\frac{n\nu e^{-\nu}}{k}\big(f_2(C,\zeta,\rho)+o(1)\big)\Big)\right).
\end{align}
Thus, to ensure that (\ref{eq:345}) behaves as $1-o(1)$, it suffices to have $\log(p-k)-\frac{n\nu e^{-\nu}}{k}(f_2(C,\zeta,\rho)+o(1))\to\infty$. Substituting $\log(p-k)\sim \log p\sim \frac{1}{1-\theta}\log\frac{p}{k}$, we obtain that if     \begin{equation}\label{eq:ber_con_c2_first}
    \frac{1}{1-\theta}\log\frac{p}{k}-\frac{n\nu e^{-\nu}}{k}\big(f_2(C,\zeta,\rho)+o(1)\big)\geq \eta'_2\log\frac{p}{k} 
\end{equation}
holds for some $ \eta'_2>0$, then \textbf{(C2)} holds with $1-o(1)$ probability.  Re-arranging, we obtain the following condition on $n$ for arbitrarily small $\eta_2 > 0$:
\begin{equation}\label{eq:ber_fail_2}
    n\leq  \frac{(1-\eta_2)\frac{1}{1-\theta}k\log\frac{p}{k}}{\nu e^{-\nu}\big(f_2(C,\zeta,\rho)+o(1)\big)}.
\end{equation}
% for some $\eta_2>0$ possibly different from the one in (\ref{eq:ber_con_c2_first}). 

\subsubsection*{Wrapping up}
By merging $\eta_1,~\eta_2$ and $o(1)$ into a single parameter $\eta$,
  the statement  ``(\ref{eq:ber_fail_1}) and (\ref{eq:ber_fail_2}) simultaneously hold'' can be written as 
\begin{equation}
   n\leq (1-\eta)\frac{k\log\frac{p}{k}}{(1-\theta)\nu e^{-\nu}}\frac{1}{\max\{\frac{1}{\theta}f_1(C,\zeta,\rho) ,f_2(C,\zeta,\rho)\}},
\end{equation}
where $\eta>0$ is arbitrarily small.   
Under this condition, the optimal decoder MLE fails with $1-o(1)$ probability. Further optimizing over $(C,\zeta)$  to render the tightest converse bound, we arrive at the threshold
\begin{equation}\label{eq:352}
   \frac{k\log\frac{p}{k}}{(1-\theta)\nu e^{-\nu}}\frac{1}{\min_{C>0,\zeta\in (0,1)}\max\{\frac{1}{\theta}f_1(C,\zeta,\rho) ,f_2(C,\zeta,\rho)\}},
\end{equation}
where $f_1(C,\zeta,\rho)$ is given in (\ref{eq:f1_ber_con}), and $f_2(C,\zeta,\rho)$ is defined in (\ref{eq:f2_ber_con}) and (\ref{eq:g_ber_con}).  This establishes the second term in \eqref{eq:ber_threshold}.
\section{Low-$\ell$ Achievability Analysis for Bernoulli Designs} \label{app:ber_achi}

% We will use the original information density approach to deal with the case of $\ell>\frac{k}{\log k}$, which gives rise to the capacity bound of $\frac{k\log\frac{p}{k}}{H_2(e^{-\nu}\star\rho)-H_2(\rho)}$ {\color{blue}[TODO: To Fill]}. 
Recall that in this part of the analysis, we restrict our attention to  $\ell\in[1,\frac{k}{\log k}]$. Our strategy is to show that there is $o(1)$ probability of the restricted MLE decoder (\ref{eq:mle_restricted}) failing. This is more involved than the converse analysis because we need to consider the failure events associated with {\it all}   $\ell\in [1,\frac{k}{\log k}]$ and the corresponding feasible $(C,\zeta)$ pairs, and finally bound the overall failure probability.  Due to the symmetry of the design with respect to re-ordering items, it suffices to again consider the fixed defective set $s=[k]$, meaning that items $1,2,\dotsc,k$ are defective and items $k+1,k+2,\dotsc,p$ are non-defective.

\subsubsection*{Reduction to two conditions}
 
 For $\mathcal{J}\subset s$ with $|\mathcal{J}|=\ell$, we will use the notations $\mathcal{M}_{\mathcal{J}},\mathcal{M}_{\mathcal{J}0},\mathcal{M}_{\mathcal{J}1},\mathcal{N}_0,\mathcal{N}_{00},\mathcal{N}_{01}$ and their corresponding cardinalities (e.g., $M_{\mathcal{J}}=|\mathcal{M}_{\mathcal{J}}|$, $N_0=|\mathcal{N}_0|$) from Lemma \ref{lem:restate}. 
 Conditioned on $\Xv_s$ and $\Yv$,  these sets and their cardinalities are deterministic. Given $\ell\in[1,\frac{k}{\log k}]$, we say $(C,\zeta)\in[0,\infty)\times[0,1]$ is \emph{feasible} if $\frac{Cn\nu e^{-\nu}\ell}{k}$ and $\frac{\zeta\cdot Cn\nu e^{-\nu}\ell}{k}$ are integers, subject to the restriction that $(C,\zeta)=(0,0)$ is the only feasible pair with $C=0$.   Given a feasible pair $(C,\zeta)$, we define 
    $$\mathcal{K}_{\ell,C,\zeta}=\Big\{\mathcal{J}\subset s,|\mathcal{J}|=\ell:M_{\mathcal{J}}=\frac{Cn\nu e^{-\nu}\ell}{k},~M_{\mathcal{J}0}=\frac{\zeta\cdot Cn\nu e^{-\nu}\ell}{k}\Big\}$$ 
 as in (\ref{eq:defi_KellCzeta}).     For given $(\mathcal{J},\mathcal{J}')\in \mathcal{K}_{\ell,C,\zeta}\times ([p]\setminus s)$ with $|\mathcal{J}'|=\ell$, recall that  $G_{\mathcal{J},\mathcal{J}',1}$ denotes the number of tests in $\mathcal{N}_{01}\cup \mathcal{M}_{\mathcal{J}1}$ that contain some item from $\mathcal{J}'$, and ${G}_{\mathcal{J},\mathcal{J}',2}$ denotes the number of tests in $\mathcal{N}_{00}\cup \mathcal{M}_{\mathcal{J}0}$ that contain some item from $\mathcal{J}'$. All of these definitions are consistent with Lemma \ref{lem:restate}.

To bound the probability that restricted MLE fails, we first invoke Lemma \ref{lem:restate}(b) to obtain 
a necessary condition for failure. 
In particular, Lemma \ref{lem:restate}(b) implies the following: If restricted MLE fails and returns some $s'$ satisfying $|s\setminus s'|\in [1,\frac{k}{\log k}]$, then there exists some feasible pair $(C,\zeta)\in [0,\infty)\times [0,1]$ with respect to the specific $\ell:=|s\setminus s'|$,  
such that \textbf{(C1)} and \textbf{(C2)} below
simultaneously hold:
%it suffices that for all $1\leq \ell\leq \frac{k}{\log k}$ and  for all feasible $(C,\zeta)\in[0,\infty)\times [0,1]$ under this specific $\ell$,   {\it either} \textbf{(C1)} {\it or} \textbf{(C2)} below holds: 
\begin{itemize}
    \item \textbf{(C1)} $k_{\ell,C,\zeta}\ge 1$ (i.e., $\mathcal{K}_{\ell,C,\zeta}\neq \varnothing$); 
    \item \textbf{(C2)} There exist some $\mathcal{J}\in\mathcal{K}_{\ell,C,\zeta}$ and some  $\mathcal{J}'\subset [p]\setminus s$ with $|\mathcal{J}'|=\ell$ such that
    \begin{equation}
        \label{eq:35411}G_{\mathcal{J},\mathcal{J}',1}-G_{\mathcal{J},\mathcal{J}',2} \ge  (1-2\zeta)\frac{Cn\nu e^{-\nu}\ell}{k}.
    \end{equation}    
\end{itemize}
Therefore, it remains to identify the thresholds for $n$ above which there is $o(1)$ probability of \textbf{(C1)} and \textbf{(C2)} simultaneously holding. In the following, we will consider a fixed  $\ell\in[1,\frac{k}{\log k}]$, and we will later apply a union bound to cover all $\ell$. 

\subsubsection*{The case of $(C,\zeta)\in (C_0,\infty)\times [0,1]$ under a given $\ell$}

In general, we can have integer-valued $\frac{Cn\nu e^{-\nu}\ell}{k}$ with arbitrarily large $C$. In this part, we first address the cases of $C> C_0$ for some sufficiently large absolute constant $C_0$, so that we can restrict to $C=O(1)$ later.  For specific $\mathcal{J}\subset s$ with $|\mathcal{J}|=\ell$,
we let the probability of a test   belonging to $\mathcal{M}_{\mathcal{J}}$ be
\begin{equation}
    P_1:= \Big(1-\big(1-\frac{\nu}{k}\big)^\ell\Big)\Big(1-\frac{\nu}{k}\Big)^{k-\ell} = \frac{\ell \nu e^{-\nu}(1+o(1))}{k},
\end{equation}
where the last equality follows from $\ell=o(k)$. We thus have  $M_{\mathcal{J}}\sim\text{Bin}(n,P_1)$, and under the given $\ell$ we can bound the probability 
\begin{align}\label{eq:boundk1}
&\mathbb{P}\Big(\textbf{(C1)}\text{ and }\textbf{(C2)} \text{ hold for some feasible }(C,\zeta)\in (C_0,\infty)\times [0,1]\Big)
\\ \label{eq:C1_only}
&\le \mathbb{P}\Big(\mathcal{K}_{\ell,C,\zeta}\neq \varnothing \text{ holds for some feasible } (C,\zeta)\in (C_0,\infty)\times [0,1]\Big)\\\label{eq:K_first_con}
    &\leq\mathbb{P}\Big(\exists~ \mathcal{J}\subset s\text{ with }|\mathcal{J}|=\ell,~\text{s.t. }M_{\mathcal{J}}> \frac{C_0n\nu e^{-\nu}\ell}{k}\Big)\\
    &\leq \binom{k}{\ell}\mathbb{P}\left(\text{Bin}\Big(n,\frac{\ell \nu e^{-\nu}(1+o(1))}{k}\Big)>\frac{C_0n\nu e^{-\nu}\ell}{k}\right)\label{eq:union_calJ}\\
    &\leq\exp\left((1+o(1))\ell \log k - \frac{\ell n}{k}\nu e^{-\nu}\big(C_0\log C_0-C_0+1+o(1)\big)\right)\label{eq:361}\\ &\leq  k^{-5}, \label{eq:boundk2}
\end{align} 
where in (\ref{eq:C1_only}) we get an upper bound by only considering \textbf{(C1)} and not \textbf{(C2)}, in (\ref{eq:K_first_con}) we further relax $\mathcal{K}_{\ell,C,\zeta}\neq \varnothing$ to the first condition of $\mathcal{K}_{\ell,C,\zeta}$ (namely $M_{\mathcal{J}}=\frac{Cn\nu e^{-\nu}\ell}{k}$; see (\ref{eq:defi_KellCzeta})), in (\ref{eq:union_calJ}) we apply a union bound over all $\mathcal{J}\subset s$ with $|\mathcal{J}|=\ell$,
in (\ref{eq:361}) we use the Chernoff bound (see (\ref{eq:chernoff2}) in Lemma \ref{binoconcen}), and \eqref{eq:boundk2} holds by letting $C_0=\Theta(1)$ be sufficiently large (recall that we have $\frac{n}{k}=\Theta(\log k)$). Therefore,   for a given $\ell\in[1,\frac{k}{\log k}]$, \textbf{(C1)} and  \textbf{(C2)} simultaneously hold for some feasible $(C,\zeta)\in (C_0,\infty)\times [0,1]$ with probability no higher than $k^{-5}$.

\subsubsection*{Condition for (C1) under a given $(\ell,C,\zeta)$}
In this part, we identify the more explicit condition for \textbf{(C1)} under the given $\ell$ and a specific feasible $(C,\zeta)$ with $C\leq C_0$. 
We seek to bound $k_{\ell,C,\zeta}$. Our strategy is to first calculate $\mathbbm{E}k_{\ell,C,\zeta}$ and then deduce the high-probability behavior of $k_{\ell,C,\zeta}$ via Markov's inequality. 
Since $k_{\ell,C,\zeta}=\sum_{\mathcal{J}\subset s,|\mathcal{J}|=\ell}\mathbbm{1}(M_{\mathcal{J}}=\frac{Cn\nu e^{-\nu}\ell}{k},~M_{\mathcal{J}0}=\zeta M_{\mathcal{J}})$, we have \begin{align}
    \mathbbm{E}k_{\ell,C,\zeta}&= \binom{k}{\ell}\mathbb{P}\left(\text{for fixed }\mathcal{J}\subset s\text{ with }|\mathcal{J}|=\ell,~M_{\mathcal{J}}=\frac{Cn\nu e^{-\nu}\ell}{k},~M_{\mathcal{J}0}=\zeta M_{\mathcal{J}}\right)\\
    &\label{eq:ber_Mj}=\binom{k}{\ell}\mathbb{P}\left(\text{Bin}\big(n,P_1\big)=\frac{Cn\nu e^{-\nu}\ell}{k}\right)\mathbb{P}\left(\text{Bin}\Big(\frac{Cn\nu e^{-\nu}\ell}{k},\rho\Big)=\frac{\zeta\cdot Cn\nu e^{-\nu}\ell}{k}\right)\\\label{eq:3655}
    &\leq \binom{k}{\ell}\exp\left(-n\cdot D\Big(\frac{C \nu e^{-\nu}\ell}{k}\big\|\frac{\nu e^{-\nu}\ell(1+o(1))}{k}\Big)-\frac{Cn\nu e^{-\nu}\ell}{k}D(\zeta\|\rho)\right)\\
    &=\binom{k}{\ell}\exp \Big(-\frac{n\nu e^{-\nu}\ell}{k}\Big[f_1(C,\zeta,\rho)+o(1)\Big]\Big), \label{eq:364}
\end{align}
where in (\ref{eq:3655}) we use the Chernoff bound (see (\ref{eq:chernoff1}) in Lemma \ref{binoconcen}), and in (\ref{eq:364}) we define 
\begin{equation}\label{eq:f1_ber_achi}
    f_1(C,\zeta,\rho):= C\log C-C+C\cdot D(\zeta\|\rho)+1,
\end{equation}
and then use $D(\frac{C\nu e^{-\nu}\ell}{k}\|\frac{\nu e^{-\nu}\ell(1+o(1))}{k})=\frac{\nu e^{-\nu}\ell}{k}(C\log C-C+1+o(1))$.
We note that $f_1(C,\zeta,\rho)$ coincides with (\ref{eq:f1_ber_con}) in the proof of the converse bound for Bernoulli designs.

We now apply Markov's inequality to bound  $k_{\ell,C,\zeta}$. We divide this step into two cases, namely $1\leq \ell\leq \log k$  and  $\log k<\ell\leq \frac{k}{\log k}$.

\textbf{The case of $1\leq \ell\leq \log k$:} For specific $(\ell,C,\zeta)$, Markov's inequality gives   \begin{equation}\label{eq:365}
    \mathbb{P}\Big(k_{\ell,C,\zeta}\geq (\ell\log k)^5 \mathbbm{E}[k_{\ell,C,\zeta}]\Big) \leq (\ell\log k )^{-5}.
\end{equation}
%Since $C=O(1)$ and $\frac{n\ell }{k}=\Theta(\ell\log k)$, there are $O(\ell\log k)$ possibilities of $C$ such that $\frac{Cn\nu e^{-\nu}\ell}{k}$ is an integer; also, for a specific $0<C\le C_0$, there are at most $O(\ell\log k)$ possibilities of $\zeta\in[0,1]$ such that $\frac{\zeta\cdot Cn\nu e^{-\nu}\ell}{k}$ is an integer. Therefore, for specific $\ell$, we only need to  consider $O([\ell\log k]^2)$ choices of $(C,\zeta)$, and so based on (\ref{eq:365}), a union bound over $(C,\zeta)$ and further over $1\leq \ell\leq \log k$ yields that, with $1-O((\log k)^{-2})= 1-o(1)$ probability,
Using the bound on $\mathbbm{E}k_{\ell,C,\zeta}$ from (\ref{eq:364}), we know that   with probability at least $1-(\ell\log k)^{-5}$, it holds that
\begin{align}
    k_{\ell,C,\zeta}< (\ell\log k)^5\mathbbm{E}k_{\ell,C,\zeta}\leq & (\ell\log k)^5
    \binom{k}{\ell}\exp \Big(-\frac{n\nu e^{-\nu}\ell}{k}\Big[f_1(C,\zeta,\rho)+o(1)\Big]\Big)
    \\\leq & \exp\left((1+o(1))\ell \log\frac{k}{\ell}- \frac{n\nu e^{-\nu}\ell}{k}\Big[f_1(C,\zeta,\rho)+o(1)\Big]\right).
\end{align}
%holds for any $1\leq \ell \leq \log k$, $C\leq C_0$, $0\leq \zeta\le 1$. 

\textbf{The case  of $\log k<\ell\leq \frac{k}{\log k}$:}
  For specific $(\ell,C,\zeta)$, Markov's inequality gives 
    \begin{equation}\label{eq:367}
    \mathbb{P}\Big(k_{\ell,C,\zeta}\geq   k^5 \mathbbm{E}[k_{\ell,C,\zeta}]\Big) \leq k^{-5}.
\end{equation}
%Thus, by a union bound on $O([\ell\log k]^2)=O(k^2)$ possibilities of $(C,\zeta)$ and $\log k<\ell\leq \frac{k}{\log k}$, we obtain that with probability at least $1-O(k^{-2})=1-o(1)$ probability, 
Combining with (\ref{eq:364}), we know that  with probability at least $1-k^{-5}$, it holds that  
\begin{align}
    k_{\ell,C,\zeta}<k^5 \mathbbm{E}k_{\ell,C,\zeta} \leq &k^5\binom{k}{\ell}\exp \Big(-\frac{n\nu e^{-\nu}\ell}{k}\Big[f_1(C,\zeta,\rho)+o(1)\Big]\Big)\\
    \leq & \exp \Big((1+o(1))\ell\log\frac{k}{\ell}-\frac{n\nu e^{-\nu}\ell}{k}\Big[f_1(C,\zeta,\rho)+o(1)\Big]\Big).
\end{align}
Combining the above two cases, we obtain the following implication for given $(\ell,C,\zeta)$ for any $\eta_1 > 0$ and sufficiently large $k$:
\begin{align}\label{eq:C1_threshold}
    &\frac{n\nu e^{-\nu}\ell}{k}\big[f_1(C,\zeta,\rho)+o(1)\big]\geq (1+\eta_1)\ell\log\frac{k}{\ell} \\ %~\text{for some }\eta_1>0\\
    &\Longrightarrow ~
    \mathbb{P}\big(k_{\ell,C,\zeta}\ge 1\big)\le \hat{P}_\ell,~~\text{where }\hat{P}_\ell = \begin{cases}
        (\ell\log k)^{-5},~~\text{if }1\le \ell\le \log k\\
        ~~~~~~k^{-5},~~~~\text{if }\log k<\ell\le \frac{k}{\log k}.
    \end{cases}\label{eq:c1_condition}
\end{align}
%where we implicitly assume that $k$ is sufficiently large. 
Therefore, if (\ref{eq:C1_threshold}) holds, then \textbf{(C1)} holds for some given $(\ell,C,\zeta)$ with probability at most $\hat{P}_\ell$.

\subsubsection*{Condition for (C2) for a given $(\ell,C,\zeta)$}

We again consider a given $\ell$ and a specific feasible pair $(C,\zeta)$ with $C\leq C_0$, and now switch to analyzing \textbf{(C2)}.
For $\mathcal{J}'\subset [p]$ with $|\mathcal{J}'|=\ell$, the probability of a test containing some item from $\mathcal{J}'$ is: \begin{equation}\label{eq:ber_achi_defi_P2}
    P_2:= 1- \big(1-\frac{\nu}{k}\big)^\ell = \frac{\nu\ell(1+o(1))}{k}
.\end{equation}
Given $(\ell,C,\zeta)$,     \textbf{(C2)} states that (\ref{eq:35411}) holds for some $\mathcal{J}\in \mathcal{K}_{\ell,C,\zeta}$ and some $\mathcal{J}'\subset [p]\setminus s$ with $|\mathcal{J}'|=\ell$. In (\ref{eq:35411}),
recall that $G_{\mathcal{J},\mathcal{J}',1}$ (resp., $G_{\mathcal{J},\mathcal{J}',2}$) is the number of tests in $\mathcal{N}_{01}\cup \mathcal{M}_{\mathcal{J}1}$ (resp., $\mathcal{N}_{00}\cup \mathcal{M}_{\mathcal{J}0}$) that contain some item from $\mathcal{J}'$, and it will be convenient to further decompose these into 
\begin{equation}
G_{\mathcal{J},\mathcal{J}',1}=\widetilde{G}_{\mathcal{J}',1}+U_{\mathcal{J},\mathcal{J}',1},~G_{\mathcal{J},\mathcal{J}',2}=\widetilde{G}_{\mathcal{J}',2}+U_{\mathcal{J},\mathcal{J}',2},\label{eq:decompose}
\end{equation}
where we define $\widetilde{G}_{\mathcal{J}',1}$  (resp., $\widetilde{G}_{\mathcal{J}',2}$)   as the number of tests in $\mathcal{N}_{01}$ (resp., $\mathcal{N}_{00}$) that contain some item from $\mathcal{J}'$ (thus $\widetilde{G}_{\mathcal{J}',1}$ and $\widetilde{G}_{\mathcal{J}',2}$ have no dependence on $\mathcal{J}$, as also reflected by our notation), and $U_{\mathcal{J},\mathcal{J}',1}$ (resp., $U_{\mathcal{J},\mathcal{J}',2}$) as the number of tests in $\mathcal{M}_{\mathcal{J}1}$ (resp., $\mathcal{M}_{\mathcal{J}0}$) that contain some item from $\mathcal{J}'$. Given  $(\Xv_s, \Yv)$ and for specific $\mathcal{J}$ and $\mathcal{J}'$, it follows that
\begin{gather}\label{eq:tildeG_distri}
    \widetilde{G}_{\mathcal{J}',1}\sim \text{Bin}(N_{01},P_2),~\widetilde{G}_{\mathcal{J}',2}\sim \text{Bin}(N_{00},P_2)\\
    U_{\mathcal{J},\mathcal{J}',1}\sim \text{Bin}(M_{\mathcal{J}1},P_2),~ U_{\mathcal{J},\mathcal{J}',2}\sim \text{Bin}(M_{\mathcal{J}0},P_2)
\end{gather}
are independent, as the tests comprising $\mathcal{N}_{01},\mathcal{N}_{00},\mathcal{M}_{\mathcal{J}0},\mathcal{M}_{\mathcal{J}1}$ are disjoint. Under these definitions, (\ref{eq:35411}) can be written as 
\begin{equation}\label{eq:376}
    \underbrace{\widetilde{G}_{\mathcal{J}',1}-\widetilde{G}_{\mathcal{J}',2}}_{:=\widetilde{G}_{\mathcal{J}'}}+U_{\mathcal{J},\mathcal{J}',1}-U_{\mathcal{J},\mathcal{J}',2}\ge (1-2\zeta)\frac{Cn\nu e^{-\nu}\ell}{k},
\end{equation}
which   further implies $\widetilde{G}_{\mathcal{J}'}\ge  (1-2\zeta)\frac{Cn\nu e^{-\nu}\ell}{k}-U_{\mathcal{J},\mathcal{J}',1}$ since $U_{\mathcal{J},\mathcal{J}',2}\ge 0$. Thus,  if \textbf{(C2)} holds, then we have  $\widetilde{G}_{\mathcal{J}'}\ge (1-2\zeta)\frac{Cn\nu e^{-\nu}\ell}{k}-U_{\mathcal{J},\mathcal{J}',1}$ for some $\mathcal{J}\in \mathcal{K}_{\ell,C,\zeta}$ and some $\mathcal{J}'\subset [p]\setminus s$ with $|\mathcal{J}'|=\ell$, which further implies
\begin{equation}\label{eq:suff_C2}
    \max_{\substack{\mathcal{J}'\subset [p]\setminus s\\|\mathcal{J}'|=\ell}}\widetilde{G}_{\mathcal{J}'} \ge  (1-2\zeta)\frac{Cn\nu e^{-\nu}\ell}{k}- \max_{\substack{\mathcal{J}\in\mathcal{K}_{\ell,C,\zeta}\\\mathcal{J}'\subset [p]\setminus s,|\mathcal{J}'|=\ell}}~ U_{\mathcal{J},\mathcal{J}',1}.
\end{equation}
Thus, to bound the probability of \textbf{(C2)}, it suffices to bound the probability of its necessary condition given in (\ref{eq:suff_C2}).

\textbf{The effect of $U_{\mathcal{J},\mathcal{J}',1}$}: For convenience, we write the last term in (\ref{eq:suff_C2}) using the shorthand $\max_{\mathcal{J},\mathcal{J}'}U_{\mathcal{J},\mathcal{J}',1}$ with the constraints left implicit. 
We first show that the term  $\max_{\mathcal{J},\mathcal{J}'}U_{\mathcal{J},\mathcal{J}',1}$  in the right-hand side of (\ref{eq:suff_C2}) 
only has a minimal impact, in the sense that scales as $o(\frac{\ell n}{k})=o(\ell\log k)$. In fact, by $U_{\mathcal{J},\mathcal{J}',1}\sim \text{Bin}(M_{\mathcal{J}1},P_2)$ and $M_{\mathcal{J}1}=\frac{(1-\zeta)Cn\nu e^{-\nu}\ell}{k}\leq \frac{C_0n \ell}{k}$ (recall that we are now considering $C\leq C_0=O(1)$ and  $\nu e^{-\nu}\leq 1$ holds trivially), for fixed $\mathcal{J},\mathcal{J}'$ we have   
\begin{align}
    &\mathbb{P}\left(U_{\mathcal{J},\mathcal{J}',1}\geq \frac{\ell n}{k (\log\frac{k}{\ell})^{1/2}}\right) \nonumber \\ \label{eq:nc_small_effect}&\leq \mathbb{P}\left(\text{Bin}\Big(\frac{C_0n\ell}{k},\frac{\nu\ell(1+o(1))}{k}\Big)\geq \frac{\ell n}{k (\log\frac{k}{\ell})^{1/2}}\right)\\\label{eq:379}
    &\leq\exp\left(-\frac{\ell n}{k}\Big(\Big[\log\frac{k}{\ell}\Big]^{-1/2}\Big[\log\frac{k}{C_0\nu \ell(\log\frac{k}{\ell})^{1/2}}- 1\Big]+\frac{C_0\nu\ell}{k}\Big)\right)\\&=\exp\Big(-\Omega\Big(\sqrt{\log\frac{k}{\ell}}\Big)\cdot \ell \log k\Big), \label{eq:nc_small_end}
\end{align}
where in (\ref{eq:379}) we use the Chernoff bound (see (\ref{eq:chernoff2}) in Lemma \ref{binoconcen}), and   (\ref{eq:nc_small_end}) holds since $\frac{\ell n}{k}=\Theta(\ell\log k)$. 
By a union bound over at most $\binom{k}{\ell}\times \binom{p}{\ell}$ choices of $(\mathcal{J},\mathcal{J}')$, we obtain that 
\begin{align}
\mathbb{P}\left(\max_{\mathcal{J},\mathcal{J}'}U_{\mathcal{J},\mathcal{J}',1}\geq \frac{\ell n}{k(\log\frac{k}{\ell})^{1/2}}\right)\leq \exp\Big(\ell\log\frac{ek}{\ell}+ \ell\log\frac{ep}{\ell}-\Omega\Big(\sqrt{\log\frac{k}{\ell}}\Big)\cdot \ell\log k\Big)\leq k^{-10\ell}, \label{eq:k-10ell}
\end{align}
where the last inequality holds for large enough $k$ because $\frac{k}{\ell}\to \infty$ and $k=\Theta(p^\theta)$ for some $\theta\in(0,1)$. Combining with $\frac{\ell n}{k(\log\frac{k}{\ell})^{1/2}}=o(\frac{\ell n}{k})$, with probability at least $1-k^{-10\ell}$, we have
\begin{align}
    \label{eq:ber_minimal_U}
\max_{\mathcal{J},\mathcal{J}'}~U_{\mathcal{J},\mathcal{J}',1} =o\Big(\frac{\ell n}{k}\Big), 
\end{align}
under which the sufficient condition for (\textbf{C2}) given in (\ref{eq:suff_C2}) can be further simplified to 
\begin{equation}\label{eq:suff_C2_simplify}
     \max_{\substack{\mathcal{J}'\subset [p]\setminus s\\|\mathcal{J}'|=\ell}}\widetilde{G}_{\mathcal{J}'} \ge (1-2\zeta-o(1))\frac{Cn\nu e^{-\nu}\ell}{k}.
\end{equation}

\textbf{Decomposing the event (\ref{eq:suff_C2_simplify}):} With respect to the randomness of $\Xv_s$ and $\Yv$, we have already shown that (\ref{eq:N00N01}) holds with $1-o(1)$ probability, and thus we can proceed on the event
\begin{align}
    \mathscr{A}_1= \big\{N_{01}=\rho n e^{-\nu}(1+o(1)),~N_{00}=(1-\rho)ne^{-\nu}(1+o(1)) \big\}
    \label{eq:ber_achi_n0}
\end{align}
by a simple union bound. % \footnote{Strictly speaking, this means that subsequent probabilities $\mathbb{P}(\cdot)$ should actually be interpreted as $\mathbb{P}(\cdot \cap \mathscr{A}_1)$, but this is left implicit to avoid cumbersome notation, and is handled at the end of our analysis.}
In the following, for a generic event $\mathscr{E}$, an upper bound on $\mathbb{P}(\mathscr{E})$ should be more precisely understood as a bound on $\mathbb{P}(\mathscr{E}\cap\mathscr{A}_1)$, but to avoid cumbersome notation we will only make this explicit in the final stage (see (\ref{eq:ber_c2_prob_boundhh}) below).

Given $\mathcal{J}'\subset [p]\setminus s $ with $|\mathcal{J}'|=\ell$, by conditioning on $(\Xv_s,\Yv)$, the randomness of $\Xv_{\mathcal{J}'}$ gives \begin{align}
    \widetilde{G}_{\mathcal{J}'}=\widetilde{G}_{\mathcal{J}',1}-\widetilde{G}_{\mathcal{J}',2}\sim \text{Bin}(\rho n e^{-\nu}(1+o(1)),P_2)-\text{Bin}((1-\rho)ne^{-\nu}(1+o(1)),P_2), \label{eq:ber_achi_dis1}
\end{align}
where $P_2$ is the probability defined in (\ref{eq:ber_achi_defi_P2}). 
 We can now introduce 
 \begin{align}
     d_0 := \frac{C(1-2\zeta)+1-\rho}{\rho} 
 \end{align}
 and proceed as follows:
\begin{align}
    \label{eq:393}&\mathbb{P}\Big(\widetilde{G}_{\mathcal{J}'}\geq (1-2\zeta-o(1))\frac{Cn\nu e^{-\nu}\ell}{k}\Big)\leq \mathbb{P}\Big(\widetilde{G}_{\mathcal{J}',1}\geq \frac{d_0\cdot \rho \nu \ell n e^{-\nu}}{k}\Big)\\
    &\quad\quad+ \sum_{\substack{0\leq d<d_0\\\frac{d\rho\nu e^{-\nu}\ell n}{k}\in\mathbb{Z}}}\mathbb{P}\Big(\widetilde{G}_{\mathcal{J}',1}=\frac{d\cdot \rho \nu \ell n e^{-\nu}}{k}\Big)\mathbb{P}\Big(\widetilde{G}_{\mathcal{J}',2}\leq \big[\rho d-(1-2\zeta-o(1))C\big]\frac{n\nu e^{-\nu}\ell}{k}\Big).\label{eq:387}
\end{align}
To   bound the terms in (\ref{eq:393}) and (\ref{eq:387}), we consider the following two cases.

\textbf{Case 1: $(1-2\zeta)C> 2\rho -1$.} In this case we have $d_0>1.$

\underline{Bounding the   term in (\ref{eq:393}):}
Consider conditioning on $(\Xv_s,\Yv)$ and utilizing the randomness of $\Xv_{\mathcal{J}'}$, 
and recall from (\ref{eq:ber_achi_dis1}) and (\ref{eq:ber_achi_defi_P2}) that  we have $\widetilde{G}_{\mathcal{J}',1}\sim \text{Bin}(\rho ne^{-\nu}(1+o(1)),\frac{\nu\ell(1+o(1))}{k})$. 
Since $d_0>1$, we  can apply Chernoff bound ((\ref{eq:chernoff1}) in Lemma \ref{binoconcen}) to bound the term in (\ref{eq:393}) as  
\begin{align}\label{eq:minorbound}
\mathbb{P}\Big(\widetilde{G}_{\mathcal{J}',1}\geq \frac{d_0\cdot \rho \nu \ell n e^{-\nu}}{k}\Big)&\leq \exp\left(-\rho ne^{-\nu}(1+o(1))D\Big(\frac{d_0\ell\nu }{k}\big\|\frac{\ell\nu}{k}\Big)\right).
\end{align}

\underline{Bounding the term in (\ref{eq:387}):} First observe that $\mathbb{P}(\widetilde{G}_{\mathcal{J}',2}\le [\rho d-(1-2\zeta-o(1))C]\frac{n\nu e^{-\nu}\ell}{k})=0$ if $d< \frac{C(1-2\zeta)}{\rho}-o(1)$, and so we can further restrict the summation over $d$ in (\ref{eq:387}) to the range \begin{align} \label{eq:ber_achi_drange}
    \mathcal{D}:=\Big\{ \max\Big\{0,\frac{C(1-2\zeta)}{\rho}\Big\}\le d<d_0: \frac{d\rho \nu e^{-\nu}\ell n}{k}\in \mathbb{Z}\Big\} 
\end{align}
up to a $o(1)$ term that has no impact on our subsequent analysis. Since $d_0=O(1)$, 
the term in (\ref{eq:387}) involves no more than $O(\frac{\ell n}{k})$ summands, so using the Chernoff bound (see (\ref{eq:chernoff1}) in Lemma \ref{binoconcen}) we can bound it as
\begin{align}
 \label{eq:390}   &O\Big(\frac{\ell n}{k}\Big)\max_{d\in \mathcal{D}}\Big\{\exp\Big(-\rho n e^{-\nu}(1+o(1))\cdot D\Big(\frac{d\nu\ell(1+o(1))}{k}\big\|\frac{\nu\ell(1+o(1))}{k}\Big)\Big)\\&\quad\quad \quad\cdot \exp\Big(-(1-\rho)ne^{-\nu}(1+o(1))\cdot D\Big(\frac{[\rho d-(1-2\zeta)C]\nu\ell(1+o(1))}{(1-\rho)k}\big\|\frac{\nu\ell(1+o(1))}{k}\Big)\Big)\Big\}\label{eq:chernoff_justi} \\&=\exp\Big(-ne^{-\nu}(1+o(1))\min_{d\in\mathcal{D}}\Big[\rho\cdot D\Big(\frac{d\nu\ell}{k}\big\|\frac{\nu\ell}{k}\Big)+(1-\rho) D\Big(\frac{[\rho d-(1-2\zeta)C]\nu\ell}{(1-\rho)k}\big\|\frac{\nu\ell}{k}\Big)\Big]\Big)\label{eq:399}
 \\
    &\le \exp \Big(-\frac{\ell n\nu e^{-\nu}}{k} \big[\min_{d\ge \max\{0,\frac{A}{\rho}\}} g(C,\zeta,d,\rho)+o(1)\big]\Big),\label{eq:392}
\end{align}
where:
\begin{itemize}
    \item the Chernoff bound in (\ref{eq:chernoff_justi}) is justified by $\widetilde{G}_{\mathcal{J}',2}\sim\text{Bin}\big((1-\rho)ne^{-\nu}(1+o(1)),\frac{\nu\ell(1+o(1))}{k}\big)$ (note also that $d<d_0$ implies $\rho d-(1-2\zeta)C<1-\rho$);
    \item to obtain (\ref{eq:392}), we define $A=C(1-2\zeta)$ and relax the range of $d$ from $d\in \mathcal{D}$ to $d \ge \max\big\{0,\frac{A}{\rho}\big\}$, and we introduce the shorthand 
    \begin{equation}\label{eq:g_ber_achi}
        g(C,\zeta,d,\rho) = \rho d\log d+\big(\rho d-A\big)\log\Big(\frac{\rho d-A}{1-\rho}\Big) + 1-2\rho d+A.
    \end{equation}
     Then, (\ref{eq:392}) follows from (\ref{eq:399}), by using $D(\frac{a\nu\ell}{k}\|\frac{\nu\ell}{k})=\frac{\nu\ell}{k}(a\log a-a+1+o(1))$, and noting that $g(C,\zeta,d,\rho)$ defined here coincides with (\ref{eq:g_ber_con}) in the proof of converse bound for Bernoulli designs.
\end{itemize}
Thus, by the same reasoning as (\ref{eq:338})--(\ref{eq:340}), we know that \begin{equation}
     \min_{d\ge \max\{0,\frac{A}{\rho}\}}g(C,\zeta,d,\rho) =  g(C,\zeta,d^*,\rho):=f_2(C,\zeta,\rho),
\end{equation}
with $d^*$ being the same as in (\ref{eq:340}):
\begin{equation}\label{eq:ber_achi_dstar}
    d^* = \frac{A+\sqrt{A^2+4\rho(1-\rho)}}{2\rho}. 
\end{equation}
Therefore, we obtain a bound on the term in (\ref{eq:387}) as 
\begin{equation}\label{eq:bigbound}
    \exp\Big(-\frac{\ell n \nu e^{-\nu}}{k}\big(f_2(C,\zeta,\rho)+o(1)\big)\Big),
\end{equation}
 with $f_2(C,\zeta,\rho)$ being the same as in the proof of Bernoulli design converse bound; see (\ref{eq:f2_ber_con}).

\underline{Comparing the two bounds (\ref{eq:bigbound}) and (\ref{eq:minorbound}):} We return to (\ref{eq:399})--(\ref{eq:392}), in which we used $D(\frac{a\nu\ell}{k}\| \frac{\nu\ell}{k})= \frac{\nu\ell}{k}(a\log a - a + 1+o(1))$ and established that 
\begin{align}
    \rho \cdot D\Big(\frac{d\nu\ell}{k}\big\|\frac{\nu\ell}{k}\Big)+(1-\rho)\cdot D\Big(\frac{[\rho d-(1-2\zeta)C]\nu\ell}{(1-\rho)k}\big\|\frac{\nu\ell}{k}\Big) = \frac{\nu\ell}{k}[g(C,\zeta,d,\rho)+o(1)].\label{eq:divergence_form_g}
\end{align}
This observation, along with $d_0=\frac{C(1-2\zeta)+1-\rho}{\rho}>1$ (recall that we are considering  $C(1-2\zeta)>2\rho-1$), allows us to compare the relevant terms as follows:
\begin{align}
   \frac{\ell n\nu e^{-\nu}}{k}f_2(C,\zeta,\rho) &= \frac{\ell n\nu e^{-\nu}}{k}\min_{d\ge\max\{0,\frac{A}{\rho}\}}g(C,\zeta,d,\rho) \\ 
   &\le \frac{\ell n\nu e^{-\nu}}{k} g(C,\zeta,d_0,\rho)\\
    &= ne^{-\nu}\Big(\rho D\Big(\frac{d_0\nu\ell}{k}\big\|\frac{\nu\ell}{k}\Big)+(1-\rho)D\Big(\frac{[\rho d_0-(1-2\zeta)C]\nu\ell}{(1-\rho)k}\big\|\frac{\nu\ell}{k}\Big)\Big) - o\Big(\frac{\ell n}{k}\Big)\label{eq:use_g_divergence}\\
    & = \rho n e^{-\nu}D\Big(\frac{d_0\nu\ell}{k}\big\|\frac{\nu\ell}{k}\Big) -o\Big(\frac{\ell n}{k}\Big) \label{eq:substi_d0_0} \\
    & = \rho n e^{-\nu}(1+o(1))D\Big(\frac{d_0\nu\ell}{k}\big\|\frac{\nu\ell}{k}\Big),\label{eq:substi_d0}
\end{align}
where in (\ref{eq:use_g_divergence}) we substitute (\ref{eq:divergence_form_g}), and in (\ref{eq:substi_d0_0}) we substitute $d_0 = \frac{(1-2\zeta)C+1-\rho}{\rho}$ and observe that the second term in (\ref{eq:use_g_divergence}) vanishes.  
Therefore, the bound (\ref{eq:bigbound}) dominates the one in (\ref{eq:minorbound}), and it follows that 
\begin{equation}
    \label{eq:bound_diff_ber}\mathbb{P}\Big(\widetilde{G}_{\mathcal{J}'}\geq (1-2\zeta-o(1))\frac{Cn\nu e^{-\nu}\ell}{k}\Big) \leq \exp\Big(-\frac{\ell n \nu e^{-\nu}}{k}\big(f_2(C,\zeta,\rho)+o(1)\big)\Big).
\end{equation}

\textbf{Case 2: $(1-2\zeta)C\leq 2\rho -1$.}  In this case, it suffices to apply the trivial bound \begin{align}
    \label{eq:ber_achi_trivial}\mathbb{P}\Big(\widetilde{G}_{\mathcal{J}'}\geq (1-2\zeta-o(1))\frac{Cn\nu e^{-\nu}\ell}{k}\Big)\leq 1,
\end{align}
which in turn can trivially be written as $\exp(-\frac{\ell n\nu e^{-\nu}}{k}o(1))$.  

% Intuitively, a trivial bound suffices here because 
% We pause to briefly explain the intuition of applying a trivial bound. In fact,   it is   unlikely to get a meaningful bound from a population level when $(1-2\zeta)C\le 2\rho-1$, since from (\ref{eq:ber_achi_dis1}) and (\ref{eq:ber_achi_defi_P2}) $\widetilde{G}_{\mathcal{J}'}$ is expected to tightly concentrate about its mean that asymptotically equals to\begin{align}
%     \rho ne^{-\nu}(1+o(1))P_2- (1-\rho)ne^{-\nu}(1+o(1))P_2=(2\rho-1+o(1)) \frac{n\nu e^{-\nu}\ell}{k}.\label{eq:trivial_intu}
% \end{align} This intuition applies to similar analysis in the proof of the achievablity bound for near-constant weight designs.

\textbf{Combining Cases 1-2:}
Some algebra verifies that $f_2(C,\zeta,\rho)=0$ when $A=(1-2\zeta)C=2\rho-1$: Substituting this into  (\ref{eq:ber_achi_dstar}) gives $d^*=1$, and further we have $f_2(C,\zeta,\rho)=g(C,\zeta,1,\rho)=0$ by using (\ref{eq:g_ber_achi}) and $A=2\rho-1$.  Thus, we can define the continuous function 
\begin{equation}\label{eq:hatf2}
    \hat{f}_2(C,\zeta,\rho) := \begin{cases}
        f_2(C,\zeta,\rho) & \text{if }(1-2\zeta)C\ge 2\rho -1\\
        0 & \text{if }(1-2\zeta)C<2\rho -1.
    \end{cases}
\end{equation}
 Combining these two cases, we arrive at 
\begin{equation}
\label{eq:bound_diff_ber_1}\mathbb{P}\Big(\widetilde{G}_{\mathcal{J}'}\geq (1-2\zeta-o(1))\frac{Cn\nu e^{-\nu}\ell}{k}\Big) \leq \exp\Big(-\frac{\ell n \nu e^{-\nu}}{k}\big(\hat{f}_2(C,\zeta,\rho)+o(1)\big)\Big)
\end{equation}
for any feasible $(C,\zeta)$ with $C\leq C_0$.

\textbf{The condition for (C2):} We now bound the probability of \textbf{(C2)} holding for some $(\ell,C,\zeta)$ as 
\begin{align}
   &\mathbb{P}\Big(\textbf{(C2)}\text{   holds for }(\ell,C,\zeta)\Big) \nonumber \\
    &\leq \mathbb{P}\Big((\ref{eq:suff_C2_simplify})\text{   holds for }(\ell,C,\zeta)\Big)+ \mathbb{P}\Big((\ref{eq:ber_minimal_U}) \text{ does not hold}\Big)\label{eq:add_explain1}\\
    &\leq \binom{p}{\ell}\cdot \mathbb{P}\Big(\widetilde{G}_{\mathcal{J}'}\geq (1-2\zeta-o(1))\frac{Cn\nu e^{-\nu}\ell}{k}\Big) +k^{-10\ell}\label{eq:add_explain2}\\
    &\leq \exp\Big((1+o(1))\ell \log\frac{p}{\ell}-\frac{\ell n \nu e^{-\nu}}{k}\big(\hat{f}_2(C,\zeta,\rho)+o(1)\big)\Big) +k^{-10\ell},\label{eq:413}
\end{align}
where (\ref{eq:add_explain1}) follows because under the condition (\ref{eq:ber_minimal_U}) we have that (\ref{eq:suff_C2_simplify}) is a necessary condition for ``\textbf{(C2)} holds for the given $(\ell,C,\zeta)$'', and in (\ref{eq:add_explain2}) we apply a  union bound to account for the maximum over $\mathcal{J}'$ in (\ref{eq:suff_C2_simplify}), also recalling that (\ref{eq:ber_minimal_U}) holds with probability at least $1-k^{-10\ell}$. 

Therefore, for given $(\ell,C,\zeta)$, we have shown the following implication for any $\eta_2 > 0$ and sufficiently large $p$ (recall that our analysis of \textbf{(C2)} is on the event $\mathscr{A}_1$ (\ref{eq:ber_achi_n0}), and in our conclusion below we make this explicit):
\begin{align}\label{eq:c2_threshold}
    \frac{\ell n \nu e^{-\nu}}{k}\big[\hat{f}_2(C,\zeta,\rho)+o(1)\big]\geq (1+\eta_2)\ell\log\frac{p}{\ell} \\ 
    \Longrightarrow ~ \mathbb{P}\Big(\big\{\textbf{(C2)}\text{ holds for }(\ell,C,\zeta)\big\}\cap \mathscr{A}_1\Big) &\leq \exp\Big(-\frac{\eta_2}{2}\ell\log\frac{p}{\ell}\Big)+k^{-10\ell} \\ &\le \hat{P}_{\ell}: =\begin{cases}
        (\ell\log k)^{-5},~\text{if }1\le \ell \le \log k\\
        ~~~k^{-5},~~~~~~\text{if }\log k<\ell \le \frac{k}{\log k}, \label{eq:ber_c2_prob_boundhh}
    \end{cases} 
\end{align}
since $\exp(-\frac{\eta_2}{2}\ell\log\frac{p}{\ell})+k^{-10\ell}$ is an upper bound on (\ref{eq:413}) when \eqref{eq:c2_threshold} holds.  Recall also that $\mathbb{P}(\mathscr{A}_1) = 1-o(1)$. 
%Strictly speaking, the probability in (\ref{eq:explain}) should be $1-\hat{P}_\ell-o(1)$ to ensure the high-probability event $\mathscr{A}_1$ in  (\ref{eq:ber_achi_n0}), and the probability in (\ref{eq:413}) and (\ref{eq:explain}) should accordingly include this additional $o(1)$ term. Nevertheless, note that the event $\mathscr{A}_1$ is universal for all (i.e., does not depend on)  $(\ell,C,\zeta)$, thus for convenience, we implicitly build our analysis  on \textbf{(C2)} with respect to specific $(\ell,C,\zeta)$ (e.g., the statement in (\ref{eq:explain})) on $\mathscr{A}_1$. It suffices to remove the $o(1)$ probability (that accounts for the failure of $\mathscr{A}_1$) at a later stage; see the $o(1)$ term in  (\ref{eq:ber_achi_account}).  

\subsubsection*{Establishing the threshold}

  We are now ready to establish the threshold for $n$ above which restricted MLE has $o(1)$ probability of failing. We first pause to review our previous developments:
\begin{itemize}
    \item For given $\ell$, from (\ref{eq:boundk1})--(\ref{eq:boundk2}),
     \textbf{(C1)} and \textbf{(C2)} hold for some feasible $(C,\zeta)\in (C_0,\infty)\times [0,1]$ with probability at most $k^{-5}$.     
    \item For given $\ell$ and any feasible $(C,\zeta)$ with $C\leq C_0$, 
    from (\ref{eq:C1_threshold})--(\ref{eq:c1_condition}), if
\begin{equation}\label{eq:C1_final}
      n \geq (1+\eta_1)\frac{k\log\frac{k}{\ell}}{\nu e^{-\nu}[f_1(C,\zeta,\rho)+o(1)]} 
\end{equation}
holds for some $\eta_1>0$, then $\mathbb{P}\big(\{$\textbf{(C1)} holds for the given $(\ell,C,\zeta)\}\cap\mathscr{A}_1\big)\le \hat{P}_\ell.$
    \item For given $\ell$ and any feasible $(C,\zeta)$ with $C\leq C_0$, 
    by (\ref{eq:c2_threshold})--(\ref{eq:ber_c2_prob_boundhh}), if
    \begin{equation}\label{eq:C2_final}
    n\geq (1+\eta_2)\frac{k\log\frac{p}{\ell}}{\nu e^{-\nu}[\hat{f}_2(C,\zeta,\rho)+o(1)]}  
\end{equation}
holds for some $\eta_2>0$,\footnote{Note that $\hat{f}_2(C,\zeta,\rho)$ in the denominator equals to $0$ in some cases (see (\ref{eq:hatf2})), making the condition (\ref{eq:C2_final}) vacuous.  To show $o(1)$ failure probability for MLE it suffices to show $o(1)$ probability of conditions \textbf{(C1)} and \textbf{(C2)} \emph{simultaneously} holding, so when (\ref{eq:C2_final}) becomes vacuous, condition \textbf{(C1)} becomes the relevant one. }
%\footnote{Note that $\hat{f}_2(C,\zeta,\rho)$ in the denominator equals to $0$ in some cases (see (\ref{eq:hatf2})), then the threshold (\ref{eq:C2_final}) indicates that it is not   likely to ensure that \textbf{(C2)}  for such $(C,\zeta)$ holds with $\hat{P}_\ell$ probability. But this will not induce ambiguity to the threshold for MLE to succeed (see, e.g., (\ref{eq:411})), since to bound the failure probability of MLE as $o(1)$ it suffices to bound the probability of the event in which \textbf{(C1)} {\it and} \textbf{(C2)} simultaneously hold.}
then $\mathbb{P}\big(\{$\textbf{(C2)} holds for the given $(\ell,C,\zeta)\}\cap\mathscr{A}_1\big)\le \hat{P}_\ell.$  
\end{itemize} 
%To ensure the success of MLE (\ref{eq:mle}), the above first point addresses all feasible $(C,\zeta)$ with $C>C_0$ for a given $\ell$,
 %while the above last two points address the case of given $(\ell,C,\zeta)$ with $C\leq C_0$ (under the condition that either (\ref{eq:C1_final}) or (\ref{eq:C2_final}) holds). All that remains is to establish the threshold for ensuring the success of MLE decoder over   all $(\ell,C,\zeta)$. We pursue such uniformity in two steps.
% We are now in a position to establish the threshold for $n$ above which the MLE fails for some $\ell\in [1,\frac{k}{\log k}]$ with $o(1)$ probability. 

\underline{Bounding the failure probability for a fixed $\ell$:} 
For a fixed $\ell$, we first consider   the feasible $(C,\zeta)$ with $C\leq C_0$, and by $\frac{Cn\nu e^{-\nu}\ell}{k},\frac{\zeta\cdot Cn\nu e^{-\nu}\ell}{k}\in\mathbb{Z}$ there are no more than $O(\frac{\ell n}{k})\times O(\frac{\ell n}{k})=O((\ell\log k)^2)$ possible choices of such $(C,\zeta)$.  (Recall from Section \ref{sec:n_scaling} that $n = \Theta(k \log k)$.) If it holds for some $\eta_1>0$ that
\begin{equation}
    n\geq (1+\eta_1)\frac{k}{\nu e^{-\nu}}\max_{\substack{0<C<C_0\\0<\zeta<1}}\min\Big\{\frac{\log\frac{k}{\ell}}{f_1(C,\zeta,\rho)},\frac{\log \frac{p}{\ell}}{\hat{f}_2(C,\zeta,\rho)}\Big\}\label{eq:con_small_C},
\end{equation}
then by the last two dot points reviewed above, for any  feasible $(C,\zeta)$ with $C\leq C_0$,   
$$\mathbb{P}\big((\textbf{C1}) \text{ and } (\textbf{C2})\text{ hold for }(C,\zeta) \,\cap\,\mathscr{A}_1\text{ holds}\big)\le \hat{P}_\ell.$$ 
Moreover, by a union bound over the $O((\ell\log k)^2)$ possibilities of $(C,\zeta)$, 
$$\mathbb{P}\big(\text{\textbf{(C1)} and \textbf{(C2)} hold for some feasible $(C,\zeta)$ with $C\leq C_0$} \,\cap\, \mathscr{A}_1 \text{ holds}\big)\le O((\ell \log k)^2)\hat{P}_\ell.$$
On the other hand,   the first dot point reviewed above gives 
$$\mathbb{P}\big(\text{\textbf{(C1)} and \textbf{(C2)} hold for some feasible $(C,\zeta)$ with $C>C_0$} \,\cap\, \mathscr{A}_1 \text{ holds}\big)\le k^{-5}.$$
 Overall, for a fixed $\ell\in [1,\frac{k}{\log k}]$, if (\ref{eq:con_small_C}) holds for some $\eta_1>0$, then for the restricted MLE output $\widehat{S}'$ (see (\ref{eq:mle_restricted})) we have 
\begin{align}
 & \mathbb{P}\Big(\big\{|s\setminus\widehat{S}'|=\ell\big\}\cap\mathscr{A}_1\Big) \nonumber 
 \\\label{eq:fail_necc} &\le
  \mathbb{P}\Big(\textbf{(C1)}\text{ and }\textbf{(C2)}~\text{hold for some feasible }(C,\zeta) \cap \,\mathscr{A}_1\, \text{ holds}\Big)\\ &\le  k^{-5}+O\big((\ell\log k)^2\big)\hat{P}_\ell,\label{eq:success_fixed_ell}
\end{align}
where (\ref{eq:fail_necc}) holds because ``\textbf{(C1)} and \textbf{(C2)} hold for some feasible $(C,\zeta)$ (under the given $\ell$)'' is a necessary condition for ``$|s\setminus \widehat{S}'|=\ell$'' due to Lemma \ref{lem:restate}(b).

\underline{Bounding the overall failure probability for $\ell\in [1,\frac{k}{\log k}]$:} We further take a union bound over $\ell\in[1,\frac{k}{\log k}]$.  Suppose that the condition (\ref{eq:con_small_C}) holds for all $\ell\in [1,\frac{k}{\log k}]$, i.e., 
\begin{align}
     \label{eq:ber_uni_ell}&n\geq (1+\eta_1)\frac{k}{\nu e^{-\nu}}\max_{1\leq \ell \leq \frac{k}{\log k}}\max_{\substack{0<C<C_0\\0<\zeta<1}}\min\Big\{\frac{\log\frac{k}{\ell}}{f_1(C,\zeta,\rho)},\frac{\log \frac{p}{\ell}}{\hat{f}_2(C,\zeta,\rho)}\Big\}\\
     \label{eq:411}\iff & n\geq \frac{(1+\eta_1)k\log\frac{p}{k}}{(1-\theta)\nu e^{-\nu}}\frac{1}{\min _{C\in(0,C_0),\zeta\in(0,1)}\max\{\frac{1}{\theta}f_1(C,\zeta,\rho),\hat{f}_2(C,\zeta,\rho)\}},
\end{align}
where in (\ref{eq:411})  we observe that the maximum over $\ell$ is attained at $\ell=1$ and use $\log p\sim \frac{1}{1-\theta}\log\frac{p}{k},~\log k \sim \frac{\theta}{1-\theta}\log\frac{p}{k}$. 
Then, we can bound the probability that the restricted MLE decoder fails as
\begin{align}
    %&\mathbb{P}\Big(|s\setminus\widehat{S}|:=\ell\in\Big[1,\frac{k}{\log k}\Big]\Big)
    \mathbb{P}({\rm err}) \label{eq:ber_achi_account} &\le  \mathbb{P}(\mathscr{A}_1^c) + \sum_{\ell=1}^{k/\log k}  \mathbb{P}\Big(\big\{\big|s\setminus \widehat{S}'\big|=\ell\big\}\cap \mathscr{A}_1\Big)\\
    &\le  o(1)+  \sum_{\ell=1}^{k/\log k}\Big(k^{-5} + O((\ell\log k)^2)\hat{P}_\ell\Big)   \label{eq:ber_prob_cal} \\
    &\le o(1)+ \Big[\sum_{\ell=1}^{\log k} (\ell\log k)^{-3} + \sum_{\ell=\log k}^{k/\log k}k^{-3}\Big] = o (1), \label{eq:substi_Pell}
\end{align}
where (\ref{eq:ber_achi_account}) follows from the union bound, (\ref{eq:ber_prob_cal}) follows from (\ref{eq:success_fixed_ell}), and (\ref{eq:substi_Pell}) follows by substituting $\hat{P}_\ell$ (given in  (\ref{eq:ber_c2_prob_boundhh})). In conclusion, if (\ref{eq:411}) holds for some $\eta>0$, then the restricted MLE decoder fails to recover $s=[k]
$ with $o(1)$ probability.

\textbf{Simplifying $\hat{f}_2(C,\zeta,\rho)$ to $f_2(C,\zeta,\rho)$:}
Recall that $\hat{f}_2(C,\zeta,\rho)$ is given in (\ref{eq:hatf2}). The aim of this step is to show that the minimum over $(C,\zeta)$ in (\ref{eq:411}) is not attained in the domain of $(1-2\zeta)C<2\rho -1$, and thus we can safely simplify $\hat{f}_2(C,\zeta,\rho)$ to $f_2(C,\zeta,\rho)$ in the threshold (\ref{eq:411}).    Consider $(C,\zeta)$ satisfying $(1-2\zeta)C\leq 2\rho -1$, i.e., \begin{align}
    \zeta\geq \zeta':= \frac{1}{2}+\frac{1-2\rho}{2C}\ge \frac{1}{2},
\end{align}
in which we have $\hat{f}_2(C,\zeta,\rho)=0$ (see (\ref{eq:hatf2}) and recall that $\hat{f}_2(C,\zeta,\rho)$ is continuous), and that $f_1(C,\zeta,\rho)$ is monotonically increasing with respect to $\zeta$ in $[\frac{1}{2},1)$ (see (\ref{eq:f1_ber_achi})). Therefore, we can compare the values of $f_1(C,\zeta,\rho)$ and $f_2(C,\zeta,\rho)$ over $(C,\zeta)$ and $(C,\zeta')$ as \begin{align}
    f_1(C,\zeta,\rho)\geq f_1(C,\zeta',\rho),~\hat{f}_2(C,\zeta,\rho)=\hat{f}_2(C,\zeta',\rho)=0,
\end{align} yielding that \begin{align}
    \label{eq:ber_use_conti}
    \max\Big\{\frac{1}{\theta}f_1(C,\zeta,\rho),\hat{f}_2(C,\zeta,\rho)\Big\}\geq \max\Big\{\frac{1}{\theta}f_1(C,\zeta',\rho),\hat{f}_2(C,\zeta',\rho)\Big\}. 
\end{align} Observe that $(C,\zeta')$ satisfies $(1-2\zeta)C\ge 2\rho -1$, in which $\hat{f}_2(C,\zeta,\rho)=f_2(C,\zeta,\rho)$, thus we can proceed as
\begin{align}
    \label{eq:ber_coincide1}&\min_{C\in(0,C_0),\zeta\in(0,1)}\max\Big\{\frac{1}{\theta}f_1(C,\zeta,\rho),\hat{f}_2(C,\zeta,\rho)\Big\}\\
    =\, & \min_{\substack{C\in (0,C_0),\zeta\in(0,1)\\\zeta\leq \frac{1}{2}+\frac{1-2\rho}{2C}}}\max \Big\{\frac{1}{\theta}f_1(C,\zeta,\rho),\hat{f}_2(C,\zeta,\rho)\Big\}\\
    =\,& \min_{\substack{C\in (0,C_0),\zeta\in(0,1)\\(1-2\zeta)C\geq 2\rho -1}}\max \Big\{\frac{1}{\theta}f_1(C,\zeta,\rho),f_2(C,\zeta,\rho)\Big\}\\
    \geq\, & \min_{\substack{C>0 \\\zeta\in(0,1)}}\max  \Big\{\frac{1}{\theta}f_1(C,\zeta,\rho),f_2(C,\zeta,\rho)\Big\}.\label{eq:ber_coincide2}
\end{align}
Therefore, to ensure (\ref{eq:411}), it suffices to have 
\begin{equation}
    n\geq \frac{(1+\eta)k\log\frac{p}{k}}{(1-\theta)\nu e^{-\nu}}\frac{1}{\min_{C>0,\zeta\in(0,1)}\max\{\frac{1}{\theta}f_1(C,\zeta,\rho),f_2(C,\zeta,\rho)\}}
\end{equation}
for some $\eta>0$.  This establishes the second term in \eqref{eq:ber_threshold}.

\section{Low-$\ell$ Converse Analysis for Near-Constant Weight Designs} \label{app:nc_converse}
In the near-constant weight design, each item is independently and uniformly placed at $\Delta = \frac{\nu n}{k}$ tests with replacement. As with the Bernoulli design, it suffices to analyze the MLE decoder (\ref{eq:mle}) that finds $\widehat{S}$ having the most correct test outcomes (i.e., outcomes that would be obtained if there were no noise). By the symmetry of the test design with respect to re-ordering items, we can again consider  $s=[k]$ as the underlying defective set without loss of generality. 

%\subsubsection*{Some Notation}

%Recall that for $j\in[k]$, $\mathcal{M}_j$ indexes the tests where $j$ is the only defective, and $\mathcal{M}_{j0}$ indexes the negative tests in $\mathcal{M}_j$,  and $\mathcal{M}_{j1}$ indexes the positive ones. Moreover, $\mathcal{N}_0$ indexes the tests with no defectives. Recall that $\mathcal{K}_{C,\zeta}$ is defined in (\ref{eq:defi_K_c_zeta}), i.e., $$\mathcal{K}_{C,\zeta}=\Big\{j\in [k]:M_j = \frac{Cn\nu e^{-\nu}}{k},~
%M_{j0}=\frac{\zeta\cdot Cn\nu e^{-\nu}}{k}\Big\},$$ and we let $|\mathcal{K}_{C,\zeta}|=k_{C,\zeta}$. Here, we consider pairs  $(C,\zeta)\in (0,e^{-\nu}]\times (0,1)$ such that $\frac{Cn\nu e^{-\nu}}{k}, \frac{\zeta\cdot Cn\nu e^{-\nu}}{k}$ are integers and refer to such $(C,\zeta)$ pairs as feasible. 

 \subsubsection*{Notation} 
 For $j\in s$ and $j'\in [p]\setminus s$ we will use the notation $\mathcal{M}_j,\mathcal{M}_{j0},\mathcal{M}_{j1},\mathcal{N}_0,\mathcal{N}_{00},\mathcal{N}_{01}$ and the corresponding cardinalities (e.g., $N_{0}=|\mathcal{N}_0|$, $M_j=|\mathcal{M}_j|$) from Corollary \ref{cor:restate}. These sets and their cardinalities are deterministic given $\Xv_s$ and $\Yv$. We consider $\ell=1$, and say that $(C,\zeta)\in(0,\infty)\times(0,1)$ is feasible if $\frac{Cn\nu e^{-\nu}}{k}$ and $\frac{\zeta\cdot Cn\nu e^{-\nu}}{k}$ are integers. Given a feasible pair $(C,\zeta)$, we will consider $\mathcal{K}_{C,\zeta}=\{j\in s:M_j = \frac{Cn\nu e^{-\nu}}{k},~M_{j0}=\frac{\zeta\cdot Cn\nu e^{-\nu}}{k}\}$ with cardinality $k_{C,\zeta}=|\mathcal{K}_{C,\zeta}|$ as in (\ref{eq:defi_KCzeta}).  For given $(j,j')\in \mathcal{K}_{C,\zeta}\times([p]\setminus s)$, recall that 
     $G_{j,j',1}$ denotes the number of tests in $\mathcal{N}_{01}\cup \mathcal{M}_{j1}$ that contain the   non-defective item $j'$, and $G_{j,j',2}$ denotes number of tests in $\mathcal{N}_{00}\cup \mathcal{M}_{j0}$ that contain the non-defective item $j'$; see Corollary \ref{cor:restate}.

 We also introduce some additional definitions for analyzing the near-constant weight design. For $j\in s$ and $j'\in [p]\setminus s$, we let $\mathcal{M}_{j,j'}$ index the tests in $\mathcal{M}_j$  that contain item $j'$, and $\mathcal{N}_{0,j'}$ index the tests in $\mathcal{N}_0$ that contain item $j'$. Then we note the equivalent  interpretations of $G_{j,j',1}$ and $G_{j,j',2}$: $G_{j,j',1}$ represents the number of positive tests in $\mathcal{N}_{0,j'}\cup \mathcal{M}_{j,j'}$, and $G_{j,j',2}$ represents the number of negative tests in $\mathcal{N}_{0,j'}\cup \mathcal{M}_{j,j'}$. %; this alternative formulation happens to be more convenient for analyzing the near-constant weight design.

 In addition, for $j\in s$, we let $\mathcal{M}_j'$ index the tests where item $j$ not only appears as the only defective but is also placed precisely once, and let its cardinality be $M_j'=|\mathcal{M}_j'|$. Compared to Bernoulli designs, this is a new ingredient for handling near-constant weight designs
whose random placements are done with replacement. Evidently, it holds that  $M_{j}'\leq M_j$, and $M_{j}'<M_j$ occurs when item $j$ is placed in some test more than once but is still the only defective.   This situation only occurs for a small number of defective items, as we will formalize via event $\mathscr{A}_1$ below.

\subsubsection*{Reduction to two conditions}
 Similarly to the converse analysis for Bernoulli design,  by Corollary \ref{cor:restate}, MLE fails  if for some feasible  $(C,\zeta)\in(0,e^\nu)\times (0,1)$, 
the following two conditions \textbf{(C1)} and \textbf{(C2)} simultaneously hold:\begin{itemize}
    \item \textbf{(C1)} $k_{C,\zeta}\geq 1$ (i.e., $\mathcal{K}_{C,\zeta}\neq \varnothing$); 
    \item \textbf{(C2)} There exist $j\in \mathcal{K}_{C,\zeta}$ and   $j' \in \{k+1,\dotsc,p\}$ such that 
    \begin{equation}\label{eq:failurenc}
        G_{j,j',1}-G_{j,j',2}>(1-2\zeta)\frac{Cn\nu e^{-\nu}}{k}.
    \end{equation}
\end{itemize}

\subsubsection*{Some useful events}
Before proceeding, we first construct three high-probability events, deferring the proofs that they have high probability to Appendix \ref{app:nc_lemms}.  By Lemma \ref{lem3}(a) in Appendix \ref{app:nc_lemms}, for some sufficiently large constant $C_0 > 0$, the event   
\begin{align}
    \mathscr{A}_1=\Big\{|\{j\in [k]:M_j'<M_j\}|\le C_{0}\log k\Big\}. \label{eq:nc_A1}
\end{align}  
holds with $1-o(1)$ probability. 
In addition, recalling that $M_j'$ is the number of tests in which item $j$ is placed precisely once and is the only defective, 
we further define \begin{align}
    \widetilde{M}:=\sum_{j=1}^k M_j', \label{eq:def_tildeM}
\end{align} 
which represents the total number of tests having exactly one defective placement (and no repeated placements of that defective item).  Lemma \ref{lem3}(b) in Appendix \ref{app:nc_lemms} gives that the event \begin{align}
    \mathscr{A}_2=\Big\{\widetilde{M}=(1+o(1))e^{-\nu}k\Delta\Big\} \label{eq:nc_A2}
\end{align} holds with probability $1-o(1)$. 

Lastly, we deduce the high-probability behavior of $N_0$.  For $\mathcal{M}\subset s$, letting $W^{(\mathcal{M})}$ denote the number of tests in which some item from $\mathcal{M}$ is placed, Lemma \ref{lem44} in Appendix \ref{app:nc_lemms} gives that $\mathbbm{E}W^{(s)}= (1-e^{-\nu}+o(1))n$, and that for any $\delta>0$, it holds with probability at least $1-2\exp(-2\nu^{-1}\delta^2n)$ that $W^{(s)}=\mathbbm{E}W^{(s)}+\delta' n \sqrt{\frac{\ell}{k}}$ for some $|\delta'|<\delta$.   Therefore, by letting $\delta =\Theta(\frac{1}{\sqrt{k}})$, we obtain  $W^{(s)}=(1-e^{-\nu}+o(1))n$, which implies $N_0 = n-W^{(s)} = ne^{-\nu}(1+o(1))$, holds with $1-o(1)$ probability.
Thus, the event \begin{align}
    \mathscr{A}_3:=\big\{N_0 = ne^{-\nu}(1+o(1))\big\}\label{eq:nc_count_N0}
\end{align} holds with $1-o(1)$ probability. Throughout the proof, we will often apply the conditions in $\mathscr{A}_1\cap \mathscr{A}_2\cap\mathscr{A}_3$, which is justified via simple union bound. 
\subsubsection*{Condition for (C1)}

Let $\xi\in (\frac{1}{2},1)$ be such that $k^{\xi}$ is an integer. %Under $\mathscr{A}_1$, we can consider
%\begin{equation}
 %   M_j=M_j', \text{ when }1\leq j\leq k^{\xi},\label{eq:nc_A3}
%\end{equation}
%where we assume \textbf{without loss of generality} {\color{blue}[TODO: To check whether this is ok.]\footnote{\color{blue}Hint from Jon's note: show there are $\omega(n)$ possible $j$, so at least one must have $M_j'=M_j$ (I do not understand... My thought is that $\{j\in [k]: M_j=M_j'\}$ already relies on the randomness of $\Tv_s$, so we may need some uniformity.)}} that any defectives with $M'_j < M_j$ are among the final $1-k^{\zeta}$ indices.   
%Thus, for $1\leq j\leq k^\zeta$, the event $\mathscr{B}_j$ can be written as \begin{equation}
   % \mathscr{B}_j= \Big\{M_j' >\frac{M_j}{2}\Big\}=\Big\{  M_j'> \frac{M_j''}{2}\Big\}.\label{eq:Bj_new}
%\end{equation} 
 Given $s_0\subset s=[k]$, we let $\Tv_{s_0}$ be an \emph{unordered multi-set} of length $|s_0|\Delta$ whose entries are in $[n]$, with the overall multi-set representing the $|s_0|\Delta$ placements from the items in $s_0$ in an unordered manner.  If $s_0=\{j\}$ for some $j\in [k]$, then we simply write $\Tv_{\{j\}}=\Tv_j$ for the $\Delta$ placements of item $j$. %Note that under (\ref{eq:nc_A3}), when $1\le j\le k^{\xi}$, $j\in \mathcal{K}_{C,\zeta}$ holds if $M_j'=\frac{Cn\nu e^{-\nu}}{k}$ and $M_{j0}=\zeta M_{j}'$.   Given $M_j'=\frac{Cn\nu e^{-\nu}}{k}$, we have $M_{j0}\sim \text{Bin}(\frac{Cn\nu e^{-\nu}}{k},\rho)$ based on the randomness of $\Zv$, so $M_{j0}=\zeta M_j'$ can be easily analyzed. Thus, to derive a condition for \textbf{(C1)}, the main step is to study the distribution of $M_j'$.  

Compared to $M_j$,
it is more convenient to deal with $M_j'$. Therefore, instead of directly studying $\mathcal{K}_{C,\zeta}$, we let $M_{j0}'$ be the number of negative tests (with results   flipped by noise) in $\mathcal{M}_j'$ and
define \begin{align}
    \mathcal{K}'_{C,\zeta}  = \Big\{j\in [k]:M_j' = \frac{Cn\nu e^{-\nu}}{k},M_{j0}' =\frac{\zeta\cdot Cn\nu e^{-\nu}}{k} \Big\}, \label{eq:calK_prime}
\end{align} 
and let $k_{C,\zeta}'=|\mathcal{K}'_{C,\zeta}|$.
To quantify the difference between $\mathcal{K}_{C,\zeta}$ and $\mathcal{K}'_{C,\zeta}$, we observe that
\begin{align}
    \mathcal{K}'_{C,\zeta} \subset \mathcal{K}_{C,\zeta} \cup \big\{j\in[k]:M_j'<M_j\big\}. \label{eq:Kprime_NC}
\end{align}
Combining with the event $\mathscr{A}_1$ in (\ref{eq:nc_A1}), this implies 
\begin{align}\label{eq:diff_k_kpai}
     k'_{C,\zeta} \le k_{C,\zeta}+ C_0\log k. 
\end{align}
Given $M_j'=\frac{Cn\nu e^{-\nu}}{k}$, we have $M_{j0}'\sim \text{Bin}(\frac{Cn\nu e^{-\nu}}{k},\rho)$ based on the randomness of $\Zv$, so the event $M_{j0}'=\frac{\zeta\cdot Cn\nu e^{-\nu}}{k}$ can be easily analyzed.  The more significant challenge lies in analyzing the event $M_j'=\frac{Cn\nu e^{-\nu}}{k}$, and we now proceed to study the distribution of $M_j'$.

 \textbf{The distribution of $M_j'$:} We follow the idea in \cite[Lemma 3.3]{coja2020information} of conditioning on $\Tv_s$ (i.e., $\Tv_{[k]}$) and using the symmetry of the test design to deduce that the conditional assignments of tests to items are uniform, i.e., every partition of $\Tv_s$ into $k$ sets of size-$\Delta$ is equally likely.  In particular, note that $\widetilde{M}$ defined in (\ref{eq:def_tildeM}) is a known quantity from the given $\Tv_s$, and given $\widetilde{M}$,  
 the number of such tests assigned to any given $j \in s$, which has been defined as $M_j'$, follows a hypergeometric distribution: There are $\Delta$ draws (without replacement) from $k\Delta$ objects (that represent the overall $k\Delta$ placements in $\Tv_s$), $\widetilde{M}$ of which are of a special type {(namely, the placement is to a test that is connected to precisely one defective item)}. For instance, regarding the allocation of these $k\Delta$ placements to item 1, we have $(M_1'|\Tv_{s})\sim \mathrm{Hg}(k\Delta,\widetilde{M},\Delta)$. To understand the behavior of the first $k^{\zeta}$ items simultaneously, we interpret the above-mentioned uniform allocation from $\Tv_s$ as being done sequentially:
\begin{itemize}
    \item Given $\Tv_s$, we allocate a random size-$\Delta$ subset to item $1$, which yields $(M_1'|\Tv_{s})\sim \mathrm{Hg}(k\Delta,\widetilde{M},\Delta)$ as described above.
    \item Given $\Tv_s$ and $\Tv_1$ (the placements for item $1$), we allocate a random size-$\Delta$ subset of the \emph{reduced multi-set $\Tv_s \setminus \Tv_1$} to item $2$, which gives  $M_2'|(\Tv_s,\Tv_1)\sim \mathrm{Hg}\big((k-1)\Delta, \widetilde{M}-M_1',\Delta\big)$ (where $M'_1$ is determined via $\Tv_s$ and $\Tv_1$).
    % \item Given $\Tv_S,~\Tv_1$ and $\Tv_2$, note that $M_1''$ and $M_2''$ are known, and analogously \emph{the randomness of the placement of item 3} gives $M_3''|(\Tv_S,\Tv_1,\Tv_2)\sim \mathrm{Hg}\big((k-2)\Delta,\widetilde{M}-M_1''-M_2'',\Delta\big)$; 
    \item $\dotsc$
    \item Given $\Tv_s,\Tv_1,...,\Tv_{k^{\xi}-1}$, the quantities $M_1',M_2',...,M'_{k^{\xi}-1}$ are known, and an analogous argument gives $M'_{k^{\xi}}|(\Tv_s,\{\Tv_j\}_{j=1}^{k^{\xi}-1})\sim \mathrm{Hg}\big((k-k^{\xi}+1)\Delta,\widetilde{M}-\sum_{j=1}^{k^{\xi}-1}M_j',\Delta\big)$.
\end{itemize}
  Similarly to the analysis of the Bernoulli design, we only consider this procedure up to index $k^{\xi}$, since going all the way up to $k$ would introduce non-negligible dependencies.  For $\xi\in (\frac{1}{2},1)$ we have $k-j = (1+o(1))k$ for any $0\leq j\leq k^{\xi}$. Then, we utilize the event $\mathscr{A}_2$ in (\ref{eq:nc_A2}) that specifies  $\widetilde{M}=(1+o(1))e^{-\nu}k\Delta$, and write 
  \begin{align}
       (1+o(1))e^{-\nu}k\Delta=\widetilde{M}\geq \widetilde{M}-\sum_{j=1}^i M_j'\geq \widetilde{M}- k^{\xi}\Delta= (1+o(1))e^{-\nu}k\Delta
  \end{align}
  for any $0\leq i\leq k^{\xi}$. 
  Therefore, the following holds for any $1\leq j\leq k^{\xi}$: 
  \begin{equation}
    M_j'|(\Tv_s,\Tv_1,...,\Tv_{j-1})\sim \mathrm{Hg}((1+o(1))k\Delta,(1+o(1))e^{-\nu}k\Delta,\Delta). 
\end{equation}
The Chernoff bound and matching anti-concentration for the hypergeometric distribution (see Lemma \ref{lem4} in Appendix \ref{app:nc_lemms}) then gives that
\begin{equation}
    \PP\Big(M_j'=\frac{Cn\nu e^{-\nu}}{k}\big|(\Tv_s,\Tv_1,...,\Tv_{j-1})\Big)=\exp\Big(-\frac{\nu n}{k}\big[D(Ce^{-\nu}\|e^{-\nu})+o(1)\big]\Big)  \label{eq:Hg_summary}
\end{equation}
for any $1\le j\le k^{\xi}$.

 \textbf{Bounding $\mathbb{P}(j\in \mathcal{K}'_{C,\zeta})$ from below:}
Under $ M_j'=\frac{Cn\nu e^{-\nu}}{k}$, with respect to the randomness of $\Zv$ we have $M_{j0}'\sim \text{Bin}\big(\frac{Cn e^{-\nu}\nu}{k},\rho\big)$, which should equal $\frac{\zeta \cdot Cn\nu e^{-\nu}}{k}$ to ensure $j\in\mathcal{K}'_{C,\zeta}$. Accordingly,  for $1\le j\le k^{\xi}$, we proceed as follows:
\begin{align}\label{eq:lower_kczetapai1}
    &\PP\big(j\in\mathcal{K}'_{C,\zeta}\big|(\Tv_s,\Tv_1,...,\Tv_{j-1})\big)\\
    &\quad =  \PP\Big(M_j'=\frac{Ce^{-\nu}\nu n}{k},M_{j0}'=\frac{\zeta\cdot Cn\nu e^{-\nu}}{k}\big|(\Tv_s,\Tv_1,...,\Tv_{j-1})\Big)\\
    &\quad=\PP\Big(M_j'=\frac{Cn\nu e^{-\nu}}{k}\big|(\Tv_s,\Tv_1,...,\Tv_{j-1})\Big)\cdot \mathbb{P}\Big(\text{Bin}\Big(\frac{C\nu e^{-\nu}n}{k},\rho\Big)= \frac{\zeta\cdot C\nu e^{-\nu}n}{k}\Big)\\ \label{eq:nc_substi}
    &\quad \geq  \exp\left(-\frac{\nu n e^{-\nu}}{k}\Big[e^\nu D\big(Ce^{-\nu}\|e^{-\nu}\big)+C\cdot D(\zeta\|\rho)+o(1)\Big]\right)\\
    &\quad :=\exp\Big(-\frac{\nu n e^{-\nu}}{k}\big[f_1(C,\zeta,\rho,\nu)+o(1)\big]\Big):=P_{\text{in}},\label{eq:lower_kczetapai2}
\end{align}
where in (\ref{eq:nc_substi}) we substitute (\ref{eq:Hg_summary}) and apply  (\ref{eq:anti2}) in Lemma \ref{binoconcen} (with leading factor absorbed into the exponent), and in \eqref{eq:lower_kczetapai2} we define 
\begin{align}\label{eq:nc_f1}
    f_1(C,\zeta,\rho,\nu)= e^\nu  D\big(Ce^{-\nu}\|e^{-\nu}\big)+ C\cdot D(\zeta\|\rho).
\end{align}

\textbf{Deriving the condition for (C1):} Noting that 
\begin{align}
    k'_{C,\zeta}= \sum_{j=1}^k \mathbbm{1}(j\in \mathcal{K}'_{C,\zeta})\ge \sum_{j=1}^{k^{\xi}}\mathbbm{1}(j\in \mathcal{K}'_{C,\zeta}):=k'_{C,\zeta,\xi}, \label{eq:stochasti_domi}
\end{align}
and from (\ref{eq:lower_kczetapai2}) we have that $(k'_{C,\zeta,\xi}|\Tv_s)$ stochastically dominates $\text{Bin}(k^{\xi},P_{\text{in}})$, and thus
\begin{align}
    \label{eq:stat_domi_1}
    \mathbb{P}\big(k'_{C,\zeta,\xi}\le z \,\big|\, \Tv_s\big) \le \mathbb{P}\big(\text{Bin}(k^{\xi},P_{\text{in}})\le z\big)
\end{align}
for any $z$. Combining this property with  (\ref{eq:diff_k_kpai}) and (\ref{eq:stochasti_domi}), we can proceed as follows:
\begin{align}
    \PP\big(k_{C,\zeta}=0 \,|\, \Tv_s\big) &\le \mathbb{P}\Big(k_{C,\zeta}'\le C_0\log k\big|\Tv_s\Big)\\&\le \mathbb{P}\Big(k_{C,\zeta,\xi}'\le C_0\log k\big|\Tv_s\Big)\\
   &\le\mathbb{P}\Big(\text{Bin}(k^{\xi},P_{\text{in}})\le C_0\log k\Big),%\\
   %\le&\Big(1-P_{\text{in}}\Big)^{k^{\xi}}\le \exp\big(-k^{\xi}P_{\text{in}}\big),
\end{align}
and since this holds regardless of the conditioning variable, we obtain 
\begin{align}\label{eq:nc_con_c1}
    %\mathbb{P}\Big(\textbf{(C1)}\text{ does not hold}\Big)=
    \mathbb{P}\big(k_{C,\zeta}=0\big)\le \mathbb{P}\Big(\text{Bin}(k^{\xi},P_{\text{in}})\le C_0\log k\Big).
\end{align}
We now claim that to ensure that \textbf{(C1)} holds with $1-o(1)$ probability, it suffices to ensure that
\begin{align}
\label{eq:nc_converse_con1}k^{\xi}P_{\text{in}}=\exp\Big(\xi\log k-\frac{\nu n e^{-\nu}[f_1(C,\zeta,\rho,\nu)+o(1)]}{k}\Big)\ge \exp\Big(\eta'_1 \log\frac{p}{k}\Big) \to \infty 
\end{align}
for some $\eta'_1>0$. Indeed, if (\ref{eq:nc_converse_con1}) holds, then by the Chernoff bound (see (\ref{eq:chernoff2}) in Lemma \ref{binoconcen}), we have 
\begin{align}
    \mathbb{P}\Big(\text{Bin}(k^{\xi},P_{\text{in}})\le C_0\log k\Big)&\le \exp \Big(-k^{\xi}P_{\text{in}}\Big[\frac{C_0\log k}{k^{\xi}P_{\text{in}}}\log\Big(\frac{C_0\log k}{k^{\xi}P_{\text{in}}}\Big) + 1-\frac{C_0\log k}{k^{\xi}P_{\text{in}}}\Big]\Big)\\
    &\le \exp\Big(-\frac{1}{2}k^{\xi}P_{\text{in}}\Big),
\end{align}
where the second line holds since $\frac{C_0\log k}{k^{\xi}P_{\text{in}}} \le \frac{C_0\log k}{\exp(\eta'_1\log\frac{p}{k})}=o(1)$ for large enough $k$. The claim thus follows by combining with (\ref{eq:nc_con_c1}).

It remains to deduce the threshold from (\ref{eq:nc_converse_con1}). By $k=\Theta(p^\theta)$ we have $\log k\sim\frac{\theta}{1-\theta}\log\frac{p}{k}$, and then by choosing $\xi$ arbitrarily close to $1$ we obtain (\ref{eq:nc_converse_con1}) provided that the following holds for some $\eta_1 > 0$:
%It turns out that, the calculation is the same as Bernoulli design in (\ref{eq:substi1})--(\ref{eq:P2_bernoulli}) {\color{blue}[but actually quite some details to add!]}, and so to ensure \textbf{(C1)} we need the same threshold as in (\ref{eq:ber_fail_1}):
\begin{equation}\label{eq:nc_c1_bound}
    n\leq (1-\eta_1)\frac{\frac{\theta}{1-\theta}k\log\frac{p}{k}}{\nu e^{-\nu}f_1(C,\zeta,\rho,\nu)}.
\end{equation}
% for some $\eta_1>0$ possibly different from the one in (\ref{eq:nc_converse_con1}).  

\subsubsection*{Condition for (C2)}

The subsequent analysis is built on condition \textbf{(C1)}, which ensures that $k_{C,\zeta}\ge 1$, and we consider a single specific index $j\in \mathcal{K}_{C,\zeta}$. %We seek to formulate condition \textbf{(C2)} and need to first quantify $N_0$ (the number of tests containing no defectives).  The rest of the analysis is conditioned on $(\Xv_s, \Yv)$ with $k_{C,\zeta}\ge 1$, and the only randomness is $\Xv_{s^c}$.  
We will study the placement of each non-defective item $j'$ for $j' \in \{k+1, \dotsc, p\}$. For the given $j\in\mathcal{K}_{C,\zeta}$ and a fixed $j' \in \{k+1, \dotsc, p\}$, we define $R_{j'}=|\mathcal{N}_{0,j'}\cup\mathcal{M}_{j,j'}|$ as the number of tests in $\mathcal{N}_0\cup \mathcal{M}_j$  that contain item $j'$, and define $R_{j'}'$ as the number of placements of item $j'$ that are in $\mathcal{N}_0\cup \mathcal{M}_j$. It is evident that $R_{j'}'\ge R_{j'}$ always holds, and $R_{j'}'>R_{j'}$   when item $j'$ is placed more than once into some test in $\mathcal{N}_0 \cup \mathcal{M}_j$. To formulate condition \textbf{(C2)}, we are primarily interested in $R_{j'}$, but it is more convenient to first study $R_{j'}'$. Recall that the event $\mathscr{A}_3$ in  (\ref{eq:nc_count_N0})  states that
\begin{align}
    ne^{-\nu}(1+o(1))=N_0\le |\mathcal{N}_0 \cup \mathcal{M}_j|\le N_0+\ell\Delta=ne^{-\nu}(1+o(1)).
\end{align} 
Under this event, a placement of item $j'$ increments $R_{j'}'$ by $1$ with probability $e^{-\nu}(1+o(1))$, and we thus have $R_{j'}'\sim\text{Bin}(\Delta,e^{-\nu}(1+o(1)))$ with respect to the randomness in placing item $j'$ into tests.

\textbf{Bounding the difference between $R_{j'}$ and $R_{j'}'$:} We are ultimately interested in $R_{j'}$, but we have characterized $R_{j'}'$ as an intermediate step.  Accordingly, we proceed to quantify the difference between $R_{j'}$ and $R_{j'}'$. 
We envision that the $\Delta$ placements of  item $j'$ are done sequentially, and let $D_{j'}$ be the number of ``collisions'', defined as  
\begin{align}
    D_{j'} = \sum_{i=1}^\Delta \mathbbm{1}\big(\text{the }i\text{-th placement coincides with the }i_1\text{-th placement for some }i_1<i\big).
\end{align}
Then we have
 \begin{equation}\label{eq:441}
    R'_{j'}-D_{j'}\le R_{j'}\leq R'_{j'},
\end{equation}
and we will show that $D_{j'}$ is asymptotically negligible. 
Specifically, observe that the $i$-th placement increments $D_{j'}$ by $1$ with probability less than $\frac{\Delta}{n}$, so for any $D_0>0$, a union bound (over the $\binom{\Delta}{D_0}$ possibilities of the $D_0$ placements that increment $D_{j'}$) 
gives 
\begin{align}\label{eq:Dj'bound}
    \mathbb{P}\Big(D_{j'} \geq D_0\Big) \le \binom{\Delta}{D_0}\Big(\frac{\Delta}{n}\Big)^{D_0} \leq \Big(\frac{e\Delta^2}{D_0n}\Big)^{D_0}.
\end{align}
%Therefore, taking $D_0=100$ or other large enough constant, since $\Delta\asymp \log k$, under large enough $k$ we obtain {\color{blue}[If necessary, we can use scaling like $D_0=\frac{1}{\theta}$ to overcome a union bound over $p-k$ non-defectives]}
%\begin{align}
 %   \mathbb{P}\Big(D_{j'}\geq 100\Big) \le k^{-99}. 
%\end{align}
%It seems that the blue texts above are necessary --- it's more amenable to first give uniform bound on $\sup_{k+1\leq j'\leq p}D_{j'}$ before proceeding. 
We take $D_0 = \frac{10}{\theta}$ in (\ref{eq:Dj'bound}) and apply a union bound over $j' \in \{k+1,\dotsc,p\}$, yielding 
\begin{align}
    \mathbb{P}\Big(\max_{j' \in \{k+1,\dotsc,p\}}D_{j'} \ge \frac{10}{\theta}\Big) &\le  p\cdot\Big(\frac{e\theta\Delta^2}{10n}\Big)^{\frac{10}{\theta}}\\&= \exp \Big(\log p + \frac{10}{\theta}\log\Big(\frac{e\theta \nu^2n}{10k}\Big)-\frac{10}{\theta}\log k\Big) \\
    &\le k^{-5}, \label{eq:colli_event}
\end{align}
where (\ref{eq:colli_event}) holds when $k$ is large enough because $p=\Theta(k^{1/\theta})$ and $n=\Theta(k\log k)$ (see Section \ref{sec:n_scaling}). In the subsequent analysis, we implicitly suppose (via a union bound) that we are on the event   \begin{align}\label{eq:collision_event}
    D_{j'}< \frac{10}{\theta},~\forall j' \in \{k+1,\dotsc,p\},
\end{align} 
which holds with probability at least $1-k^{-5}$.

\textbf{Studying the condition for (\ref{eq:failurenc}):}  Given $R_{j'}$,  the randomness of the noise $\Zv$ gives
\begin{align}
    G_{j,j',1}\sim \text{Bin}(R_{j'},\rho),~~G_{j,j',2}=R_{j'} - G_{j,j',1}.
\end{align}
Thus, (\ref{eq:failurenc}) can be expressed as $2G_{j,j',1}-R_{j'}> (1-2\zeta)\frac{Cn\nu e^{-\nu}}{k}$, or equivalently
\begin{align}\label{eq:nc_con_c2_equal}
 G_{j,j',1} > \Big(\frac{1}{2}-\zeta\Big)\frac{Cn\nu e^{-\nu}}{k}+\frac{1}{2}R_{j'}. 
\end{align}
By introducing a parameter $d$ (to be chosen later) that satisfies
\begin{align}\label{eq:d_great_A}
     |C(1-2\zeta)|\le d\le e^\nu,
\end{align}
we consider the event $R_{j'}'=\frac{dn\nu e^{-\nu}}{k}$ and proceed as follows:
\begin{align}
     &\mathbb{P}\Big(G_{j,j',1} > \big(\frac{1}{2}-\zeta\big)\frac{Cn\nu e^{-\nu}}{k}+\frac{1}{2}R_{j'}\Big)\\
     &\ge  \mathbb{P}\Big(R_{j'}' = \frac{dn\nu e^{-\nu}}{k}\Big) \mathbb{P}\Big(\text{Bin}(R_{j'},\rho)> \big(\frac{1}{2}-\zeta\big)\frac{Cn\nu e^{-\nu}}{k}+\frac{1}{2}R_{j'}\Big|R_{j'}' = \frac{dn\nu e^{-\nu}}{k}\Big)\\\nonumber
     &\ge   \mathbb{P}\Big(\text{Bin}\big(\frac{\nu n}{k},e^{-\nu}+o(1)\big) = \frac{dn\nu e^{-\nu}}{k}\Big)\\\label{eq:nc_turn_bin}&\quad\quad\quad\quad\cdot \mathbb{P}\Big(\text{Bin}\big(\frac{(1+o(1))dn \nu e^{-\nu}}{k},\rho\big)> \frac{d+C(1-2\zeta)}{2}\frac{n\nu e^{-\nu}}{k}\Big)\\
     &\ge \exp\left(-\frac{\nu n e^{-\nu}}{k}\Big[e^\nu \cdot D(de^{-\nu}\|e^{-\nu})+d\cdot D\Big(\frac{1}{2}+\frac{C(1-2\zeta)}{2d}\big\|\rho\Big)+o(1)\Big]\right)\label{eq:nc_use_anti}\\
     &:=\exp\Big(-\frac{\nu n e^{-\nu}}{k}\big[g(C,\zeta,d,\rho,\nu)+o(1)\big]\Big),\label{eq:ncc2_lower}
\end{align}
where in (\ref{eq:nc_turn_bin}) we use $R_{j'}'\sim\text{Bin}(\Delta,e^{-\nu}+o(1))$ along with the event (\ref{eq:collision_event}) and $R_{j'}'=\frac{dn\nu e^{-\nu}}{k}$, from which (\ref{eq:441}) gives $R_{j'}=\frac{dn\nu e^{-\nu}(1+o(1))}{k}$ (since $\frac{n}{k}\to\infty$ and $\frac{10}{\theta}$ is   a constant); then in (\ref{eq:nc_use_anti}) we apply (\ref{eq:anti2}) in Lemma \ref{binoconcen} to the two probability terms (with leading factors absorbed into the exponent), and in \eqref{eq:ncc2_lower} we introduce the following shorthand:
\begin{align}\label{eq:nc_con_def_g}
    g(C,\zeta,d,\rho,\nu)=e^{\nu}\cdot D\big(de^{-\nu}\|e^{-\nu}\big)+d\cdot D\Big(\frac{1}{2}+\frac{C(1-2\zeta)}{2d}\big\|\rho\Big).
\end{align}
This depends on $(C,\zeta)$ only through $A:=C(1-2\zeta)$, and can be expanded as  
\begin{align}
    g(C,\zeta,d,\rho,\nu)&=d\log d +(e^\nu-d)\log \frac{1-de^{-\nu}}{1-e^{-\nu}}+\frac{d+A}{2}\log\frac{d+A}{2d\rho}+\frac{d-A}{2}\log\frac{d-A}{2d(1-\rho)} .
\end{align}

  \textbf{Optimizing $d$:} Our next step is to specify $d$ that minimizes $g(C,\zeta,d,\rho,\nu)$ to render the tightest  lower bound  (\ref{eq:ncc2_lower}). Note that range of $d$ is given by $|A|\le d\le e^\nu$ in  (\ref{eq:d_great_A}), 
and we define 
\begin{align}\label{eq:dstar_nc_con}
    d^* = \text{arg}\min~g(C,\zeta,d,\rho,\nu),~~\text{subject to }|A|\le d\le e^{\nu}. 
\end{align}
To pinpoint $d^*$, we differentiate $g(C,\zeta,d,\rho,\nu)$ with respect to $d$ to obtain
\begin{align}\label{eq:nc_deri_g}
    \frac{\partial g}{\partial d} = \log\Big(\frac{\sqrt{d^2-A^2}}{2\sqrt{\rho(1-\rho)}}\Big)-\log\Big(\frac{e^\nu -d}{e^\nu -1}\Big).
\end{align}
Note that $\frac{\partial g}{\partial d}$ has the same sign as $g^{(1)}(d) := \frac{d^2-A^2}{4\rho(1-\rho)}-\frac{(e^\nu-d)^2}{(e^\nu-1)^2}$, which can be expanded as
\begin{align}
  g^{(1)}(d) = \Big(\frac{1}{4\rho(1-\rho)}-\frac{1}{(e^\nu-1)^2}\Big)d^2 +\frac{2e^\nu}{(e^\nu-1)^2}d -\Big(\frac{A^2}{4\rho(1-\rho)}+\frac{e^{2\nu}}{(e^\nu-1)^2}\Big), \label{eq:g1_def}
\end{align}
where we only regard $d$ as the variable. We observe that $g^{(1)}(0)<0$, and perform some algebra to find that  \begin{align}\label{eq:nc_zero_obser}
    g^{(1)}(|A|)=-\Big(\frac{e^\nu-A}{e^\nu-1}\Big)^2\le 0,~g^{(1)}(e^\nu)=\frac{e^{2\nu}-A^2}{4\rho(1-\rho)}\ge 0.
\end{align}
Based on these observations, we determine $d^*$ by considering the following cases: 
\begin{itemize}
    \item When $4\rho(1-\rho)< (e^\nu-1)^2$, $g^{(1)}(d)$ has a unique zero in $d>0$. By (\ref{eq:nc_zero_obser}), this zero falls in the range of feasible $d$ (i.e., $[|A|,e^\nu]$), so it is precisely the desired $d^*$, yielding   
    \begin{align}
    d^* 
    &=\frac{-e^\nu +\sqrt{e^{2\nu}+(\frac{(e^\nu-1)^2}{4\rho(1-\rho)}-1)(\frac{A^2(e^{\nu}-1)^2}{4\rho(1-\rho)}+e^{2\nu})}}{\frac{(e^\nu -1)^2}{4\rho(1-\rho)}-1},~~\text{if }4\rho(1-\rho)< (e^\nu-1)^2.\label{eq:nc_dstar}
\end{align}
\item When $4\rho(1-\rho)=(e^\nu-1)^2$, the right-hand side of (\ref{eq:g1_def}) becomes linear in $d$, and solving gives
 \begin{align}
    d^* 
    &= \frac{A^2+e^{2\nu}}{2e^\nu},~~\text{if }4\rho(1-\rho)= (e^\nu-1)^2.\label{eq:nc_dstar_case2}
\end{align}
\item When $ 4\rho(1-\rho)>(e^\nu-1)^2$, we can verify that the two zeros of $g^{(1)}(d)$ both fall  in $(0,\infty)$. Combining with (\ref{eq:nc_zero_obser}), we find that $d^*$ equals the smaller zero of $g^{(1)}(d)$:
\begin{align}
    d^* 
    &=\frac{-e^\nu +\sqrt{e^{2\nu}+(\frac{(e^\nu-1)^2}{4\rho(1-\rho)}-1)(\frac{A^2(e^{\nu}-1)^2}{4\rho(1-\rho)}+e^{2\nu})}}{\frac{(e^\nu -1)^2}{4\rho(1-\rho)}-1},~~\text{if }4\rho(1-\rho)> (e^\nu-1)^2.\label{eq:nc_dstar_case3}
\end{align}
\end{itemize}
Therefore, $d^*$ defined in (\ref{eq:dstar_nc_con}) is given by 
\begin{align}
    d^* = \begin{cases}
       ~~~ (\ref{eq:nc_dstar}),~~~~\text{when }4\rho(1-\rho)\neq (e^\nu -1)^2 \\
        ~~~\frac{A^2+e^{2\nu}}{2e^\nu},~~\text{when }4\rho(1-\rho)= (e^\nu -1)^2 
    \end{cases}.\label{eq:d_star_pinpoint}
\end{align}
A simple limiting argument verifies that $d^*$ is continuous with respect to $(\rho,\nu)$: When $a:=\frac{(e^\nu-1)^2}{4\rho(1-\rho)}\to 1$, we have
\begin{align}
    &\lim_{a\to 1} \frac{-e^\nu+\sqrt{e^{2\nu}+(a-1)(A^2a+e^{2\nu})}}{a-1}\\
    =& \lim_{a\to 1}\frac{(a-1)(A^2a+e^{2\nu})}{(a-1)\big(\sqrt{e^{2\nu}+(a-1)(A^2a+e^{2\nu})}+e^\nu\big)}= \frac{A^2+e^{2\nu}}{2e^\nu}.
\end{align}
In summary, we can introduce the shorthand
\begin{align}\label{eq:f2_con_nc}
    f_2(C,\zeta,\rho,\nu) = g(C,\zeta,d^*,\rho,\nu)
\end{align}
and arrive at the following conclusion: For fixed $j\in\mathcal{K}_{C,\zeta}$ and $j'\in [p]\setminus s$ we have
\begin{align}
   \mathbb{P}\Big(G_{j,j',1} > \big(\frac{1}{2}-\zeta\big)\frac{Cn\nu e^{-\nu}}{k}+\frac{1}{2}R_{j'}\Big)\geq \exp \Big(-\frac{\nu n e^{-\nu}}{k}[f_2(C,\zeta,\rho,\nu)+o(1)]\Big).
\end{align}
Recall that \textbf{(C2)} requires that (\ref{eq:nc_con_c2_equal}) holds for some $j'\in [p]\setminus s$, and the placements of the $p-k$ non-defective items are independent. Therefore, by arguments analogous to (\ref{eq:ber_conclu_c2_start})--(\ref{eq:ber_fail_2}), we obtain the threshold
\begin{equation}\label{eq:nc_c2_bound}
    n\leq \frac{(1-\eta_2)\frac{1}{1-\theta}k\log\frac{p}{k}}{\nu e^{-\nu}\big(f_2(C,\zeta,\rho,\nu)+o(1)\big)}
\end{equation}
for some $\eta_2>0$, which suffices for ensuring \textbf{(C2)} holds with $1-o(1)$ probability. 

\subsubsection*{Wrapping up}
Recall that for MLE to fail, it suffices that \textbf{(C1)} and \textbf{(C2)} simultaneously hold for some feasible pair $(C,\zeta)$. By the above analyses, for given $(C,\zeta)$ we have \textbf{(C1)} and \textbf{(C2)} with $1-o(1)$ probability if both 
(\ref{eq:nc_c1_bound}) and (\ref{eq:nc_c2_bound}) hold, and by merging $\eta_1,\eta_2$ into a single parameter $\eta$, this can be written as   
\begin{align}
    n\le \frac{(1-\eta)k\log\frac{p}{k}}{(1-\theta)\nu e^{-\nu}}\frac{1}{\max\{\frac{1}{\theta}f_1(C,\zeta,\rho,\nu),f_2(C,\zeta,\rho,\nu)\}}
\end{align}
for arbitrarily small $\eta>0$.   Optimizing over $(C,\zeta)$ establishes the second term in \eqref{eq:NC_threshold}.  

% Furthermore, we 
% optimize $(C,\zeta)$  and  combine with the capacity bound $\frac{k\log\frac{p}{k}}{H_2(e^{-\nu}\star\rho)-H_2(\rho)}$ {\color{blue}[To fill]} to arrive at the desired converse bound
% \begin{align}\label{eq:nc_converse_threshold}
%   & (1-\eta) k\log\frac{p}{k}\cdot\max \Big\{\frac{1}{H_2(e^{-\nu}\star\rho)-H_2(\rho)},\nonumber\\&\quad\quad\quad\quad\quad\quad\quad\frac{1}{(1-\theta)\nu e^{-\nu}\min_{\substack{C\in(0,e^{\nu})\\\zeta\in(0,1)}}\max\{\frac{1}{\theta}f_1(C,\zeta,\rho,\nu),f_2(C,\zeta,\rho,\nu)\}}\Big\},
% \end{align}

\section{Low-$\ell$ Achievability Analysis for Near-Constant Weight Designs}  \label{app:nc_achi}

Recall that in this part of the analysis we only need to consider $\ell\in[1,\frac{k}{\log k}]$. We will establish the threshold for $n$ above which the restricted MLE decoder (\ref{eq:mle}) has $o(1)$ probability of failing.  As with the Bernoulli design, this is more challenging than the converse.  We suppose that the true defective set is $s=[k]$ without loss of generality.

% This is significantly more involved than the converse analysis because we need to consider the success/failure of MLE (\ref{eq:mle}) associated with {\it all}   $\ell\in [1,\frac{k}{\log k}]$ and the corresponding feasible $(C,\zeta)$ under a specific $\ell$.    By symmetry we only consider the case that $s=[k]$ without loss of generality. Thus,    items $1,2,...,k$ are defective and items $k+1,k+2,...,p$ are non-defective. 

\subsubsection*{Notation}

We first recap some notation that we also used when studying the Bernoulli design.  
For $\mathcal{J}\subset s$ with $|\mathcal{J}|=\ell$ we will use the sets $\mathcal{M}_{\mathcal{J}},\mathcal{M}_{\mathcal{J}0},\mathcal{M}_{\mathcal{J}1},\mathcal{N}_0,\mathcal{N}_{00},\mathcal{N}_{01}$ and their corresponding cardinalities from Lemma \ref{lem:restate}. 
 Conditioned on $\Xv_s$ and $\Yv$, these sets and   their cardinalities are deterministic. For $\ell\in [1,\frac{k}{\log k}]$, we say that $(C,\zeta)\in [0,\infty)\times[0,1]$ is \emph{feasible} if $\frac{Cn\nu e^{-\nu}\ell}{k}$ and $\frac{\zeta\cdot Cn\nu e^{-\nu}\ell}{k}$ are integers, subject to the restriction that $(C,\zeta)=(0,0)$ is the only feasible pair with $C=0$. For a feasible pair $(C,\zeta)$, we define the set 
    $$\mathcal{K}_{\ell,C,\zeta}=\Big\{\mathcal{J}\subset s, \,:\, |\mathcal{J}|=\ell, M_{\mathcal{J}}=\frac{Cn\nu e^{-\nu}\ell}{k},~M_{\mathcal{J}0}=\frac{\zeta\cdot Cn\nu e^{-\nu}\ell}{k}\Big\}$$ 
 with cardinality $k_{\ell,C,\zeta}=|\mathcal{K}_{\ell,C,\zeta}|$, as defined in (\ref{eq:defi_KellCzeta}). For given $\mathcal{J}\in \mathcal{K}_{\ell,C,\zeta}$ and $\mathcal{J}'\subset [p]\setminus s$ with $|\mathcal{J}'|=\ell$, recall that $G_{\mathcal{J},\mathcal{J}',1}$ denotes the number of tests in $\mathcal{N}_{01}\cup \mathcal{M}_{\mathcal{J}1}$ that contain at least one item from $\mathcal{J}'$, and ${G}_{\mathcal{J},\mathcal{J}',2}$ denotes the number of tests in $\mathcal{N}_{00}\cup \mathcal{M}_{\mathcal{J}0}$ that contain at least one item from $\mathcal{J}'$.

We also need some additional notation when studying the near-constant weight design.
For $\mathcal{J}'\subset [p]\setminus s$ with $|\mathcal{J}'|=\ell$, we let $\mathcal{M}_{\mathcal{J},\mathcal{J}'}$ index the tests in $\mathcal{M}_{\mathcal{J}}$ that contain some item from $\mathcal{J}'$, and $\mathcal{N}_{0,\mathcal{J}'}$ index the tests in $\mathcal{N}_0$ that contain some item in $\mathcal{J}'$.  With these notations, we can equivalently interpret $G_{\mathcal{J},\mathcal{J}',1}$ as the number of positive tests 
in $\mathcal{M}_{\mathcal{J},\mathcal{J}'}\cup \mathcal{N}_{0,\mathcal{J}'}$, and   $G_{\mathcal{J},\mathcal{J}',2}$ as the number of negative tests in $\mathcal{M}_{\mathcal{J},\mathcal{J}'}\cup \mathcal{N}_{0,\mathcal{J}'}$. Furthermore, we note that
\begin{align}
\label{eq:nega_test_J'}G_{\mathcal{J},\mathcal{J}',2}=M_{\mathcal{J},\mathcal{J'}}+N_{0,\mathcal{J}'}-G_{\mathcal{J},\mathcal{J}',1}.
\end{align}

\subsubsection*{Reduction to two conditions}

We seek to bound the probability that restricted MLE fails, first considering each $\ell\in[1,\frac{k}{\log k}]$ separately (before later applying a union bound).  As before, we  observe the following via Lemma \ref{lem:restate}(b): If restricted MLE fails and returns some $s'$ with $|s\setminus s'|\in [1,\frac{k}{\log k}]$, then for $\ell:= |s\setminus s'|$, there exists some feasible $(C,\zeta)\in [0,\infty)\times [0,1]$ with respect to the specific $\ell$  
such that \textbf{(C1)} and \textbf{(C2)} below
simultaneously hold:
%it suffices that for all $1\leq \ell\leq \frac{k}{\log k}$ and  for all feasible $(C,\zeta)\in[0,\infty)\times [0,1]$ under this specific $\ell$,   {\it either} \textbf{(C1)} {\it or} \textbf{(C2)} below holds: 
\begin{itemize}
    \item \textbf{(C1)} $k_{\ell,C,\zeta}\ge 1$ (i.e., $\mathcal{K}_{\ell,C,\zeta}\neq \varnothing$); 
    \item \textbf{(C2)} There exist $\mathcal{J}\in\mathcal{K}_{\ell,C,\zeta}$ and some  $\mathcal{J}'\subset [p]\setminus s$ with $|\mathcal{J}'|=\ell$ such that
    \begin{equation}
        \label{eq:354}G_{\mathcal{J},\mathcal{J}',1}-G_{\mathcal{J},\mathcal{J}',2} \ge  (1-2\zeta)\frac{Cn\nu e^{-\nu}\ell}{k}.
    \end{equation}    
     By substituting (\ref{eq:nega_test_J'}), this can be equivalently formulated as 
   \begin{align}
        \label{eq:use_nega_test}G_{\mathcal{J},\mathcal{J}',1}\ge\Big(\frac{1}{2}-\zeta\Big)\frac{Cn\nu e^{-\nu}\ell}{k}+\frac{1}{2}\big(N_{0,\mathcal{J}'}+M_{\mathcal{J},\mathcal{J}'}\big).
   \end{align}
\end{itemize}
%Note that $G_{\mathcal{J},\mathcal{J}',1}$ ($G_{\mathcal{J},\mathcal{J}',2}$, resp.) can be equivalently interpreted as the number of   tests in $\mathcal{N}_{01}\cup\mathcal{M}_{\mathcal{J}1}$ ($\mathcal{N}_{00}\cup\mathcal{M}_{\mathcal{J}0}$, resp.) that contain some item from $\mathcal{J}'$, thus $G_{\mathcal{J},\mathcal{J}',1},G_{\mathcal{J},\mathcal{J}',2}$ here have the same meaning as the corresponding notations in the proof of the Bernoulli design achievablity bound (Appendix \ref{app:ber_achi}). As a result, the above claim can be proved by exactly the same arguments therein, and we omit the details. 
In order to show that the overall failure probability is $o(1)$, we will first show that the probability of both $\textbf{(C1)}$ and $\textbf{(C2)}$ holding for some $(\ell,C,\zeta)$ is suitably small. Observe that for any length-$\ell$ $\mathcal{J}\subset [p]$, we have $M_{\mathcal{J}}\le \ell\Delta=\frac{\ell\nu n}{k}$, which implies that $k_{\ell,C,\zeta}=0$ holds for any $C>e^{\nu}$. Thus, \textbf{(C1)} does not hold when $C>e^{-\nu}$, so from now on  we concentrate on a specific $ \ell \in [1,\frac{k}{\log k}]$ and a given feasible pair $(C,\zeta)\in [0,e^\nu]\times [0,1]$.

%\subsubsection*{The cases of $(C,\zeta)\in [C_0,\infty)\times (0,1)$ under a given $\ell$}

%In this part, we first address the cases of $C\ge C_0$ for sufficiently large $C_0$ by showing that \textbf{(C1)} holds (for all these large $C$). {\color{blue}[Different from Bernoulli design, this step is not needed because we always have $M_{\mathcal{J}}\le \ell\Delta = \frac{\ell \nu n}{k}$. Thus, we can directly restrict the range of $C$ to $C\in (0,e^{\nu}].$]}

%From now on, we consider a fixed feasible pair $(C,\zeta)\in (0,e^\nu]\times (0,1)$ under the specific $\ell\in [1,\frac{k}{\log k}]$. 

\subsubsection*{Condition for (C1) under the given $(\ell,C,\zeta)$}

We first bound $\mathbbm{E}k_{\ell,C,\zeta}$ and then deduce the high-probability behavior of $k_{\ell,C,\zeta}$ via Markov's inequality. Since $k_{\ell,C,\zeta}=\sum_{\mathcal{J}\subset s,|\mathcal{J}|=\ell}\mathbbm{1}(M_{\mathcal{J}}=\frac{Cn\nu e^{-\nu}\ell}{k},~M_{\mathcal{J}0}=\zeta M_{\mathcal{J}})$, we have 
\begin{align}
    \mathbbm{E}k_{\ell,C,\zeta}&= \binom{k}{\ell}\mathbb{P}\left(\text{for fixed }\mathcal{J}\subset s\text{ with }|\mathcal{J}|=\ell,~M_{\mathcal{J}}=\frac{Cn\nu e^{-\nu}\ell}{k},~M_{\mathcal{J}0}=\zeta M_{\mathcal{J}}\right)\\
    &=\binom{k}{\ell}\mathbb{P}\left(M_{\mathcal{J}}=\frac{Cn\nu e^{-\nu}\ell}{k}\right)\mathbb{P}\left(\text{Bin}\Big(\frac{Cn\nu e^{-\nu}\ell}{k},\rho\Big)=\frac{\zeta\cdot Cn\nu e^{-\nu}\ell}{k}\right)\label{eq:nc_bound_exp}
\end{align}
since $\Xv_s$ (which determines $M_{\mathcal{J}}$) and $\Zv$ (which determines $M_{\mathcal{J}0}$ given $M_{\mathcal{J}}$) are independent.

\textbf{Bounding $\mathbb{P}(M_{\mathcal{J}}=\frac{Cn\nu e^{-\nu}\ell}{k})$:} We make use of Lemma \ref{lem44} in Appendix \ref{app:nc_lemms}, in which $\mathcal{W}^{(\mathcal{J})}$ is defined to index the tests that contain some item from $\mathcal{J}$, and its cardinality is defined as $W^{(\mathcal{J})}:=|\mathcal{W}^{(\mathcal{J})}|$. We thus have that $i\in\mathcal{M}_{\mathcal{J}}$ holds if and only if $i\in \mathcal{W}^{(\mathcal{J})}\setminus \mathcal{W}^{(s\setminus \mathcal{J})}$.

Recall also the following notation introduced above \eqref{eq:calK_prime}: Given $s_0\subset s=[k]$, we let $\Tv_{s_0}$ be an unordered multi-set of length $|s_0|\Delta$ whose entries are in $[n]$, with the overall multi-set representing the $|s_0|\Delta$ placements from the items in $s_0$ in an unordered manner.  It follows that $\mathcal{W}^{(\mathcal{J})}$ and $\mathcal{W}^{(s\setminus \mathcal{J})}$ are determined by the randomness of $\Tv_{\mathcal{J}}$ and $\Tv_{s\setminus \mathcal{J}}$, respectively. Given $\mathcal{J}$ satisfying $|\mathcal{J}|=\ell$ with $\frac{\ell}{k}=o(1)$, (\ref{9262}) and (\ref{9263}) in Lemma \ref{lem44} give that for any $\delta>0$ we have $W^{(s\setminus \mathcal{J})}=\big(1-e^{-\nu}+\delta'  \sqrt{\frac{\ell}{k}}+o(1)\big)n$ for some $|\delta'|<\delta$, with probability at least $1-2\exp(-2\nu^{-1}\delta^2n)$. Therefore, we set $\delta=\sqrt{C_*\frac{\nu\ell}{k}}$ (thus $\delta\sqrt{\frac{\ell}{k}}=o(1)$) for some sufficiently large constant $C_*$ to obtain that the event $W^{(s\setminus \mathcal{J})}= (1-e^{-\nu}+o(1))n$, which implies $n-W^{(s\setminus\mathcal{J})} = ne^{-\nu}(1+o(1))$, holds with probability at least $1-2\exp(-\frac{C_*\ell n}{k})$. Thus, the event
\begin{align}\label{eq:nc_achi_AJ}
    \mathscr{A}_{\mathcal{J}} = \Big\{n-W^{(s\setminus \mathcal{J})} = ne^{-\nu}(1+o(1))\Big\}
\end{align}
holds with probability at least $1-2\exp(-\frac{C_*\ell n}{k})$, where $C_*$ can be made arbitrarily large. Therefore, we can proceed as follows:
\begin{align}\label{eq:rhsrhs}
    \mathbb{P}\Big(M_{\mathcal{J}}=\frac{Cn\nu e^{-\nu}\ell}{k}\Big)& \le \mathbb{P}\Big(M_{\mathcal{J}}=\frac{Cn\nu e^{-\nu}\ell}{k}\Big| \mathscr{A}_{\mathcal{J}}\Big)   +  \mathbb{P}\big(\mathscr{A}_{\mathcal{J}}^c\big)\le \tilde{P}_1+2\exp\Big(-\frac{C_*\ell n}{k}\Big) ,
\end{align}
where we introduce the shorthand $\tilde{P}_1:= \mathbb{P}\big(M_{\mathcal{J}}=\frac{Cn\nu e^{-\nu}\ell}{k}|\mathscr{A}_{\mathcal{J}}\big)$.

To bound $\tilde{P}_1$, we divide the $n$ tests into two parts $\mathcal{W}^{(s\setminus\mathcal{J})}$ and $[n]\setminus\mathcal{W}^{(s\setminus\mathcal{J})}$.  Then, conditioning on $\Tv_{s\setminus\mathcal{J}}$ and using the randomness of $\Tv_{\mathcal{J}}$, we can identify $M_{\mathcal{J}}$ with the number of tests in $[n]\setminus \mathcal{W}^{(s\setminus \mathcal{J})}$ that contain item from $\mathcal{J}$. 
With this perspective,
Lemma \ref{lem77} in Appendix \ref{app:nc_lemms} states that 
\begin{align} 
    \tilde{P}_1 &\le \exp\Big(-\frac{\ell\nu n}{k}\Big[D\Big(Ce^{-\nu}\big\|e^{-\nu}(1+o(1))\Big)+o(1)\Big]\Big)\\&=\exp \Big(-\frac{\ell\nu n}{k}\big[D(Ce^{-\nu}\|e^{-\nu})+o(1)\big]\Big).\label{eq:hatP1}
\end{align}
  Note that $D(Ce^{-\nu}\|e^{-\nu})$, as a function of $C\in(0,e^{-\nu}]$, is uniformly bounded from above.  Hence, we can choose $C_*$ large enough so that the bound on $\tilde{P}_1$ in (\ref{eq:hatP1}) dominates $2\exp(-\frac{C_*\ell n}{k})$ in the right-hand side of (\ref{eq:rhsrhs}), and thus obtain  
  \begin{align}\label{eq:nc_achi_exp1}
      \mathbb{P}\Big(M_{\mathcal{J}}=\frac{Cn\nu e^{-\nu}\ell}{k}\Big)\le \exp \Big(-\frac{\ell\nu n}{k}\big[D(Ce^{-\nu}\|e^{-\nu})+o(1)\big]\Big).
  \end{align}

  \textbf{Bounding} $\mathbb{P}\big(\text{Bin}(\frac{Cn\nu e^{-\nu}\ell}{k},\rho)=\frac{\zeta\cdot Cn\nu e^{-\nu}\ell}{k}\big)$\textbf{ and combining:}  We apply the Chernoff bound  (see (\ref{eq:chernoff1}) in Lemma \ref{binoconcen}) to obtain 
  \begin{align}
      \label{eq:nc_achi_exp2}
     \mathbb{P}\left(\text{Bin}\Big(\frac{Cn\nu e^{-\nu}\ell}{k},\rho\Big)=\frac{\zeta\cdot Cn\nu e^{-\nu}\ell}{k}\right) \le \exp \Big(-\frac{Cn\nu e^{-\nu}\ell}{k}\cdot D\big(\zeta\|\rho\big)\Big).
  \end{align}
Substituting (\ref{eq:nc_achi_exp1}) and (\ref{eq:nc_achi_exp2}) into (\ref{eq:nc_bound_exp}) yields
\begin{align}
    \mathbbm{E}k_{\ell,C,\zeta} \le \binom{k}{\ell}\exp\Big(-\frac{\ell\nu ne^{-\nu}}{k}\underbrace{\Big[e^\nu D(Ce^{-\nu}\|e^{-\nu})+C\cdot D(\zeta\|\rho)+o(1)\Big]}_{:=f_1(C,\zeta,\rho,\nu)}\Big).\label{eq:nc_achi_defi_f1}
\end{align}
Note that $f_1(C,\zeta,\rho,\nu)$ here has appeared in the proof of converse bound; see (\ref{eq:nc_f1}).

We separately deal with the two cases $1\le \ell\le \log k$  and  $\log k <\ell \le \frac{k}{\log k}$   using Markov's inequality.  Specifically, following exactly the same arguments as the Bernoulli design (see (\ref{eq:365})--(\ref{eq:c1_condition})), we conclude the following for a given $(\ell,C,\zeta)$: 
\begin{align}
    \label{eq:nc_c1_thre}& \frac{n\nu e^{-\nu}\ell}{k}\big[f_1(C,\zeta,\rho,\nu)+o(1)\big] \ge (1+\eta_1)\ell\log\frac{k}{\ell},~\text{for some }\eta_1>0\\
    \Longrightarrow& ~
    \mathbb{P}\big(k_{\ell,C,\zeta}\ge 1\big)\le \hat{P}_\ell,~~\text{where }\hat{P}_\ell = \begin{cases}
        (\ell\log k)^{-5},~~\text{if }1\le \ell\le \log k\\
        ~~~~~~k^{-5},~~~~\text{if }\log k<\ell\le \frac{k}{\log k}
    \end{cases}.\label{eq:nc_c1_prob}
\end{align}
Therefore, if (\ref{eq:nc_c1_thre}) holds, then \textbf{(C1)} holds for some given $(\ell,C,\zeta)$ with probability at most $\hat{P}_\ell$.

\subsubsection*{Condition for (C2) for a given $(\ell,C,\zeta)$}
As in (\ref{eq:nc_count_N0}), with respect to the randomness of $\Tv_s$, the event  \begin{align}\label{eq:nc_achi_A1}
    \mathscr{A}_1=\big\{N_0=(1+o(1))e^{-\nu}n\big\}
\end{align} 
holds with $1-o(1)$ probability.  In the following, for a generic event $\mathscr{E}$, an upper bound on $\mathbb{P}(\mathscr{E})$ in this part should be understood as a bound on $\mathbb{P}(\mathscr{E}\cap\mathscr{A}_1)$, but we will only make this explicit in the concluding stage; see (\ref{eq:ber_c2_prob_bound}) below.

We consider a fixed $(\ell,C,\zeta)$ and seek to establish the condition for \textbf{(C2)}, i.e.,  (\ref{eq:use_nega_test}) holding for some $\mathcal{J}\in \mathcal{K}_{\ell,C,\zeta}$ and some $\mathcal{J}' \subset [p]\setminus s$ with $|\mathcal{J}'|=\ell$.
Note that 
  $G_{\mathcal{J},\mathcal{J}',1}$ in (\ref{eq:use_nega_test}) is the number of positive tests in $\mathcal{N}_{0,\mathcal{J}'}\cup \mathcal{M}_{\mathcal{J},\mathcal{J}'}$, and thus we can further decompose $G_{\mathcal{J},\mathcal{J}',1}$ into \begin{equation}\label{eq:decom_GJJ}
    G_{\mathcal{J},\mathcal{J}',1} = \widetilde{G}_{\mathcal{J}',1} + U_{\mathcal{J},\mathcal{J}',1},
\end{equation}
where $\widetilde{G}_{\mathcal{J}',1}$ is the number of positive tests in $\mathcal{N}_{0,\mathcal{J}'}$ (with no dependence on $\mathcal{J}$, as reflected by the notation), and $U_{\mathcal{J},\mathcal{J}',1}$ is the number of positive tests in $\mathcal{M}_{\mathcal{J},\mathcal{J}'}$. Substituting (\ref{eq:decom_GJJ}) into (\ref{eq:use_nega_test}), we can restate \textbf{(C2)} as follows: 
\begin{align}
    &\widetilde{G}_{\mathcal{J}',1}\ge \Big(\frac{1}{2}-\zeta\Big)\frac{Cn\nu e^{-\nu}\ell}{k}+\frac{1}{2}N_{0,\mathcal{J}'}+\Big(\frac{1}{2}M_{\mathcal{J},\mathcal{J}'}-U_{\mathcal{J},\mathcal{J}',1}\Big) 
\end{align}
for some $\mathcal{J}\in\mathcal{K}_{\ell,C,\zeta}$ and some $\mathcal{J}'\subset[p]\setminus s$ with $|\mathcal{J}'|=\ell$. 
Therefore, if \textbf{(C2)} holds, then the following also necessarily holds:  
\begin{align}\label{eq:nc_suff_c2}
    \max_{\substack{\mathcal{J}'\subset [p]\setminus s\\|\mathcal{J}'|=\ell}}\Big(\widetilde{G}_{\mathcal{J}',1}-\frac{1}{2}N_{0,\mathcal{J}'}\Big) \ge \Big(\frac{1}{2}-\zeta\Big)\frac{Cn\nu e^{-\nu}\ell}{k}-\max_{\substack{\mathcal{J}\in \mathcal{K}_{\ell,C,\zeta}\\\mathcal{J}'\subset[p]\setminus s,|\mathcal{J}'|=\ell}}M_{\mathcal{J},\mathcal{J}'},
\end{align}
where we express ``there exist'' (or ``for some'') in \textbf{(C2)} via the ``max'' operation, and apply the trivial inequality $\frac{1}{2}M_{\mathcal{J},\mathcal{J}'}-U_{\mathcal{J},\mathcal{J}',1}\geq -M_{\mathcal{J},\mathcal{J}'}$. %That is to say, (\ref{eq:nc_suff_c2}) is a sufficient condition for \textbf{(C2)}. 

{\textbf{The effect of $M_{\mathcal{J},\mathcal{J}'}$}:} We show that the term  $\max_{\mathcal{J},\mathcal{J}'}M_{\mathcal{J},\mathcal{J}'}$  indeed has minimal effect.  
Given specific $(\mathcal{J},\mathcal{J}')$, recall that $M_{\mathcal{J},\mathcal{J}'}$ is the number of tests in $\mathcal{M}_{\mathcal{J}}$ that contain some item from $\mathcal{J}'$; we now further let $\widetilde{M}_{\mathcal{J},\mathcal{J}'}$ be the number of placements from the items in $\mathcal{J}'$ that fall in the $\frac{Cn\nu e^{-\nu}\ell}{k}$ tests in $\mathcal{M}_{\mathcal{J}}$. Then, note that we always have $M_{\mathcal{J},\mathcal{J}'}\le \widetilde{M}_{\mathcal{J},\mathcal{J}'}$, and $M_{\mathcal{J},\mathcal{J}'}< \widetilde{M}_{\mathcal{J},\mathcal{J}'}$ happens if some test in $\mathcal{M}_{\mathcal{J}}$ receives more than one placement from the items in $\mathcal{J}'$. Therefore, we have 
\begin{align}
\max_{\mathcal{J},\mathcal{J}'}M_{\mathcal{J},\mathcal{J}'}\le\max_{\mathcal{J},\mathcal{J}'}\widetilde{M}_{\mathcal{J},\mathcal{J}'},
\end{align}
where $(\mathcal{J},\mathcal{J}')$ are implicitly subject to the same constraints as the last term in (\ref{eq:nc_suff_c2}), and similarly in the developments below. We consider fixed $\mathcal{J}\in \mathcal{K}_{\ell,C,\zeta}$ and $\mathcal{J}'\subset[p]\setminus s$ with $|\mathcal{J}'|=\ell$. Since $\mathcal{M}_{\mathcal{J}}=\frac{Cn\nu e^{-\nu}\ell}{k}$ and the test placements are uniform  in $[n]$, a placement from items in $\mathcal{J}'$ falls in $\mathcal{M}_{\mathcal{J}}$ with probability $\frac{C\nu e^{-\nu}\ell}{k}$, thus
the randomness of $\Tv_{\mathcal{J}'}$ gives $\widetilde{M}_{\mathcal{J},\mathcal{J}'}\sim\text{Bin}(\frac{\ell\nu n}{k},\frac{C\nu e^{-\nu}\ell}{k})$.
Then we can proceed analogously to the calculations in (\ref{eq:nc_small_effect})--(\ref{eq:k-10ell}). In particular, our assumption $C \le e^{\nu}$ is equivalent to $Ce^{-\nu}\le 1$, and thus
for given $(\mathcal{J},\mathcal{J}')$ we have
\begin{align}
    \mathbb{P}\Big(\widetilde{M}_{\mathcal{J},\mathcal{J}'}\ge \frac{\ell n}{k(\log\frac{k}{\ell})^{1/2}}\Big) &\leq \mathbb{P}\Big(\text{Bin}\big(\frac{\nu \ell n}{k},\frac{\nu\ell}{k}\big)\ge \frac{\ell n}{k(\log\frac{k}{\ell})^{1/2}}\Big)\\
    &\le \exp \Big(-\Omega\Big(\sqrt{\log\frac{k}{\ell}}\Big)\cdot\ell\log k\Big), \label{eq:second_ine_fol}
\end{align}
where (\ref{eq:second_ine_fol}) follows from the same argument as the Bernoulli design (see (\ref{eq:nc_small_end})), with only the constant factors changing. Moreover, by repeating the union bound over $(\mathcal{J},\mathcal{J}')$ as in (\ref{eq:k-10ell}), we obtain 
\begin{align}
    \label{eq:nc_minimal}\mathbb{P}\Big(\max_{\mathcal{J},\mathcal{J}'}\widetilde{M}_{\mathcal{J},\mathcal{J}'}\ge \frac{\ell n}{k(\log\frac{k}{\ell})^{1/2}}\Big) \leq k^{-10\ell}
\end{align}
for sufficiently large $k$, and 
thus it holds with probability at least 
$1-k^{-10\ell}$ that
\begin{align}\label{eq:minimal_MJJpai}
\max_{\mathcal{J},\mathcal{J}'}M_{\mathcal{J},\mathcal{J}'}\le \frac{\ell n}{k(\log\frac{k}{\ell})^{1/2}}=o\Big(\frac{\ell n}{k}\Big).
\end{align} 
On this high-probability event,  the necessary condition for \textbf{(C2)} given in (\ref{eq:nc_suff_c2}) can be written as
\begin{align}\label{eq:nc_suff2_c2}
    \max_{\substack{\mathcal{J}'\subset [p]\setminus s\\|\mathcal{J}'|=\ell}}\Big(\widetilde{G}_{\mathcal{J}',1}-\frac{1}{2}N_{0,\mathcal{J}'}\Big) \ge \Big(\frac{1}{2}-\zeta-o(1)\Big)\frac{Cn\nu e^{-\nu}\ell}{k}.
\end{align}
Thus, to bound the probability of \textbf{(C2)}, it suffices to bound the probability of its necessary condition (\ref{eq:nc_suff2_c2}). 
We proceed to study (\ref{eq:nc_suff2_c2}) by first studying a fixed $\mathcal{J}'$ and then applying a union bound to account for the maximum.

   For given $\mathcal{J}'$, we define 
   \begin{align}
       d_0:=\frac{C(2\zeta -1)}{1-2\rho} \label{eq:nc_achi_d0}
   \end{align}
   and write  
\begin{align}
   \label{eq:nc_decom_1} &\mathbb{P}\Big(\widetilde{G}_{\mathcal{J}',1}-\frac{1}{2}N_{0,\mathcal{J}'}\ge \big(\frac{1}{2}-\zeta-o(1)\big)\frac{Cn\nu e^{-\nu}\ell}{k}\Big)\le  \mathbb{P}\Big(N_{0,\mathcal{J}'}\le \frac{d_0\nu e^{-\nu} \ell n}{k}\Big)\\
    &+\sum_{\substack{d_0< d\le e^{\nu}\\\frac{d\nu e^{-\nu}\ell n}{k}\in\mathbb{Z}}}\mathbb{P}\Big(N_{0,\mathcal{J'}}=\frac{d\nu e^{-\nu}\ell n}{k}\Big)\mathbb{P}\Big(\widetilde{G}_{\mathcal{J}',1}\ge\frac{d+C(1-2\zeta)}{2}\frac{\nu e^{-\nu}\ell n}{k}\Big|N_{0,\mathcal{J}'}=\frac{d\nu e^{-\nu}\ell n}{k}\Big).\label{eq:nc_decom_2}
\end{align}

 % This trivial bound is of similar spirit   to (\ref{eq:ber_achi_trivial}) in the proof of achievablity bound for Bernoulli designs, and we refer to the intuition given there.  

\textbf{Case 1: $d_0\in [0,1]$ (i.e., $0\le C(2\zeta-1)\le 1-2\rho$).} We will separately bound the probability terms in the right hand side of (\ref{eq:nc_decom_1}) and in (\ref{eq:nc_decom_2}).

\underline{Bounding the term in   (\ref{eq:nc_decom_1}):} We first bound the term in the right-hand side of (\ref{eq:nc_decom_1}).  
By definition, $ {N}_{0,\mathcal{J}'}$ is the number of tests in $\mathcal{N}_0$ that contain some item from $\mathcal{J}'$, thus on the event of  $N_0=(1+o(1))ne^{-\nu}$ (as explained in the text below (\ref{eq:nc_achi_A1})),  Lemma \ref{lem77} gives 
\begin{align}\label{eq:apply_lem77_1}    \mathbb{P}\Big(N_{0,\mathcal{J}'}= \frac{d\nu e^{-\nu}\ell n}{k}\Big)&\le \exp\Big(-\frac{\ell\nu n}{k}\big[D\big(de^{-\nu}\|e^{-\nu}+o(1)\big)+o(1
)\big]\Big)  
    \\&\le \exp\Big(-\frac{\ell\nu n}{k}\big[D\big(d_0e^{-\nu}\|e^{-\nu}\big)+o(1)\big]\Big) \label{eq:second_ine}
\end{align}
for any $d\le d_0$ such that $\frac{d\nu e^{-\nu}\ell n}{k}\in \mathbb{Z}$, where (\ref{eq:second_ine}) follows from $d_0\le 1$ and the monotonicity of $D(de^{-\nu}\|e^{-\nu})$ with respect to $d$. Therefore, the term in the right-hand side of (\ref{eq:nc_decom_1}) can be bounded as 
\begin{align}
     \PP\Big(N_{0,\mathcal{J}'}\le\frac{d_0\nu e^{-\nu}\ell n}{k}\Big) &=\sum_{0\le d\le d_0:\frac{d\nu e^{-\nu}
    \ell n}{k}\in \mathbb{Z}} \mathbb{P}\Big(N_{0,\mathcal{J}'}=\frac{d\nu e^{-\nu}\ell n}{k}\Big) \label{eq:count_summand}\\&\le O\big(\frac{\ell n}{k}\big)\exp \Big(-\frac{\ell\nu n}{k}\big[D\big(d_0e^{-\nu}\|e^{-\nu}\big)+o(1)\big]\Big) \\&\le \exp \Big(-\frac{\ell\nu n}{k}\big[D\big(d_0e^{-\nu}\|e^{-\nu}\big)+o(1)\big]\Big), \label{eq:first_bound_nc} 
\end{align}
where (\ref{eq:count_summand}) follows from $d_0=O(1)$ (see (\ref{eq:nc_achi_d0})) and hence there are at most $O(\frac{\ell n}{k})$ summands, and in  (\ref{eq:first_bound_nc}) the leading factor $O(\frac{\ell n}{k})$ is absorbed into the $o(1)$ term in the exponent.

\underline{Bounding the term in (\ref{eq:nc_decom_2}):} Analogously to  (\ref{eq:apply_lem77_1}), on the event of $N_0=(1+o(1))ne^{-\nu}$,   Lemma \ref{lem77} gives 
\begin{align}\label{eq:first_nc}
    \mathbb{P}\Big(N_{0,\mathcal{J}'}=\frac{d\nu e^{-\nu}\ell n}{k}\Big) \le \exp\Big(-\frac{\ell \nu n}{k}\cdot \big[D \big(de^{-\nu}\|e^{-\nu}\big)+o(1)\big]\Big).
\end{align}
Then,   conditioned on $N_{0,\mathcal{J}'}=\frac{d\nu e^{-\nu}\ell n}{k}$ we have $\widetilde{G}_{\mathcal{J}',1}\sim \text{Bin}(\frac{d\nu e^{-\nu}\ell n}{k},\rho)$ via the randomness of $\Zv$, and thus the condition $d> d_0$ (which implies $\frac{d+C(1-2\zeta)}{2}>d\rho$) allows us to apply the Chernoff bound (see (\ref{eq:chernoff1}) in Lemma \ref{binoconcen}) to obtain  
\begin{align}
   & \mathbb{P}\Big(\text{Bin}\Big(\frac{d\nu e^{-\nu}\ell n}{k},\rho\Big)\ge\frac{d+C(1-2\zeta)}{2}\frac{\nu e^{-\nu}\ell n}{k}\Big)\le \exp \left(-\frac{d\nu e^{-\nu}
    \ell n}{k}\cdot D\Big(\frac{1}{2}+\frac{C(1-2\zeta)}{2d}\big\|\rho\Big)\right). \label{eq:second_nc}
\end{align}
Combining the bounds in (\ref{eq:first_nc}) and (\ref{eq:second_nc}), without accounting for the summation over $d$, a single summand in (\ref{eq:nc_decom_2}) corresponding to a specific $d$ can be bounded by
\begin{align}
    \exp\Big(-\frac{\ell \nu n e^{-\nu}}{k}\Big[ \underbrace{e^{\nu}\cdot D(de^{-\nu}\|e^{-\nu})+d\cdot D\Big(\frac{1}{2}+\frac{C(1-2\zeta)}{2d}\big\|\rho\Big)}_{:=g(C,\zeta,d,\rho,\nu)} +o(1)\Big]\Big).
    \label{eq:second_bound_nc}
\end{align}
Note that $g(C,\zeta,d,\rho,\nu)$ here has appeared in the   proof of converse bound; see (\ref{eq:nc_con_def_g}).

To further account for the summation over $\{  d_0 <d\le e^\nu:\frac{d\nu e^{-\nu}\ell n}{k}\in \mathbb{Z}\}$, which contains at most $O(\frac{\ell n}{k})$ elements, we bound all these summands by their common upper bound \begin{align}\label{eq:nc_common_bound}
    \exp\Big(-\frac{\ell\nu ne^{-\nu}}{k}\Big[g(C,\zeta,d^*,\rho,\nu)+o(1)\Big]\Big), 
\end{align} 
where $d^*$ is given by 
\begin{align}\label{eq:dstar_nc_achi}
    d^* =\argmin_{d}~g(C,\zeta,d,\rho,\nu),~~\text{subject to }|C(2\zeta-1)|<d\le e^\nu,
\end{align}
and where we relax the range of $d$ from $\frac{C(2\zeta-1)}{1-2\rho}=d_0\le d\le e^{\nu}$ to $|C(2\zeta-1)|\le d\le e^\nu$ so that it matches the corresponding development in the converse proof (see (\ref{eq:dstar_nc_con})). We again write $A:=C(1-2\zeta)$, and note that $d^*$
in (\ref{eq:dstar_nc_achi}) has been pinpointed in (\ref{eq:d_star_pinpoint}). 
Therefore, the term in (\ref{eq:nc_decom_2}) is bounded by 
\begin{align}\label{eq:bound_332}
    O\Big(\frac{\ell n}{k}\Big)\exp\Big(-\frac{\ell\nu n e^{-\nu}}{k} [g(C,\zeta,d^*,\rho,\nu)+o(1)]\Big)=\exp\Big(-\frac{\ell\nu n e^{-\nu}}{k} [\underbrace{g(C,\zeta,d^*,\rho,\nu)}_{:=f_{2}(C,\zeta,\rho,\nu)}+o(1)]\Big),
\end{align}
where $f_2(C,\zeta,\rho,\nu)$ coincides with the one in the proof of near-constant weight design converse bound; see (\ref{eq:f2_con_nc}) in Appendix \ref{app:nc_converse}.

 \underline{Comparing the bounds in  (\ref{eq:first_bound_nc}) and (\ref{eq:bound_332}):} Observe that (\ref{eq:dstar_nc_achi}) and $C(2\zeta-1)\le d_0\le e^\nu$ give
\begin{align}
    f_{2}(C,\zeta,\rho,\nu)=g(C,\zeta,d^*,\rho,\nu)\le g(C,\zeta,d_0,\rho,\nu)=e^\nu D(d_0e^{-\nu}\|e^{-\nu}), 
\end{align}
where the last equality follows from $D(\frac{1}{2}+\frac{C(1-2\zeta)}{2d}\|\rho)=0$ when $d=d_0$. Thus, the bound in (\ref{eq:bound_332}) dominates the one in (\ref{eq:first_bound_nc}), and by substituting into (\ref{eq:nc_decom_1})--(\ref{eq:nc_decom_2}), we obtain 
\begin{align} \label{eq:nc_achi_maincasebound}
    \mathbb{P}\Big(\widetilde{G}_{\mathcal{J}',1}-\frac{1}{2}N_{0,\mathcal{J}'}\ge\big(\frac{1}{2}-\zeta-o(1)\big)\frac{Cn\nu e^{-\nu}\ell}{k}\Big)\le \exp\Big(-\frac{\ell\nu ne^{-\nu}}{k}\big[f_{2}(C,\zeta,\rho,\nu)+o(1)\big]\Big)
\end{align}
for the case $C(2\zeta-1)\le1-2\rho$.

\textbf{Case 2: $d_0<0$ (i.e., $ C(2\zeta-1)<0$).} In this case, the term in (\ref{eq:nc_decom_1}) vanishes, and we focus on bounding the term in (\ref{eq:nc_decom_2}). As in Case 1, given $N_{0,\mathcal{J}'}=\frac{d\nu e^{-\nu}\ell n}{k}$ we have $\widetilde{G}_{\mathcal{J}',1}\sim \text{Bin}(\frac{d\nu e^{-\nu}\ell n}{k},\rho)$, and so we have the summand being $0$ if $\frac{d+C(1-2\zeta)}{2}>d$, i.e., $d< C(1-2\zeta)$. Therefore, we can restrict the summation over $d$ in (\ref{eq:nc_decom_2}) to 
\begin{align}
    \mathcal{D}= \Big\{C(1-2\zeta)\le d\le e^{\nu}: \frac{d\nu e^{-\nu}\ell n}{k}\in\mathbb{Z}\Big\}. 
\end{align}
Then similarly to Case 1, we bound the $O(\frac{\ell n}{k})$ summands by their common upper bound as per (\ref{eq:nc_common_bound}) with $d^*$ given in (\ref{eq:dstar_nc_achi}), and obtain the same bound as (\ref{eq:nc_achi_maincasebound}).

\textbf{Case 3: $d_0>1$ (i.e., $C(2\zeta-1)> 1-2\rho$).}  In this case,   we simply apply the trivial bound 
\begin{align}\label{eq:nc_achi_triv}
    \mathbb{P}\Big(\widetilde{G}_{\mathcal{J}',1}-\frac{1}{2}N_{0,\mathcal{J}'}\ge\big(\frac{1}{2}-\zeta-o(1)\big)\frac{Cn\nu e^{-\nu}\ell}{k}\Big)  \leq 1,
\end{align}
which in turn trivially behaves as $\exp(-\frac{\ell \nu n}{k}o(1))$.

\textbf{Combining the Cases 1-3:}  We recall that $f_2(C,\zeta,\rho,\nu)$ is given in (\ref{eq:bound_332}) and  define  
\begin{align}\label{eq:nc_achi_hatf}
    \hat{f}_2(C,\zeta,\rho,\nu)=\begin{cases}
        f_2(C,\zeta,\rho,\nu) & \text{when }C(2\zeta-1)\le 1-2\rho \\
        0 & \text{when }C(2\zeta-1)> 1-2\rho.
    \end{cases},
\end{align}
Combining the three cases discussed above, we obtain
\begin{align}
      \mathbb{P}\Big(\widetilde{G}_{\mathcal{J}',1}-\frac{1}{2}N_{0,\mathcal{J}'}\ge\big(\frac{1}{2}-\zeta-o(1)\big)\frac{Cn\nu e^{-\nu}\ell}{k}\Big)\le \exp\Big(-\frac{\ell\nu ne^{-\nu}}{k}\big[\hat{f}_2(C,\zeta,\rho,\nu)+o(1)\big]\Big).
\end{align}
Additionally, we can verify that $f_2(C,\zeta,\rho,\nu)=0$ when $C(2\zeta-1)=1-2\rho$: Using (\ref{eq:d_star_pinpoint}) and some simple algebra, we find that $d^*=1$ holds when $A=C(1-2\zeta)=2\rho -1$; substituting this into $g(C,\zeta,d,\rho,\nu)$ in (\ref{eq:second_bound_nc}) yields \begin{align}
    g(C,\zeta,1,\rho,\nu) = D\Big(\frac{1+C(1-2\zeta)}{2}\big\|\rho\Big) =0. 
\end{align} Thus,  $\hat{f}_2(C,\zeta,\rho,\nu)$ is continuous. %this fact will be implicitly used in the subsequent (\ref{eq:use_continuous_hat}).  

\textbf{The condition for (C2):}
We now proceed as follows:
\begin{align}
    &\mathbb{P}\Big(\text{\textbf{(C2)}   holds for }(\ell,C,\zeta)\Big)\\
    &\quad\le \mathbb{P}\Big((\ref{eq:nc_suff2_c2})\text{   hold for }(\ell,C,\zeta)\Big) + \mathbb{P}\Big((\ref{eq:minimal_MJJpai})\text{ does not hold}\Big)  \label{eq:add_explain3}
    \\ &\quad\le \binom{p}{\ell} \exp\Big(-\frac{\ell\nu ne^{-\nu}}{k}\big[\hat{f}_2(C,\zeta,\rho,\nu)+o(1)\big]\Big) +k^{-10\ell}\label{eq:add_explain4}\\
    &\quad\le \exp\Big((1+o(1))\ell\log\frac{p}{\ell}-\frac{\ell\nu n e^{-\nu}}{k}\big[\hat{f}_2(C,\zeta,\rho,\nu)+o(1)\big]\Big)+k^{-10\ell},\label{eq:nc_c2_prob}
\end{align}
where (\ref{eq:add_explain3}) holds because under the condition (\ref{eq:minimal_MJJpai}) we have that (\ref{eq:nc_suff2_c2}) is a necessary condition for ``\textbf{(C2)} holds for $(\ell,C,\zeta)$'', and in (\ref{eq:add_explain4}) we apply a union bound over no more than $\binom{p}{\ell}$ possibilities of $\mathcal{J}'$ to account for the maximum in (\ref{eq:nc_suff2_c2}), and recall that (\ref{eq:minimal_MJJpai}) holds with probability at least $1-k^{-10\ell}$.  Combining these developments, we deduce that the following implication holds for any $\eta_2 > 0$ and sufficiently large $p$ (where we now make the role of $\mathscr{A}_1$ explicit):
\begin{align}
   \label{eq:nc_threshold2}\frac{\ell\nu n e^{-\nu}}{k}\hat{f}_2(C,\zeta,\rho,\nu)\ge (1+\eta_2)\ell\log\frac{p}{\ell} \\ % ~\text{ for some }\eta_2>0\\
    \Longrightarrow~ \mathbb{P}\Big(\big\{\textbf{(C2)}\text{ holds for }(\ell,C,\zeta)\big\}\cap \mathscr{A}_1\Big) 
    &\le \exp\Big(-\frac{\eta_2}{2}\ell\log\frac{p}{\ell}\Big)+k^{-10\ell} \\ &\le  \hat{P}_{\ell}: =\begin{cases}
        (\ell\log k)^{-5},~\text{if }1\le \ell \le \log k\\
        ~~~k^{-5},~~~~~~\text{if }\log k<\ell \le \frac{k}{\log k},\label{eq:ber_c2_prob_bound}
    \end{cases} 
\end{align}
since $\exp(-\frac{\eta^2}{2}\ell \log\frac{p}{\ell})$ is an upper bound on (\ref{eq:nc_c2_prob}) when  (\ref{eq:nc_threshold2}) holds. 
% which states that the $\widetilde{M}_{\mathcal{J},\mathcal{J}'}$ only has minimal effect on \textbf{(C2)}. %{\color{blue}[TODO]} Strictly speaking, the probability in (\ref{eq:nc_c2_prob_1}) should be $1-\hat{P}_\ell -o(1)$ to ensure the high-probability event $\mathscr{A}_1$ in (\ref{eq:nc_achi_A1}), and the probability in (\ref{eq:nc_c2_prob}) and (\ref{eq:prob_summary}) should include this $o(1)$ term accordingly. But we note that the event $\mathscr{A}_1$ is universal for all $(\ell,C,\zeta)$, thus we simply build our analysis  on \textbf{(C2)} with respect to specific $(\ell,C,\zeta)$ (e.g., the statement in (\ref{eq:c2_condition})) on $\mathscr{A}_1$, and it will be sufficient to remove the $o(1)$ probability (that accounts for the failure of $\mathscr{A}_1$) at a later stage;; see the $o(1)$ term in  (\ref{eq:nc_achi_account}). 

\subsubsection*{Establishing the threshold}

We are now in a position to derive the threshold for $n$ above which restricted MLE has $o(1)$ probability of failing.  We pause to review our previous developments:
\begin{itemize}
    \item For given $\ell$ and any feasible $(C,\zeta)$, by (\ref{eq:nc_c1_thre})--(\ref{eq:nc_c1_prob}), if
    \begin{align}
        \label{eq:nc_c1_final_threshold}
        n\ge \frac{(1+\eta_1)k\log\frac{k}{\ell}}{\nu e^{-\nu}f_1(C,\zeta,\rho,\nu)} 
    \end{align}
    holds for some $\eta_1>0$ with $f_1(C,\zeta,\rho,\nu)$ being defined in (\ref{eq:nc_achi_defi_f1}), then \textbf{(C1)} holds for the given $(\ell,C,\zeta)$ with probability at most $\hat{P}_\ell$.
    
    \item For given $\ell$ and any feasible $(C,\zeta)$, by  (\ref{eq:nc_threshold2})--(\ref{eq:ber_c2_prob_bound}), if 
    \begin{align}
        \label{eq:nc_c2_final_threshold}
        n\ge \frac{(1+\eta_2)k\log\frac{p}{\ell}}{\nu e^{-
        \nu}[\hat{f}_2(C,\zeta,\rho,\nu)+o(1)]} 
    \end{align}
    holds for some $\eta_2>0$ with  $\hat{f}_2(C,\zeta,\rho,\nu)$ being given in (\ref{eq:nc_achi_hatf}) and (\ref{eq:bound_332}), then the probability of ``\textbf{(C2)} holds for the given $(\ell,C,\zeta)$ and $\mathscr{A}_1$ holds'' is no higher than $\hat{P}_\ell$.
      %$\mathbb{P}(\{$\textbf{(C2)}     holds for the given $(\ell,C,\zeta)\}\cap\mathscr{A}_1)\le \hat{P}_\ell.$  
\end{itemize}
% We are in a position to establish the threshold for $n$ above which the MLE fails for some $\ell\in [1,\frac{k}{\log k}]$ with $o(1)$ probability. 

\underline{Bounding the failure probability for a fixed $\ell$:} 
For fixed $\ell\in[1,\frac{k}{\log k}]$, we first consider the feasible $(C,\zeta)\in [0,e^\nu]\times[0,1]$ such that $\frac{Ce^{-\nu}\ell\nu n}{k},\frac{\zeta\cdot Ce^{-\nu}\ell\nu n}{k}$ are integers, and note that   there are at most $O((\frac{\ell n}{k})^2)=O((\ell\log k)^2)$ possibilities of feasible pairs of $(C,\zeta)$. If it holds for some $\eta>0$ that
\begin{align}\label{eq:nc_con_fixedell}
    n\ge (1+\eta)\frac{k}{\nu e^{-\nu}}\max_{\substack{C\in(0,e^\nu)\\\zeta\in(0,1)}}\min \Big\{\frac{\log\frac{k}{\ell}}{f_1(C,\zeta,\rho,\nu)},\frac{\log\frac{p}{\ell}}{\hat{f}_2(C,\zeta,\rho,\nu)}\Big\}
\end{align}
 then   the two dot points reviewed above give that, 
for any feasible $(C,\zeta)\in [0,e^\nu]\times [0,1]$ (with respect to the given $\ell $) we have 
$$\mathbb{P}\big( \{ \text{(\textbf{C1}) and (\textbf{C2})  simultaneously hold for $(\ell,C,\zeta)\}$} \,\cap\, \mathscr{A}_1\big)\le \hat{P}_\ell.$$
  Since it is impossible for \textbf{(C1)} to hold when $C>e^\nu$, for the restricted MLE decoder in (\ref{eq:mle_restricted}) we obtain that 
\begin{align}
& \mathbb{P}\Big(\big\{\big|s\setminus\widehat{S}'\big|=\ell\big\}\cap\mathscr{A}_1\Big) \nonumber \\\label{eq:fail_necc_nc_1} & \le\mathbb{P}\Big(\{\textbf{(C1)}\text{ and }\textbf{(C2)}~\text{hold for some feasible }(C,\zeta)\in [0,e^\nu]\times [0,1]\} \cap\mathscr{A}_1 \text{ holds}\Big) \\ & \le O((\ell\log k)^2)\hat{P}_\ell, \label{eq:bound_single_ell}
\end{align}
where (\ref{eq:fail_necc_nc_1}) holds due to Lemma \ref{lem:restate}(b), and (\ref{eq:bound_single_ell}) follows from a union bound over all feasible $(C,\zeta)$ under the given $\ell\in[1,\frac{k}{\log k}]$.

\underline{Bounding the overall failure probability for $\ell\in [1,\frac{k}{\log k}]$:} We further take a union bound over all  $\ell\in[1,\frac{k}{\log k}]$.  Suppose that the condition (\ref{eq:nc_con_fixedell}) holds for all $\ell\in [1,\frac{k}{\log k}]$; by  the same reasoning as (\ref{eq:ber_uni_ell})--(\ref{eq:411}),   this amounts to  \begin{align}\label{eq:nc_smallell_thre}
    n\geq \frac{(1+\eta)k\log\frac{p}{k}}{(1-\theta)\nu e^{-\nu}}\frac{1}{\min _{C\in(0,e^\nu),\zeta\in(0,1)}\max\{\frac{1}{\theta}f_1(C,\zeta,\rho,\nu),\hat{f}_2(C,\zeta,\rho,\nu)\}}
\end{align}
  for some $\eta>0$. Then, if (\ref{eq:nc_smallell_thre}) holds, we can bound the overall failure probability as
  \begin{align}
    %&\mathbb{P}\Big(|s\setminus\widehat{S}|\in\Big[1,\frac{k}{\log k}\Big]\Big)\\
    \label{eq:union_bound_hh} \mathbb{P}({\rm err}) &\le \mathbb{P}(\mathscr{A}_1^c) +   \sum_{\ell=1}^{k/\log k}\mathbb{P}\Big(\Big\{ |s\setminus\widehat{S}'|=\ell\Big\}\cap \mathscr{A}_1\Big)\\
    &\le o(1)+  \sum_{\ell=1}^{k/\log k} O((\ell\log k)^2)\hat{P}_\ell \label{eq:sub_bound}\\ &\le  o(1)+ \Big[\sum_{\ell=1}^{\log k} (\ell\log k)^{-3} + \sum_{\ell=\log k}^{k/\log k}k^{-3}\Big] = o (1),\label{eq:sub_Pell}
    \end{align}
 where (\ref{eq:union_bound_hh}) follows from the union bound, (\ref{eq:sub_bound}) follows from (\ref{eq:bound_single_ell}), and in (\ref{eq:sub_Pell}) we substitute $\hat{P}_\ell$ (given by (\ref{eq:ber_c2_prob_bound})). In conclusion, if (\ref{eq:nc_smallell_thre}) holds for some $\eta>0$, then restricted MLE has $o(1)$ probability of failure.

\textbf{Simplifying $\hat{f}_2(C,\zeta,\rho,\nu)$ to $f_2(C,\zeta,\rho,\nu)$:}
Recall that $\hat{f}_2(C,\zeta,\rho,\nu)$ is given in (\ref{eq:nc_achi_hatf}). In this step, we show that the minimum over $(C,\zeta)$ in (\ref{eq:nc_smallell_thre}) is not attained in the domain of $C(2\zeta-1)> 1-2\rho$, thus we can safely simplify $\hat{f}_2(C,\zeta,\rho,\nu)$ to $f_2(C,\zeta,\rho,\nu)$ in the threshold (\ref{eq:nc_smallell_thre}). 
   To this end,  note that $C(2\zeta-1)>1-2\rho$ is equivalent to  $$\zeta > \zeta':=\frac{1}{2}+\frac{1-2\rho}{2C}\ge\frac{1}{2},$$ and we note that $\hat{f}_2(C,\zeta,\rho,\nu)=0$ holds for any $\zeta\ge\zeta'$ (see (\ref{eq:nc_achi_hatf}) and recall that $\hat{f}_2(C,\zeta,\rho,\nu)$ is continuous), and that $f_1(C,\zeta,\rho,\nu)$ is monotonically increasing with respect to $\zeta$ in $[\frac{1}{2},1)$ (see (\ref{eq:nc_achi_defi_f1})). Therefore, for any given $(C,\zeta)$ with $C(2\zeta-1)>1-2\rho$ we can also consider $(C,\zeta')$  that 
  falls in the case of $C(2\zeta-1)\le 1-2\rho$, and we can compare the values of $f_1(C,\zeta,\rho,\nu)$ and $\hat{f}_2(C,\zeta,\rho,\nu)$ as 
  $$f_1(C,\zeta,\rho,\nu)\geq f_1(C,\zeta',\rho,\nu),~\hat{f}_2(C,\zeta',\rho,\nu)=\hat{f}_2(C,\zeta,\rho,\nu)=0,$$ 
  yielding that \begin{align}
    \label{eq:nc_use_conti}
    \max\Big\{\frac{1}{\theta}f_1(C,\zeta,\rho,\nu),\hat{f}_2(C,\zeta,\rho,\nu)\Big\}\geq \max\Big\{\frac{1}{\theta}f_1(C,\zeta',\rho,\nu),\hat{f}_2(C,\zeta',\rho,\nu)\Big\}. 
\end{align}
Thus,     the minimum over $(C,\zeta)$ in (\ref{eq:nc_smallell_thre}) is attained at the case $C(2\zeta-1)\le 1-2\rho$, in which $\hat{f}_2(C,\zeta,\rho,\nu)=f_2(C,\zeta,\rho,\nu)$. Then similarly to (\ref{eq:ber_coincide1})--(\ref{eq:ber_coincide2}) we have
\begin{align}
    &\min_{C\in(0,e^\nu),\zeta\in(0,1)}\max\Big\{\frac{1}{\theta}f_1(C,\zeta,\rho,\nu),\hat{f}_2(C,\zeta,\rho,\nu)\Big\}\\
    \ge & \min_{C\in(0,e^\nu),\zeta\in(0,1)}\max\Big\{\frac{1}{\theta}f_1(C,\zeta,\rho,\nu),f_2(C,\zeta,\rho,\nu)\Big\},
\end{align}
implying that to ensure (\ref{eq:nc_smallell_thre}) it suffices to have 
\begin{align}\label{eq:nc_achi_final_thre}
    n\geq \frac{(1+\eta)k\log\frac{p}{k}}{(1-\theta)\nu e^{-\nu}}\frac{1}{\min _{C\in(0,e^\nu),\zeta\in(0,1)}\max\{\frac{1}{\theta}f_1(C,\zeta,\rho,\nu),f_2(C,\zeta,\rho,\nu)\}}. 
\end{align}
This establishes the second term in \eqref{eq:NC_threshold}.

\section{Technical Lemmas for Near-Constant Weight Designs (Low-$\ell$)}\label{app:nc_lemms}

Compared to the Bernoulli design, an additional challenge in the near-constant weight design is that a test may receive multiple placements from a single item. We present two high-probability events that are useful in overcoming this difficulty.  

\begin{lem}\label{lem3}
    {\em (\cite[Prop. II.3]{coja2020information} and \cite[Eq. (40)]{coja2020information})}
    Under the near-constant design where each item is uniformly placed (with replacement) to $\Delta= \frac{\nu n}{k}$ tests, we have the following:

    \noindent
    (a) With $1-o(1)$ probability, the number of defective items appearing in any test more than once scales at most as $O(\log k)$.

    \noindent
    (b) With probability at least $1-p^{-7}$, the number of tests assigned to precisely one defective item (referred to as a degree-1 test in \cite{coja2020information}) is given by $(1+o(1))e^{-\nu}k\Delta$.
\end{lem}

As another useful technical result, we present the (anti-)concentration of the  hypergeometric distribution.  The Chernoff bound (upper bound) can be found in  \cite{hoeffding1963hyper}, and we expect that the matching lower bound is also known, but we provide a short proof for completeness. 

\begin{lem}
    \label{lem4}
    If $U\sim \mathrm{Hg}((1+o(1))k\Delta,(1+o(1))e^{-\nu}k\Delta,\Delta)$ with $k \to \infty$, $\Delta = \frac{\nu n}{k}$, and $n=\Theta(k\log k)$, then for a given $C\in (0,1)$ such that $C\Delta$ is an integer, we have  \begin{align}
        \PP\big(U=C\Delta\big)=\exp\big(-\Delta\big[D(C\|e^{-\nu})+o(1)\big]\big). \label{eq:hyper_diws}
    \end{align}
\end{lem}
\begin{proof}
    By letting $k_0=(1+o(1))k$ and $k_\nu = (1+o(1))e^{-\nu}k$ 
    we can write $U\sim \mathrm{Hg}(k_0\Delta,k_\nu\Delta,\Delta)$, which implies that
 \begin{equation}\label{poi}
     \PP(U=C\Delta)=\frac{\binom{k_\nu\Delta}{C\Delta}\binom{(k_0-k_\nu)\Delta}{(1-C)\Delta}}{\binom{k_0\Delta}{\Delta}}.
 \end{equation} 
 Given positive integers $m$ and $\alpha m$ for some $\alpha\in (0,1)$, we have $\frac{\sqrt{\pi}}{2}G\leq \binom{m}{\alpha m}\leq G$ where $G=\frac{\exp(mH_2(\alpha))}{\sqrt{2\pi m\alpha(1-\alpha)}}$ (see Lemma \ref{lem:coeff_bounds}), which gives $$\binom{m}{\alpha m}=\Theta(1) \frac{\exp(mH_2(\alpha))}{\sqrt{m\alpha(1-\alpha)}}.$$ Substituting this into (\ref{poi}), we obtain 
  \begin{align}\label{eq:prob_exp}
 \PP\big(U=C\Delta\big)
 &=\Theta(1) \frac{\frac{\exp(k_\nu\Delta H_2(\frac{C}{k_\nu}))}{\sqrt{C\Delta (1-\frac{C}{k_\nu})}}\cdot \frac{\exp((k_0-k_\nu)\Delta H_2(\frac{1-C}{k_0-k_\nu}))}{\sqrt{(1-C)\Delta(1-\frac{1-C}{k_0-k_\nu})}}}{\frac{\exp(k_0\Delta H_2(\frac{1}{k_0}))}{\sqrt{\Delta(1-\frac{1}{k_0})}}}
 \\&=\frac{\Theta(1)}{\sqrt{C(1-C)\Delta}} \exp\Big(\Delta\Big[\underbrace{k_\nu H_2\Big(\frac{C}{k_\nu}\Big)+(k_0-k_\nu)  H_2\Big(\frac{1-C}{k_0-k_\nu}\Big)-k_0H_2\Big(\frac{1}{k_0}\Big)}_{:=\mathcal{T}}\Big]\Big), \label{eq:useTheta1}
 \end{align}
 where (\ref{eq:useTheta1}) follows from $1-\frac{C}{k_\nu}$, $1-\frac{1-C}{k_0-k_\nu}$, and $1-\frac{1}{k_0}$ all behaving as $\Theta(1)$. Now we seek to further simplify $\mathcal{T}$.
 Note that $\frac{C}{k_{\nu}}$, $\frac{1-C}{k_0-k_{\nu}}$, and $\frac{1}{k_0}$ all scale as $o(1)$, and it holds for any $t = o(1)$ that
 \begin{equation}\begin{aligned}
     H_2(t)&=-t\log t-(1-t)\log(1-t)=-t(\log t-1+o(1)).
     \end{aligned}
 \end{equation}
  This yields
     \begin{align} 
         \mathcal{T}&=k_\nu \Big(-\frac{C}{k_\nu}\Big[\log\Big(\frac{C}{k_\nu}\Big)-1+o(1)\Big]\Big)-(k_0-k_\nu)\cdot\frac{1-C}{k_0-k_\nu}\cdot \Big[\log \Big(\frac{1-C}{k_0-k_\nu}\Big)-1+o(1)\Big]\\
         &\quad\quad \quad ~~~~~~~~~~~~~+k_0\cdot\frac{1}{k_0}\cdot \Big[\log\Big(\frac{1}{k_0}\Big)-1+o(1)\Big]\\
         &=-C\log\Big(\frac{C}{k_\nu}\Big)-(1-C)\log\Big(\frac{1-C}{k_0-k_\nu}\Big)+\log\Big(\frac{1}{k_0}\Big)+o(1)\\
         &= -C \log\Big(\frac{C}{e^{-\nu}}\Big)-(1-C)\log\Big(\frac{1-C}{1-e^{-\nu}}\Big)+o(1) \label{eq:subst_k0_knu_0} \\
         &=-D(C\|e^{-\nu})+o(1), \label{eq:subst_k0_knu}
     \end{align}
 where in (\ref{eq:subst_k0_knu_0}) we substitute $k_0=(1+o(1))k$ and $k_\nu=(1+o(1))e^{-\nu}k$. To complete the proof, we substitute (\ref{eq:subst_k0_knu}) into (\ref{eq:useTheta1}), and then note that the leading factor $\frac{\Theta(1)}{\sqrt{C(1-C)\Delta}}$ with fixed $C\in [\Delta^{-1},1-\Delta^{-1}]$ (this follows from $C\in(0,1)$ and hence  $1\le C\Delta \le\Delta-1$) can be written as $\exp(- o(\Delta))$ since $\Delta=\frac{\nu n}{k}=\Theta(\log k)\to \infty$.
\end{proof}
Next, we present a useful concentration bound   analogously to \cite[Lem. 1]{Joh16}.  We consider $\mathcal{J}\subset [p]$ with  $|\mathcal{J}|=\ell$ and use $\mathcal{W}^{(\mathcal{J})}$ to index the tests that contain some item from $\mathcal{J}$, with cardinality given by $W^{(\mathcal{J})}:=|\mathcal{W}^{(\mathcal{J})}|.$
\begin{lem}\label{lem44}
    {\em (Concentration of $W^{(\mathcal{J})}$)}
    Given $\mathcal{J}\subset [p]$ with $|\mathcal{J}|=\ell$, 
    for any $\varepsilon>0$  we have \begin{equation}\label{2083}
         \PP\Big(\big|W^{(\mathcal{J})}-\mathbbm{E}W^{(\mathcal{J})}\big|\geq \varepsilon\Big)\leq 2\exp\Big(-\frac{2\varepsilon^2k}{\nu\ell n}\Big),
     \end{equation} 
    where the expectation satisfies
    \begin{align}\label{9262} 
        \mathbbm{E} W^{(\mathcal{J})} &=n \Big(1-\exp\Big(-\frac{\ell\nu(1+o(1))}{k}\Big)\Big)\\&= \begin{cases}
            (1-e^{-\nu\alpha}+o(1))n, ~~\text{if }\frac{\ell}{k}\to \alpha\in (0,1]\\
~~~~~ \frac{n\nu\ell}{k}(1+o(1)),~~~~\text{if }\frac{\ell}{k}\to 0.
        \end{cases} \label{eq:simpli9262}
    \end{align}
    %when $\frac{\ell}{k}\to \alpha \in (0,1]$, $\mathbbm{E}W^{(\mathcal{M})}=(1-e^{-\nu\alpha}+o(1))n$, and for any $\delta>0$,
    %with probability at least $1-2\exp\big(-(2+o(1))\delta^2\nu\alpha n\big)$ 
    %{\color{red}[JR: Previously miss the term of $\ell\log\frac{ek}{\ell}$, and this also harms Lemma \ref{nc-theta1}]} 
    %we have
    Moreover, given any $\delta>0$, with probability at least $1-2\exp(-2\nu^{-1}\delta^2n)$, we have 
    \begin{equation}\label{9263}
        W^{(\mathcal{J})}= \mathbbm{E}W^{(\mathcal{J})}+\delta' n \sqrt{\frac{\ell}{k}},\text{  for some }|\delta'|<\delta.
    \end{equation}
\end{lem}
\begin{proof}  
%We only need to prove (\ref{2083}), (\ref{9262}) and (\ref{9263}) because the statements after "More particularly" are their direct outcomes. 
    We first prove (\ref{9262}). We can write 
    $ W^{(\mathcal{J})}=\sum_{i=1}^n \mathbbm{1}(\text{test }i\text{ contains some item from }\mathcal{J}),$  which leads to  
\begin{align} \mathbbm{E}W^{(\mathcal{J})}&=n\cdot\PP(\text{test }i\text{ contains some item from }\mathcal{J})
    \\&=n\Big[1-\Big(1-\frac{1}{n}\Big)^{\ell\Delta}\Big] 
    \\&= n\Big[1-\exp\Big(\frac{\ell\nu n}{k}\log\big(1-\frac{1}{n}\big)\Big)\Big]\\
    &=n\Big[1-\exp\Big(-\frac{\ell\nu(1+o(1))}{k}\Big)\Big],\label{9264}
\end{align} 
where we use $\log(1-\frac{1}{n})=-\frac{1}{n}(1+o(1))$ in (\ref{9264}). Thus, (\ref{9262}) follows, and to obtain (\ref{eq:simpli9262})  we substitute $\frac{\ell}{k}=\alpha+o(1)$ if $\frac{\ell}{k}\to\alpha\in(0,1]$, or $1-\exp(-\beta)=\beta(1+o(1))$ as $\beta \to 0$ if $\frac{\ell}{k}=o(1)$.

Next, we establish (\ref{2083}).  Because the $\ell\Delta$ placements of items in $\mathcal{J}$ are independent and changing one placement changes the quantity $W^{(\mathcal{J})}$ by at most one,   McDiarmid's inequality \cite{mcdiarmid1989method} gives $$\PP\big(|W^{(\mathcal{J})}-\mathbbm{E}W^{(\mathcal{J})}|\geq \varepsilon\big)\leq 2\exp\Big(-\frac{2\varepsilon^2}{\ell\Delta}\Big)$$ for any $\varepsilon>0$,
which yields (\ref{2083}) by substituting $\ell\Delta = \frac{\nu\ell n}{k}$. 
Finally, we show (\ref{9263}) by using (\ref{2083}). To this end, we take $\varepsilon =\delta  n\sqrt{\frac{\ell}{k}}$ in (\ref{2083}), and note that the event in (\ref{9263}) holds true if and only if $|W^{(\mathcal{J})}-\mathbbm{E}W^{(\mathcal{J})}| < \delta  n\sqrt{\frac{\ell}{k}}$.  
\end{proof}

Finally, the following lemma characterizes the probability mass function for the random variable indicating how many tests from a certain subset of $[n]$ contain items from a certain subset of $[p]$.

 \begin{lem}\label{lem77}
     Suppose that the $n$ tests are divided into two parts  consisting of $n_1$ tests and $n_2$ tests, with the first part and the second part respectively indexed by $\mathcal{N}_1\subset [n]$ and $\mathcal{N}_2\subset [n]$, where $|\mathcal{N}_1|=n_1$, $|\mathcal{N}_2|=n_2$, and $n_1+n_2=n$. Given $\mathcal{J}\subset [p]$ with $|\mathcal{J}|=\ell$, consider the $\ell\Delta = \frac{\ell\nu n}{k}$ placements of the $\ell$ items in $\mathcal{J}$, and let   $M$ be the number of tests in $\mathcal{N}_2$ that contain some item from $\mathcal{J}$.  Then for any $0\leq m\leq \ell\Delta$, we have \begin{equation}\label{2897}
       \mathbb{P}(M=m)\leq \exp \left(-\ell\Delta \Big[D\Big(\frac{m}{\ell\Delta}\Big\|\frac{n_2}{n}\Big)+\frac{\ell\Delta}{4n_1}\Big]\right),
    \end{equation}
    where the probability is with respect to the randomness in placing each item from $\mathcal{J}$ in $\Delta = \frac{\nu n}{k}$ tests chosen uniformly at random with replacement.
    
    % the randomness of $\Tv_{\mathcal{J}}$, defined (as was done shortly before \eqref{eq:calK_prime}) to be the unordered multi-set of length $|\mathcal{J}|\Delta$ with entries in $[n]$, representing the $|\mathcal{J}|\Delta$ placements from the items in $\mathcal{J}$ in an unordered manner.
    % \footnote{\color{blue}[TODO]: We may need to define this recurring notation; currently it seems only appear in the proof of converse.} 
   % In particular, if we have $n_2=e^{-\nu}n(1+o(1))$ as in (\ref{n1n2bound}), then (\ref{2897}) gives \begin{equation}\label{2898}
    %    \Psf(M=m)\leq \exp \left(-\ell\Delta \Big[D\Big(\frac{m}{\ell\Delta}\Big\| e^{-\nu}(1+o(1))\Big)+o(1)\Big]\right).
    %\end{equation}
\end{lem}

\begin{proof}  Our argument builds on \cite[Sec. IV-C]{Joh16}, which characterizes a similar distribution (for different purposes) under the near-constant weight design. Our Lemma \ref{lem77} generalizes the one therein in two ways. First, we consider the placement of $\ell$ items, while \cite[Lemma 4]{Joh16} only concerns the case that $\ell=1$. Second, our statement is more general since we leave $\mathcal{N}_1$ and $\mathcal{N}_2$ unspecified, while \cite[Lemma 4]{Joh16} concerns a more specific choice.  As was done in \cite{Joh16}, we make use of the {\it Stirling numbers of the second kind}, defined as $\stirlingii{n}{k}:= \frac{1}{k!}\sum_{j=0}^k(-1)^{k-j}\binom{k}{j}j^n$, e.g., see \cite[P. 204]{lieb1968concavity}.
    
    %Note that the $\ell$ items in $\sdif$ provide $\ell\Delta = \frac{\ell \nu n}{k}$ placements in total, and hence $0\leq M\leq \ell\Delta$. 
    % Firstly, we seek to characterize $\mathbb{P}(M=m)$    for specific $0\leq m\leq\ell\Delta$.  
    %
    The probability of the event that exactly $u$ placements fall in $\mathcal{N}_1$ is $\PP\big(\mathrm{Bin}(\ell\Delta, \frac{n_1}{n})=u\big)$. Conditioning on this event, we need to calculate the probability that the remaining $\ell\Delta -u$ placements are at $m$ distinct tests in $\mathcal{N}_2$. Evidently, we require $\ell\Delta -u \geq m$, or equivalently $u\leq \ell\Delta -m$, in order to have $M=m$. Note that there are overall $n_2^{\ell\Delta -u}$ possibilities for performing $\ell\Delta-u$ placements in $n_2$ tests. By a standard counting argument, we can choose these $m$ distinct tests with $\binom{n_2}{m}$ ways, then $\stirlingii{\ell\Delta -u}{m}$ ways of arranging the $\ell\Delta -u$ placements into these $m$ unlabeled bins (which represent the $m$ tests) such that none of them are empty, and finally $m!$ different labelings of the bins. These findings give 
    \begin{align}\label{15.1}
   \PP(M=m) &= \sum_{u=0}^{\ell\Delta -m}\PP\big(\mathrm{Bin}\big(\ell\Delta,\frac{n_1}{n}\big)=u\big)\cdot \PP(\ell\Delta -u \text{ placements  occupy}\text{ $m$ tests in }\mathcal{N}_2)\\
&= \sum_{u=0}^{\ell\Delta -m}\binom{\ell\Delta}{u}\Big(\frac{n_1}{n}\Big)^u \Big(\frac{n_2}{n}\Big)^{\ell\Delta-u}\binom{n_2}{m} \cdot\frac{ \stirlingii{\ell\Delta -u}{m}m!}{n_2^{\ell\Delta-u}}
\\&=\frac{\binom{n_2}{m}m!}{n^{\ell\Delta}}  \cdot\underbrace{\sum_{u=0}^{\ell\Delta -m}\binom{\ell\Delta}{u}n_1^u \stirlingii{\ell\Delta -u}{m}}_{:= \Phi}.
%\\&:= \frac{\binom{n_2}{m}m!}{n^{\ell\Delta}}  \cdot I. 
    \end{align}
We bound $\Phi$ as follows:
    \begin{align}\label{15.2}
        \Phi&  =  \sum_{t=0}^{\ell\Delta -m}\binom{\ell\Delta}{m+t}n_1^{\ell\Delta - m -t}
        \stirlingii{m+t}{m}\\
        &\label{eq:15.2ii} \leq \sum_{t=0}^{\ell\Delta -m}\binom{\ell\Delta}{m+t}\binom{t+m}{m}m^tn_1^{\ell\Delta - m -t}\\
        &\label{eq:15.2iii} = \binom{\ell\Delta}{m}n_1^{\ell\Delta-m}\sum_{t=0}^{\ell\Delta-m}\binom{\ell\Delta-m}{t}\Big(\frac{m}{n_1}\Big)^t\\& \label{eq:15.2iv}=   \binom{\ell\Delta}{m}n_1^{\ell\Delta -m}\Big(1+\frac{m}{n_1}\Big)^{\ell\Delta-m}
        \\ \label{eq:15.2v}&:= C\binom{\ell\Delta}{m}n_1^{\ell\Delta -m}, 
    \end{align}
where (\ref{15.2}) follows from a change of the variable $t = \ell\Delta - m -u$, (\ref{eq:15.2ii}) follows from $\stirlingii{t+m}{m}\leq \binom{t+m}{m}m^t$ \cite{rennie1969stirling}, (\ref{eq:15.2iii}) follows from $\binom{\ell\Delta}{m+t}\binom{t+m}{m}=\binom{\ell\Delta}{m}\binom{\ell\Delta-m}{t}$, (\ref{eq:15.2iv}) follows from $\sum_{t=0}^{\ell\Delta-m}\binom{\ell\Delta-m}{t}\big(\frac{m}{n_1}\big)^t= \big(1+\frac{m}{n_1}\big)^{\ell\Delta-m}$, and in (\ref{eq:15.2v}), $C:= (1+\frac{m}{n_1})^{\ell\Delta -m}$ is bounded by \begin{equation}\begin{aligned}\label{Cbound}
    C &= \exp\Big([\ell\Delta-m]\log\Big(1+\frac{m}{n_1}\Big)\Big)\leq \exp\Big(\frac{m[\ell\Delta -m]}{n_1}\Big)\leq \exp\Big(\frac{\ell^2\Delta^2}{4n_1}\Big).
    \end{aligned}
\end{equation}
%Because with high prob. $n_1=\Theta(n)$ and $\frac{\ell}{k}=o(1)$, we have \begin{equation}
 %   C =\exp\Big(\ell\Delta o(1)\Big)
%\end{equation}
%which may not have large impact. 
Substituting (\ref{eq:15.2v}) back into (\ref{15.1}), we obtain
    \begin{align}
        \PP\big(M=m\big)&\leq C\binom{\ell\Delta}{m}n_1^{\ell\Delta-m}\frac{\binom{n_2}{m}m!}{n^{\ell\Delta}}
        \\&\label{eq:369i}\leq C\binom{\ell\Delta}{m}\frac{n_1^{\ell\Delta-m}n_2^m}{n^{\ell\Delta}}\\\label{eq:369ii}
        & = \exp\Big(\frac{\ell^2\Delta^2}{4n_1}\Big)\PP\Big(\mathrm{Bin}\big(\ell\Delta,\frac{n_2}{n}\big)=m\Big)\\
        & \label{eq:369iii}\leq  \exp\Big(\frac{\ell^2\Delta^2}{4n_1}-\ell\Delta \cdot D\Big(\frac{m}{\ell\Delta}\Big\|\frac{n_2}{n}\Big)\Big), 
    \end{align}
where (\ref{eq:369i}) follows from $\binom{n_2}{m}m!\leq n_2^m$, (\ref{eq:369ii}) follows from (\ref{Cbound}), and in (\ref{eq:369iii}) we apply the Chernoff bound (see (\ref{eq:chernoff1}) in Lemma \ref{binoconcen}).  This completes the proof.  
\end{proof}

\section{High-$\ell$ Achievability Analysis for Both Designs} \label{app:high_both}

\subsection{Initial Results}

With the information-theoretic tools from Section \ref{sec:it_tools} in place (recalling in particular the information density $\imath^n(\Xvdif; \Yv | \Xveq )$ and mutual information $I_{\ell}^n$ with $\ell = |\sdif|$), we proceed to state an initial non-asymptotic bound from \cite{Sca15b} that we use as a starting point.  Recalling that we already split the analysis into $\ell \le \frac{k}{\log k}$ and $\ell > \frac{k}{\log k}$ (and here we consider the latter), we define $\ell_{\min} = 1 + \lfloor\frac{k}{\log k}\rfloor$ for notational convenience. 

From \cite[Thm.~4]{Sca15b},\footnote{More precisely, this is the adapted version for approximate recovery in \cite[App.~D]{Sca15b}, because as noted in Appendix \ref{app:roadmap}, error events for incorrect sets within distance $\frac{k}{\log k}$ of the true defective set are considered separately (in our analysis of the the low-$\ell$ regime) so are not considered here.} we know that by a suitable choice of $\{\gamma_{\ell}\}_{\ell = \ell_{\min}}^{k}$, the error probability is upper bounded as follows for any $\delta_1 > 0$:
\begin{equation}
    \PP({\rm err}) \le \PP\bigg[ \bigcup_{(\sdif,\seq)\,:\,|\sdif| \ge \ell_{\min}} \bigg\{ \imath^n(\Xvdif; \Yv | \Xveq ) \le \log {{p-k} \choose |\sdif|} + \log\bigg(\frac{k}{\delta_1}{k \choose |\sdif|}\bigg) \bigg\} \bigg] + \delta_1. \label{eq:info_density}
\end{equation}
As a result, for any $\delta_2 \in (0,1)$, if the number of tests $n$ is large enough to ensure that
\begin{equation}
   \max_{\ell=\ell_{\min},\dotsc,k} \frac{ \log{{p-k} \choose \ell} + \log\big(\frac{k}{\delta_1}{k \choose \ell}\big) }{ I_\ell^n (1-\delta_{2}) } \le 1, \label{eq:final_ach}
\end{equation}
and if each information density satisfies a (one-sided) concentration bound of the form
\begin{equation}
     \PP\big[ \imath^n(\Xvdif; \Yv | \Xveq) \le (1 - \delta_{2})I_\ell^n \big] \le \psi_{\ell}(n,\delta_{2}), \label{eq:psi_ach}
\end{equation}
for some functions $\{\psi_{\ell}\}_{\ell=1}^{k}$, then we have
\begin{equation}
    \PP({\rm err}) \le \sum_{\ell=\ell_{\min}}^k{k \choose \ell}\psi_\ell(n,\delta_{2}) + \delta_1. \label{eq:nonasymp_ach}
\end{equation}
Note that while \cite{Sca15b} focuses on the Bernoulli design, the analysis up to this point does not use any specific properties of that design except symmetry with respect to re-ordering items, so it remains valid for the near-constant weight design.  The only minor difference above compared to \cite{Sca15b} is that we use the full mutual information $I_{\ell}^n$ instead of simplifying it to $n I_{\ell}$ for some single-test counterpart $I_{\ell}$ (which can be done for the Bernoulli design but not in general).

We can take $\delta_1 \to 0$ sufficiently slowly to have an asymptotically negligible impact in \eqref{eq:final_ach}, so the key step in ensuring $\PP({\rm err}) \to 0$ is showing that the summation in \eqref{eq:nonasymp_ach} approaches 0.  We will take $\delta_2$ to be a constant arbitrarily close to zero, as was done in \cite{Sca15b},\footnote{For low $\ell$ values it should instead be taken as a parameter to be optimized \cite{Sca15b}, but we are only using these techniques for the high-$\ell$ regime.} and we will show that the condition \eqref{eq:final_ach} simplifies to a requirement on $n$ corresponding to the first term in \eqref{eq:ber_threshold}. 

Naturally, to simplify \eqref{eq:final_ach} we need to characterize the behavior of $I_{\ell}^n$.  For the Bernoulli design, the following asymptotic characterization was given in \cite{Sca15b}.

\begin{lem} \label{lem:mi_bern}
   {\rm (Mutual Information for the Bernoulli Design, \cite[Prop. 5]{Sca15b})} For the noisy group testing problem under Bernoulli design with i.i.d. $\mathrm{Bernoulli}(\frac{\nu}{k})$ entries, consider arbitrary sequences of $k\to\infty$ and $\ell\in\{1,...,k\}$ (both indexed by $p$). 
    If $\frac{\ell}{k}=o(1)$, then \begin{equation}\label{12.7}
        I_\ell^n=\left(\frac{n\nu e^{-\nu}\ell}{k}(1-2\rho)\log\frac{1-\rho}{\rho}\right)\big(1+o(1)\big).
    \end{equation}
    Moreover, if $\frac{\ell}{k}\to\alpha\in(0,1]$, then \begin{equation}\label{2077}
        I_\ell^n=ne^{-(1-\alpha)\nu}\big(H_2(e^{-\alpha\nu}\star\rho)-H_2(\rho)\big)\big(1+o(1)\big).
    \end{equation}
\end{lem}

We will show that the same result holds for the near-constant weight design in Lemma \ref{lem:mi_ncc} below.

\subsection{Bernoulli Design}

As we have mentioned earlier, the high-$\ell$ analysis of the Bernoulli design comes directly from \cite{Sca15b} with only minor changes.  We thus only provide a brief outline for the unfamiliar reader, citing the relevant results therein.

The following concentration term was shown to be valid (i.e., \eqref{eq:psi_ach} holds) in \cite[Prop.~3]{Sca15b}:
\begin{equation}
    \psi_{\ell}(n,\delta_2) = 2\exp\bigg( - \frac{n(\delta_2 I_{\ell})^2}{ 4(8+\delta_2 I_{\ell}) } \bigg), \label{eq:psi_bern}
\end{equation}
where $I_{\ell} = I_{\ell}^n / n$ is the mutual information associated with a single test.  
Under this choice, it is shown in \cite[Prop.~7(ii)]{Sca15b} that $\sum_{\ell > \frac{k}{\log k}} {k \choose \ell} \psi_\ell(n,\delta_{2}) \to 0$ as long as $n = \Omega(k \log \frac{p}{k})$ (regardless of the implied constant).  Moreover, in  \cite[Sec.~C.2]{Sca15b} it is shown for $\nu = \log 2$ that the maximum in \eqref{eq:final_ach} is asymptotically attained by $\ell=k$, and that the limiting behavior exactly corresponds to the first term in \eqref{eq:ber_threshold}.  

A slight difference in our work is that we need to handle general $\nu > 0$, rather than only $\nu = \log 2$, with the non-trivial step being to show that the maximum in \eqref{eq:final_ach} is still asymptotically attained by $\ell=k$.\footnote{The proof of $\sum_{\ell > \frac{k}{\log k}} {k \choose \ell} \psi_\ell(n,\delta_{2}) \to 0$ is not impacted by changing $\nu$, since doing so only amounts to changing constant factors in the exponent in \eqref{eq:psi_bern}.}  Fortunately, this follows from another well-known result in the literature \cite[Thm.~3]{Mal80} (see also \cite[Lemma 2]{Yav20}) that we now proceed to outline.

% Fortunately, this is essentially just a matter of replacing the analysis of $\nu = \log 2$ from \cite[Thm.~3a]{Mal78} by that of \cite[Thm.~3]{Mal80} (see also \cite[Lemma 2]{Yav20}), which applies not only for general $\nu$ but also for a much broader class of noise models.

First note that for $\ell > \frac{k}{\log k}$ and $k=o(p)$, the second term in the numerator of \eqref{eq:final_ach} is asymptotically negligible compared to the first.  We further have $\log {p-k \choose \ell} = \big(\ell\log\frac{p}{\ell}\big)(1+o(1)) = \big(\ell\log\frac{p}{k}\big)(1+o(1))$, so the maximization over $\ell$ simplifies to maximizing $\frac{\ell}{I_{\ell}^n}$, or equivalently minimizing $\frac{I_{\ell}^n}{\ell}$.  The above-mentioned results in \cite[Thm.~3]{Mal80} and \cite[Lemma 2]{Yav20} directly state that $\ell=k$ minimizes $\frac{I_{\ell}^n}{\ell}$ (for a broader class of noise models that includes ours), as desired.

Having established that $\sum_{\ell > \frac{k}{\log k}} {k \choose \ell} \psi_\ell(n,\delta_{2}) \to 0$ and that $\ell=k$ asymptotically attains the maximum in \eqref{eq:final_ach}, the desired threshold on $n$ follows by writing the numerator therein as $\big(k\log\frac{p}{k}\big)(1+o(1))$ and the denominator as $n(H_2(e^{-\nu} \star \rho) - H_2(\rho))(1+o(1)) \times (1-\delta_2)$ (see \eqref{2077} with $\alpha = 1$), and recalling that $\delta_2$ can be arbitrarily small.  This establishes the first term in \eqref{eq:ber_threshold}.

\subsection{Near-Constant Weight Design}

For the near-constant weight design, we adopt the same choices of $\delta_1$ and $\delta_2$ as those for the Bernoulli design.  Notice that the study of the condition \eqref{eq:final_ach} has no direct relation to the choice of $\psi_{\ell}$.  Thus, to simplify \eqref{eq:final_ach}, we can apply the exact same analysis as the Bernoulli case as long as we can show that the mutual information terms $I_{\ell}^n$ have the same asymptotic behavior; this is done in the following lemma, whose proof is given in Appendix \ref{sec:pf_mi_ncc}.
    
\begin{lem} \label{lem:mi_ncc}
    {\rm (Mutual Information, Near-Constant Weight Design)}
  For the noisy group testing problem under the near-constant weight design with $\Delta = \frac{\nu n}{k}$, consider sequences of sparsity levels $k\to \infty$ and $\ell\in \{1,...,k\}$ (both indexed by $p$). If $\frac{\ell}{k}=o(1)$, then \begin{equation}\label{1976}
        I_\ell^n=\left(\frac{n\nu e^{-\nu}\ell}{k}(1-2\rho)\log\frac{1-\rho}{\rho}\right)\big(1+o(1)\big).
    \end{equation}
    Moreover, if $\frac{\ell}{k}\to\alpha\in(0,1]$, then \begin{equation}\label{2077nc}
        I_\ell^n=ne^{-(1-\alpha)\nu}\big(H_2(e^{-\alpha\nu}\star\rho)-H_2(\rho)\big)\big(1+o(1)\big).
    \end{equation}
\end{lem}

From this lemma and the above discussion, we deduce that \eqref{eq:final_ach} again leads to the first term in \eqref{eq:ach_ncc}.  It only remains to show that $\sum_{\ell > \frac{k}{\log k}} {k \choose \ell} \psi_\ell(n,\delta_{2}) \to 0$.

While we cannot directly use \eqref{eq:psi_bern}, the proof of \cite[Prop.~7(ii)]{Sca15b} reveals that we do not precisely need \eqref{eq:psi_bern}, but rather, the term inside the exponent can be scaled be any finite constant factor.  Since $I_{\ell} = O(1)$ and $\delta_2 = \Theta(1)$, the denominator in \eqref{eq:psi_bern} is dominated by the constant term (to within a constant factor), and thus, what we seek is a tail bound of the form $e^{-\Theta(n I_{\ell}^2)} = e^{-\Theta((I_{\ell}^n)^2 / n)}$.  We provide this in the following lemma, whose proof is given in Appendix \ref{sec:pf_conc_ncc}.

\begin{lem}  \label{nc-theta1} {\rm (Concentration Bound for $\frac{\ell}{k}\to \alpha$, Near-Constant Weight Design)}
For the noisy group testing problem under near-constant design, consider sequences $k\to\infty$ and $\ell\geq 1$ (indexed by $p$) such that $\frac{\ell}{k}\to\alpha\in [0,1]$. For any given $\delta$ smaller than a quantity depending on $(\alpha,\nu ,\rho)$, the following holds for sufficiently large $p$ and any $(\sdif,\seq)$ with $|\sdif|=\ell$: %{\color{red}\bf[JR: again similarly to footnote 2, if we only want this statement, do we really need $\ell\log\frac{ek}{\ell}$ below to guarantee uniformity?]}
\begin{equation}
    \label{2907}\PP\Big(\imath^n(\Yv|\Xveq,\Xvdif)\leq (1-\delta)I_\ell^n\Big)\leq  \exp\Big(-\frac{C_{\rho,\nu,\alpha}(I_\ell^n)^2 \delta^2}{n}\Big),
\end{equation}
where $C_{\rho,\mu,\alpha}>0$ is a constant depending on  $(\rho,\nu,\alpha)$, and $I_\ell^n$ satisfies (\ref{1976}) when $\alpha =0$, or (\ref{2077nc}) when $\alpha\in (0,1]$.%{\color{black}Moreover, for $\frac{\ell}{k}\to 0$ we similarly have \begin{equation}
   %\label{tozero}\PP\Big(\imath^n(\Yv|\Xveq,\Xvdif) \leq (1-\delta)I_\ell^n\Big)\leq 4\exp \Big(-\frac{C'_{\rho,\nu}\delta^2 I^2(\ell)}{n}\Big)
%\end{equation}for some constant $C'_{\rho,\nu}>0$ depending on $(\rho,\nu)$, 
%where $I_\ell^n$ is given in (\ref{1976}).} 
\end{lem}

The proof is given in Appendix \ref{sec:pf_conc_ncc}.  With this result in place, the rest of the analysis again follows that of the Bernoulli design to establish the first term in \eqref{eq:NC_threshold}.

\section{High-$\ell$ Converse Analysis for Both Designs} \label{app:high_conv}

Recall that the thresholds in (\ref{eq:ber_threshold}) and (\ref{eq:NC_threshold}) contain two terms; this appendix concerns only the first, which corresponds to the high-$\ell$ regime.

\subsection{Bernoulli Design}

For the Bernoulli design, the first term in \eqref{eq:ber_threshold} was already given in \cite[Thm.~7 \& App.~C]{Sca15b}, and mainly comes down to studying the lower tail probability of the information density $\imath^n(\Xv_s,\Yv)$ (i.e., \eqref{12.2} with $\seq = \emptyset$, corresponding to $\ell = k$).  Here establishing concentration is relatively straightforward due of the i.i.d.~nature of the Bernoulli design.  We omit any further details, since (i) they are already known from \cite{Sca15b}, and (ii) for the near-constant weight design (analyzed below) we will follow a similar structure and we will give more detail for it.

\subsection{Near-Constant Weight Design}

Unlike the Bernoulli design, the first term in \eqref{eq:NC_threshold} does not readily follow from \cite{Sca15b}, since their analysis was specific to the Bernoulli design.  However, it turns out that we can get the desired result by combining certain tools from \cite{Sca15b} with a fairly simple analysis of the information density. 

We start with the following non-asymptotic bound from \cite[Thm.~5]{Sca15b}, which only relies on the design being symmetric with respect to re-ordering items:
\begin{equation}
    \PP({\rm err}) \ge \PP\bigg( \imath^n(\Xv_s,\Yv) \le \log\Big( \delta_1 {p \choose k} \Big) \bigg) - \delta_1 \label{eq:conv_init}
\end{equation}
for any $\delta_1 > 0$, where we implicitly condition on the defective set being $S = s = [k]$ (without loss of generality), and $\imath^n(\Xv_s,\Yv) = \log\frac{\PP(\Yv|\Xv_s)}{\PP(\Yv)}$ is a shorthand for the information density \eqref{eq:info_density} with $\seq = \emptyset$ (and $\sdif = [k]$).  Hence, its mean is $I_k^n$ (i.e., $I_{\ell}^n$ from \eqref{eq:mi} with $\ell = k$), which behaves as $n(H_2(e^{-\nu}\star\rho)-H_2(\rho))(1+o(1))$ by Lemma \ref{lem:mi_ncc}.

The main step of the proof is to show the following one-sided concentration bound for $\imath^n$ with arbitrarily small $\delta_2 > 0$:
\begin{equation}
    \PP\big( \imath^n(\Xv_s,\Yv) \le (1+\delta_2) I_k^n \big) \to 1. \label{eq:conv_goal}
\end{equation}
To do so, we write $\imath^n(\Xv_s,\Yv) = \log\PP(\Yv|\Xv_s) - \log\PP(\Yv)$ and handle the two terms separately.  For the first term, we note that given $\Xv_s$, the sequence $\Yv$ has independent entries depending on the ${\rm Bernoulli}(\rho)$ noise, from which a simple concentration argument (e.g., Hoeffding's inequality) gives $\log\PP(\Yv|\Xv_s) = -nH_2(\rho)(1+o(1))$ with probability approaching one.  For the second term, we decompose $\PP(\Yv) = \PP(V)\PP(\Yv|V)$, where $V$ is the number of 1s in $\Yv$.  We will utilize two useful lemmas regarding $V$, the first of which is the following concentration bound.

\begin{lem} \label{lem:conc_V}
    {\em (Concentration of $V$)} Under the preceding setup and notation, we have with probability approaching one (as $p \to \infty$) that $V = (e^{-\nu}\star\rho+o(1))n$.
\end{lem}
\begin{proof}
    See (\ref{855}) in the proof of Lemma \ref{nc-theta1} in Appendix \ref{app:high_ncc}.  Note that (\ref{855}) considers general $(\seq,\sdif)$ and defines $V$ to be the number of 1s in tests that contain no items from $\seq$, but this reduces to the above definition when $\seq = \emptyset$ (and thus $\alpha = 1$ in (\ref{855})).
\end{proof}

Now observe that for any fixed $v$ satisfying the high-probability condition $v = (e^{-\nu}\star\rho+o(1))n$, and any corresponding $\yv$ whose number of $1$s is $v$, we have by symmetry that $\PP(\yv|v) = \frac{1}{{n \choose v}}$, which in turn behaves as $e^{n(H_2( e^{-\nu}\star\rho ) + o(1))}$.  

Combining the preceding two findings with the fact that $I_k^n = n(H_2(e^{-\nu}\star\rho)-H_2(\rho))(1+o(1))$ (Lemma \ref{lem:mi_ncc}), we see that  in order to establish \eqref{eq:conv_goal}, it only remains to show that the impact of the remaining term $\log P(V)$ is asymptotically negligible; namely, it should scale as $o(n)$ to be insignificant compared to the leading $\Theta(n)$ terms.  This is stated in the following lemma, which is proved in Appendix \ref{app:high_ncc}.

\begin{lem} \label{lem:P_V}
    {\em (Behavior of $P(V)$)} Under the preceding setup and notation, we have with probability approaching one (as $p \to \infty$) that $P(V) = e^{-o(n)}$.
\end{lem}

Intuitively, this result holds because $V$ is produced via binomial random variables (due to the independent noise variables); it can readily be deduced from anti-concentration bounds (Lemma \ref{binoconcen}) that for $U \sim {\rm Bin}(N,q)$   with $N \to \infty$ and constant $q \in (0,1)$, we have that $P(U=u) \ge e^{-o(N)}$ whenever $u=Nq(1+o(1))$ (which holds with high probability).  Since the full details are slightly more technical without being much more insightful than this intuition, we defer the formal proof of Lemma \ref{lem:P_V} to Appendix \ref{sec:pf_P_V}.

Combining the above findings, we have established that \eqref{eq:conv_goal} holds for arbitrarily small $\delta_2$.  Moreover, in \eqref{eq:conv_init}, we can take $\delta_1$ to decay with $p$ sufficiently slowly such that $\log\big(\delta_1 {p \choose k}\big) = \big(k \log \frac{p}{k}\big)(1+o(1))$ (i.e., the impact of $\delta_1$ is asymptotically negligible), and comparing this $\big(k \log \frac{p}{k}\big)(1+o(1))$ term with the term $I_k^n = n(H_2(e^{-\nu}\star\rho)-H_2(\rho))(1+o(1))$ in \eqref{eq:conv_goal}, we deduce the desired threshold on $n$ given by the first term in \eqref{eq:NC_threshold}.

\section{Proofs of Technical Results in the High-$\ell$ Analysis for the Near-Constant Weight Design} \label{app:high_ncc}

Recall that each item is placed in $\Delta:=\frac{\nu n}{k}$ tests for some $\nu=\Theta(1)$, and that the information density $\imath^n(\Xvdif;\Yv|\Xveq)$ is defined in (\ref{12.2}). % \footnote{Here, note that (\ref{12.1}) is no longer valid for near-constant designs where different tests are dependent.}
%Our goal is to bound $\PP\big(\imath^n(\Xvdif;\Yv|\Xveq)\leq (1-\delta)I_\ell^n\big)$ for the cases of $\frac{\ell}{k}=\Theta(1)$ and $\frac{\ell}{k}=o(1)$.  We first introduce some useful notations. 
Throughout much of the analysis in this appendix, it will be useful to condition on $\Xveq$ (or its specific realization $\xveq$ depending on the context), and we accordingly adopt the shorthand $\Psf(\cdot)=\PP(\cdot|\Xveq)$ (or $\Psf(\cdot)=\PP(\cdot|\xveq)$).  

Given $\Xveq$, we split $\Yv$ into $\Yv_1$ and $\Yv_2$, where $\Yv_1$ represents the results of tests corresponding to $X_{\seq}^{(i)} \neq 0$, and $\Yv_2$ represents the remaining results corresponding to $X_{\seq}^{(i)} = 0$.  Due to the ``OR'' operation in the group testing model, $\Yv_1$ is completely determined by the noise in its tests, which implies that $\Yv_1$ and $\Yv_2$ are conditionally independent given $\Xveq$. Therefore, we have \begin{align}
\label{788}    \imath^n(\Yv|\Xveq,\Xvdif)&=\log\frac{\PP(\Yv|\Xveq,\Xvdif)}{\PP(\Yv|\Xveq)}\\
    &=\log \frac{\Psf(\Yv_1| \Xvdif)\Psf(\Yv_2|\Xvdif)}{\Psf(\Yv_1)\Psf(\Yv_2)} 
    \\&=\log\frac{\Psf(\Yv_2|\Xvdif)}{\Psf(\Yv_2)}, \label{eq:final_i_ele}
\end{align}
%where starting from $(i)$ we introduce the shorthand $\Psf(\cdot):=\PP(\cdot|\Xveq)$ to implicitly condition on $\Xveq$ (note that   the division of $\Yv$ into $\Yv_1,\Yv_2$ is fixed via the given $\Xveq$), %{\color{red}[JR: I try to describe the effect of "implicitly" condition on $X_{eq}$.]} 
where (\ref{eq:final_i_ele}) follows from $\Psf(\Yv_1|\Xvdif)=\Psf(\Yv_1)$ due to the fact that $\Yv_1$ only depends on the noise. We let $n_i$ be the length of $\Yv_i$ ($i=1,2$), which is deterministic given $\Xveq$.

\subsection{Proof of Lemma \ref{nc-theta1} (Concentration Bound for $\frac{\ell}{k}\to\alpha\in[0,1]$)} \label{sec:pf_conc_ncc}

We consider the fixed defective set $S = s = [k]$ without loss of generality, since the design is symmetric with respect to re-ordering items.

\subsubsection*{The case $\alpha\in (0,1]$}

\noindent
\textbf{Concentration behavior of various quantities:} We first consider the case $\alpha \in (0,1]$ and establish (\ref{2907}).
Recall that we let $n_i$ be the length of $\Yv_i$ for $i=1,2$. Let $M$   denote the number of tests corresponding to ``$X_{\seq}^{(i)}=0, X_{\sdif}^{(i)} \ne 0$'',\footnote{Here, $M$ represents the random variable, while we will use lowercase letter $m$ for its specific value or actual realization; we will use a similar convention for other variables below, such as $V$ and $v$.}
meaning that $n_2-M$ is the number of tests corresponding to ``$X_{\seq}^{(i)}=0,X_{\sdif}^{(i)}=0$'', or equivalently  $X^{(i)}_s=0$.  This definition of $M$ coincides with $M_{\sdif}$ in the notation introduced above Lemma \ref{lem:restate}, but we omit the subscript to lighten notation, since we only consider one particular choice of $\sdif$ here.

% {\color{blue}[To decide: whether add this notation statement.] We mention that $M$     can actually be denoted by $M_{\sdif}$ from the notation in Lemma \ref{lem:restate}, while we can adopt the lighter notation $M$ here since Lemma \ref{nc-theta1} is stated with respect to a given $\sdif \subset S$}. %We first cope with the case of $\alpha\in (0,1)$ (note that the following probability term is problematic when $\alpha=1$). 

Note that $\frac{|\seq|}{k}=\frac{k-\ell}{k}=1-\frac{\ell}{k} $, and $\Yv_1$ corresponds to the tests where ``$X_{\seq}^{(i)} \neq 0$'' so that we can write $n_1=W^{(\seq)}$ in the notation of Lemma \ref{lem44}. Thus, given $\delta>0$, (\ref{9262}) and (\ref{9263}) in Lemma \ref{lem44} give that with probability at least $1-2\exp\big(-2\nu^{-1}\delta^2n\big)$, we have%, and then in $(ii)$ we let $\delta'=\delta_0'(1-\alpha)$,
\begin{equation}
\begin{aligned}
    n_1&=\Big(1-e^{-\nu(1-\frac{\ell}{k}) (1+o(1))}+\delta'\sqrt{1-\frac{\ell}{k}}\Big)n,\text{ for some } |\delta'|<\delta,\label{2966}
    \end{aligned}
\end{equation} 
which implies  
\begin{equation}\label{27n2}
    n_2=n-n_1= \Big(e^{-\nu(1-\frac{\ell}{k})(1+o(1))}-\delta'\sqrt{1-\frac{\ell}{k}}\Big)n,\text{ for some } |\delta'|<\delta.
\end{equation} 
We note that (\ref{27n2}) holds for all $\ell=1,\dotsc,k$. For the case of $\frac{\ell}{k} \to \alpha \in (0,1]$, we  obtain from (\ref{27n2})  that the event \begin{equation}\label{eq:n2update}
    n_2 = \big(e^{-\nu(1-\alpha)}+o(1)-\delta'\big)n,\text{ for some } |\delta'|<\delta.
\end{equation}
holds with probability at least $1-2\exp(-2\nu^{-1}\delta^2n)$. 

%Then we deal with the case of $\alpha=1$ and establish result similarly to (\ref{27n2}). For this case, we turn back the the intermediate result (\ref{2083}) with $\mathcal{M}=\seq$ satisfying $|\mathcal{M}|=k-\ell$, and note that we have $n_1=W^{(\mathcal{M})}$ and $n_2=n-n_1$. 
%Thus, (\ref{2083}) implies \begin{equation}
 %   \PP\big(|n_2-\mathbbm{E}n_2|\geq \varepsilon\big)\leq 2\exp\Big(-\frac{2\varepsilon^2 }{\nu (1-\ell/k) n}\Big)\label{chosen}
%\end{equation} for any  $\varepsilon>0$, where (\ref{9262}) gives $\mathbbm{E}n_2=n-\mathbbm{E}n_1=n-\frac{n\nu(k-\ell)}{k}(1+o(1))=(1+o(1))n$. We let $\varepsilon=\frac{\delta\nu n}{2}$ in (\ref{chosen}),\footnote{\color{red}[JR: Is it possible that $\ell=k$?]} it follows that with probability at least $1-2\exp(-\delta^2\nu n)$, we have $|n_2-\mathbbm{E}n_2|\leq \frac{\delta\nu n}{2}$, which implies $ n_2-\mathbbm{E}n_2 =\delta_0'\nu n$ for some $|\delta_0'|\leq \frac{\delta}{2}$. Using $\mathbbm{E}n_2=(1+o(1))n$ and then  incorporating $o(1)$ term into $\delta \nu$, we obtain $n_2= (1+o(1))n+\delta_0'\nu n=(1+o(1)+\delta_0'\nu)n=(1-\delta'\nu)n$, for some $|\delta'|\leq \delta$. Therefore, with the promised probability, (\ref{27n2}) is valid when $\frac{\ell}{k}\to \alpha\in (0,1]$. 

Next, we seek to obtain a similar concentration result for $M$. Note that $n-(n_2-M)$ represents the number of tests where $X_s^{(i)}\neq 0$, so we can write $n-(n_2-M)=W^{(s)}$. Hence, (\ref{9262}) and (\ref{9263}) in
 Lemma \ref{lem44} yield that  with probability at least $1-2\exp(-2\nu^{-1}\delta^2 n )$, we have
\begin{equation}
    n-(n_2-M) = (1-e^{-\nu}+o(1)+\delta'')n,~\text{for some }|\delta''|<\delta. \label{eq:sappear}
\end{equation}
Then, (\ref{eq:n2update}) and (\ref{eq:sappear}) imply
\begin{equation}\label{27m}
M= (e^{-\nu(1-\alpha)}-e^{-\nu}+o(1)+\delta^{(3)})n,~\text{ for some }    |\delta^{(3)}|< 2\delta.
\end{equation}
%These bounds  hold uniformly over all $\ell$-length $\sdif$, and 
%We implicitly condition on these events in subsequent analysis. %Note that the quantities $m$ when conditioning on 
% Thus, by a union bound that removes probability at most $4\exp(-2\nu^{-1}\delta^2 n)$, we   suppose  that (\ref{eq:n2update}) and (\ref{27m}) hold in subsequent analysis. 
In the subsequent analysis, we will suppose that (\ref{eq:n2update}) and (\ref{27m}) hold, which is justified via a simple union bound for the multiple high-probability events.

Let $V$ be the number of positive tests in $\Yv_2$. Conditioned on some $\Xv=\xv$ (and hence on some $M=m$), the randomness of the noise gives $(V|M=m) \sim \mathrm{Bin}(m,1-\rho)+\mathrm{Bin}(n_2-m,\rho)$, which in turn can be interpreted as a sum of $n_2$ independent (but non-identical) Bernoulli variables. 
  This allows us to utilize Hoeffding's inequality \cite[Thm. 2.8]{Bou13} with respect to the randomness of the noise 
to obtain that, for any $t\geq 0$, 
\begin{equation}\label{eq:HoeffV}
    \Psf\Big(\big|V-\big[M(1-\rho)+(n_2-M)\rho\big]\big|\geq t\Big) \leq 2\exp\Big(-\frac{2t^2}{n_2}\Big).
\end{equation}
We take $t=n_2 \delta $ %and incorporate $o(1)n$ by slightly changing the multiplicative factors for $\delta_1$, it allows us to assume
to obtain that $|V-[M(1-\rho)+(n_2-M)\rho]|< n_2\delta$, or equivalently, $V=(1-2\rho)M+(\rho+\delta^{(4)})n_2$ for some $|\delta^{(4)}|<\delta$, with probability at least $1-2\exp\big(-2n_2\delta^2\big)$. Substituting the high-probability values of $M$ and $n_2$ in (\ref{eq:n2update}) and (\ref{27m}), when $\delta<\frac{1}{4}e^{-\nu(1-\alpha)}$, with probability at least $1-2\exp(-e^{-\nu(1-\alpha)}\delta^2n)$ we have
%for some $\delta^{(4)}$ with $|\delta^{(4)}|<\delta$, 
\begin{align}\label{8555}
        V&= (1-2\rho)\big(e^{-\nu(1-\alpha)}-e^{-\nu}+o(1)+\delta^{(3)}\big)n+(\rho+\delta^{(4)})\big(e^{-\nu(1-\alpha)}+o(1)-\delta'\big)n\\
        &\label{eq:finalele_ii}=ne^{-\nu(1-\alpha)}\Big((1-2\rho)(1-e^{-\nu\alpha})+\rho+o(1)+c_{1,\rho,\nu,\alpha}\delta\Big)\\&= ne^{-\nu(1-\alpha)}\big(\rho\star e^{-\nu\alpha}+o(1)+c_{1,\rho,\nu,\alpha}\delta\big), \label{855}
\end{align}
%holds with probability at least $1-2\exp(-e^{-\nu(1-\alpha)}\delta^2n)$.  
where in (\ref{eq:finalele_ii}) we introduce $c_{1,\rho,\nu,\alpha}$ that is bounded by an absolute constant depending on $(\rho,\nu,\alpha)$,\footnote{Similarly, we will use $c_{2,\rho,\nu,\alpha}$, $c_{3,\rho,\nu,\alpha}$, etc.~to denote constants depending on $(\rho,\nu,\alpha)$.} 
in (\ref{855})  we make the observation $(1-2\rho)(1-e^{-\nu\alpha})+\rho=\rho\star e^{-\nu\alpha}$ 
(recalling that $a \star b = ab+(1-a)(1-b)$).

Recall that (\ref{eq:final_i_ele}) gives $\imath^n(\Yv|\Xveq,\Xvdif)=\log \Psf(\Yv_2|\Xvdif)-\log \Psf(\Yv_2)$, and in both terms we implicitly condition on $\Xveq$. To establish the concentration bound for $\imath^n(\Yv|\Xveq,\Xvdif)$, the main work is to bound these two terms separately. %and $v$ represents the number of positive tests in $\Yv_2$. 

\textbf{Bounding $\log \Psf(\Yv_2)$:} To handle the term $\Psf(\Yv_2)$, consider a fixed realization $\yv_2$ of $\Yv_2$.  This implies that the realization $v$ of $V$ is also fixed (since it is a function of $\yv_2$), and we suppose that it satisfies the high-probability condition  (\ref{855}). 
By the symmetry of the test design with respect to re-ordering the test indices, all $\binom{n_2}{v}$ possibilities of $\Yv_2$ having weight $v$ occur with equal probability, and hence,
we have $\Psf(\yv_2)=\frac{\Psf(v)}{\binom{n_2}{v}}$. Then, applying\footnote{This step may appear loose, but the idea is that if we were to consider all $v$ and take the dominant one, it would have a ``large'' value of $\PP(v)$ anyway, i.e., not decaying fast enough to impact the final result.} $\Psf(v)\leq 1$ and substituting (\ref{eq:n2update}) and (\ref{855}), % $n_2= n[e^{-\nu(1-\alpha)}-\delta'\nu(1-\alpha)]=ne^{-\nu(1-\alpha)}(1+c_{2,\rho,\nu,\alpha}\delta)$ 
we obtain 
\begin{equation}
    \Psf(\yv_2) \leq \frac{1}{\binom{ne^{-\nu(1-\alpha)}(1+o(1)+c_{2,\rho,\nu,\alpha}\delta)}{ne^{-\nu(1-\alpha)}(e^{-\alpha\nu}\star\rho+o(1)+c_{1,\rho,\nu,\alpha}\delta)}}. \label{eq:Py2_init}
\end{equation}
From Lemma \ref{lem:coeff_bounds}, we have $\binom{N}{\beta N}\geq \frac{\exp(NH_2(\beta))}{\sqrt{2N}}=\exp(N[H_2(\beta)+o(1)])$ for constant $\beta \in (0,1)$, which implies that if $\delta$ is smaller than some constant depending on $(\alpha,\nu,\rho)$ such that $\frac{e^{-\alpha\nu}\star\rho+o(1)+c_{1,\rho,\nu,\alpha}\delta}{1+o(1)+c_{2,\rho,\nu,\alpha}\delta}$ is sufficiently close to $e^{-\alpha\nu}\star\rho$, we can upper bound the right-hand side of \eqref{eq:Py2_init} by $\exp\big(-[1+o(1)+c_{3,\rho,\nu,\alpha}\delta]ne^{-\nu(1-\alpha)} H_2(e^{-\alpha\nu}\star\rho) \big)$ for some constant $c_{3,\rho,\nu,\alpha}$.  Since this holds for any realization of $\yv_2$ whose corresponding $v$ value satisfies (\ref{855}), we conclude that
\begin{equation}\label{1987}
    \Psf(\Yv_2)\leq \exp\Big(-[1+o(1)+c_{3,\rho,\nu,\alpha}\delta]ne^{-\nu(1-\alpha)} H_2(e^{-\alpha\nu}\star\rho) \Big)
\end{equation}
with probability at least $1-2\exp(-e^{-\nu(1-\alpha)}\delta^2 n)$ (over the randomness of $\Yv_2$).

\textbf{Bounding $\log \Psf(\Yv_2|\Xvdif)$:}
Next, we consider $\Psf(\Yv_2|\Xvdif)$. 
%Given $\Xvdif$, and recalling that we implicitly condition on a fixed realization $\xveq$ of $\Xveq$, the $M$ (the number of tests corresponding to $\Xveq=0,\Xvdif\neq0$) is known, and we denote this actual realization by $m$. 
Here the entries of $\Yv_2$ are  conditionally independent, since given $(\Xveq,\Xvdif)$ they only depend on the noise variables (which are independent across tests). This allows us to decompose 
\begin{equation}\begin{aligned}\label{888}
    \log \Psf(\Yv_2|\Xvdif)=\log \prod_{i=1}^{n_2}\Psf(Y_2^{(i)}|X_{\sdif}^{(i)})=\sum_{i=1}^{n_2}\log \Psf(Y_2^{(i)}|X_{\sdif}^{(i)}),
    %\\
    %&=\sum_{i=1}^{m^*} \log \PP(\Yv_2^{(i)}|\Xvdif^{(i)})+\sum_{i=m^*+1}^{n_2} \log \PP(\Yv_2^{(i)}|\Xvdif^{(i)}),
    \end{aligned}
\end{equation}
where without loss of generality we index the tests in $\Yv_2$ from $1$ to $n_2$. %and the $m^*$ tests indexed from $1$ to $m^*$ correspond to $\Xveq=0,\Xvdif\neq 0$, the remaining $n_2-m^*$ tests correspond to $\Xveq=0,\Xvdif= 0$.
Note that (\ref{888}) indicates that $\log\Psf(\Yv_2|\Xvdif)$ is a sum of $n_2$ independent random variables, which have a common expected value $-H_2(\rho)$ and are bounded by $|\log \rho|$ (since for any $\rho \in \big(0,\frac{1}{2}\big)$ we have $|\rho\log\rho+(1-\rho)\log(1-\rho)|\le \rho|\log\rho|+(1-\rho)|\log\rho|=|\log\rho|$), so using Hoeffding's inequality \cite[Thm. 2.8]{Bou13} with respect to the noise randomness gives the following for any $t>0$:
\begin{equation}\label{eq:27n22}
    \Psf\Big(\big|\log \Psf(\Yv_2|\Xvdif)+n_2H_2(\rho)\big|\geq t\Big)\leq2 \exp\Big(-\frac{2t^2}{n_2\log^2\rho}\Big).
\end{equation}
%It's not hard to see that  $\mathbbm{E}[\log\PP(\Yv_2|\Xvdif)]=-n_2H(\rho)$. 
Moreover, we use (\ref{eq:27n22}) and let $t=\delta n_2 H_2(\rho)$ to obtain that $|\log\Psf(\Yv_2|\Xvdif)+n_2H_2(\rho)|\leq \delta n_2H_2(\rho)$ with probability at least $1-2\exp\big(-2\delta^2(\log\rho)^{-2}(H_2(\rho))^2n_2\big)$. Further substituting (\ref{eq:n2update}), when $\delta<\frac{1}{4}e^{-\nu(1-\alpha)}$, we obtain that
\begin{equation}\label{1990}
    \log\Psf(\Yv_2|\Xvdif)=-ne^{-\nu(1-\alpha)}H_2(\rho)\big(1+c_{4,\rho,\nu,\alpha}\delta\big)
\end{equation}
for some $c_{4,\rho,\nu,\alpha} = \Theta(1)$, with probability at least $1-2\exp\big(-\delta^2 (\log \rho)^{-2}(H_2(\rho))^2 e^{-\nu(1-\alpha)}n\big)$. %{\color{red}[\textbf{JR}: start from here next time]}

\textbf{Putting the pieces together:}
We have shown that (\ref{1987}) and (\ref{1990}) hold with probability at least $1-\exp\big(-C_{\rho,\nu,\alpha}\delta^2n\big)$ for some   $C_{\rho,\nu,\alpha}>0$ depending on $(\rho,\nu,\alpha)$ only (here we suppose that $\delta$ is given and $n$ is large enough, so that any leading constant before $\exp(-C_{\rho,\nu,\alpha}\delta^2 n)$ can be absorbed into the exponent).   Now 
 we start with (\ref{eq:final_i_ele}) and proceed as  
\begin{align}\label{2976}
   &\imath^n(\Yv|\Xveq,\Xvdif) =\log\Psf(\Yv_2|\Xvdif)-\log\Psf(\Yv_2)\\\label{29761}&\geq -ne^{-\nu(1-\alpha)}H_2(\rho)(1+c_{4,\rho,\nu,\alpha}\delta)+[1+o(1)+c_{3,\rho,\nu,\alpha}\delta]ne^{-\nu(1-\alpha)}H_2(e^{-\alpha\nu}\star\rho)\\\label{29762}
    &\geq (1-o(1))I_\ell^n+\big(c_{3,\rho,\nu,\alpha}H_2(e^{-\alpha\nu}\star\rho)-c_{4,\rho,\nu,\alpha}H_2(\rho)\big)\delta n e^{-\nu(1-\alpha)}\\ \label{29763.5}
    &\geq (1-o(1))I_\ell^n-C_{1,\rho,\nu,\alpha}n\delta  \\ 
    &= \Big( 1-o(1)- \frac{C_{1,\rho,\nu,\alpha}n\delta}{I_\ell^n}\Big)I_\ell^n, \label{29763}
\end{align}
where in (\ref{29761}) we substitute (\ref{1987}) and (\ref{1990}),  in (\ref{29762}) we recall $I_\ell^n$ given in (\ref{2077nc}) in Lemma \ref{lem:mi_ncc}, and in \eqref{29763.5} $C_{1,\rho,\nu,\alpha}$ is a constant depending on  $(\rho,\nu,\alpha)$. Thus, given any $\delta_1>0$, we can set the above $\delta$ as $\delta=\frac{I_\ell^n\delta_1}{2nC_{1,\rho,\nu,\alpha}}$, and we obtain $\imath^n(\Yv|\Xveq,\Xvdif)\geq (1-o(1)-\frac{\delta_1}{2})I_\ell^n$  with probability at least $1-4\exp\big(-\frac{C_{\rho,\nu,\alpha}'(I_\ell^n)^2\delta_1^2}{n}\big)$ for some  $C'_{\rho,\nu,\alpha}>0$ only depending on $(\rho,\nu,\alpha)$, as desired.

\subsubsection*{The case $\alpha=0$}

\noindent
It remains to deal with the case of $\frac{\ell}{k}\to 0$, for which the proof is generally similar to the case $\alpha\in(0,1]$ above but needs some minor modifications.   We focus on the differences and avoid repeating most near-identical steps.

To get started, we note that (\ref{27n2}) remains valid when $\alpha=0$, which yields that with probability at least $1-2\exp(-2\nu^{-1}\delta^2 n)$, we have
\begin{equation}
    \label{eq:n2tozero}
    n_2 =\big (e^{-\nu(1-\frac{\ell}{k})}+o(1)-\delta'\big)n,~\text{for some }|\delta'|<\delta.
\end{equation}
Next, we note that (\ref{eq:sappear}) remains valid,  which together with (\ref{eq:n2tozero}) yields the following with probability at least $1-2\exp(-2\nu^{-1}\delta^2 n)$:  
\begin{equation}\label{eq:Mvalidtozero}
    M = \big(e^{-\nu(1-\frac{\ell}{k})}-e^{-\nu}+o(1)+\delta^{(3)}\big)n,~\text{for some }|\delta^{(3)}|< 2\delta.
\end{equation}
Moreover, via the randomness of noise, Hoeffding's inequality still gives (\ref{eq:HoeffV}), in which we set $t=n_2\delta$, substitute (\ref{eq:n2tozero}) and (\ref{eq:Mvalidtozero}), and then      perform some simple algebra (as in (\ref{8555})--(\ref{855})) to obtain that \begin{equation}
    \label{eq:Vdaluetozero}
    V = ne^{-\nu(1-\frac{\ell}{k})}\big(\rho\star e^{-\frac{\nu\ell}{k}}+o(1)+c_{1,\rho,\nu}\delta\big)
\end{equation}
holds with probability at least $1-2\exp(-e^{-\nu}\delta^2 n)$, 
where $c_{1,\rho,\nu}$ is a quantity bounded by a constant only depending on $(\rho,\nu)$.  By a union bound, the events (\ref{eq:n2tozero}), (\ref{eq:Mvalidtozero}) and (\ref{eq:Vdaluetozero}) collectively hold with probability at least $1-\exp(-C_{\rho,\nu}\delta^2 n)$ for some $C_{\rho,\nu}$ depending on $(\rho,\nu)$. 

Then, it is not hard to check the arguments for bounding $\log \Psf(\Yv_2)$ and $\log \mathsf{P}(\Yv_2|\Xvdif)$ directly carry over here, which give the following (corresponding to the earlier (\ref{1987}) and (\ref{1990})):
\begin{gather}
    \log\Psf(\Yv_2) \leq -[1+o(1)+c_{3,\rho,\nu}\delta]ne^{-\nu(1-\frac{\ell}{k})} H_2(e^{-\frac{\nu\ell}{k}}\star\rho),\label{2977} \\
    \log \Psf(\Yv_2|\Xvdif)=-ne^{-\nu(1-\frac{\ell}{k})}H_2(\rho) \big(1+c_{4,\rho,\nu}\delta\big),\label{2978}
\end{gather}
which hold with probability at least $1-\exp\big(-C'_{\rho,\nu}\delta^2 n\big)$ for some $C'_{\rho,\nu}>0$ depending on $(\rho,\nu)$. %\footnote{\color{red}[JR: Go back to Footnote 2, need to check whether need uniformity and then modify.]} 
Starting with (\ref{eq:final_i_ele}), we  can proceed as in (\ref{2976})--(\ref{29763}) and obtain 
\begin{align} \label{2979}\imath^n(\Yv|\Xveq,\Xvdif)&=\log\Psf(\Yv_2|\Xvdif)-\log\Psf(\Yv_2)\\&\label{29791}\ge ne^{-\nu(1-\frac{\ell}{k})}\big[H_2(e^{-\frac{\nu \ell}{k}}\star\rho)-H_2(\rho)\big]-n C_{5,\rho,\nu}\delta (1+o(1))\\
   \label{29792}&\ge (1-o(1))I_\ell^n-2nC_{5,\rho,\nu}\delta 
   \\&=\Big(1-o(1)-\frac{2nC_{5,\rho,\nu}\delta}{I_\ell^n}\Big)I_\ell^n,
\end{align}
where in (\ref{29791}) we use (\ref{2977}) and (\ref{2978}), and in (\ref{29792}) we apply Lemma \ref{calculation} below with
$C_{5,\rho,\nu}$ being some constant depending on $(\rho,\nu)$.  Thus, given any $\delta_1>0$, we can take the above $\delta$ as $\delta=\frac{I_\ell^n\delta_1}{4nC_{5,\rho,\nu}}$ to ensure $\imath^n(\Yv|\Xveq,\Xvdif)\geq (1-o(1)-\frac{\delta_1}{2})I_\ell^n$, which holds with probability at least $1-4\exp\big(-\frac{C'_{\rho,\nu}(I_{\ell}^n)^2\delta_1^2}{n}\big)$. The proof of Lemma \ref{nc-theta1} is now complete. 

% \subsubsection{The Lemma to get (\ref{29792})}

In the final steps of the above analysis, we used the following.

\begin{lem}
    \label{calculation}
    If $\frac{\ell}{k}\to 0$, then we have 
    \begin{equation}
        ne^{-\nu(1-\frac{\ell}{k})}\big[H_2(e^{-\frac{\nu\ell}{k}}\star\rho)-H_2(\rho)\big]\geq (1-o(1))I_\ell^n, \label{eq:calculation}
    \end{equation}
    where $I_\ell^n$ satisfies (\ref{1976}).
\end{lem}
\begin{proof}
    Since $\frac{\ell}{k}\to 0$, we have $ne^{-\nu(1-\frac{\ell}{k})}= ne^{-\nu}(1+o(1))$.  In view of the desired inequality \eqref{eq:calculation} and the asymptotics of $I_{\ell}^n$ in (\ref{1976}), some simple algebra reveals that it suffices to prove \begin{equation}\label{eq:suffice1}
       H_2(e^{-\frac{\nu\ell}{k}}\star\rho)-H_2(\rho)\geq (1-o(1))\frac{\nu\ell(1-2\rho)}{k}\log\frac{1-\rho}{\rho}. 
    \end{equation}
    To this end, by using $e^{-\frac{\nu\ell}{k}}=1-(1+o(1))\frac{\nu\ell}{k}$, we first write  
        \begin{align}\label{2981}
            e^{-\frac{\nu\ell}{k}}\star\rho&= \rho e^{-\frac{\nu\ell}{k}}+(1-\rho)\big(1-e^{-\frac{\nu\ell}{k}}\big)\\
            &=\rho\Big(1-(1+o(1))\frac{
            \nu\ell 
            }{k}\Big)+(1-\rho)(1+o(1))\frac{\nu\ell}{k}\\&=\rho +(1-2\rho+o(1))\frac{\nu\ell}{k}.
            %:=\rho+A
        \end{align}
   % where $A=(1-2\rho+o(1))\frac{\nu\ell}{k}=o(1)$. By using (\ref{2981}) and expanding $H(\cdot)$ in $(i)$, rearranging the terms in $(ii)$, and then use $A=o(1)$ and  $\log(1+t)=t(1+o(1))$ for $t=o(1)$ in $(iii)$, we obtain 
    %\begin{equation}
       % \begin{aligned}
       %     &H(e^{-\frac{\nu\ell}{k}}\star\rho)-H(\rho) \\ \stackrel{(i)}{=}& -(\rho+A)\log(\rho+A)-\big(1-(\rho+A)\big)\log\big(1-(\rho+A)\big) +\rho\log\rho +(1-\rho)\log(1-\rho) \\
          %  \stackrel{(ii)}{=}& -\rho \log\Big(1+\frac{A}{\rho}\Big)-(1-\rho)\log \Big(1-\frac{A}{1-\rho}\Big)+A\log\frac{1-(\rho+A)}{\rho+A} \\
            %\stackrel{(iii)}{=}&-A(1+o(1))+A(1+o(1))+ A \Big[\log\Big(\frac{1-\rho}{\rho}\Big)+o(1)\Big]= (1-2\rho +o(1)) \frac{\nu\ell}{k}\log\frac{1-\rho}{\rho}.
       % \end{aligned}
    %\end{equation}
   Then, we only need to apply a Taylor expansion to $H_2(e^{-\frac{\nu\ell}{k}}\star\rho)-H_2(\rho)$ with $H'_2(\rho)=\log\frac{1-\rho}{\rho}$ to obtain the desired bound (\ref{eq:suffice1}). 
\end{proof}

\subsection{Proof of Lemma \ref{lem:mi_ncc} (Mutual Information for the Near-Constant Weight Design)} \label{sec:pf_mi_ncc}

% The following Lemma indicates that under near-constant design, the expected value of the information density $\imath^n(\Xvdif; \Yv | \Xveq)$ where $\ell=|s_{dif}|$, denoted by $I_\ell^n=\mathbbm{E}[\imath(\Xvdif; \Yv | \Xveq)]$, coincides with the one under Bernoulli design \cite[Prop. 5]{Sca15b}. 
We are interested in $I_\ell^n = \mathbbm{E}(\imath^n(\Xvdif;\Yv|\Xveq))$ when $\Xv$ is generated via the near-constant column design with parameter $\nu$ (i.e., $\Delta = \frac{\nu n}{k}$ tests per item).  We define \begin{equation}\label{eq:def_con_eq}
    I(\Xvdif;\Yv|\Xveq=\xveq) := \mathbbm{E}_{\Xvdif,\Yv}(\imath^n(\Xvdif;\Yv|\Xveq=\xveq))
\end{equation}
where the expectation is taken with respect to the randomness of $\Xvdif,\Yv$ with  a given $\Xveq = \xveq$. Further averaging over $\Xveq$ will provide the quantity $I_\ell^n$ of interest, i.e., \begin{equation}
    I_\ell^n = \mathbbm{E}_{\Xveq}\big[I(\Xvdif;\Yv|\Xveq=\xveq)\big].\label{eq:fur_exp}
\end{equation} 
%This can equivalently be viewed as $I(\Tvdif;\Yv|\Tveq)$ , where for item $j=1,\dotsc,k$, we let $\Tv_j$ be a length-$\Delta$ list of integers in $\{1,\dotsc,n\}$ indexing the test placement.  {\color{red} [Note: Didn't end up using $\Tv$ in this write-up.]}
    Recall that $|\sdif|=\ell$ and $|\seq|=k-\ell$. We first consider $I(\Xvdif;\Yv|\Xveq=\xveq)$ for fixed $\xveq$, which we expand as (see (\ref{12.2})) 
 \begin{align}
	I(\Xvdif;\Yv|\Xveq=\xveq) &= H(\Yv|\Xveq=\xveq) - H(\Yv | \Xvdif, \Xveq=\xveq)\\&= H(\Yv|\Xveq=\xveq) -nH_2(\rho),\label{eq:con_on_eq_1}
 \end{align}
where (\ref{eq:con_on_eq_1}) holds because with the full conditioning on $\Xv_s$, the only remaining uncertainty is the i.i.d.~Bernoulli($\rho$) noise. It  remains to bound $ H(\Yv|\Xveq=\xveq).$ 
Given $\Xveq=\xveq$, let $\Yv_1$ be the collection of outcomes of tests containing at least one item from $\seq$, let $\Yv_2$ be the remaining test outcomes, and let the corresponding lengths of $\Yv_1$ and $\Yv_2$ be $n_1$ and $n_2$, respectively.  %By McDiarmid's inequality, with probability approaching one, we have $n_1 = n (1- e^{-(1-\alpha)\nu}) + o(n)$, where $\alpha = \frac{\ell}{k}$.
%By Lemma \ref{lem44}, 
Note also that given $\Xveq=\xveq$, $\Yv_1$ only depends on the corresponding $n_1$ noise terms, so the vectors $\Yv_1$ and $\Yv_2$ are conditionally independent.  Thus, we get
\begin{align}
	H(\Yv|\Xveq=\xveq) &= H(\Yv_1|\Xveq=\xveq) + H(\Yv_2|\Xveq=\xveq)\\&= n_1 H_2(\rho) + H(\Yv_2|\Xveq=\xveq),\label{eq:con_depen}
 \end{align}
where (\ref{eq:con_depen}) is again because given $\Xveq=\xveq$, $\Yv_1$ only depends on the randomness of the $n_1$ noise variables. Substituting (\ref{eq:con_depen}) into (\ref{eq:con_on_eq_1}) yields %Combining with  $H(\Yv | \Xvdif, \Xveq=\xveq) = n H(\rho) = (n_1+n_2)H(\rho)$, we conclude that
\begin{equation}
	I(\Xvdif;\Yv|\Xveq=\xveq) = H(\Yv_2|\Xveq=\xveq) - n_2 H_2(\rho). \label{eq:mi_equals}
\end{equation}
In addition, note that we can upper bound
\begin{equation}
	H(\Yv_2|\Xveq=\xveq) \le n_2 H(Y' | \Xveq=\xveq) \label{eq:H_ub}
\end{equation}
where $Y'$ is an arbitrary single test in $\Yv_2$ (e.g., the first one).  This follows by sub-additivity of entropy \cite[Thm.~2.5.1 \& 2.6.5]{Cov01} and the fact that each test in $\Yv_2$ has the same probability of being positive by symmetry.

\subsubsection*{The case $\frac{\ell}{k} \to \alpha > 0$}

We use \eqref{eq:H_ub}
to upper bound the conditional entropy $H(\Yv_2|\Xveq= \xveq)$. Without noise, a given test in $\Yv_2$ would be negative if and only if it does not contain any item from $\sdif$. Thus,
the conditional probability of a single test in $\Yv_2$ being negative \emph{without noise} would be $\big(1-\frac{1}{n}\big)^{\ell\Delta}=\big(1-\frac{1}{n}\big)^{\frac{\ell\nu n}{k}} = e^{-\alpha\nu(1+o(1))}$. After noise, the probability of a single test being positive becomes $e^{-\alpha \nu(1+o(1))} \star \rho$ (recalling the notation $a\star b= ab+(1-a)(1-b)$).  Thus, the conditional entropy of a single test is %\footnote{\color{red}[JR: Let us stick to the logarithmic base of $e$ and does not use the subscript of $2$.]}  
$H_2(e^{-\alpha \nu(1+o(1))} \star \rho)$, and overall (\ref{eq:H_ub}) yields  
\begin{align}
    H(\Yv_2|\Xveq=\xveq) \le n_2 H_2(e^{-\alpha \nu(1+o(1))} \star \rho).
\end{align}
Substituting this into (\ref{eq:mi_equals}), we find that \begin{align}\label{eq:be_averaged}
    I(\Xvdif;\Yv|\Xveq = \xveq)\leq n_2(H_2(e^{-\alpha\nu}\star\rho)-H_2(\rho))(1+o(1)).
\end{align} Since (\ref{9262}) in Lemma \ref{lem44} gives $\mathbbm{E}(n_2) = n-\mathbbm{E}(n_1) = n-(1-e^{-\nu(1-\alpha)}+o(1))n = (e^{-\nu(1-\alpha)}+o(1))n$, by further averaging (\ref{eq:be_averaged}) over $\Xveq$ as in (\ref{eq:fur_exp}), we obtain \begin{align}
    I_\ell^n\leq ne^{-\nu(1-\alpha)}(H_2(e^{-\alpha\nu}\star\rho)-H_2(\rho))(1+o(1)).
\end{align}

For the lower bound, we define $V$ to be the number of positive tests in $\Yv_2$ and 
first count $(n_2,V)$ as in first part of the proof of Lemma \ref{nc-theta1} (Appendix \ref{sec:pf_conc_ncc}). In particular, for any given $\delta>0$, (\ref{eq:n2update})   states that   \begin{align}
    n_2= \big(e^{-\nu(1-\alpha)}+o(1)-\delta' \big)n,~\text{for some }|\delta'|<\delta
\end{align}
with probability at least $1-2\exp(-2\nu^{-1}\delta^2n)$; on this event, (\ref{855}) further gives with probability at least $1-2\exp(-e^{-\nu(1-\alpha)}\delta^2n)$ that
\begin{align}
    V = ne^{-\nu(1-\alpha)}\big(\rho\star e^{-\nu\alpha}+o(1)+c_{\rho,\nu,\alpha}\delta\big),~\text{for some }c_{\rho,\nu,\alpha}\leq c_1(\rho,\nu,\alpha),
\end{align}  
where $c_1(\rho,\nu,\alpha)$ is a constant depending on $(\rho,\nu,\alpha)$. 
Thus, we can set $\delta = n^{-1/4}=o(1)$ to obtain that the two events 
\begin{align}\label{eq:A1_for_n2}
    &\mathscr{A}_1=\big\{n_2=ne^{-\nu(1-\alpha)}(1+o(1))\big\},\\
    &\mathscr{A}_2=\big\{V=ne^{-\nu(1-\alpha)}(\rho\star e^{-\nu\alpha})(1+o(1))\big\},\label{eq:A2_for_V}
\end{align}
hold with probability at least $\PP(\mathscr{A}_1)\geq 1-2\exp(-2\nu^{-1}\sqrt{n})=1-o(1)$ and $\PP(\mathscr{A}_2)\geq 1-2\exp(-e^{-\nu(1-\alpha)}\sqrt{n})=1-o(1)$. 
%use concentration to deduce that $\Yv_2$ has weight roughly $n_2(e^{-\alpha \nu} \star \rho)$ with high probability, and note that anything outside this concentration bound only contributes $o(n)$ to the entropy. 
Since conditioning reduces entropy \cite[Thm.~2.6.5]{Cov01}, conditioning on $V$ gives %{\color{red}[JR: \textbf{Question}: Here, can we directly pick $V=ne^{-\nu(1-\alpha)}(\rho\star e^{-\nu\alpha})(1+o(1))$? or we need to rely on the event $\mathscr{A}$ above to give the value of $v$? If we can directly pick $V$, it seems that counting $V$ is unnecessary.]}
\begin{align}
	H(\Yv_2|\Xveq=\xveq) &\ge H(\Yv_2 | \Xveq = \xveq, V)\\&= \sum_{v=0}^{n_2}\PP(V=v)H(\Yv_2| \Xveq = \xveq,V=v)\\
 &\ge \PP(\mathscr{A}_2)H\big(\Yv_2| \Xveq = \xveq,V=ne^{-\nu(1-\alpha)}(\rho\star e^{-\nu\alpha})(1+o(1))\big)\label{eq:remove1}\\
 &\geq (1-o(1))\log \binom{n_2}{ne^{-\nu(1-\alpha)}(\rho\star e^{-\nu\alpha})(1+o(1))},\label{eq:H_lb}
\end{align}
where in (\ref{eq:remove1}) we only count the summands in which $V=ne^{-\nu(1-\alpha)}(\rho\star e^{-\nu\alpha})(1+o(1))$ as specified in the event $\mathscr{A}_2$ (\ref{eq:A2_for_V}), and (\ref{eq:H_lb}) 
holds because given $V=v$ (and $\xveq$), $\Yv_2$ is uniformly distributed over the $\binom{n_2}{v}$ possibilities of $n_2$-dimensional  $\{0,1\}$-valued vector with $v$ $1$'s. Substituting (\ref{eq:H_lb}) into (\ref{eq:mi_equals}) yields 
\begin{equation}
    \label{eq:lowerentropy}
    I(\Xvdif;\Yv|\Xveq=\xveq)\geq (1-o(1))\log \binom{n_2}{ne^{-\nu(1-\alpha)}(\rho\star e^{-\nu\alpha})(1+o(1))}-n_2H_2(\rho).
\end{equation}
%$$I(\Xvdif;\Yv|\Xveq=\xveq)\geq (1-o(1))\log \binom{n_2}{ne^{-\nu(1-\alpha)}(\rho\star e^{-\nu\alpha})(1+o(1))}-n_2H(\rho).$$
 To evaluate $I_\ell^n$, we further 
 average over $\Xveq$ as in (\ref{eq:fur_exp}) to obtain  that
     \begin{align}
         I_\ell^n& \geq \mathbbm{E}_{\Xveq}\big(I(\Xvdif;\Yv|\Xveq=\xveq)\big)\\ 
         \label{eq:removeiii}&\geq (1-o(1))\log\binom{ne^{-\nu(1-\alpha)}(1+o(1))}{ne^{-\nu(1-\alpha)}(\rho\star e^{-\nu\alpha})(1+o(1))} -\big(ne^{-\nu(1-\alpha)}(1+o(1))\big)H_2(\rho)\\
         &= ne^{-\nu(1-\alpha)}(H_2(\rho\star e^{-\nu\alpha})-H_2(\rho))(1+o(1)),\label{eq:removeii} 
     \end{align}
where (\ref{eq:removeiii}) follows by using the non-negativity of mutual information to sum only over $\xveq$ satisfying the event $\mathscr{A}_1$ (\ref{eq:A1_for_n2}) and then using (\ref{eq:lowerentropy}) along with $n_2= ne^{-\nu(1-\alpha)}(1+o(1))$ (under $\mathscr{A}$), and in (\ref{eq:removeii}) we apply
  $\log{N \choose \beta N} = N H_2(\beta)(1+o(1))$ (see Lemma \ref{lem:coeff_bounds}). Since this matches the upper bound we derived above, (\ref{2077nc}) follows. %Combining the upper and lower bounds and substituting into \eqref{eq:mi_equals} along with $n_2 \sim n e^{-(1-\alpha)\nu}$ gives
%\begin{equation}
%	I(\Xvdif;\Yv|\Xveq=\xveq) \sim n e^{-(1-\alpha)\nu} \big( H_2(e^{-\alpha \nu} \star \rho) - H_2(\rho) \big).
%\end{equation}

\subsubsection*{The case $\frac{\ell}{k} \to 0$}

In order to upper bound $H(\Yv_2|\Xveq = \xveq)$, we again use \eqref{eq:H_ub}, but we now trivially use the fact that the probability of a particular test containing an item from $\sdif$ is at most $\frac{\nu \ell}{k} \to 0$ (a union bound of $\frac{\nu n \ell}{k}$ events having probability $\frac{1}{n}$ each).  Hence, we have 
\begin{equation}
	\PP[Y' = 1] \le \frac{\nu \ell}{k} (1-\rho) + \bigg(1-\frac{\nu \ell}{k}\bigg)\rho = \rho + \frac{\nu \ell}{k} (1-2\rho).
\end{equation}
The corresponding entropy upper bound is
\begin{equation}
	H(Y'|\Xveq=\xveq)\leq H_2\bigg( \rho + \frac{\nu \ell}{k} (1-2\rho) \bigg) = H_2(\rho) + \frac{\nu \ell}{k} (1-2\rho) \log\Big(\frac{1-\rho}{\rho}\Big)(1+o(1)),
\end{equation}
where we use a Taylor expansion with $H'_2(\rho) =  \log\frac{1-\rho}{\rho}$ in   the last step.  Combining with \eqref{eq:mi_equals} and \eqref{eq:H_ub} gives
\begin{equation}
	I(\Xvdif;\Yv|\Xveq=\xveq) \leq n_2 \frac{\nu \ell}{k} (1-2\rho) \log\Big(\frac{1-\rho}{\rho}\Big)(1+o(1)). \label{eq:I_ub}
\end{equation}
%which further simplifies by writing $n_2 \sim n e^{-\nu}$.
Note that $\frac{|\seq|}{k}= 1-\frac{\ell}{k}\to 1$, so (\ref{9262}) in Lemma \ref{lem44} gives \begin{align}
    \mathbbm{E}(n_2)=n-\mathbbm{E}(n_1)=n-\mathbbm{E}(W^{(\seq)})=n-(1-e^{-\nu}+o(1))n=ne^{-\nu}(1+o(1)).
\end{align} Thus, by averaging over $\Xveq$ in (\ref{eq:I_ub}), we obtain $I_\ell^n\leq \frac{n\nu e^{-\nu}\ell}{k}(1-2\rho)\log(\frac{1-\rho}{\rho})(1+o(1))$.

It remains to prove a matching lower bound. 
To this end, we write
\begin{align}
	H(\Yv_2|\Xveq=\xveq)&=H(\Yv_2,V|\Xveq=\xveq)\label{eq:H_next1}\\\label{eq:H_next2}&=H(\Yv_2|\Xveq=\xveq,V)+H(V|\Xveq=\xveq)
    \\&\ge   \sum_{v=0}^{n_2} \PP(V=v) H(\Yv_2|\Xveq=\xveq,V=v)+\Big(\frac{1}{2}+o(1)\Big)\log (n_2), \label{eq:H_next}
 \end{align}
where (\ref{eq:H_next1}) is because $V$ is deterministic given $\Yv_2$, (\ref{eq:H_next2}) holds by the chain rule for entropy \cite[Thm.~2.5.1]{Cov01}, and (\ref{eq:H_next}) holds due to the following argument showing that the entropy of $V$ (when given $\Xv_{\seq}=\xv_{\seq}$) is at least $\big(\frac{1}{2}+o(1)\big) \log(n_2)$:
\begin{itemize}[itemsep=0ex]
    \item Since conditioning reduces entropy \cite[Thm.~2.6.5]{Cov01}, we can lower bound the entropy of $V$ by a conditional entropy of $V$ given the \emph{noiseless} test results (still conditioned on the fixed $\xveq$).
    \item Given any fixed realization of noiseless test results, we have $V = B_1 + B_2$, where $B_1,B_2$ are independent binomial random variables with probability parameters $\rho,1-\rho$, and count parameters $n',n''$ satisfying $n'+n''=n_2$.  In particular, $\max\{n',n''\} \ge \frac{n_2}{2}$.
    \item Again using that conditioning reduces entropy, we deduce that the resulting entropy for $V$ (with the conditioning described above) is lower bounded by the higher of the two entropies associated with $B_1$ and $B_2$.  Specifically, after conditioning on $B_2$ the uncertainty in $V$ reduces to that in $B_1$ alone, and vice versa.
    \item The deisred claim then follows from the fact that Binomial$(N,q)$ has entropy $\frac{1}{2}\log(2\pi e N q(1-q)) + O(1/N)$ as $N \to \infty$ \cite{Ade10}, which simplifies to $\big(\frac{1}{2}+o(1)\big) \log N$ when $q$ is constant.  (Here we use $q \in \{\rho,1-\rho\}$ and $N = \max\{n',n''\} \ge \frac{n_2}{2}$, which we established above.)
\end{itemize}
As before, we seek to identify high-probability events that simplify the further analysis of \eqref{eq:H_next}. Fortunately, what we need has already been shown in the proof of Lemma \ref{nc-theta1}, as we now detail.

We first quantify $n_2$.  By setting $\delta = \Theta\big( \sqrt{\frac{\ell}{k}} \big)$ in (\ref{eq:n2tozero}) in the proof of Lemma \ref{nc-theta1}, the event \begin{align}
    \mathscr{A}_3=\big\{n_2=ne^{-\nu}(1+o(1))\big\}\label{eq:A3_for_n2}
\end{align} holds with probability at least $1-2\exp(-\frac{C_{*,1} \ell n}{k})$, where $C_{*,1}$ can be chosen arbitrarily large as a function of $(\nu,\rho)$.  Similarly, regarding $V$, we set $\delta = \Theta\big( \sqrt{\frac{\ell}{k}} \big)$ in \eqref{eq:Vdaluetozero} in the proof of Lemma \ref{nc-theta1}, that the event \begin{align}
    \label{eq:A4_for_V}
    \mathscr{A}_4 = \big\{V = \rho n e^{-\nu}(1+o(1))\big\}
\end{align} 
holds with probability at least $1-2\exp(-\frac{C_{*,2}\ell n}{k})$, where we can set $C_{*,2}$ to be arbitrarily large as a function of $(\nu,\rho)$. %Therefore, with  probability  at least $1-4\exp(-\frac{C_*\ell n}{k})$, where we can ensure that $C_*$ is greater than a given quantity depending on $(\nu,\rho)$, the following event holds: 
%\begin{equation}\label{eq:eventA1}
 %   \mathscr{A}_1= \big\{n_2=ne^{-\nu}(1+o(1)),~V=\rho ne^{-\nu}(1+o(1))\big\}.
%\end{equation} 

% to be larger than a given quantity depending on  $(\nu,\rho)$. 

We argue that it suffices to consider $\Xveq = \xveq$ such that $n_2= ne^{-\nu}(1+o(1))$, as given in the high-probability event $\mathscr{A}_1$. Recall that our goal is to characterize $I(\Xvdif;\Yv|\Xveq=\xveq)$ in (\ref{eq:def_con_eq}) and then average it over $\Xveq$ (as in (\ref{eq:fur_exp})) to characterize $I_{\ell}^n$.  Since $\Yv$ is binary-valued and has length $n$, we always have $I(\Xvdif;\Yv|\Xveq=\xveq) \le n \log 2 = O(n)$, which means that the total contribution to $I_{\ell}^n$ from the low-probability event $\mathscr{A}_1^c$ is bounded by at most the following:
  % In fact, from (\ref{eq:mi_equals}) and (\ref{eq:H_ub}), it holds deterministically that $I(\Xvdif;\Yv|\Xveq=\xveq)=O(n)$ for some implied constant depending on $\rho$, so the cases where the value of $n_2$ is not as in the event $\mathscr{A}_1$, which occurs with probability less than $2\exp(-\frac{C_{*,1}\nu^{-1}\ell n}{k})$, contributes to $I_\ell^n$  a term 
\begin{equation}
    O(n)\exp\Big(-\frac{C_{*,1}\nu^{-1}\ell n}{k}\Big) =\exp\big(O(\log k)-C_{*,1}\nu^{-1}\ell\cdot \Omega(\log k) \big) =o\Big(\frac{\ell n}{k}\Big), \label{eq:bad_n2_contri}
\end{equation}
where the first equality holds since $n=\Theta(k \log k)$ (see Section \ref{sec:n_scaling}), and the second equality follows by taking $C_{*,1}$ to be sufficiently large. Thus, the total contribution under $\mathscr{A}_1^c$ does not affect the desired bound (\ref{1976}).

For specific $\Xveq = \xveq$ (with $n_2= ne^{-\nu}(1+o(1))$  as in  the event $\mathscr{A}_1$), we  define $\mathcal{V}_{\xveq}$ as the  corresponding set of $V$ values belonging to the event $\mathscr{A}_2$ (all satisfying $V=\rho ne^{-\nu}(1+o(1))$). Then, we obtain from (\ref{eq:H_next}) that 
\begin{align}
    H(\Yv_2|\Xveq=\xveq)&\geq \sum_{v\in \mathcal{V}_{\xveq}}\PP(V=v) H(\Yv_2|\Xveq=\xveq,V=v)+\Big(\frac{1}{2}+o(1)\Big)\log(n_2) 
    \\\label{eq:star1}
    &\geq \sum_{v\in \mathcal{V}_{\xveq}}\PP(V=v)\log \binom{n_2}{v}+\Big(\frac{1}{2}+o(1)\Big)\log(n_2)
    \\\label{eq:star2}
    & \geq \sum_{v\in \mathcal{V}_{\xveq}}\PP(V=v) \Big[n_2 H_2\Big(\frac{v}{n_2}\Big)-\frac{1}{2}\log (n_2)-\log(2 )\Big]+\Big(\frac{1}{2}+o(1)\Big)\log(n_2)
    \\\label{eq:star3}
    &=\sum_{v\in \mathcal{V}_{\xveq}}\PP(V=v)n_2 H_2\Big(\frac{v}{n_2}\Big)+o\Big(\frac{\ell n}{k}\Big),
\end{align}
where (\ref{eq:star1}) holds because  $\Yv_2|\{V=v\}$ follows a uniform distribution over $\binom{n_2}{v}$ possibilities, in (\ref{eq:star2}) we use $\log\binom{N}{\beta N} \geq NH_2(\beta)-\frac{1}{2}\log N-\log(2)$ (see Lemma \ref{lem:coeff_bounds}), and in (\ref{eq:star3}) we perform some rearrangement to cancel out the terms of $\frac{1}{2}\log(n_2)$ and then write $o(\log n_2)=o(\log n)=o(\log k)=o(\frac{\ell n}{k})$, which can be seen analogously to (\ref{eq:bad_n2_contri}). Combining with (\ref{eq:mi_equals}), it follows that 
\begin{align}
    &I(\Xvdif;\Yv|\Xveq=\xveq) \geq \sum_{v\in \mathcal{V}_{\xveq}}\PP(V=v)\Big[n_2 H_2\Big(\frac{v}{n_2}\Big)-n_2H_2(\rho)\Big]\\
    &~~~~~~-\Big(\sum_{v\in [n_2]\setminus \mathcal{V}_{\xveq}}\PP(V=v)\Big) n_2H_2(\rho)+o\Big(\frac{\ell n}{k}\Big):=\mathcal{T}_1-\mathcal{T}_2+o\Big(\frac{\ell n}{k}\Big).
    \end{align}
By using $n=\Theta(k\log k)$ and choosing $C_{*,2}$ large enough, analogously to (\ref{eq:bad_n2_contri}), we obtain  \begin{align}\label{eq:anal1}
    |\mathcal{T}_2|\leq \PP(\mathscr{A}_2^c)n_2H_2(\rho)&\leq 2\exp\Big(-\frac{C_{*,2}\ell n}{k}\Big)nH_2(\rho)\\&=O\Big(\exp\big(O(\log k)-C_{*,2}\ell\cdot \Omega(\log k)\big)\Big)=o\Big(\frac{\ell n}{k}\Big),\label{eq:anal2}
\end{align} 
which only amounts to a lower-order term in the desired result (\ref{1976}).

Thus, by further averaging over $\Xveq$ as in (\ref{eq:fur_exp}), we only need to prove $\mathbbm{E}(\mathcal{T}_1)$   is lower bounded by the right-hand side of (\ref{1976}). Since we are considering $n_2=ne^{-\nu}(1+o(1))$ and $v=\rho ne^{-\nu}(1+o(1))$ (see the high-probability events in \eqref{eq:A1_for_n2}--\eqref{eq:A2_for_V}), a Taylor expansion with derivative $H'_2(\rho)=\log\frac{1-\rho}{\rho}$ gives 
   \begin{align}
       H_2\Big(\frac{v}{n_2}\Big)-H_2(\rho)= \log\Big(\frac{1-\rho}{\rho}\Big)\Big(\frac{v}{n_2}-\rho\Big)(1+o(1)).
   \end{align}
We substitute this into $\mathcal{T}_1$, then decompose $\sum_{v\in \mathcal{V}_{\xveq}}=\sum_{v=0}^{n_2}-\sum_{v\in [n_2]\setminus \mathcal{V}_{\xveq}}$
to write $\mathcal{T}_1$ as 
\begin{align}
    &\mathcal{T}_1 = \sum_{v=0}^{n_2}\PP(V=v)\Big[\log\Big(\frac{1-\rho}{\rho}\Big)\big(v-n_2\rho\big)(1+o(1))\Big] \nonumber \\&~~~~~~~~~- \sum_{v\in [n_2]\setminus \mathcal{V}_{\xveq}}\PP(V=v)\Big[\log\Big(\frac{1-\rho}{\rho}\Big)\big(v-n_2\rho\big)(1+o(1))\Big]:=\mathcal{T}_{11}-\mathcal{T}_{12}.
    \end{align}
Note that $\mathcal{T}_{12}$ is again insignificant, since analogously to (\ref{eq:bad_n2_contri}) and (\ref{eq:anal2}) we have
    \begin{align}
        |\mathcal{T}_{12}|\leq O(n)\PP(\mathscr{A}_2^c)\leq O(n)\exp\Big(-\frac{C_{*,2}\ell n}{k}\Big)=o\Big(\frac{\ell n}{k}\Big).
    \end{align} 
    For $\mathcal{T}_{11}$, we write
    \begin{align}
    \mathcal{T}_{11}&=\mathbbm{E}_{V|\Xveq=\xveq} \Big[\log\Big(\frac{1-\rho}{\rho}\Big)\big(V-n_2\rho\big)(1+o(1))\Big]
    \\
    &=\log\Big(\frac{1-\rho}{\rho}\Big)\big(\mathbbm{E}_{V|\Xveq=\xveq}(V)-n_2\rho\big)(1+o(1)).\label{eq:T11simply}
\end{align}
Thus, it suffices to show $\mathbbm{E}[\mathcal{T}_{11}]$ is lower bounded by the right-hand side of (\ref{1976}), and now we 
  calculate $\EE_{V|\Xveq=\xveq}[V]$. While $V$ represents the number of $1$s in the noisy outcomes $\Yv_2$, it is useful to also define $U$ as the number of  tests in $\Yv_2$ that would be positive in the absence of noise (i.e., the tests corresponding to $\Xv_{\seq}=0,\Xv_{\sdif}\neq 0$).  Then, conditioned on $U=u$, the effect of noise gives 
    \begin{align}
    \EE_{V|\Xveq=\xveq,U=u}[V]&=\EE\big[\mathrm{Bin}(u,1-\rho)+\mathrm{Bin}(n_2-u,\rho)\big]\\&=u(1-\rho)+(n_2-u)\rho\\&=(1-2\rho)u+\rho n_2.
    \end{align}
By averaging over $U$, we obtain  \begin{equation}\label{eq:everUfir}
    \EE_{V|\Xveq=\xveq}[V]=(1-2\rho)\mathbbm{E}_{U|\Xveq=\xveq}[U]+\rho n_2.
\end{equation} 
To characterize $\mathbbm{E}_{U|\Xveq=\xveq}[U]$, we note that $U$ represents the number of tests (among those corresponding to $\Yv_2$) that get at least one placement from overall $\ell\Delta = \frac{\ell\nu n}{k}$ placements of the $\ell$ items in $\sdif$.  %We construct $U'$ that counts the number of $\ell\Delta$ placements that are at $\Yv_2$, and we have $U\leq U'$ since $U'$ counts repetition; since $U'\sim \mathrm{Bin}(\ell\Delta,\frac{n_2}{n})$, we obtain  $\mathbbm{E}U\leq \mathbbm{E}U'=\frac{\ell\nu n_2}{k}$. Moreover, 
% and then consider conducting $\ell\Delta$ placements in a sequential manner. We define $U'$ with initial value $0$ as the following quantity: 
% \begin{itemize}[itemsep=0ex]
%  \item Let the initial value of $U'$ be $U'=0$. Let $\Yv_1',\Yv_2'$ be ``dynamic'' variants of $\Yv_1,\Yv_2$ with initial values $\Yv_1'=\Yv_1,\Yv_2'=\Yv_2$. We realize the randomness of $\Xvdif$ by conducting the $1$-th, $2$-th, ..., $i$-th,  ..., $(\ell\Delta)$-th placement in a sequential manner. 
%  \item Suppose that the $i$-th placement is in test $j$.  If test $j$ is {\it currently} in $\Yv_2'$,  then $U'=U'+1$ and we move test $j$ from $\Yv_2'$ to $\Yv_1'$; otherwise, if $j$ is {\it currently} in $\Yv_1'$, then we move a test (e.g., the first one) from $\Yv_2'$ to $\Yv_1'$.
% \end{itemize}
%   Note that if the $i$-th placement contributes an addition of $1$ to $U$ if it is in $\Yv_2$, and none of the previous $i-1$ placements is in $\Yv_2$. However, it may not contribute to $U'$ since it may not be in $\Yv_2'$. Thus, for any realization of $\Xvdif$ we have  
%  $U\geq U'$, which yields $\mathbbm{E}_{U|\Xveq=\xveq}[U]\geq\mathbbm{E}[U'|\Xveq=\xveq]$. The expectation of $U'$ is easier to calculate  because before the $i$-th placement, it always holds that $|\Yv_2'|=n_2-i+1$, thus we obtain 
To understand the resulting expectation, we envision performing the placements sequentially, and note that for a given placement to increment $U$ by one, it suffices that (i) the placement is into one of the $n_2$ tests that form $\Yv_2$, \emph{and} (ii) the placement is not into the same test as any of the previous placements.  Thus, we have
\begin{align}
 \mathbbm{E}[U|\Xveq=\xveq]&\ge\mathbbm{E}\Big[\sum_{i=1}^{\ell\Delta}\mathbbm{1}(\text{the $i$-th placement satisfies the preceding two conditions})\Big] 
 \\\label{eq:392i}&\ge\sum_{i=1}^{\ell\Delta}\frac{n_2-i+1}{n} 
 \\&= \frac{n_2\ell\Delta}{n}-\frac{(\ell\Delta-1)\ell\Delta}{2n}\\&\geq  \frac{\ell\nu n_2}{k}\Big(1-\frac{\ell\nu n}{2kn_2}\Big),\label{eq:lowerEU}
\end{align}
where (\ref{eq:392i}) holds since at most $i-1$ out of the $n_2$ tests are ruled out by the previous $i-1$ placements, and (\ref{eq:lowerEU}) uses $\Delta = \frac{\nu n}{k}$.
%
%where the last step is because we are considering $n_2=$
%Ignoring the effect of repeated placements into the same tests {[TODO: More rigorous]}, $U$ is ${\rm Binomial}( \ell \frac{\nu n}{k}, \frac{n_2}{n} )$, and we have we have $\EE[V|U=u] = u(1-\rho) + (n_2-u)\rho$.  So overall we have 
%\begin{align}
%	\EE[V] &= \EE[U](1-\rho) + (n_2-\EE[U])\rho \\
%		&= n_2\rho + \EE[U](1-2\rho) \\
%		&\sim n_2\rho + n_2 \frac{\ell \nu}{k}(1-2\rho) = 
%\end{align}
Substituting (\ref{eq:everUfir}) and (\ref{eq:lowerEU}) into \eqref{eq:T11simply} gives 
    \begin{align}
        \mathcal{T}_{11}&\geq (1+o(1))\log\Big(\frac{1-\rho}{\rho}\Big)(1-2\rho)\frac{\ell\nu n_2}{k}\Big(1-\frac{\ell\nu n}{2kn_2}\Big)\\
        &=\left(\frac{n\nu e^{-\nu}\ell}{k}(1-2\rho)\log\frac{1-\rho}{\rho}\right)\big(1+o(1)\big),  \label{eq:T11_bound}
    \end{align}
%and combining with \eqref{eq:mi_equals} gives us a lower bound on mutual information that matches \eqref{eq:I_ub}.
where (\ref{eq:T11_bound}) holds because we are considering $n_2=ne^{-\nu}(1+o(1))$ (as explained earlier in (\ref{eq:bad_n2_contri})).  

We have now established that the right-hand side of \eqref{eq:T11_bound} is a lower bound for $I_\ell^n$ (with other terms such as $\mathcal{T}_{12}$ being factored into the $o(1)$ part).  This matches the upper bound on $I_\ell^n$ derived above, and the proof is complete.

\subsection{Proof of Lemma \ref{lem:P_V} (Characterization of $\PP(V)$)} \label{sec:pf_P_V}

We specialize some of the developments from the proof of Lemma \ref{nc-theta1} to the case that $\seq = \varnothing$, $\sdif = s$, and $\ell = k$.  Hence, in the notation used therein, we have $\Yv_1=\varnothing$ and $\Yv_2=\Yv$, $M$ denotes the number of tests that would be positive if there were no noise, and $V$ represents for the number of positive tests in the noisy results (i.e., the number of 1s in $\Yv$).  Our goal is to show that $\PP(v)=e^{-o(n)}$ for all $v$ within a high-probability set (i.e., $V$ lies in that set with $1-o(1)$ probability).  Since we trivially have $\PP(v) \le 1$, it suffices to only lower bound $\PP(v)$.

By setting $\delta= n^{-1/4}=o(1)$ in (\ref{27m}), we have with $1-o(1)$ probability that
\begin{equation}\label{eq:M_val1}
    M = (1-e^{-\nu}+o(1))n.
\end{equation}
Moreover, by setting $\delta=n^{-1/4}=o(1)$ in (\ref{855}), we have with $1-o(1)$ probability that
\begin{equation}\label{eq:V_val1}
    V = (\rho\star e^{-\nu}+o(1))n,
\end{equation}
We will consider values $M=m$ and $V=v$ that satisfy both of these high-probability events.

Give $M=m$, we have that $V$ follows a sum of two independent binomial variables: $(V|M=m) = B_1 + B_2$, where $(B_1|M=m) \sim \mathrm{Bin}(m,1-\rho)$ and $(B_2|M=m) \sim \mathrm{Bin}(n-m,\rho)$.  Given fixed $m$ and $v$ values satisfying \eqref{eq:M_val1}--\eqref{eq:V_val1}, we can find $b_1,b_2$ with $b_1+b_2=v$ such that
\begin{equation}
    b_1 = n(1-\rho)(1-e^{-\nu}) (1+o(1)), \quad b_2 = n\rho e^{-\nu} (1+o(1)), \label{eq:b_choices}
\end{equation}
i.e., such that the two are asymptotically close to the mean values of $\mathrm{Bin}(m,1-\rho)$ and $\mathrm{Bin}(n-m,\rho)$, and also sum to $v$.  Letting $\mathcal{M}_{\rm typical}$ be a high-probability set of $m$ values that all satisfy \eqref{eq:M_val1}, we can then lower bound $\PP(V = v)$ as follows:
\begin{align}% \label{eq:con_M}
    \PP(V = v) & \geq \sum_{m \in \mathcal{M}_{\rm typical}} \PP(M=m) \PP(B_1=b_1|M=m) \PP(B_2=b_2|M=m). 
    \label{eq:divide2}
\end{align}
We now observe that $\PP(B_1 = b_1|M=m)$ and $\PP(B_2 = b_2|M = m)$ both scale as $e^{-o(n)}$; this directly follows from the anti-concentration of binomial random variables (Lemma \ref{binoconcen}) and $b_1,b_2$ being arbitrarily close to the respective means (see \eqref{eq:b_choices}).  Since $\mathcal{M}_{\rm typical}$ is a high-probability set of $m$ values, we deduce from \eqref{eq:divide2} that $\PP(V = v) \ge e^{-o(n)}$, which completes the proof.

% \section*{Acknowledgment}

% This work was supported by the Singapore National Research Foundation (NRF) under grant number A-0008064-00-00.

\bibliographystyle{myIEEEtran}
\bibliography{JS_References}

% Generated by IEEEtranS.bst, version: 1.14 (2015/08/26)
\begin{thebibliography}{10}
\providecommand{\url}[1]{#1}
\csname url@samestyle\endcsname
\providecommand{\newblock}{\relax}
\providecommand{\bibinfo}[2]{#2}
\providecommand{\BIBentrySTDinterwordspacing}{\spaceskip=0pt\relax}
\providecommand{\BIBentryALTinterwordstretchfactor}{4}
\providecommand{\BIBentryALTinterwordspacing}{\spaceskip=\fontdimen2\font plus
\BIBentryALTinterwordstretchfactor\fontdimen3\font minus
  \fontdimen4\font\relax}
\providecommand{\BIBforeignlanguage}[2]{{%
\expandafter\ifx\csname l@#1\endcsname\relax
\typeout{** WARNING: IEEEtranS.bst: No hyphenation pattern has been}%
\typeout{** loaded for the language `#1'. Using the pattern for}%
\typeout{** the default language instead.}%
\else
\language=\csname l@#1\endcsname
\fi
#2}}
\providecommand{\BIBdecl}{\relax}
\BIBdecl

\bibitem{Ade10}
J.~A. Adell, A.~Lekuona, and Y.~Yu, ``Sharp bounds on the entropy of the
  {P}oisson law and related quantities,'' \emph{IEEE Trans. Inf. Theory},
  vol.~56, no.~5, pp. 2299--2306, 2010.

\bibitem{Ald18}
M.~{Aldridge}, ``Individual testing is optimal for nonadaptive group testing in
  the linear regime,'' \emph{IEEE Trans. Inf. Theory}, vol.~65, no.~4, pp.
  2058--2061, April 2019.

\bibitem{Ald14a}
M.~Aldridge, L.~Baldassini, and O.~Johnson, ``Group testing algorithms: Bounds
  and simulations,'' \emph{IEEE Trans. Inf. Theory}, vol.~60, no.~6, pp.
  3671--3687, June 2014.

\bibitem{Ald15}
M.~Aldridge, ``The capacity of {B}ernoulli nonadaptive group testing,''
  \emph{IEEE Trans. Inf. Theory}, vol.~63, no.~11, pp. 7142--7148, 2017.

\bibitem{Ald14}
M.~Aldridge, L.~Baldassini, and K.~Gunderson, ``Almost separable matrices,''
  \emph{J. Comb. Opt.}, pp. 1--22, 2015.

\bibitem{Ald19}
M.~Aldridge, O.~Johnson, and J.~Scarlett, ``Group testing: An information
  theory perspective,'' \emph{Found. Trend. Comms. Inf. Theory}, vol.~15, no.
  3--4, pp. 196--392, 2019.

\bibitem{arenbaev1977asymptotic}
N.~Arenbaev, ``Asymptotic behavior of the multinomial distribution,''
  \emph{Theory of Probability \& Its Applications}, vol.~21, no.~4, pp.
  805--810, 1977.

\bibitem{arratia1989tutorial}
R.~Arratia and L.~Gordon, ``Tutorial on large deviations for the binomial
  distribution,'' \emph{Bulletin of Mathematical Biology}, vol.~51, no.~1, pp.
  125--131, 1989.

\bibitem{ash1990information}
R.~Ash, \emph{Information Theory}, ser. Dover books on advanced
  mathematics.\hskip 1em plus 0.5em minus 0.4em\relax Dover Publications, 1990.

\bibitem{Bal13}
L.~Baldassini, O.~Johnson, and M.~Aldridge, ``The capacity of adaptive group
  testing,'' in \emph{IEEE Int. Symp. Inf. Theory}, July 2013, pp. 2676--2680.

\bibitem{Bay22}
W.~H. Bay, J.~Scarlett, and E.~Price, ``Optimal non-adaptive probabilistic
  group testing in general sparsity regimes,'' \emph{Information and Inference:
  A Journal of the IMA}, vol.~11, no.~3, pp. 1037--1053, 2022.

\bibitem{Bou13}
S.~Boucheron, G.~Lugosi, and P.~Massart, \emph{Concentration Inequalities: A
  Nonasymptotic Theory of Independence}.\hskip 1em plus 0.5em minus 0.4em\relax
  Oxford University Press, 2013.

\bibitem{Cha11}
C.~L. Chan, P.~H. Che, S.~Jaggi, and V.~Saligrama, ``Non-adaptive probabilistic
  group testing with noisy measurements: Near-optimal bounds with efficient
  algorithms,'' in \emph{Allerton Conf. Comm., Ctrl., Comp.}, Sep. 2011, pp.
  1832--1839.

\bibitem{Che09}
M.~Cheraghchi, ``Noise-resilient group testing: Limitations and
  constructions,'' in \emph{Int. Symp. Found. Comp. Theory}, 2009, pp. 62--73.

\bibitem{coja2020information}
A.~Coja-Oghlan, O.~Gebhard, M.~Hahn-Klimroth, and P.~Loick,
  ``Information-theoretic and algorithmic thresholds for group testing,''
  \emph{IEEE Trans. Inf. Theory}, vol.~66, no.~12, pp. 7911--7928, 2020.

\bibitem{Coj19a}
A.~Coja-Oghlan, O.~Gebhard, M.~Hahn-Klimroth, and P.~Loick, ``Optimal group
  testing,'' in \emph{Conf. Learn. Theory (COLT)}, 2020.

\bibitem{Coj22}
A.~Coja-Oghlan, O.~Gebhard, M.~Hahn-Klimroth, A.~S. Wein, and I.~Zadik,
  ``Statistical and computational phase transitions in group testing,'' in
  \emph{Conf. Learn. Theory (COLT)}, 2022.

\bibitem{Coj24}
A.~Coja-Oghlan, M.~Hahn-Klimroth, L.~Hintze, D.~Kaaser, L.~Krieg, M.~Rolvien,
  and O.~Scheftelowitsch, ``Noisy group testing via spatial coupling,'' 2024,
  https://arxiv.org/abs/2402.02895.

\bibitem{Cov01}
T.~M. Cover and J.~A. Thomas, \emph{Elements of Information Theory}.\hskip 1em
  plus 0.5em minus 0.4em\relax John Wiley \& Sons, Inc., 2006.

\bibitem{Du93}
D.~Du and F.~K. Hwang, \emph{Combinatorial group testing and its
  applications}.\hskip 1em plus 0.5em minus 0.4em\relax World Scientific, 2000,
  vol.~12.

\bibitem{Dya83}
A.~G. D'yachkov and V.~V. Rykov, ``A survey of superimposed code theory,''
  \emph{Prob. Contr. Inf.}, vol.~12, no.~4, pp. 1--13, 1983.

\bibitem{Dya82}
A.~G. D'yachkov and V.~V. Rykov, ``Bounds on the length of disjunctive codes,''
  \emph{Problemy Peredachi Informatsii}, vol.~18, no.~3, pp. 7--13, 1982.

\bibitem{Fei54}
A.~Feinstein, ``A new basic theorem of information theory,'' \emph{IRE Prof.
  Group. on Inf. Theory}, vol.~4, no.~4, pp. 2--22, Sept. 1954.

\bibitem{Geb21}
O.~Gebhard, O.~Johnson, P.~Loick, and M.~Rolvien, ``Improved bounds for noisy
  group testing with constant tests per item,'' \emph{IEEE Trans. Inf. Theory},
  vol.~68, no.~4, pp. 2604--2621, 2021.

\bibitem{Han03}
T.~S. Han, \emph{Information-Spectrum Methods in Information Theory}.\hskip 1em
  plus 0.5em minus 0.4em\relax Springer, 2003.

\bibitem{hoeffding1963hyper}
W.~Hoeffding, ``Probability inequalities for sums of bounded random
  variables,'' \emph{J. Amer. Stat. Assoc.}, vol.~58, no. 301, pp. 13--30,
  1963.

\bibitem{Hwa72}
F.~Hwang, ``A method for detecting all defective members in a population by
  group testing,'' \emph{J. Amer. Stats. Assoc.}, vol.~67, no. 339, pp.
  605--608, 1972.

\bibitem{Ili21}
F.~Iliopoulos and I.~Zadik, ``Group testing and local search: Is there a
  computational-statistical gap?'' in \emph{Conf. Learn. Theory (COLT)}, 2021.

\bibitem{Joh16}
O.~Johnson, M.~Aldridge, and J.~Scarlett, ``Performance of group testing
  algorithms with near-constant tests-per-item,'' \emph{IEEE Trans. Inf.
  Theory}, vol.~65, no.~2, pp. 707--723, Feb. 2019.

\bibitem{Kau64}
W.~Kautz and R.~Singleton, ``Nonrandom binary superimposed codes,'' \emph{IEEE
  Trans. Inf. Theory}, vol.~10, no.~4, pp. 363--377, 1964.

\bibitem{lieb1968concavity}
E.~H. Lieb, ``Concavity properties and a generating function for {S}tirling
  numbers,'' \emph{J. Comb. Theory}, vol.~5, no.~2, pp. 203--206, 1968.

\bibitem{Mal80}
M.~B. Malyutov and P.~S. Mateev, ``Screening designs for non-symmetric response
  function,'' \emph{Mat. Zametki}, vol.~29, pp. 109--127, 1980.

\bibitem{Mal78}
M.~Malyutov, ``\BIBforeignlanguage{English}{The separating property of random
  matrices},'' \emph{\BIBforeignlanguage{English}{Math. Notes Acad. Sci.
  {USSR}}}, vol.~23, no.~1, pp. 84--91, 1978.

\bibitem{mcdiarmid1989method}
C.~McDiarmid, ``On the method of bounded differences,'' \emph{Surveys in
  Combinatorics}, vol. 141, no.~1, pp. 148--188, 1989.

\bibitem{Nil21}
J.~Niles-Weed and I.~Zadik, ``It was ``all''' for ``nothing''': Sharp phase
  transitions for noiseless discrete channels,'' in \emph{Conf. Learn. Theory
  (COLT)}, 2021.

\bibitem{rennie1969stirling}
B.~C. Rennie and A.~J. Dobson, ``On {S}tirling numbers of the second kind,''
  \emph{J. Comb. Theory}, vol.~7, no.~2, pp. 116--121, 1969.

\bibitem{Sca18}
J.~{Scarlett}, ``Noisy adaptive group testing: Bounds and algorithms,''
  \emph{IEEE Trans. Inf. Theory}, vol.~65, no.~6, pp. 3646--3661, June 2019.

\bibitem{Sca16b}
J.~Scarlett and V.~Cevher, ``Converse bounds for noisy group testing with
  arbitrary measurement matrices,'' in \emph{IEEE Int. Symp. Inf. Theory},
  Barcelona, 2016.

\bibitem{Sca15b}
J.~Scarlett and V.~Cevher, ``Phase transitions in group testing,'' in
  \emph{ACM-SIAM Symp. Disc. Alg. (SODA)}, 2016.

\bibitem{Sca17b}
J.~Scarlett and V.~Cevher, ``Near-optimal noisy group testing via separate
  decoding of items,'' \emph{IEEE J. Sel. Topics Sig. Proc.}, vol.~2, no.~4,
  pp. 625--638, 2018.

\bibitem{Sca20}
J.~Scarlett and O.~Johnson, ``Noisy non-adaptive group testing: A
  (near-)definite defectives approach,'' \emph{IEEE Trans. Inf. Theory},
  vol.~66, no.~6, pp. 3775--3797, 2020.

\bibitem{tan2022performance}
N.~Tan, W.~Tan, and J.~Scarlett, ``Performance bounds for group testing with
  doubly-regular designs,'' \emph{IEEE Trans. Inf. Theory}, vol.~69, no.~2, pp.
  1224--1243, 2022.

\bibitem{Teo22}
B.~Teo and J.~Scarlett, ``Noisy adaptive group testing via noisy binary
  search,'' \emph{IEEE Trans. Inf. Theory}, vol.~68, no.~5, pp. 3340--3353,
  2022.

\bibitem{Tru20}
L.~V. Truong, M.~Aldridge, and J.~Scarlett, ``On the all-or-nothing behavior of
  {B}ernoulli group testing,'' \emph{IEEE J. Sel. Areas in Inf. Theory},
  vol.~1, no.~3, pp. 669--680, 2020.

\bibitem{Yav20}
R.~C. Yavas, V.~Kostina, and M.~Effros, ``Random access channel coding in the
  finite blocklength regime,'' \emph{IEEE Trans. Inf. Theory}, vol.~67, no.~4,
  pp. 2115--2140, 2020.

\end{thebibliography}

\end{appendix}
\end{document}